%% file: SubtrajectorySimilarity.tex
\documentclass[sigconf, nonacm]{acmart}

\input{head.tex}

\usepackage[title]{appendix}
\newcommand{\revision}[1]{\color{black}{#1} \color{black}}
\newcommand\vldbdoi{10.14778/3611479.3611512}
\newcommand\vldbpages{XXX-XXX}
\newcommand\vldbvolume{16}
\newcommand\vldbissue{11}
\newcommand\vldbyear{2023}
\newcommand\vldbauthors{\authors}
\newcommand\vldbtitle{\shorttitle} 
\newcommand\vldbavailabilityurl{https://github.com/inabao/trajcSimilar}
\newcommand\vldbpagestyle{empty} 

\begin{document}
	\title{Efficient Non-Learning Similar Subtrajectory Search}
	
	\author{Jiabao Jin}
	\email{jiabaojin@stu.ecnu.edu.cn}
	\affiliation{%
		\institution{East China Normal University}
		\state{Shanghai}
		\country{China}
	}
	
	\author{Peng Cheng}
	\email{pcheng@sei.ecnu.edu.cn}
	\orcid{0000-0002-9797-6944}
	\affiliation{%
		\institution{East China Normal University}
		\state{Shanghai}
		\country{China}
	}
	
	\author{Lei Chen}
	\affiliation{%
		\institution{Hong Kong University of Science and Technology}
		\city{Hong Kong SAR}
		\country{China}
	}
	\email{leichen@cse.ust.hk}

	\author{Xuemin Lin}
	\affiliation{%
		\institution{Shanghai Jiaotong University}
		\city{Shanghai}
		\country{China}
	}
	\email{xuemin.lin@gmail.com}
	
	\author{Wenjie Zhang}
	\affiliation{%
		\institution{University of New South Wales}
		\city{Sydney}
		\country{Australia}
	}
	\email{wenjie.zhang@unsw.edu.au}

\input{abstract.tex}
	
	\maketitle
	
	\pagestyle{\vldbpagestyle}
	\begingroup\small\noindent\raggedright\textbf{PVLDB Reference Format:}\\
	\vldbauthors. \vldbtitle. PVLDB, \vldbvolume(\vldbissue): \vldbpages, \vldbyear.\\
	\href{https://doi.org/\vldbdoi}{doi:\vldbdoi}
	\endgroup
	\begingroup
	\renewcommand\thefootnote{}\footnote{\noindent
		This work is licensed under the Creative Commons BY-NC-ND 4.0 International License. Visit \url{https://creativecommons.org/licenses/by-nc-nd/4.0/} to view a copy of this license. For any use beyond those covered by this license, obtain permission by emailing \href{mailto:info@vldb.org}{info@vldb.org}. Copyright is held by the owner/author(s). Publication rights licensed to the VLDB Endowment. \\
		\raggedright Proceedings of the VLDB Endowment, Vol. \vldbvolume, No. \vldbissue\ %
		ISSN 2150-8097. \\
		\href{https://doi.org/\vldbdoi}{doi:\vldbdoi} \\
	}\addtocounter{footnote}{-1}\endgroup
	
	\ifdefempty{\vldbavailabilityurl}{}{
		\vspace{.3cm}
		\begingroup\small\noindent\raggedright\textbf{PVLDB Artifact Availability:}\\
		The source code, data, and/or other artifacts have been made available at \url{\vldbavailabilityurl}.
		\endgroup
	}
	
\input{introduction.tex}

\input{problemDefinition.tex}
\input{exsitingAlgorithm.tex}
\input{exactAlgorithm.tex}

\input{algorithmDetail.tex}
 \input{experimentalStudy.tex}

\input{relatedWork.tex}
\input{conclusion.tex}

 \input{ack.tex}

\balance

\bibliographystyle{ACM-Reference-Format}

\bibliography{add}

\newpage
\input{appendix.tex}

\end{document}

%% file: head.tex
\usepackage{booktabs} 
\usepackage{times} 
\usepackage{color} 
\usepackage{balance} 
\usepackage{url} 
\usepackage{tabularx}
\usepackage{siunitx} 
\usepackage{epstopdf} 
\usepackage{epsfig,tabularx,subfigure,multirow, graphicx}
\usepackage{enumitem}
\usepackage{caption}
\usepackage{pifont} 
\usepackage{makecell}
\usepackage{amsthm,amsmath}  
\usepackage{amsfonts}
\newcolumntype{L}[1]{>{\raggedright\arraybackslash}p{#1}}
\newcolumntype{C}[1]{>{\centering\arraybackslash}p{#1}}
\newcolumntype{R}[1]{>{\raggedleft\arraybackslash}p{#1}}

\usepackage[linesnumbered,ruled,vlined]{algorithm2e}
\SetKwRepeat{Do}{do}{while}
\SetCommentSty{mycommfont}

\long\def\comment#1{}

\newcommand{\nop}[1]{}

\newtheorem{theorem}{\bf Theorem}[section]
\newtheorem{lemma}{\bf Lemma}[section]

\newtheorem{example}{\bf Example}

\theoremstyle{remark}

\theoremstyle{definition}
\newtheorem{definition}{\bf Definition}

%% file: abstract.tex
\begin{abstract}
Similar subtrajectory search is a finer-grained operator that can better capture the similarities between one query trajectory and a portion of a data trajectory than the traditional similar trajectory search, which requires that the two checking trajectories are similar in their entirety. Many real applications (e.g., trajectory clustering and trajectory join) utilize similar subtrajectory search as a basic operator.
It is considered that the time complexity is $O(mn^2)$ for exact algorithms to solve the similar subtrajectory search problem under most trajectory distance functions in the existing studies, where $m$ is the length of the query trajectory and $n$ is the length of the  data trajectory.
In this paper, to the best of our knowledge, we are the first to propose an exact algorithm to solve the similar subtrajectory search problem in $O(mn)$ time for most of widely used trajectory distance functions (e.g., WED, DTW, ERP, EDR and Frechet distance). 
Through extensive experiments on three real datasets, we demonstrate the efficiency and effectiveness of our proposed algorithms.
\end{abstract}


%% file: introduction.tex
\section{Introduction}







The increasing popularity of mobile devices flourishes the generation of trajectory data, which is widely used in many fields (e.g., traffic flow prediction~\cite{Hui2021, Hui0CK21}, route planning~\cite{WangCZFCLW20}). 
With the focus of researchers on trajectory data, more and more methods are proposed for analyzing and processing trajectory data.

A significant problem in analyzing trajectory data is to query the most similar trajectory to a given trajectory among the vast amount of trajectories in the database \cite{ChenN04,ChenOO05,Xie14,ranu2015indexing,li2018deep,YiJF98}.
In real scenarios, it is hard to guarantee that the lengths of two trajectories are same or close to each other. Thus, similar subtrajectory search attracts much attention recently as a more practical method \cite{agarwal2018subtrajectory,lee2007trajectory,buchin2011detecting,tampakis2020distributed, WangLCJ19}, which uses a part of a long data trajectory as the basic unit to test its similarity to the short query trajectory. For example, as shown in Figure \ref{fig:subtrajectorysearch}, there are two trajectories: data trajectory $\tau_d$ and query trajectory $\tau_q$. They are not similar when the whole trajectories are considered, while $\tau_q$ is similar to a portion of $\tau_d$. 


Searching similar subtrajectories is usually a basic operator in real applications (e.g., subtrajectory join \cite{tampakis2020distributed} and subtrajectory clustering \cite{buchin2011detecting, agarwal2018subtrajectory}) and will be frequently invoked, thus its efficiency is very important. One application scenario of subtrajectory query is to analyze the performance of players by their trajectory data in a sport (e.g., soccer or basketball)~\cite{WangLCJ19}.  

Subtrajectory search is a highly related but different problem from trajectory search~\cite{WangLCL20, KoideXI20, SakuraiFY07, abs-2203-10364, faloutsos1994fast}. Compared with trajectory search, subtrajectory search has to not only consider the data trajectory itself but also determine whether there are subtrajectories of the data trajectory with a smaller distance from the query trajectory. The state-of-the-art study on similar subtrajectory search utilize reinforcement learning methods to accelerate the detecting speed and achieve the time complexity of $O(mn)$ \cite{WangLCL20}, where $m$ is the length of the query trajectory and $n$ is the length of the  data trajectory. However, the reinforcement learning based algorithms are approximation algorithms, which have no theoretical guarantee on the accuracy of the returned results.
In this paper, we find that \textit{the similar subtrajectory search problem can be \underline{solved exactly} with the time complexity of $O(mn)$ for most trajectory distance functions} (e.g., DTW, WED, ERP, EDR and FD. Details will be discussed in Section \ref{sec:detail_design_algorithm}), which had not been discovered to the best of our knowledge.

\begin{figure}
	\begin{small}
		\begin{center}
			\includegraphics[width=0.22\textwidth]{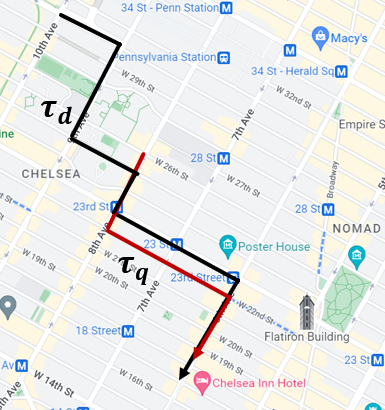}
		\end{center}
		\caption{Subtrajectory Search}\vspace{-1ex}
		\label{fig:subtrajectorysearch}
	\end{small}
\end{figure}

\noindent\textbf{Challenges.} 
For a data trajectory, the number of its subtrajectories is quadratic to its length. Let $n$ be the length of a data trajectory, there will be $\frac{n(n+1)}{2}$ its subtrajectories. Assuming that the length of the query trajectory is $m$ and the length of the data trajectory is $n$, the time complexity of directly searching for the most similar subtrajectory is $O(mn^3)$ (through traversal searching $\frac{n(n+1)}{2}$ subtrajectories of the data trajectory, and the time complexity of directly computing the similarity of two trajectories of length $x$ and $y$ by dynamic programming is $O(xy)$).  Although a recent work~\cite{WangLCL20} optimizes the time complexity of a single subtrajectory query problem from $O(mn^3)$ to $O(mn^2)$ through dynamic programming techniques, it is still unaffordable for most applications that need to find the optimal subtrajectory in a few seconds. 
In existing studies, only for \textit{dynamic time wrapping distance} (DTW) and \textit{Frechet distance} (FD), the similar subtrajectory search problem can be exactly solved in $O(nm)$ time complexity with particular algorithms \cite{SakuraiFY07, abs-2203-10364}. However, it cannot be extended to other trajectory distance functions.

In this paper, we propose the conversion-matching algorithm (CMA) to find the optimal subtrajectory by computing the minimum cost of converting the query trajectory into the data trajectory. With carefully tailored methods and transformation of the trajectory distance functions, we can incrementally fast track the optimal start position of the optimal subtrajectory in the data trajectory in $O(1)$ time. Given a query trajectory and a data trajectory, we search for the optimal subtrajectory with the time complexity of $O(nm)$. Meanwhile, the algorithm is applicable for the vast majority of distance functions. We use \textit{weighted edit distance} (WED)~\cite{KoideXI20} and \textit{dynamic time warping} (DTW)~\cite{YiJF98} as examples to analyze the design of the algorithm. We also discuss how to apply our methods to other most popular trajectory distance functions. Experiments show that the performance of our algorithm is better than other existing methods.


To summarize, we make the following contributions:
\begin{itemize}[leftmargin=*]
	\item We propose CMA with the time complexity of $O(nm)$ to find the most similar subtrajectory for a query trajectory under most order-insensitive trajectory distance functions in Section \ref{sec:cma_algorithm}. 
	\item We describe the design idea of the algorithm in detail and simplify the calculation of conversion cost, using WED and DTW as examples in Section \ref{sec:detail_design_algorithm}.
	\item We conduct experiments on three different real data sets to verify the superiority of our framework with the state-of-the-art similar subtrajectory query methods in Section \ref{sec:experimental}.
\end{itemize}


%% file: problemDefinition.tex
\section{Problem Definition}
\label{sec:problemDefinition}
\begin{figure}[t!]
	\centering\vspace{-2ex}
	\subfigure[][{\scriptsize $Cost_{del}$}]{
		\scalebox{0.2}[0.2]{\includegraphics{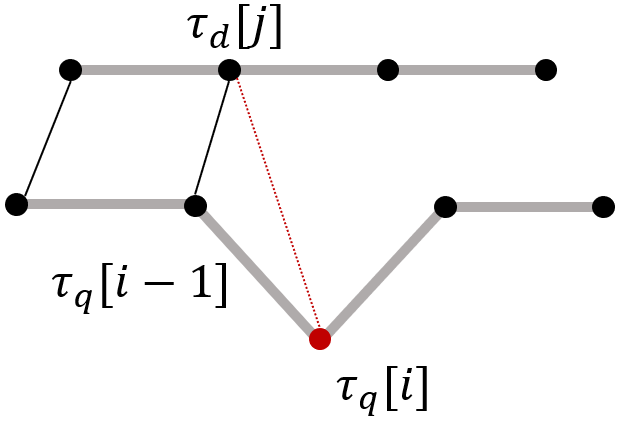}}
		\label{fig:del}}
	\subfigure[][{\scriptsize $Cost_{sub}$}]{
		\scalebox{0.2}[0.2]{\includegraphics{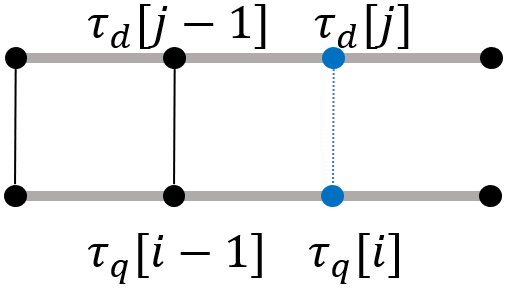}}
		\label{fig:sub}}
	\subfigure[][{\scriptsize $Cost_{ins}$}]{
		\scalebox{0.2}[0.2]{\includegraphics{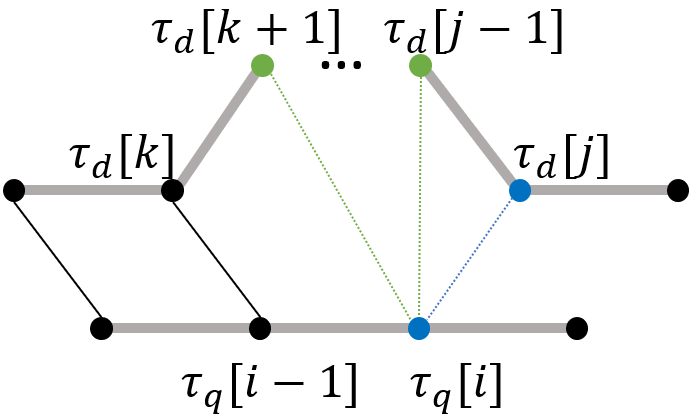}}
		\label{fig:ins}}
	\caption{\small Demonstration of the conversion cost}
	\label{fig:conversion_cost}\vspace{-1ex}
\end{figure}

\subsection{Basic Concepts}
There are two types of trajectories for the Similar Subtrajectory Search (SSS) problem: query and data trajectories. We expect to search for the most similar subtrajectory for a given query trajectory under a specific distance function among a large volume of data trajectories. We first provide the definitions of trajectories and subtrajectories as follows:
\begin{definition}(Trajectory)
    A trajectory $\tau$ with the length of $n$ consists of a series of points denoted as $\langle p_1,p_2,p_3,\dots,p_n\rangle$. 
\end{definition}

\revision{We denote the query trajectory as $\tau_q$ with the length of $m$ and the data trajectory as $\tau_d$ with the length of $n$. The points of trajectories can be specific physical locations, nodes on a road network or edges on a road network. In particular, we denote a trajectory without any point as $\tau_{\emptyset}$.}

\begin{definition}(Subtrajectory)
	Given a trajectory $\tau$ with the length of $n$, its subtrajectory is a portion of consecutive points, $\tau[i:j]=\langle p_i,p_{i+1},\dots, p_j\rangle$ ($1\leq i \leq j \leq n$).
\end{definition}

 In particular, we denote the $i^{th}$ point in $\tau$ as $\tau[i:i]$, abbreviated as $\tau[i]$. If $i>j$, we have $\tau[i:j]=\tau_{\emptyset}$. 
 
Usually, we have a set of data trajectories. In this paper, we focus on finding the optimal subtrajectory from a data trajectory among many data trajectories to match the query trajectory. We have also implemented two pruning methods, Grid-Based Prune (GBP) and Key Points Filter (KPF), to help filter the irrelevant trajectories quickly. Please refer to Appendix B and C for more details.

\subsection{Distance Function}

The distance function between trajectories represents the cost of converting the points of the query trajectory into the data trajectory plus the cost of inserting prefix subtrajectory and suffix subtrajectory.

\begin{definition}[Matching Sequence]
	For a query trajectory $\tau_q$ and a data trajectory $\tau_d$, we define its matching sequence as $\mathcal{A}_{\tau_q:\tau_d}=[a_1,a_2, a_3, \dots, a_{m}]$. \revision{If $\tau_q[i]$'s matching point is $\tau_d[j]$, we let $a_i=j$ to indicate the index of  $\tau_d[j]$ in the data trajectory.} For any $i\leq j$, we must have $a_i\leq a_j$.
\end{definition}

According to the definition, if a trajectory $\tau_d[s:t]$ is a subtrajectory of another $\tau_d[i:j]$, we have $\mathcal{A}_{\tau_q:\tau_d[s:t]}\subseteq \mathcal{A}_{\tau_q:\tau_d[i:j]}$. \revision{For example, the matching sequence for Figure \ref{fig:example1} is $[1,1,2,4,5,6,7,8,9]$, and the matching sequence for Figure \ref{fig:example2} is $[1,1,2,2,3,3,5,6,9]$. Note that, given a data trajectory $\tau_d$ and a query trajectory $\tau_q$, there may be many matching sequences. One matching sequence is valid as long as its matching index value is not decreasing (i.e., $a_i\leq a_j, \forall i\leq j$).
}

\begin{definition}[Point Matching-Conversion Cost]
	For a data trajectory $\tau_d$ and a query trajectory $\tau_q$, when $\tau_q[i]$ matches $\tau_d[j]$ (i.e., $a_i=j$), depending on the different matches of $\tau_q[i-1]$ as shown in Figure \ref{fig:conversion_cost}, we define the cost of converting $\tau_q[i]$ into  $\tau_d[j]$ in the following three cases:
	\begin{enumerate}[label=(\alph*),align=right,itemindent=2em,labelsep=2pt,labelwidth=1em,leftmargin=0pt,nosep]
		\item \textbf{$Cost_{del}$}. When $a_{i-1}=j$, we need to remove $\tau_q[i]$, and denote the conversion cost as {\small$Cost(\tau_q[i], \tau_d[a_i])=Cost_{del}(\tau_q[i], \tau_d[j])$}. 
		\item \textbf{$Cost_{sub}$}. When $a_{i-1}=j-1$, we replace $\tau_q[i]$ with $\tau_d[j]$, and denote the conversion cost as {\small$Cost(\tau_q[i], \tau_d[a_i])=Cost_{sub}(\tau_q[i], \tau_d[j])$}. 
		\item \textbf{$Cost_{ins}$}. When $a_{i-1}=k$ ($1\leq k <j-1$), we substitute $\tau_q[i]$ with $\tau_d[j]$ and insert $\tau_d[k+1:j-1]$. We denote the conversion cost as {
			\small$Cost(\tau_q[i], \tau_d[a_i])=Cost_{ins(k)}(\tau_q[i], \tau_d[j])$}.
	\end{enumerate}
\end{definition}


\begin{figure}[t!]
	\begin{small}\vspace{-2ex}
		\begin{center}
			\includegraphics[width=0.36\textwidth]{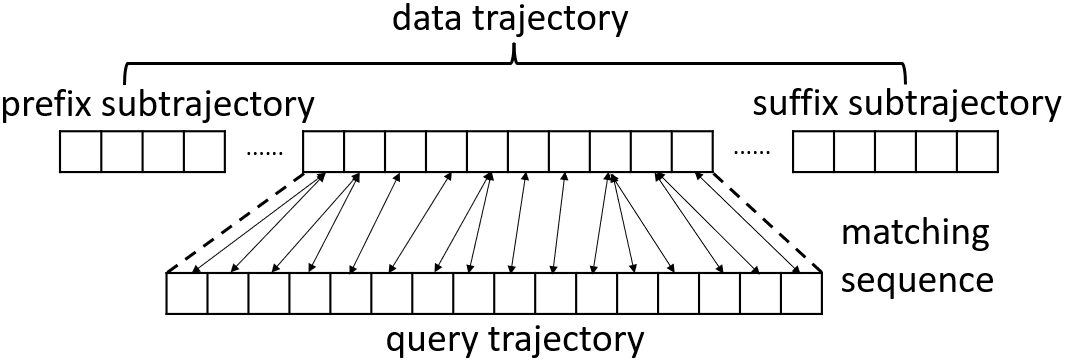}
		\end{center}
		\caption{Demonstration of Matching Process}\vspace{-1ex}
		\label{fig:matching_process}
	\end{small}
\end{figure}

We can calculate the cost of converting the query trajectory into a data trajectory for each matching sequence, which includes the cost of converting each point in the query trajectory into a matching point in the data trajectory and the cost of inserting \textit{prefix trajectory} and \textit{suffix trajectory} of the data trajectory as shown in Figure \ref{fig:matching_process}. We denote the cost of inserting a subtrajectory $\tau_d[x:y]$ as {\small$Insert(\tau_d[x:y])$}. Given a matching sequence $\mathcal{A}_{\tau_q:\tau_d}$,  its \textit{matching-conversion cost}  is {\small$\sum_{a_i \in \mathcal{A}_{\tau_q:\tau_d}}{Cost(\tau_q[i], \tau_d[a_i])} + Insert(\tau_d[1:a_1-1])+Insert(\tau_d[a_m+1:n])$}. \revision{We give an example to demonstrate the calculation of the matching-conversion cost in Example \ref{exa:dtw_example}.}

\begin{definition}[General Distance Function]
We denote the set of all possible matching sequences between the query trajectory $\tau_q$ and the data trajectory $\tau_d$ as $\mathbb{A}$. Then, we define the general distance $\Theta(\tau_q,\tau_d)$ between the query trajectory and the data trajectory as follows:\vspace{-2ex}
{\scriptsize
\begin{eqnarray}
\Theta(\tau_q,\tau_d)&=&\mathop{\min}\limits_{\mathcal{A}_{\tau_q:\tau_d}\in \mathbb{A}}\sum_{a_i \in \mathcal{A}_{\tau_q:\tau_d}}{Cost(\tau_q[i], \tau_d[a_i])} \notag\\
&+&Insert(\tau_d[1:a_1-1])+Insert(\tau_d[a_m+1:n])
\end{eqnarray}}\vspace{-2ex}
\end{definition}

There are many trajectory distance functions, such as DTW~\cite{YiJF98}, ERP~\cite{ChenN04}, EDR~\cite{ChenOO05}, and WED~\cite{KoideXI20}. In this paper, we use WED and DTW as examples to illustrate our definition. \revision{We discuss the generality of the general distance $\Theta(\tau_q,\tau_d)$ in the Appendix A.}

\noindent\textbf{WED.} WED is a general distance function that allows the user-defined cost functions and contains several important cost functions (e.g., EDR and ERP). WED defines the distance $wed(\tau_q, \tau_d)$ between $\tau_q$ and $\tau_d$ as the minimum cost of converting $\tau_q$ to $\tau_d$ by a finite number of insertion, deletion and substitution. Given two points $\tau_q[i]$ and $\tau_d[j]$, we denote the cost of insertion, deletion and substitution by $ins(\tau_d[j])$, $del(\tau_q[i])$ and $sub(\tau_q[i],\tau_d[j])$. Besides, the cost of deleting the subtrajectory $\tau_q[i:j]$ and inserting the subtrajectory $\tau_d[i:j]$ are denoted as $del(\tau_q[i:j])$ and $ins(\tau_d[i:j])$. We have $del(\tau_q[i:j])=\sum_{i\leq k\leq j}del(\tau_q[k])$ and $ins(\tau_d[i:j])=\sum_{i\leq k\leq j}ins(\tau_q[k])$.

\begin{figure}[t!]
	\centering
	\subfigure[][{\scriptsize Best Matching Sequence of WED}]{
		\scalebox{0.26}[0.26]{
			\includegraphics{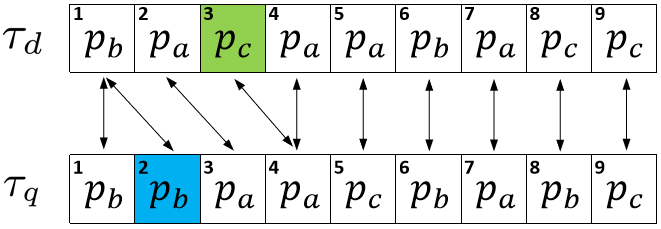}
			\label{fig:example1}
		}
	}
	\subfigure[][{\scriptsize Best Matching Sequence of DTW}]{
		\scalebox{0.21}[0.21]{
			\includegraphics{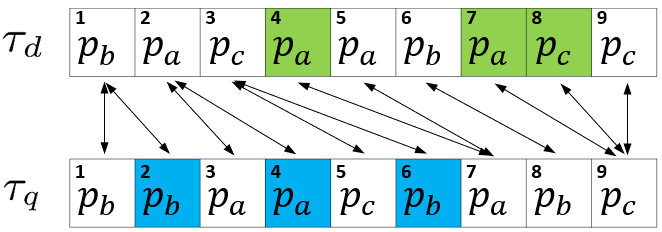}
			\label{fig:example2}
		}
	}
	\caption{\small Examples of Matching for Difference Distance Function}
	\label{fig:wed_example}\vspace{-1ex}
\end{figure}

\begin{example}
	Given two trajectories $\tau_q$ and $\tau_d$ as shown in Figure \ref{fig:wed_example}, we use WED to calculate the distance between them. The blue points indicate the deleted points in the conversion of $\tau_q$ into $\tau_d$, while the green points indicate the inserted points. We set the cost of $ins(\tau_d[j])$, $del(\tau_q[i])$ to $1$. In addition, we set the cost of $sub(\tau_q[i],\tau_d[j])$ to 1 if $\tau_q[i]\neq\tau_d[j]$; otherwise, it is set to 0. Figure \ref{fig:example1} shows an optimal matching sequence that converts $\tau_q$ into $\tau_d$ by deleting $\tau_q[2]$, inserting $\tau_d[3]$ and substituting $\tau_q[5]$ with $\tau_d[5]$ and $\tau_q[8]$ with $\tau_d[8]$. Since there are no redundant prefix and suffix subtrajectories, we have $Insert(\tau_d[1:a_1-1])+Insert(\tau_d[a_m+1:n])=0$. Therefore, the distance between $\tau_q$ and $\tau_d$ is 4($=del(\tau_q[2])+ins(\tau_d[3])+sub(\tau_q[5],\tau_d[5])+sub(\tau_q[8],\tau_d[8])$).
\end{example}

Moreover, we can compute the distance between $\tau_q$ and $\tau_d$ by a dynamic programming algorithm \cite{KoideTYXI18}. We have $wed(\tau_q[i:j], \tau_{\emptyset})=del(\tau_q[i:j])=\sum_{k=i}^{j}{del(\tau_q[k])}$ and $wed(\tau_{\emptyset}, \tau_d[i:j])=ins(\tau_d[i:j])=\sum_{k=i}^{j}{ins(\tau_d[k])}$. The $wed(\tau_q, \tau_d)$ is defined recursively:
{\scriptsize\begin{eqnarray}
	&&wed(\tau_q[1:i], \tau_d[1:j])\notag\\
	&=&\min\left\{
	\begin{array}{ll}
	wed(\tau_q[1:i-1], \tau_d[1:j-1]) + sub(\tau_q[i],\tau_d[j]) \\
	wed(\tau_q[1:i], \tau_d[1:j-1]) + ins(\tau_d[j])\\
	wed(\tau_q[1:i-1], \tau_d[1:j]) + del(\tau_q[i])\\
	\end{array}
	\right. 
	\label{eq:raw_wed}
\end{eqnarray}}

\noindent\textbf{DTW.} Another well-known distance function is DTW. 
Unlike WED, there is no deletion and insertion in DTW, instead, multiple points are allowed to be substituted for the same point in another trajectory. However, we try to interpret DTW from a different perspective to make it applicable to the algorithm proposed in this paper. We interpret the original substitution relation as a matching. We consider that only one point $\tau_q[i]$ is substituted for a point $\tau_d[j]$ in another trajectory, while other points that substitute $\tau_d[j]$ are deleted. We can define the insertion in the same way. The cost of deleting a point or inserting a point in the query trajectory is different, depending on which point it matches with, that is, {\small$del(\tau_q[i])=sub(\tau_q[i],\tau_d[j])$} and {\small$ins(\tau_d[j])=sub(\tau_q[i],\tau_d[j])$} if $\tau_q[i]$ matches $\tau_d[j]$.

Here, we give an example about the optimal matching when using DTW as distance function.

\begin{example}
	\label{exa:dtw_example}
	 The optimal matching when converting the $\tau_q$ into $\tau_d$ is shown in Figure \ref{fig:example2}. We set the distance between two points to 1 in case the two points are not equal; otherwise, it is set to 0. \revision{When the matching sequence is $[1,1,2,4,5,6,7,8,9]$, the conversion cost corresponding to each point is $[0,0,1,1,1,0,0,1,0]$. When $\tau_q[4]$ is converted into $\tau_d[4]$, $\tau_d[3]$ needs to be inserted. Therefore, although $\tau_q[4]=\tau_d[4]$, the required cost is still 1. Therefore, the conversion cost of the matching sequence corresponding to Figure \ref{fig:example1} is 4. When the matching sequence is $[1,1,2,2,3,3,5,6,9]$, the conversion cost corresponding to each point is $[0,0,0,0,0,1,0,0,1]$. When converting $\tau_q[9]$ into $\tau_d[9]$, the cost of inserting $\tau_d[7]$ is 1. Therefore, the conversion cost of this matching sequence corresponding to Figure \ref{fig:example2} is 2.}
\end{example}

Finally, we also give the dynamic process for the calculation of $dtw(\tau_q, \tau_d)$ as follows:
{\scriptsize\begin{eqnarray}
	dtw(\tau_q[1:i], \tau_d[1:j])
	=\left\{
	\begin{array}{ll}
		\sum_{k=1}^{j}{sub(\tau_q[1],\tau_d[k])}, &i=1\\
		\sum_{k=1}^{i}{sub(\tau_q[k],\tau_d[1])}, &j=1\\
	\min\{dtw(\tau_q[1:i-1], \tau_d[1:j]),\\
	dtw(\tau_q[1:i], \tau_d[1:j-1]),\\
	dtw(\tau_q[1:i-1], \tau_d[1:j-1])\}\\
	 + sub(\tau_q[i],\tau_d[j]), &else\\
	\end{array}
	\right. 
	\label{eq:raw_dtw}
\end{eqnarray}}

The algorithm proposed in this paper requires that the distance function to satisfy a specific property: the distance of points between different trajectories is independent of the position of the point in the trajectory. We will explain this in Section~\ref{sec:other}.

\subsection{Problem Definition}
\begin{definition}[Similar Subtrajectory Search Problem, SSS]
	Given a query trajectory $\tau_q$ and a data trajectory $\tau_d$, we expect a closest subtrajectory $\tau_d[i^{*}:j^{*}]$ under a specific distance function $\Theta$ (e.g., WED or DTW) from the data trajectory for the query trajectory $\tau_q$: {\small$$\left(i^{*},j^{*}\right)=\mathop{\arg\min}\limits_{1\leq i \leq  j \leq n}{\Theta(\tau_q,\tau_d[i:j])}$$} 
\end{definition}

A more general query is to find the $top$-$K$ similar subtrajectories from massive data trajectories for the query trajectory. Instead, we can follow such a search process in previous work~\cite{WangLCJ19} that maintains the most similar $K$ trajectories and updates it when a more similar subtrajectory appears. Therefore, we mainly consider querying the most similar subtrajectory from the data trajectory. \revision{Details of top-K SSS  can be found in the Appendix E.}

Suppose the length of a data trajectory is $n$, which means that a data trajectory has $\frac{n(n+1)}{2}$ subtrajectories. Assuming that the length of a query trajectory is $m$ and the complexity of computing the distance between the data trajectory and the query trajectory is {\small$O(mn)$}~\cite{KoideTYXI18}. Therefore, given a query trajectory $\tau_q$ and a data trajectory $\tau_d$, the time complexity of searching a subtrajectory of $\tau_d$ with the smallest distance from $\tau_q$ in $\tau_d$ is {\small$O(mn^{3})$}. Table \ref{table0} summarizes the commonly used notations in this paper.

\begin{table}
	\centering \vspace{-2ex}
	{\small
		\caption{\small Symbols and Descriptions.} \label{table0}
		\begin{tabular}{l|l}
			{\bf Symbol} & {\bf \qquad \qquad \qquad\qquad\qquad Description} \\ \hline \hline
			$\tau_d$   & a data trajectory\\
			$\tau_q$   & a query trajectory\\
			$\tau[i:j]$   & a subtrajectory of $\tau$ from $i^{th}$ point to $j^{th}$ point \\
			$\tau[i]$   & the $i^{th}$ point in trajectory $\tau$\\
			$\mathcal{A}_{\tau_q:\tau_d}$ & a matching sequence between $\tau_q$ and $\tau_d$\\
			$a_i$ & the matches of $\tau_q[i]$ and $\tau_d[a_i]$\\
			$\Theta$ & the distance function\\
			\hline
			\hline
		\end{tabular}
	}\vspace{-1ex}
\end{table}

%% file: exsitingAlgorithm.tex
\section{Review of Existing Solutions}
\label{sec:existing_solutions}

We  briefly review the existing exact algorithms for the SSS problem.
\subsection{ExactS}
The vast majority of distance functions~\cite{YiJF98, ChenN04, ChenOO05, KoideXI20, VlachosGK02, 0007BCXLQ18, Yuan019, Xie14, AltG95} are defined via recursive processes. Using dynamic programming, we can compute the trajectory distance of a query trajectory and a subtrajectory of the data trajectory in $O(mn)$, where $m$ and $n$ are the lengths of the query trajectory and the data trajectory, respectively. For a query trajectory $\tau_q$ and a data trajectory $\tau_d$, let $M_{x,y}$ denote the trajectory distance between $\tau_q[1:x]$ and $\tau_d[i:i+y]$ for a given iteration $i$. ExactS~\cite{WangLCL20} can compute $M_{x,y}$ from $M_{x,y-1}$ using a dynamic programming technique. Thus, line 4 in Algorithm \ref{algo:efficiencyE} can be solved in $O(mn)$. There are $n$ iterations, thus the overall time complexity of ExactS is $O(mn^2)$. ExactS can be applied to most of the distance functions.

\subsection{Spring}

Spring algorithm~\cite{SakuraiFY07} is based on the existing dynamic programming computational procedure of DTW and changes the initialization procedure of $dtw(\tau_q[1:i], \tau_d[1:j])$ in the Equation~\ref{eq:raw_dtw} when $i=1$. Spring considers $\tau_d[1:j-1]$ to be redundant when $i=1$; therefore, they modify the equation for $dtw(\tau_q[1:i], \tau_d[1:j])$ when $i=1$ to be as follows:\vspace{-2ex}
\begin{eqnarray}
    dtw(\tau_q[1:i], \tau_d[1:j])=sub(\tau_q[1],\tau_d[j]) \label{eq:spring}
\end{eqnarray}

In addition, the authors demonstrate that a modification of the Equation \ref{eq:raw_dtw} enables it to compute the optimal subtrajectory. However, this trick can only be applied to the DTW function and cannot be extended to other distance functions (e.g., ERP, EDR, and WED).

\subsection{Greedy Backtracking (GB)}
GB~\cite{abs-2203-10364} investigates finding the optimal subtrajectory in a data trajectory when using FD as the distance function. It constructs a matrix $X$, where $X_{i,j}$ denotes the Euclidean distance between $\tau_q[i]$ and $\tau_d[j]$. Assuming that $X_{1,1}$ denotes the upper left corner of the matrix, GB finds a path from the top to the right or down until it reaches the bottom. The path's cost is the maximum value in the matrix through which the path passes, and GB finds the optimal subtrajectory by finding the path with the lowest cost. Since FD only considers substitution operations between the trajectory point and trajectory point, it can construct the matrix $S$. However, the cost of converting $\tau_q[i]$ into $\tau_d[j]$ in other distance functions that consider insertion and deletion operations (e.g., ERP, EDR, and WED) is uncertain; thus, the matrix $S$ cannot be constructed and GB is not suitable.

\begin{algorithm}[t!]
	{\small
		\DontPrintSemicolon
		\KwIn{\small a query trajectory $\tau_q$, a data trajectory $\tau_d$}
		\KwOut{\small a subtrajectory $\tau_d[i^{*},j^{*}]$}
		$i^{*}\gets 0, j^{*}\gets 0$\;
		$score\gets \infty$\;
		\ForAll{$1\leq i \leq n$}{
			$M\gets DP(\tau_q,\tau_d[i:n])$\;
			$y^{*} \gets \mathop{\arg\min}\limits_{ 1\leq y\leq n-i+1}{M_{m,y}}$\;
			\If{$M_{i,y^{*}}<score$}{
				$score\gets M_{i,y^{*}}$\;
				$i^{*}\gets i$\;
				$j^{*}\gets y^{*} + i - 1$\;
			}
		}
		\Return{$\tau_d[i^{*},j^{*}]$}
		\caption{\small $ExactS(\tau_q,\tau_d)$~\cite{WangLCL20}}
		\label{algo:efficiencyE}}
\end{algorithm}

%% file: exactAlgorithm.tex
\section{Conversion-Matching Algorithm }
\label{sec:cma_algorithm}

This section presents an efficient and exact subtrajectory search algorithm, namely Conversion-Matching Algorithm (CMA). Firstly, we transform the problem of finding the optimal subtrajectory into a problem of finding the optimal matching sequence. Meanwhile, we introduce \textit{the cost of optimal partial matching} $C_{i,j}$ to find the optimal matching sequence. Here, $C_{i,j}$ denotes the minimal cost of converting $\tau_q[1:i]$ into a subtrajectory of $\tau_d[1:j]$ when $\tau_q[i]$ matches $\tau_d[j]$ (i.e., $a_i=j$). Note that, converting $\tau_q[1:i]$ into $\tau_d[1:j]$ does not means that $\tau_q[1]$ must match $\tau_d[1]$. Finally, we propose the Conversion-Matching Algorithm (CMA) to calculate $C_{i,j}$ and find the optimal subtrajectory.

\subsection{Optimal Matching Sequence }
\label{sec:optima}

Although previous work \cite{WangLCJ19} has optimized the time complexity of this problem to {\small$O(mn^2)$}, it still makes the computational cost increase dramatically when the length of the data trajectory is large. This paper reduces this time complexity to {\small$O(mn)$} by a different dynamic programming algorithm based on a newly introduced concept.

Different from existing algorithms, the algorithm is not based on the existing dynamic programming method for calculating the distance. Instead, the basic idea of the algorithm is to calculate the minimum cost of converting the points in the query trajectory to the data trajectory by three operations: insertion, deletion and substitution. Each point in the query trajectory is converted to its matching point in the data trajectory at a specific cost in the conversion process. We can prove that the optimal subtrajectory do not contain redundant prefix trajectories and suffix trajectories by following theorem.
\begin{theorem}
    \label{the:abundent}
    Assume that $\tau_d[i:j]$ is the optimal subtrajectory in $\tau_d$, i.e., $\Theta(\tau_q,\tau_d[i:j])=\mathop{\min}\limits_{1\leq s \leq t \leq n}\Theta(\tau_q,\tau_d[s:t])$. Then, we have
    {\small$$\Theta(\tau_q,\tau_d[i:j])=\sum_{a_k \in \mathcal{A}_{\tau_q:\tau_d[i:j]}^{o}}{Cost(\tau_q[k], \tau_d[a_k])}$$}
    where {\small$\mathcal{A}_{\tau_q:\tau_d[i:j]}^{o}$} is the optimal match sequence of $\tau_q$ and $\tau_d[i:j]$.
\end{theorem}
\begin{proof}
We will prove that $a_1=i$ and $a_m=j$ in $\mathcal{A}_{\tau_q:\tau_d[i:j]}$ when $\tau_d[i:j]$ is the optimal subtrajectory.

Suppose $a_1=s$ and $a_m=t$ ($s\geq i, t \leq j$), then $\mathcal{A}_{\tau_q:\tau_d[i:j]}^{o}$ is also a matching sequence of $\tau_d[s:t]$. Therefore, we have
{\scriptsize\begin{eqnarray}
    \Theta(\tau_q,\tau_d[s:t])&\leq&\sum_{a_k \in \mathcal{A}_{\tau_q:\tau_d[i:j]}^{o}\backslash \{a_1, a_m\}}{Cost(\tau_q[k], \tau_d[a_k])} \notag \\
    &&+Cost(\tau_q[1], \tau_d[s])+Cost(\tau_q[m], \tau_d[t]) \notag \\
    &=& \sum_{a_k \in \mathcal{A}_{\tau_q:\tau_d[i:j]}^{o}}{Cost(\tau_q[k], \tau_d[a_k])} \notag \\
    &\leq& \Theta(\tau_q,\tau_d[i:j])\notag
\end{eqnarray}}

If $s>i$ or $t<j$, then $\tau_d[i:j]$ is not the optimal subtrajectory, which contradicts what is known. Therefore, we have $a_1=i$ and $a_m=j$. Further, we can obtain
{\scriptsize\begin{eqnarray}
    \Theta(\tau_q,\tau_d[i:j])&=&\mathop{\min}\limits_{\mathcal{A}_{\tau_q:\tau_d[i:j]}\in \mathbb{A}}\sum_{a_k \in \mathcal{A}_{\tau_q:\tau_d[i:j]}}{Cost(\tau_q[k], \tau_d[a_k])}  \notag\\
&&+Insert(\tau_d[i:a_1-1])+Insert(\tau_d[a_m+1:j])\notag\\
&=& \sum_{a_k \in \mathcal{A}_{\tau_q:\tau_d[i:j]}^{o}}{Cost(\tau_q[k], \tau_d[a_k])}  \notag\\
&&+Insert(\tau_d[i:a_1-1])+Insert(\tau_d[a_m+1:j])\notag\\
&=& \sum_{a_k \in \mathcal{A}_{\tau_q:\tau_d[i:j]}^{o}}{Cost(\tau_q[k], \tau_d[a_k])}\notag
\end{eqnarray}}\vspace{-4ex}
\end{proof}

Theorem \ref{the:abundent} proves that we do not need to consider redundant prefix subtrajectory and suffix subtrajectory in the optimal subtrajectory problem but only need to consider minimizing the conversion cost of all matching points. Then, we will prove that the optimal matching sequence of optimal subtrajectory is also optimal among all matching sequences between query trajectory and data trajectory.
\begin{theorem}
    \label{the:opt}
    Assume that $\mathcal{A}_{\tau_q:\tau_d[i:j]}^{o}$ is the optimal matching sequence for the optimal subtrajectory $\tau_d[i:j]$, then it is also the optimal among all matching sequences, i.e.$\mathcal{A}_{\tau_q:\tau_d[i:j]}^{o}=\mathop{\arg\min}\limits_{\mathcal{A}_{\tau_q:\tau_d}\in \mathbb{A}}$ $\sum_{a_i \in \mathcal{A}_{\tau_q:\tau_d}}{Cost(\tau_q[i], \tau_d[a_i])}$.
\end{theorem}
\begin{proof}
    We assume that the matching sequence $\mathcal{A}_{\tau_q:\tau_d}^{p}$ is better than $\mathcal{A}_{\tau_q:\tau_d[i:j]}^{o}$. If $\mathcal{A}_{\tau_q:\tau_d}^{p}$ is the matching sequence of sub-trajectories $\tau_d[i:j]$ and query trajectories, then it contradicts the condition that $\mathcal{A}_{\tau_q:\tau_d[i:j]}^{o}$ is the optimal matching sequence for $\tau_d[i:j]$; conversely, if $\mathcal{A}_{\tau_q:\tau_d}^{p}$ is a matching sequence of the subtrajectory $\tau_d[a_1:a_m]$, then $\tau_d[i:j]$ is not an optimal subtrajectory, which contradicts what is known.
\end{proof}
By using the theorems \ref{the:abundent} and \ref{the:opt}, we can conclude
{\scriptsize\begin{eqnarray}
    \mathop{\min}\limits_{1\leq i \leq j \leq n}{\Theta( \tau_q,\tau_d[i:j])}
    =\mathop{\min}\limits_{\mathcal{A}_{\tau_q:\tau_d}\in \mathbb{A}}\sum_{a_i \in \mathcal{A}_{\tau_q:\tau_d}}{Cost(\tau_q[i], \tau_d[a_i])} \label{eq:reduce}
\end{eqnarray}}

According to the Equation \ref{eq:reduce}, we reduce the problem of finding the optimal subtrajectory to finding the optimal matching sequence. We split all match sequences $\mathbb{A}$ of the query trajectory $\tau_q$ with the data trajectory $\tau_d$ according to the matches at different points. We use $\mathbb{A}[a_i=j]$ to denote the set of all matching sequences in $\mathbb{A}$ that satisfy the condition that $\tau_q[i]$ matches $\tau_d[j]$. 

\begin{definition}[Optimal Partial Matching-Conversion Cost]
We denote by $C_{i,j}$ the minimum value of the cost of converting $\tau_q[1:i]$ into a subtrajectory of $\tau_d[1:j]$ when $\tau_q[i]$ matches $\tau_d[j]$, that is, 
{\scriptsize$$C_{i,j}=\mathop{\min}\limits_{\mathcal{A}_{\tau_q:\tau_d}\in \mathbb{A}[a_i=j]}\sum_{k=1,a_k \in \mathcal{A}_{\tau_q:\tau_d}}^{k=i}{Cost(\tau_q[k], \tau_d[a_k])}$$}
\end{definition}

Once we have calculated $C_{i,j}$, the distance between the query trajectory and the optimal subtrajectory is the minimum conversion cost when $\tau_q[m]$ matches a point in the data trajectory because 
{\scriptsize\begin{eqnarray}
    &&\mathop{\min}\limits_{1\leq i \leq j \leq n}{\Theta( \tau_q,\tau_d[i:j])}
    =\mathop{\min}\limits_{\mathcal{A}_{\tau_q:\tau_d}\in \mathbb{A}}\sum_{a_i \in \mathcal{A}_{\tau_q:\tau_d}}{Cost(\tau_q[i], \tau_d[a_i])} \notag \\
    &=&\mathop{\min}\limits_{1\leq j\leq n}\mathop{\min}\limits_{\mathcal{A}_{\tau_q:\tau_d}\in \mathbb{A}[a_m=j]}\sum_{k=1,a_k \in \mathcal{A}_{\tau_q:\tau_d}}^{k=m}{Cost(\tau_q[k], \tau_d[a_k])} \\\notag
    &=&\mathop{\min}\limits_{1\leq j\leq n}C_{m, j}\notag
\end{eqnarray}}

Therefore, we will mainly discuss how to compute $C_{i,j}$ in the subsequent section; meanwhile, we will use DTW and WED as examples to illustrate our algorithm in detail in Section \ref{sec:detail_design_algorithm}.

\subsection{Universal Calculation of $C_{i,j}$}
This section will discuss how to calculate $C_{i,j}$ and find the subtrajectory with the shortest distance to the query trajectory from the data trajectory for a given query trajectory $\tau_q$ and a data trajectory $\tau_d$.

\noindent\textbf{Calculate $C_{i,j}$.} We will discuss the computation process of $C_{i,j}$ in three cases:
\begin{enumerate}[label=\arabic*),align=right,itemindent=1em,labelsep=2pt,labelwidth=1em,leftmargin=0pt,nosep]
	\item $i=1$.
	When $i=1$, we substitute $\tau_q[1]$ with $\tau_d[j]$, which is $C_{i,j}=Cost_{sub}(\tau_q[i], \tau_d[j])$.
	\item $j=1$.
	There are two possible ways of converting $\tau_q[i]$ when $j=1$: deleting $\tau_q[i]$, which means that $\tau_q[i-1]$ matches $\tau_d[1]$ so that we have {\small$C_{i,j}=C_{i-1,j}+Cost_{del}(\tau_q[i], \tau_d[j])$}; the other way is to substitute $\tau_q[i]$ with $\tau_d[1]$, which means that $\tau_q[1:i-1]$ will be deleted, resulting in {\scriptsize$C_{i,j}=Cost_{sub}(\tau_q[i], \tau_d[j])+\sum_{k=1}^{i-1}Cost_{del}(\tau_q[k], \tau_d[j])$}. Therefore, we have {\scriptsize$C_{i,j}=\min \{C_{i-1,j}+Cost_{del}(\tau_q[i], \tau_d[j]), Cost_{sub}(\tau_q[i], \tau_d[j])+\sum_{k=1}^{i-1}Cost_{del}(\tau_q[k], \tau_d[j])\}$}.
	\item $1<i\leq m,1<j\leq n$.
	Considering that the point $\tau_q[i]$ matches $\tau_d[j]$, there are three different conversion possibilities for $\tau_q[i]$ and $C_{i,j}=\min\{delCost_{i,j},subCost_{i,j},insCost_{i,j}\}$:
	\begin{enumerate}
		\item $delCost_{i,j}$: deleting $\tau_q[i]$. When $\tau_q[i]$ is deleted, by the definition of matching, $\tau_q[i-1]$ and $\tau_d[j]$ are matched; thus, we have $delCost_{i,j}=C_{i-1,j}+Cost_{del}(\tau_q[i], \tau_d[j])$.
		\item $subCost_{i,j}$: substituting $\tau_q[i]$ with $\tau_d[j]$. In this case, $\tau_q[i-1]$ matches $\tau_d[j-1]$; thus, we have $subCost_{i,j}=C_{i-1,j-1}+Cost_{sub}(\tau_q[i], \tau_d[j])$.
		\item $insCost_{i,j}$: substituting $\tau_q[i]$ and inserting $\tau_d[k+1:j-1]$. In this situation, $\tau_q[i-1]$ may match $\tau_d[k]$ ($1\leq k<j-1$). We  insert $\tau_d[k+1:j-1]$ and have $C_{i,j}=C_{i-1,k}+Cost_{ins(k)}(\tau_q[i], \tau_d[j])$. Considering all possible values of $k$, $insCost_{i,j}=\min_{1\leq k<j-1} C_{i-1,k}+Cost_{ins(k)}(\tau_q[i], \tau_d[j])$.
	\end{enumerate}
\end{enumerate}

\noindent In case 3.(b), our substitution of $\tau_q[i]$ for $\tau_d[j]$ can be seen as inserting an empty trajectory along with the substitution. Therefore, we will discuss 3.(b) and 3.(c) together in the subsequent sections.

To record the start position the optimal subtrajectory, we use \underline{$s_{i,j}$ to denote the index of $\tau_q[1]$'s matched point in $\tau_d$}, when $\tau_q[i]$ \\
matches $\tau_d[j]$, i.e., the start position of the subtrajectory. Based on the computation process of $C_{i,j}$, we are able to determine which point $\tau_q[i-1]$ matches when $\tau_q[i]$ matches $\tau_d[j]$. Suppose $\tau_q[i-1]$ matches $\tau_d[k]$ ($1\leq k\leq j$), then we have $s_{i,j}=s_{i-1,k}$. Finally, we propose CMA to solve the SSS problem as shown in Algorithm \ref{algo:efficiencyA}.

\noindent\textbf{Complexity.} Since $Cost_{sub}(\tau_q[i], \tau_d[j])$ and $Cost_{del}(\tau_q[i], \tau_d[j])$ involve only the substitution and deletion of one trajectory point, their time complexity is $O(1)$; therefore, when $i=1$, the time complexity of $C_{i,j}$ is $O(1)$. We can calculate $\sum_{k=1}^{i-1}Cost_{del}(\tau_q[k], \tau_d[j])$ when $j=1$ in advance for any $i$ by preprocessing, and thus we can compute $C_{i,j}$ within the time complexity of $O(1)$. In other cases, we need to calculate $\min\{delCost_{i,j},subCost_{i,j},insCost_{i,j}\}$. Therefore, the time complexity of CMA is {\small$O(mn)$}.
\revision{We will discuss how to compute $C_{i,j}$ in $O(1)$ time complexity for a specific distance function (e.g., DTW and WED) in Section~\ref{sec:detail_design_algorithm}.}

\revision{\noindent\textbf{Discussion.} Spring and GB can also achieve $O(mn)$ time complexity for SSS problem under DTW and FD distance function, respectively. CMA is different from them. CMA and Spring are all DP methods. The main difference between CMA and Spring is that their recursive formulas are different: Spring's recursive formula is well-designed for DTW function, and just can support DTW; CMA's recursive formula is a more general one, which can support abstract insertion, substitution and deletion operations thus can be applied under most commonly used trajectory distance functions (e.g., DTW, WED and FD). Secondly, Spring will output all the subtrajectories with distances less than a given threshold to the query trajectory, thus some additional computations are involved in the process of Spring, which do not exist in CMA.
Greedy Backtracking is in general a breadth-first search method with memorizing techniques.}

\begin{algorithm}[t]
	{\small
	\DontPrintSemicolon
	\KwIn{\small a query trajectory $\tau_q$, a data trajectory $\tau_d$}
	\KwOut{\small a subtrajectory $\tau_d[i^{*},j^{*}]$}
    \ForAll{$1\leq i \leq m$}{
        \ForAll{$1\leq j \leq n$}{
            \uIf{$i=1$}{
                {\scriptsize$C_{i,j}\gets Cost_{sub}(\tau_q[i], \tau_d[j])$\;}
                $s_{i,j}\gets j$\;
            }\ElseIf{$j=1$}{
                {\scriptsize$C_{i,j}\gets \min \{C_{i-1,j}+Cost_{del}(\tau_q[i], \tau_d[j]),$ $ Cost_{sub}(\tau_q[i], \tau_d[j])+\sum_{k=1}^{i-1}Cost_{del}(\tau_q[k], \tau_d[j])\}$\;}
                $s_{i,j}\gets 1$\;
            }\Else{
                {\scriptsize$C_{i,j}\gets \min\{delCost_{i,j},subCost_{i,j},insCost_{i,j}\}$\;}
                update $s_{i,j}$ according to the matches of $\tau_q[i-1]$\;
            }
        }
    }
    $j^{*}\gets \mathop{\arg\min}\limits_{1\leq j\leq n}C_{m,j}$\;
    $i^{*}\gets s_{m,j^{*}}$\;
    \Return{$\tau_d[i^{*},j^{*}]$}
	\caption{\small $CMA(\tau_q,\tau_d)$}
	\label{algo:efficiencyA}}
\end{algorithm}

%% file: algorithmDetail.tex
\section{Fast Calculating $C_{i,j}$ on Specific $\Theta$}
\label{sec:detail_design_algorithm}

In this section, we discuss how to calculate the conversion cost and $C_{i,j}$ for each point of the query trajectory with  WED and DTW. Meanwhile, we will explain how $insCost_{i,j}$ can be computed in $O(1)$ time for the two distance functions WED and DTW.
\subsection{Minimum Cost $C_{i,j}$ of WED}
By introducing the concept of matching, we can convert the distance between trajectories into the cost required to convert points in $\tau_q$ into points in $\tau_d$. Let's discuss the cost of converting each point $\tau_q[i]$ to its matched point $\tau_d[j]$ in $\tau_d$.\\
\textbf{Conversion Cost.} There are three cases:
\begin{enumerate}[label=(\alph*),align=right,itemindent=2em,labelsep=2pt,labelwidth=1em,leftmargin=0pt,nosep]
\item $\tau_q[i-1]$ matches $\tau_d[j]$. We delete $\tau_q[i-1]$ so that $Cost_{del}(\tau_q[i],$ $ \tau_d[j])=del(\tau_q[i])$. 
\item $\tau_q[i-1]$ matches $\tau_d[j-1]$. We substitute $\tau_q[i]$ with $\tau_d[j]$, i.e., $Cost_{sub}(\tau_q[i], \tau_d[j])=sub(\tau_q[i],\tau_d[j])$. 
\item $\tau_q[i-1]$ matches $\tau_d[k]$, where $1\leq k<j-1$. The cost of converting $\tau_q[i]$ to $\tau_d[j]$ is the summation of $sub(\tau_q[i], \tau_d[j])$ and the cost of inserting the trajectory $\tau_d[k+1:j-1]$. Therefore, we have $Cost_{ins(k)}(\tau_q[i], \tau_d[j])=ins(\tau_d[k+1:j-1])+sub(\tau_q[i],\tau_d[j])$.
\end{enumerate}

\begin{example}
	Consider the example in Figure \ref{fig:wed_example}, where $\tau_q$ is converted into $\tau_d$. Since $\tau_q[1]$ has no predecessor node, $\tau_q[1]$ is only substituted for $\tau_d[1]$ with the cost of $sub(\tau_q[1],\tau_d[1])$. $\tau_q[2]$ matches $\tau_d[1]$, but since $\tau_q[1]$ matches $\tau_d[1]$, $\tau_q[2]$ has to be deleted, with the cost of $del(\tau_q[2])$. $\tau_q[4]$ matches $\tau_d[4]$ and $\tau_q[3]$ matches $\tau_d[2]$, thus $\tau_q[4]$ is converted to $\tau_d[4]$ with the cost of $sub(\tau_q[4],\tau_d[4])+ins(\tau_d[3])$.
\end{example}

\begin{figure*}[t!]
	\centering
	\subfigure[][{\scriptsize $C_{i,j}$}]{
		\scalebox{0.47}[0.47]{
			\includegraphics{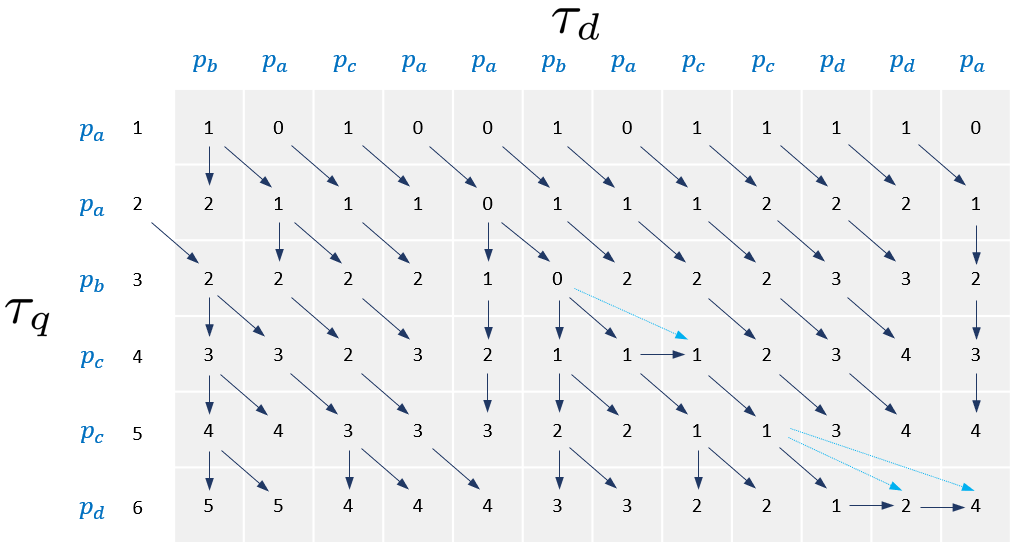}
			\label{fig:cost_matrix}
		}
	}
	\subfigure[][{\scriptsize $S_{i,j}$}]{
		\scalebox{0.47}[0.47]{
			\includegraphics{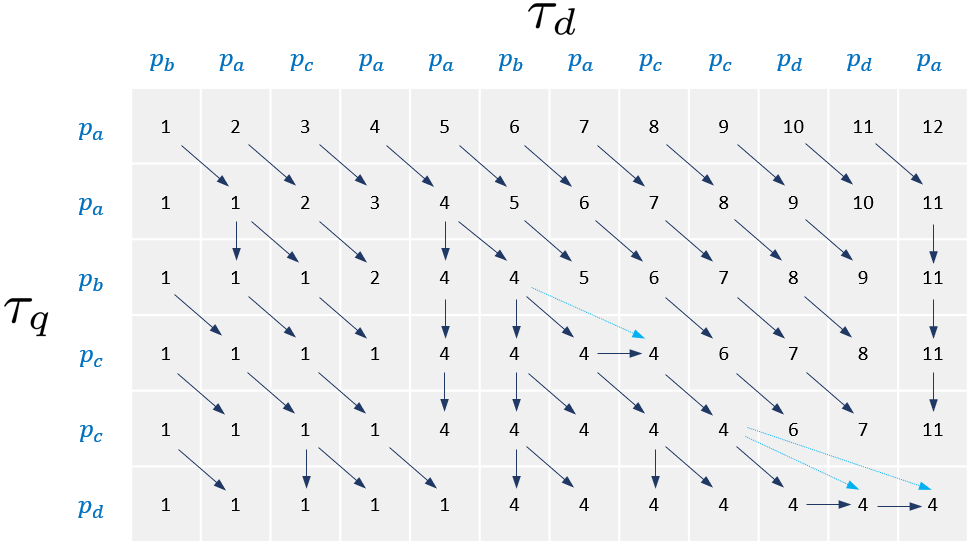}
			\label{fig:start_matrix}
		}
	}
	\caption{\small Demonstration of Calculating $C_{i,j}$ and $S_{i,j}$ when using WED as distance function}
	\label{fig:demonstration}\vspace{-1ex}
\end{figure*}

\noindent \textbf{Calculate $C_{i,j}$.} After obtaining the conversion cost of the points using WED as a distance function, we can calculate $C_{i,j}$. We will discuss the relational equation for $C_{i,j}$ in three cases:
\begin{enumerate}[label=\arabic*),align=right,itemindent=1em,labelsep=2pt,labelwidth=1em,leftmargin=0pt,nosep]
	\item $i=1$. {\small$C_{i,j}=sub(\tau_q[1],\tau_d[j])$}.
	\item $j=1$. {\small$C_{i,j}=\min\{C_{i-1,j}+del(\tau_q[i]), sub(\tau_q[i],\tau_d[1])+del(\tau_q[1:i-1])\}$}. 
	\item $1<i\leq m,1<j\leq n$.
	Considering that the point $\tau_q[i]$ matches $\tau_d[j]$, $\tau_q[i]$ may need to be deleted or substituted. Thus, the update of $C_{i,j}$ depends mainly on whether $\tau_q[i]$ is deleted or replaced:
	\begin{enumerate}
		\item $delCost_{i,j}$. When $\tau_q[i]$ is deleted, by the definition of matching, $\tau_q[i-1]$ and $\tau_d[j]$ are matched and we have $delCost_{i,j}=C_{i-1,j}+del(\tau_q[i])$.
		\item $subCost_{i,j}$ and $insCost_{i,j}$. $\tau_q[i-1]$ may match $\tau_d[k]$ ($1\leq k<j-1$) while substituting $\tau_q[i]$. In this case, we insert $\tau_d[k+1:j-1]$ and have $C_{i,j}=C_{i-1,k}+ins(\tau_d[k+1:j-1])+sub(\tau_q[i],\tau_d[j])$. Considering all possible values of $k$, {\scriptsize$insCost_{i,j}=\min_{1\leq k<j-1} C_{i-1,k}+ins(\tau_d[k+1:j-1])+sub(\tau_q[i],\tau_d[j])$}. Another situation is that $\tau_q[i-1]$ matches $\tau_d[j-1]$ and we have $subCost_{i,j}=C_{i-1,j-1}+sub\left(\tau_q[i],\tau_d[j]\right)$. Combining these two situation, we have {\scriptsize$C_{i,j}=\min\{insCost_{i,j}, subCost_{i,j}\}=\min_{1\leq k <j} C_{i-1,k}+ins(\tau_d[k+1:j-1])+sub(\tau_q[i],\tau_d[j])$}. Then, the calculation of $C_{i,j}$ can be simplified by follows:
		{\scriptsize\begin{align*}
				C_{i,j}&=\min_{1\leq k<j} C_{i-1,k}+ins\left(\tau_d[k+1:j-1]\right)\\
				&+sub\left(\tau_q[i],\tau_d[j]\right) \label{equ:e1}\\
				&=\min \{\min_{1\leq k<j-1}C_{i-1,k}+ins\left(\tau_d[k+1:j-1]\right) \\
				&+sub\left(\tau_q[i],\tau_d[j]\right),C_{i-1,j-1}+sub\left(\tau_q[i],\tau_d[j]\right)\}\\
				&=\min\{C_{i,j-1}+ins\left(\tau_d[j-1]\right)-sub\left(\tau_{q}[i],\tau_d[j-1]\right)\\
				&+sub\left(\tau_{q}[i],\tau_d[j]\right),C_{i-1,j-1}+sub\left(\tau_q[i],\tau_d[j]\right)\}
		\end{align*}}
		
	\end{enumerate}
\end{enumerate}

By the above analysis, we can obtain the expression for the calculation of $C_{i,j}$ while using WED as distance function
{\scriptsize\begin{equation}
		C_{i,j}=\left\{
		\begin{array}{ll}
			sub(\tau_q[i],\tau_d[j]), & i=1 \\
			\min\{C_{i-1,j}+del(\tau_q[i]), sub(\tau_q[i],\tau_d[1])\\
			+del(\tau_q[1:i-1])\}, & j=1,i\neq1 \\
			\min\{C_{i-1,j}+del(\tau_q[i]),\\
			C_{i,j-1}+ins\left(\tau_d[j-1]\right)-\\
			sub\left(\tau_{q}[i],\tau_d[j-1]\right)+sub\left(\tau_{q}[i],\tau_d[j]\right),\\
			C_{i-1,j-1}+sub\left(\tau_q[i],\tau_d[j]\right)\},& otherwise
		\end{array}
		\right.
		\label{eq:wed_cal}
\end{equation}}

Finally, we illustrate algorithm with an example as follows:
\begin{example}
	Given two trajectories as shown in Figure \ref{fig:demonstration}, we need to find the subtrajectory from $\tau_d$ that is closest to $\tau_q$. The insertion, deletion, and substitution costs are the same as the settings in Example \ref{fig:example1}. At the beginning, we will initialize $C_{1,j}$ (i.e., $C_{1,j}=sub(\tau_q[1],\tau_d[j])$). Then initialize $C_{i,1}$ based on whether $\tau_q[i]$ matches $\tau_d[1]$. Figure \ref{fig:cost_matrix} shows that $\tau_d[1]=b$, thus only $\tau_q[3]$ is substituted with it. For $\tau_q[4]$, when it matches $\tau_d[8]$, we need to determine which point is optimal for $\tau_q[3]$ to match with. From Figure \ref{fig:cost_matrix}, we can see that the cost of $\tau_q[3]$ when matching with $\tau_d[6]$ is 0, being the minimum, which means we need to insert $\tau_d[7]$. Therefore, considering $\tau_q[4]=\tau_d[8]$, we can compute the result of {\small$C_{4,8}$ from $C_{3,6}$, i.e., $C_{4,8}=C_{3,6}+ins(\tau_d[7])+sub(\tau_q[4],\tau_d[8])=0+1+0=1$}. In the actual implementation of the algorithm \ref{algo:efficiencyA}, we will compute $C_{4,8}$ by $C_{4,7}$, i.e. {\small$C_{4,8}=C_{4,7}+ins(\tau_d[7])-sub(\tau_q[4],\tau_d[7])+sub(\tau_q[4 ],\tau_d[8])=1+1-1+0=1$}.
	
	On the other hand, the algorithm updates $S_{i,j}$ as it executes. For example, when $\tau_q[4]$ matches $\tau_d[8]$, $\tau_q[3]$ is matched with $\tau_d[6]$ and we have $S_{4,8}=S_{3,6}$ as shown in Figure \ref{fig:start_matrix}.
\end{example}

\subsection{Minimum Cost $C_{i,j}$ of DTW}
Unlike WED, the cost required to delete a point or insert a point in DTW is different. Firstly, we analyze the cost of converting each point $\tau_q[i]$ in the query trajectory to its matched  point $\tau_d[j]$ in the data trajectory.\\
\textbf{Conversion Cost.} There are three cases:
\begin{enumerate}[label=(\alph*),align=right,itemindent=2em,labelsep=2pt,labelwidth=1em,leftmargin=0pt,nosep]
\item $\tau_q[i-1]$ matches $\tau_d[j]$. The cost of deleting $\tau_q[i]$ is equal to the cost of substituting $\tau_q[i]$ with $\tau_d[j]$, thus {\small$Cost_{del}(\tau_q[i], \tau_d[j])=sub(\tau_q[i],\tau_d[j])$}. 
\item $\tau_q[i-1]$ matches $\tau_d[j-1]$. We substitute $\tau_q[i]$ with $\tau_d[j]$, i.e., $Cost_{sub}(\tau_q[i], \tau_d[j])=sub(\tau_q[i],\tau_d[j])$. 
\item $\tau_q[i-1]$ matches $\tau_d[k]$, where $1\leq k<j-1$. The cost of converting $\tau_q[i]$ to $\tau_d[j]$ is the summation of $sub(\tau_q[i], \tau_d[j])$ and the cost of inserting the trajectory $\tau_d[k+1:j-1]$. The cost for inserting subtrajectories $\tau_d[k+1:j-1]$ depends on the points matched by $\tau_d[k+1:j-1]$ at $\tau_q$. Suppose $\tau_q[i-1]$ matches $\tau_d[k]$ and $\tau_q[i]$ will be matched into $\tau_d[j]$. Thus, the cost to insert $\tau_d[k+1:j-1]$ is {\scriptsize$\mathop{\min}\limits_{k\leq t\leq j-1}\sum_{p=k+1}^{t}{sub(\tau_q[i-1],\tau_d[p])}+\sum_{p=t+1}^{j-1}{sub(\tau_q[i],\tau_d[p])}$}. Thus, we have {\scriptsize$Cost_{ins(k)}(\tau_q[i], \tau_d[j])=\mathop{\min}\limits_{k\leq t\leq j-1}$ $\sum_{p=k+1}^{t}{sub(\tau_q[i-1],\tau_d[p])}+ \sum_{p=t+1}^{j-1}{sub(\tau_q[i],\tau_d[p])}+sub(\tau_q[i],\tau_d[j])$}.
\end{enumerate}
\begin{example}
	Let's take Figure \ref{fig:example2} as an example, the cost of converting $\tau_q[1]$ to $\tau_d[1]$ when $i=1$ and $j=1$ is {\small$sub(\tau_q[1],\tau_d[1])=sub(b,b)$}. When $i=2$ and $j=1$, $\tau_q[2]$ can only be converted to $\tau_d[1]$, thus the cost of the conversion is {\small$sub(\tau_q[2],\tau_d[1])=sub(b,b)$}. By the time $\tau_q[4]$ matches $\tau_q[2]$, we need to delete $\tau_q[4]$ requiring a cost of {\small$del(\tau_q[4])=sub(\tau_q[4],\tau_q[2])$} because $\tau_q[3]$ matches $\tau_q[2]$. For i=9 and j=9, since $\tau_q[8]$ matches $\tau_d[7]$, the cost of converting $\tau_q[9]$ to $\tau_d[9]$ consists of not only the cost of the substitution $sub(\tau_q[9],\tau_d[9])$, but also the cost of inserting $\tau_d[8]$, that is, $\mathop{\min}\limits_{7\leq t\leq 8}\sum_{p=8}^{t}{sub(\tau_q[8],\tau_d[p])}$ $+\sum_{p=t+1}^{8}{sub(\tau_q[9],\tau_d[p])}$. It is equal to $\min\{sub(\tau_q[8],\tau_d[8]),sub(\tau_q[9],\tau_d[8])\}$. It can be understood in another way that when $\tau_q[8]$ matches $\tau_d[7]$ and $\tau_q[9]$ matches $\tau_d[9]$, inserting $\tau_d[8]$ is equivalent to replacing $\tau_d[8]$ with $\tau_q[8]$ or $\tau_q[9]$.
\end{example}

\noindent\textbf{Calculate $C_{i,j}$.} After analyzing the conversion cost, similarly, we discuss the computation of $C_{i,j}$ in three cases:
\begin{enumerate}[label=(\alph*),align=right,itemindent=2em,labelsep=2pt,labelwidth=1em,leftmargin=0pt,nosep]
\item $i=1$.
When $i=1$, $\tau_q[1]$ can only be substituted with $\tau_d[j]$ as the same as WED and we have {\small$C_{i,j}=sub(\tau_q[1],\tau[j])$}.
\item $j=1$.
Considering that the cost of deleting $\tau_q[i]$ and substituting $\tau_q[i]$ is the same when $j=1$, we have
{\scriptsize\begin{eqnarray}
		C_{i,j}&=&\min \{Cost_{sub}(\tau_q[i], \tau_d[j])+\sum_{k=1}^{i-1}Cost_{del}(\tau_q[k], \tau_d[j]), \notag\\
		&&C_{i-1,j}+Cost_{del}(\tau_q[i], \tau_d[j])\} \notag\\
		&=&\min \{\sum_{k=1}^{i}sub(\tau_q[k], \tau_d[j]), C_{i-1,j}+sub(\tau_q[i],\tau_d[1])\}\notag\\
		&=&C_{i-1,j}+sub(\tau_q[i],\tau_d[1]) \notag
\end{eqnarray}}
\item $1<i\leq m,1<j\leq n$.
If we delete $\tau_q[i]$, we have $delCost_{i,j}=C_{i-1,j}+sub(\tau_q[i], \tau_q[j])$. Another conversion is substitution. $\tau_q[i-1]$ may be matched with any $\tau_d[k]$ ($1\leq k<j$), and $C_{i,j}$ denotes the smallest of all possible values. Thus, we have
{\scriptsize
	\begin{eqnarray}
		C_{i,j}&=& \min\{insCost_{i,j}, subCost_{i,j}\}\notag\\
		&=&\min \{\min_{1\leq k<j-1} C_{i-1,k}+Cost_{ins(k)}(\tau_q[i], \tau_d[j]),\notag \\ 
		&&C_{i-1,j-1}+Cost_{sub}(\tau_q[i],\tau_d[j])\}\notag \\
		&=&\mathop{\min}\limits_{1\leq k<j}C_{i-1,k}+\mathop{\min}\limits_{k\leq t< j-1 }\sum_{p=k+1}^{t}{sub(\tau_q[i-1],\tau_d[p] )}\notag \\
		&+&\sum_{p=t+1}^{j-1}{sub(\tau_q[i],\tau_d[p])}+sub(\tau_q[i],\tau_d[j])\notag \\
		&=&\mathop{\min}\limits_{1\leq k<j}\mathop{\min}\limits_{k\leq t< j-1 }C_{i-1,k}+\sum_{p=k+1}^{t}{sub(\tau_q[i-1],\tau_d[p] )}\notag \\
		&+&\sum_{p=t+1}^{j-1}{sub(\tau_q[i],\tau_d[p])}+sub(\tau_q[i],\tau_d[j]) \notag
\end{eqnarray}}
\end{enumerate}

The time complexity of computing $C_{i,j}$ ($1<i<m, 1<j<n$) directly from the above expression is very high, and therefore we are required to simplify the computation of $C_{i,j}$ by Theorem \ref{the:simplify}.
\begin{theorem}
	\label{the:simplify}
	When $i\geq 2,j\geq 2$, we have {\small$C_{i,j}=\mathop{\min}\limits_{1\leq k < j}C_{i-1,k}+\sum_{t=k+1}^{j}sub(\tau_q[i], \tau_d[t])$}.
\end{theorem}
\begin{proof}
	We use mathematical induction to prove this theorem. To simplify the proof, we denote $\sum_{t=k+1}^{j}sub(\tau_q[i], \tau_d[t])$ as $sub(i,k+1:j)$. For $\forall j\geq 2$ when $i=2$, we have
	{\scriptsize \begin{eqnarray}
			C_{2,j}&=&\mathop{\min}\limits_{1\leq k < j}\mathop{\min}\limits_{k\leq t < j}C_{1,k}+sub(1,k+1:t)+sub(2,t+1:j)\notag \\
			&=&\mathop{\min}\limits_{1\leq t < j}\mathop{\min}\limits_{1\leq k \leq t}{sub(1,k:t)+sub(2,t+1:j)}\notag \\
			&=&\mathop{\min}\limits_{1\leq t < j}{sub(\tau_q[1],\tau_d[t])+sub(2,t+1:j)}\notag \\
			&=&\mathop{\min}\limits_{1\leq t < j}{C_{1,t}+\sum_{k=t+1}^{j}{sub(\tau_q[2],\tau_d[k])}}\notag\\
			&=&\mathop{\min}\limits_{1\leq k < j}{C_{1,k}+\sum_{t=k+1}^{j}{sub(\tau_q[2],\tau_d[t])}}\notag
	\end{eqnarray}}
	Suppose $i=h-1$, and we have {\small$C_{h-1,j}=\mathop{\min}\limits_{1\leq k < j}C_{h-2,k}+\sum_{t=k+1}^{j}sub(\tau_q[h-1], \tau_d[t])=\mathop{\min}\limits_{1\leq k < j}C_{h-2,k}+sub(h-1,k+1:j)$}. Next, we have to prove that the theorem also holds when $i=h$.
	{\scriptsize \begin{eqnarray}
			C_{h,j}&=&\mathop{\min}\limits_{1\leq k < j}\mathop{\min}\limits_{k\leq t < j}C_{h-1,k}+sub(h-1,k+1:t) + sub(h,t+1:j)\notag \\
			&=&\mathop{\min}\limits_{1\leq t < j}\mathop{\min}\limits_{1\leq k \leq t}C_{h-1,k}+sub(h-1,k+1:t) + sub(h,t+1:j)\notag \\
			&=&\mathop{\min}\limits_{1\leq t < j}\mathop{\min}\limits_{1\leq k \leq t}\mathop{\min}\limits_{1\leq l < k}C_{h-2,l}+sub(h-1,l+1:k) \notag\\
			&&+ sub(h-1,k+1:t) + sub(h,t+1:j)\notag \\
			&=& \mathop{\min}\limits_{1\leq t < j}\mathop{\min}\limits_{1\leq l < t}C_{h-2,l} + sub(h-1,l+1:t)+sub(h,t+1:j) \notag \\
			&=&\mathop{\min}\limits_{1\leq t < j}C_{h-1,t}+sub(h,t+1:j)\notag \\
			&=&\mathop{\min}\limits_{1\leq k < j}C_{h-1,k}+\sum_{t=k+1}^{j}sub(\tau_q[h], \tau_d[t])\notag 
	\end{eqnarray}}
	The above analysis shows that the theorem holds when $i=2$ and the theorem holds when $i=h-1$ can infer that the theorem holds when $i=h$. Therefore, the theorem holds.
\end{proof}
After obtaining the expression for $C_{i,j}$ from the theorem \ref{the:simplify}, we can further simplify it.
{\scriptsize\begin{eqnarray}
		C_{i,j}&=&\mathop{\min}\limits_{1\leq k < j}C_{i-1,k}+\sum_{t=k+1}^{j}sub(\tau_q[i], \tau_d[t])\notag\\
		&=&\min\{\mathop{\min}\limits_{1\leq k < j-1}C_{i-1,k}+\sum_{t=k+1}^{j-1}sub(\tau_q[i], \tau_d[t]), \notag\\
		&&C_{i-1,j-1}\}+sub(\tau_q[i], \tau_d[j])\notag \\
		&=&\min\{C_{i,j-1},C_{i-1,j-1}\}+sub(\tau_q[i], \tau_d[j])\notag
\end{eqnarray}}
Finally, integrating all the previous analysis results, we can get the computational expression of $C_{i,j}$. With the Equation \ref{eq:dtw}, we can quickly adapt the Algorithm \ref{algo:efficiencyA} to get the optimal subtrajectory using DTW as the distance function.
{\scriptsize\begin{equation}
		C_{i,j}=\left\{
		\begin{array}{ll}
			sub(\tau_q[i],\tau_d[j]), & i=1 \\
			C_{i-1,j}+sub(\tau_q[i],\tau_d[1]), & j=1,i\neq1 \\
			\min\{C_{i-1,j},C_{i,j-1},C_{i-1,j-1}\}\\
			+sub\left(\tau_q[i],\tau_d[j]\right),& otherwise
		\end{array}
		\right.
		\label{eq:dtw}
\end{equation}}

\subsection{Other Similarity Functions}
\label{sec:other}
In addition to DTW and WED, our method is also valid for other order-insensitive distance functions. EDR and ERP are specific cases of WED functions. Therefore, we only need to define $sub$, $ins$, and $del$ in Equation \ref{eq:wed_cal}.
We denote the euclidean distance between two points $\tau_q[i]$ and $\tau_d[j]$ as $d(\tau_q[i], \tau_d[j])$. We can convert WED to ERP and EDR by defining $sub$, $ins$ and $del$: (i) ERP. We can convert WED into ERP by making $sub(\tau_q[i], \tau_d[j])=d(\tau_q[i], \tau_d[j])$, $del(\tau_q[i])=d(\tau_q[i], q_c)$, $ins(\tau_d[j])=d(\tau_d[j], q_c)$, where $q_c$ is a fixed point on the map (e.g., the center of the region). (ii) EDR. $ins(\tau_d[j])$ and $del(\tau_q[i])$ in EDR are both 1, while $sub(\tau_q[i], \tau_d[j])$ takes a value of 0 if and only if $d(\tau_d[j], q_c) < \epsilon$ holds; otherwise, $sub(\tau_q[i], \tau_d[j])=1$.

FD is similar to DTW. In the same way, we can obtain the expressions for $C_{i,j}$ when FD is the distance function.
{\scriptsize\begin{equation}
	C_{i,j}=\left\{
	\begin{array}{ll}
		sub\left(\tau_q[i],\tau_d[j]\right), & i=1 \\
		\max\left\{C_{i-1,j},sub(\tau_q[i],\tau_d[1])\right\}, & j=1,i\neq1 \\
		\max\{\min\{C_{i-1,j},C_{i,j-1},C_{i-1,j-1}\},\\
		sub\left(\tau_q[i],\tau_d[j]\right)\},& otherwise
	\end{array}
	\right.
	\label{eq:fd}
\end{equation}}
When the order-insensitive distance functions are used as the trajectory distance functions, the calculation of the conversation cost does not depend on the position of the current point in the trajectory.

Unfortunately, our method cannot be applied to the subtrajectory search problem when an order-sensitive trajectory distance function (such as LCSS) is used. This is because we do not consider the position from which the subtrajectory starts when computing $C_{i,j}$. When $\tau_q[i]$ matches $\tau_d[j]$, the cost of converting $\tau_q[i]$ to $\tau_d[j]$ is only related to the matching relationship between $\tau_q[i-1]$ and $\tau_d[k]$ ($1\leq k \leq j$). However, when LCSS is used as the distance function, the cost of converting $\tau_q[i]$ to $\tau_d[j]$ is also related to the matching relation of $\tau_q[1]$, i.e., the starting position of the subtrajectory. The starting position of the subtrajectories has a great influence on judging the distance between the points in two trajectories when LCSS is used as the distance function. Therefore, our algorithm is not suitable for a class of distance functions that considers the position of points in the trajectory, such as LCSS.

%% file: experimentalStudy.tex
\section{Experimental Study}
\label{sec:experimental}


\begin{table*}[t!]\vspace{1ex}
	\begin{center}
		{\small 
			\caption{\small Effectiveness of Algorithms.} \label{tab:performance}
			\vspace{1ex}
			\begin{tabular}{c|c|c|c|c|c|c|c|c|c|c|c|c|c}
				\hline
				\multirow{2}{*}{Dataset}&\multirow{2}{*}{Algorithm}&\multicolumn{3}{c|}{DTW}&\multicolumn{3}{c|}{EDR}&\multicolumn{3}{c|}{ERP}&\multicolumn{3}{c}{FD} \\\cline{3-14}
				&&AR&MR&RR&AR&MR&RR&AR&MR&RR&AR&MR&RR \\ \hline\hline
				
				\multirow{8}{*}{Porto}&POS&3.033461&351.01&13.09\%&1.432727&321.91&15.35\%&1.498225&58.99&1.83\%&2.936241&210.86&5.02\%\\ 
				&PSS&1.976035&128.91&6.80\%&1.352727&237.01&9.11\%&2.532003&139.10&5.71\%&1.378006&154.50&2.60\%\\ 
				&RLS&1.739143&97.56&5.13\%&1.343469&190.54&7.32\%&2.232406&114.62&4.70\%&1.384809&134.26&2.25\%\\ 
				&RLS-Skip&2.033048&142.68&7.01\%&1.354480&234.13&9.28\%&2.447045&134.79&5.48\%&1.643495&173.68&3.08\%\\ 
				&CMA&\textbf{1}&\textbf{1}&\textbf{0\%}&\textbf{1}&\textbf{1}&\textbf{0\%}&\textbf{1}&\textbf{1}&\textbf{0\%}&\textbf{1}&\textbf{1}&\textbf{0\%}\\ 
				&ExactS&1&1&0\%&1&1&0\%&1&1&0\%&1&1&0\%\\ 
				&Spring&1&1&0\%&-&-&-&-&-&-&-&-&-\\ 
				&GB&-&-&-&-&-&-&-&-&-&1&1&0\%\\ 
				\hline
				\multirow{8}{*}{Xi'an}&POS&35.563473&10505.00&18.12\%&1.516196&286.84&1.14\%&1.453240&34.51&0.15\%&20.502515&3771.50&5.31\%\\ 
				&PSS&4.374571&676.99&2.99\%&1.460815&378.33&1.34\%&1.703792&41.49&0.17\%&1.383835&25.90&0.03\%\\ 
				&RLS&3.613057&511.53&2.26\%&1.434072&304.03&1.08\%&1.564813&34.32&0.14\%&1.389806&22.68&0.02\%\\ 
				&RLS-Skip&7.318057&1567.09&4.25\%&1.459621&352.07&1.26\%&1.691469&41.94&0.17\%&3.531339&434.34&0.60\%\\ 
				&CMA&\textbf{1}&\textbf{1}&\textbf{0\%}&\textbf{1}&\textbf{1}&\textbf{0\%}&\textbf{1}&\textbf{1}&\textbf{0\%}&\textbf{1}&\textbf{1}&\textbf{0\%}\\ 
				&ExactS&1&1&0\%&1&1&0\%&1&1&0\%&1&1&0\%\\ 
				&Spring&1&1&0\%&-&-&-&-&-&-&-&-&-\\ 
				&GB&-&-&-&-&-&-&-&-&-&1&1&0\%\\ 
				\hline
			\end{tabular}
		}
	\end{center}
\end{table*}

\subsection{Experimental Settings}

\noindent\textbf{Data Sets.} \revision{We conduct experiments on three real data sets: (i) Porto \cite{porto} is a dataset describing a whole year (i.e., from July 1st, 2013 to June 30th, 2014) of the trajectories for all the 442 taxis running in the city of Porto (i.e., size: $23.44km\times 24.7km$, longitude: \ang{-8.75}$\sim$\ang{-8.47}, latitude: \ang{41.02}$\sim$\ang{ 41.25}). There are 1,710,670 trajectories with 15-seconds point intervals, whose average length is 67. (ii)  Xi'an Taxi Trip Dataset. DiDi Chuxing GAIA Open Dataset \cite{gaia} provides a dataset of taxi trips in Xi'an area (i.e., size: $33.43km\times 23.5km$, longitude: \ang{108.78}$\sim$\ang{109.05}, latitude: \ang{34.14}$\sim$\ang{34.38}). We use the taxi trip records on October 1st. There are 149,742 trajectories with 3-seconds point intervals, whose average length is 401. (iii) T-Drive Data~\cite{yuan2010t, yuan2011driving}. T-Drive Data provides taxi trips in Beijing area (i.e., size:  $49.80km\times 42.11km$, longitude: \ang{116.15}$\sim$\ang{116.60}, latitude: \ang{39.75}$\sim$\ang{40.10}). There are 10,357 trajectories with 300-seconds point intervals, whose average length is 1705.}

In this experiment, we generate $Q$ query trajectories from all trajectories and take the average of the results. Specifically, we select $Q$ trajectories in uniform random as query trajectory, while the other trajectories are used as data trajectories. We set $Q$ to 100 by default. 

\noindent\textbf{Searching Algorithms.}
We mainly compare our CMA algorithm with the following existing methods:

\begin{enumerate}[label=\arabic*),align=right,itemindent=1em,labelsep=2pt,labelwidth=1em,leftmargin=0pt,nosep]
\item ExactS. When it computes the distance between the query trajectory and some subtrajectories of the data trajectory, it records these intermediate results. Then, ExactS can utilize a dynamic programming technique to optimize the time complexity of searching the optimal subtrajectory from a data trajectory to $O(mn^2)$.

\item PSS and POS. The main idea of PSS is to traverse each point of a data trajectory to find the appropriate splitting position. The current optimal subtrajectory is updated by comparing the distance between the subtrajectory before the splitting point and the subtrajectory after the splitting point and the query trajectory. Then, the next suitable splitting point is found starting from the current splitting point. PSS can find an approximate solution of the optimal subtrajectory within the time complexity of $O(mn)$. As a variant of PSS, POS does not consider the subtrajectory after the splitting point. Therefore, the efficiency of POS is substantially improved compared with PSS, but the result quality of  PSS is better than that of POS.

\item RLS and RLS-Skip. RLS is an algorithm based on reinforcement learning to determine whether to split the current point, and RLS takes a different action based on the state of the current point. RLS-Skip, on the other hand, adds a new action to RLS by not segmenting the current point and skipping the next point to traverse the entire trajectory faster. As a result, RLS-Skip can get a solution in less time, while RLS can find a better solution.

\item Spring and Greedy Backtracking (GB). Both algorithms are of time complexity $O\left(mn\right)$. However, unlike the method proposed in this paper, Spring and GB can only be applied to specific distance functions, DTW and FD, respectively. Therefore, we will test the performance and results of these two algorithms under particular functions in different data sets.
\end{enumerate}

\begin{figure*}[t!]
	\centering\vspace{-2ex}
	\subfigure{
		\scalebox{0.4}[0.4]{\includegraphics{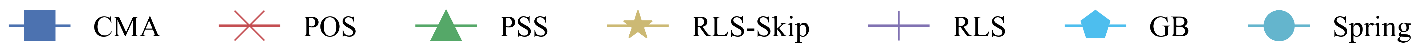}}}\hfill\\
	\addtocounter{subfigure}{-1}\vspace{-2ex}
	\subfigure[][{\scriptsize DTW (Porto)}]{
		\scalebox{0.19}[0.19]{\includegraphics{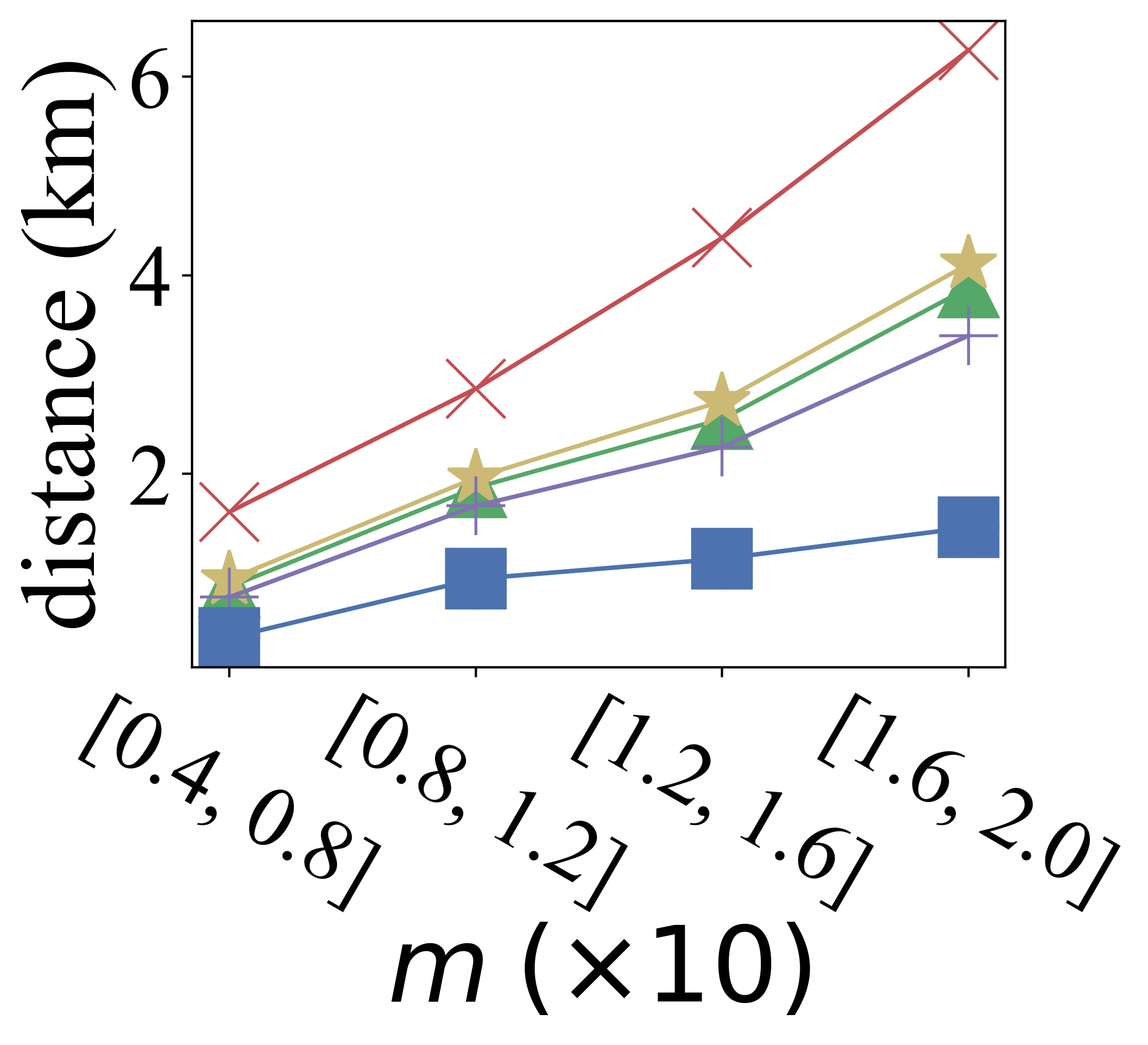}}
		\label{fig:porto_dtw_score}}
	\subfigure[][{\scriptsize DTW (Xi'an)}]{
		\scalebox{0.19}[0.19]{\includegraphics{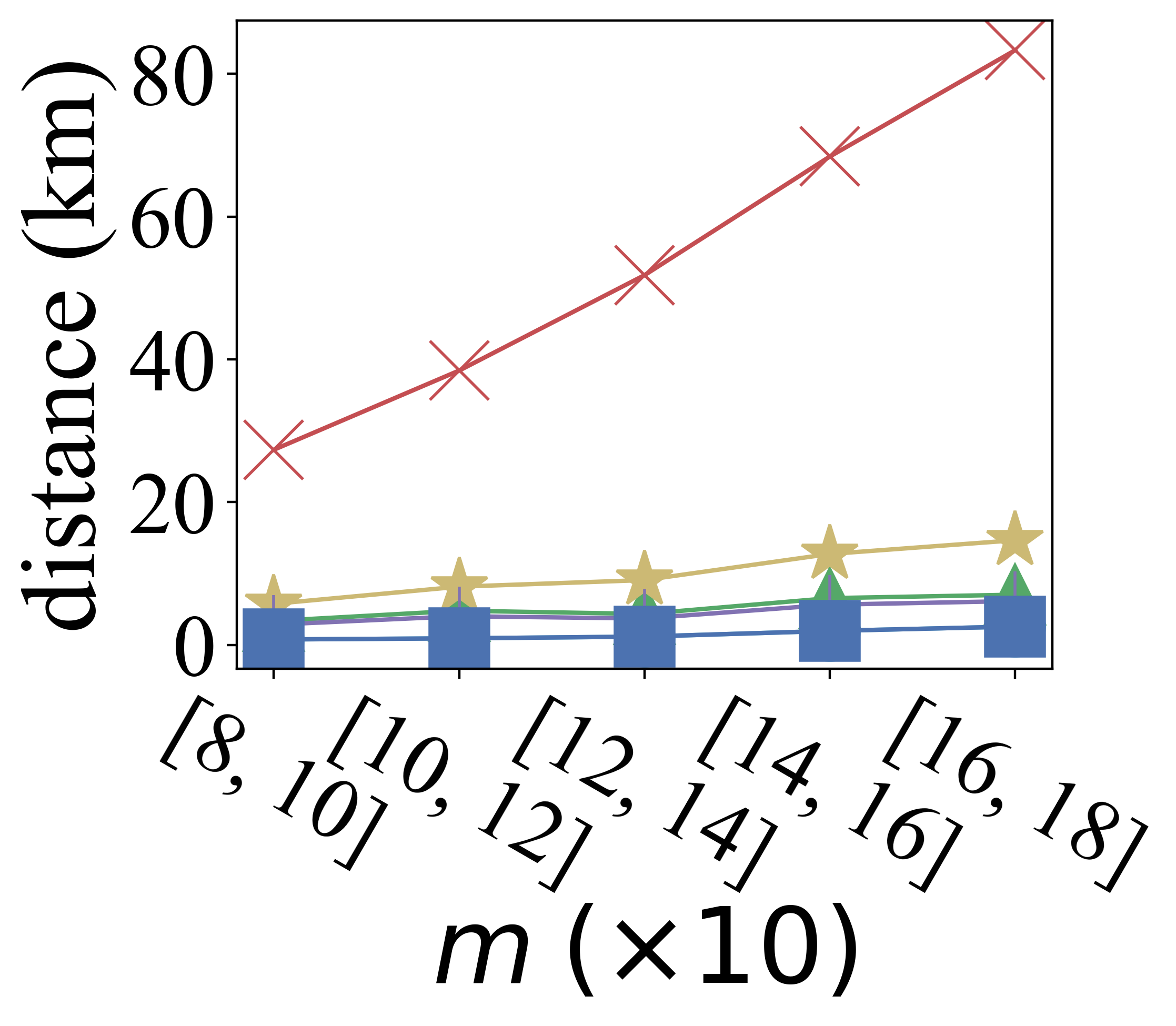}}
		\label{fig:xian_dtw_score}}
	\subfigure[][{\scriptsize DTW (Beijing)}]{
		\scalebox{0.19}[0.19]{\includegraphics{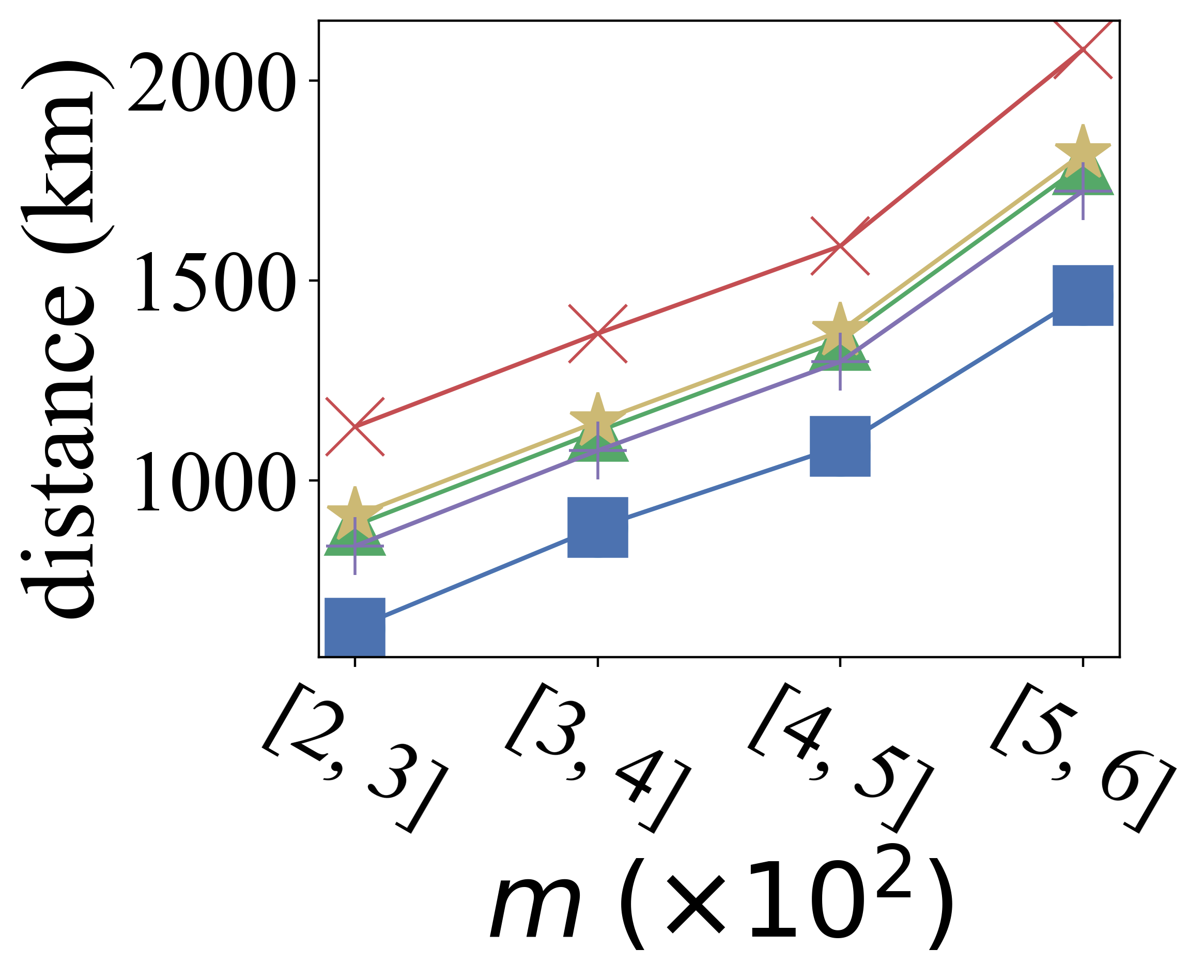}}
		\label{fig:beijing_dtw_score}}
	\subfigure[][{\scriptsize DTW (Porto)}]{
		\scalebox{0.19}[0.19]{\includegraphics{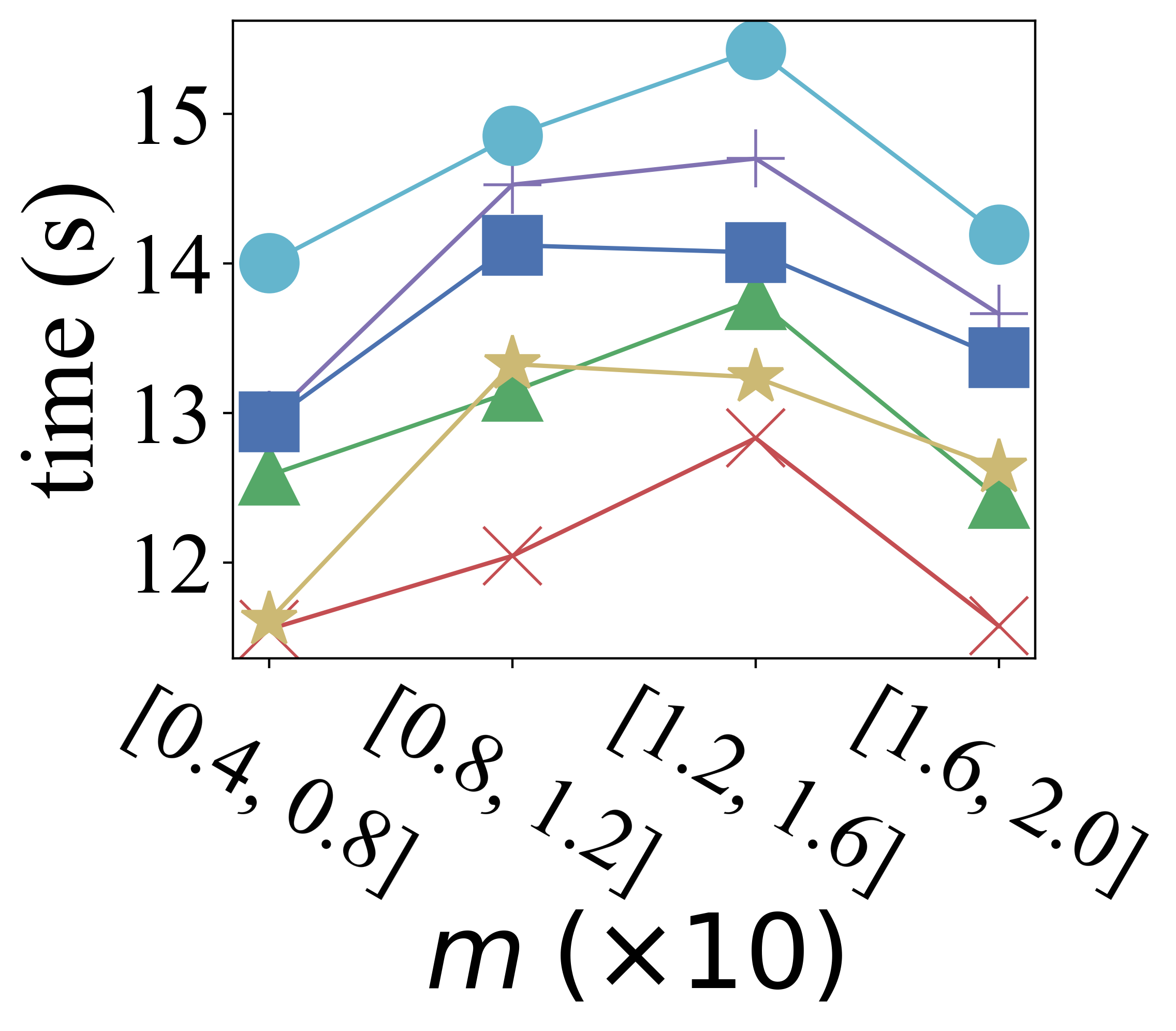}}
		\label{fig:porto_dtw}}
	\subfigure[][{\scriptsize DTW (Xi'an)}]{
		\scalebox{0.19}[0.19]{\includegraphics{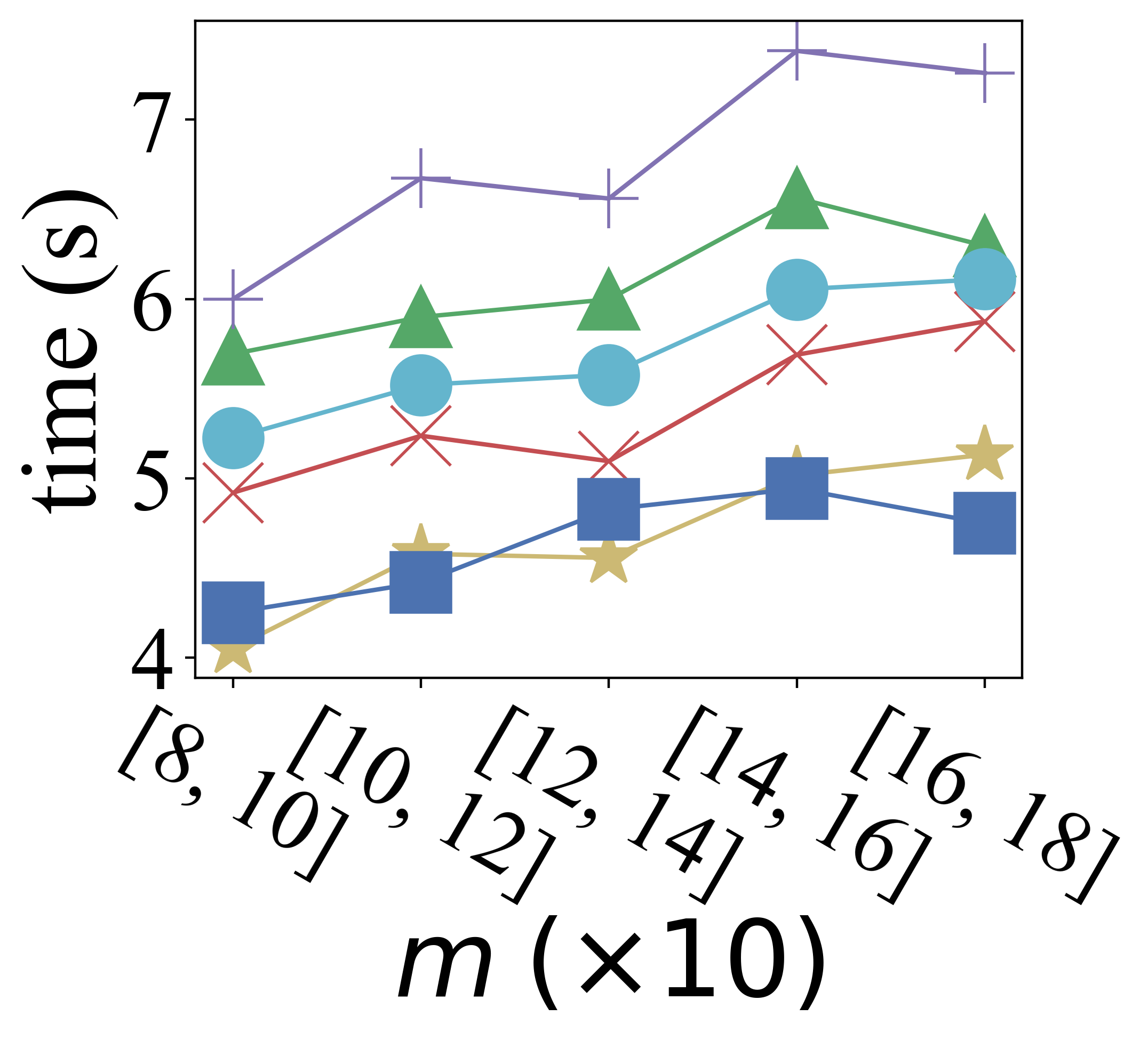}}
		\label{fig:xian_dtw}}
	\subfigure[][{\scriptsize DTW (Beijing)}]{
		\scalebox{0.19}[0.19]{\includegraphics{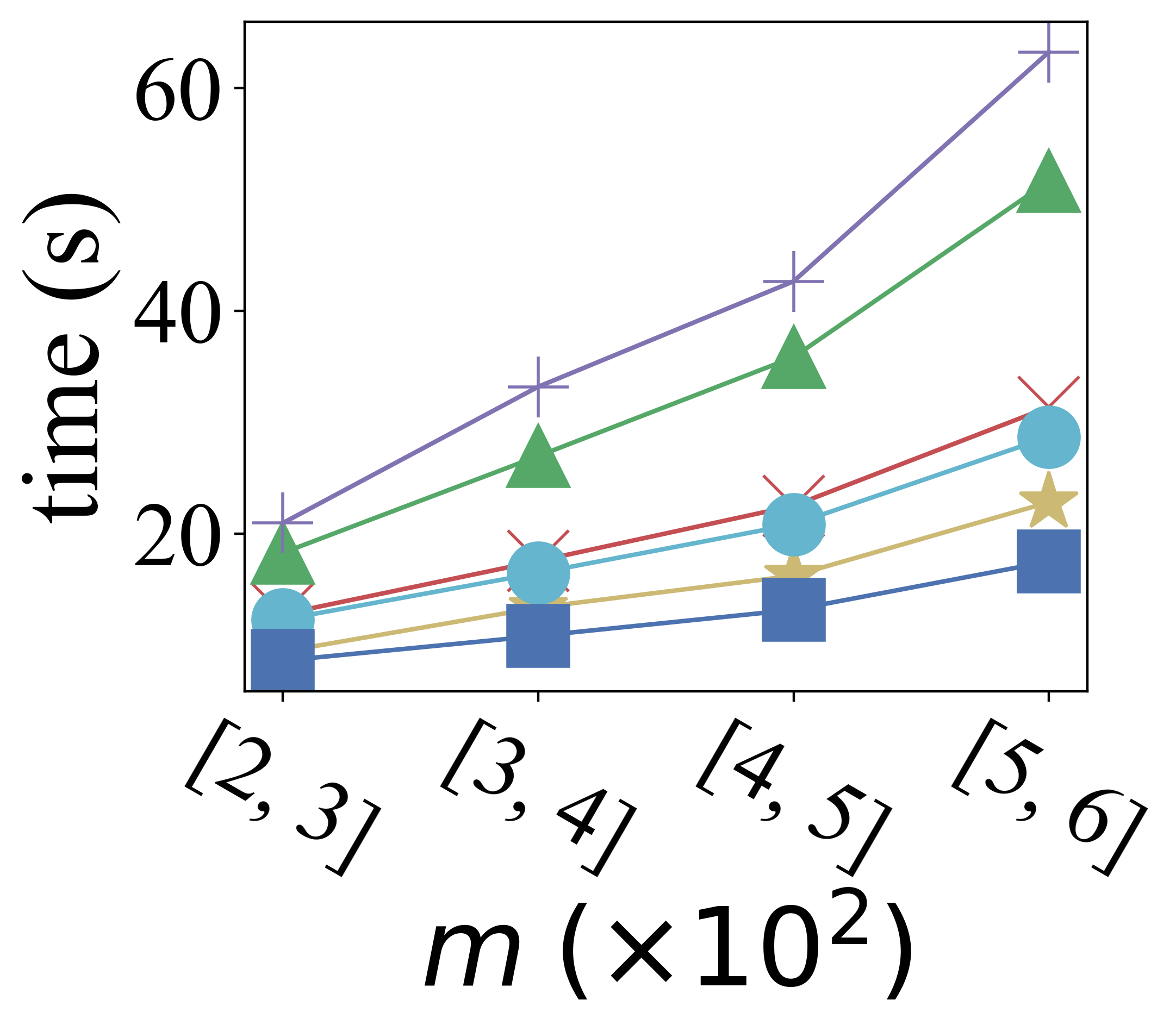}}
		\label{fig:beijing_dtw}}
	\\
	\subfigure[][{\scriptsize EDR (Porto)}]{
		\scalebox{0.19}[0.19]{\includegraphics{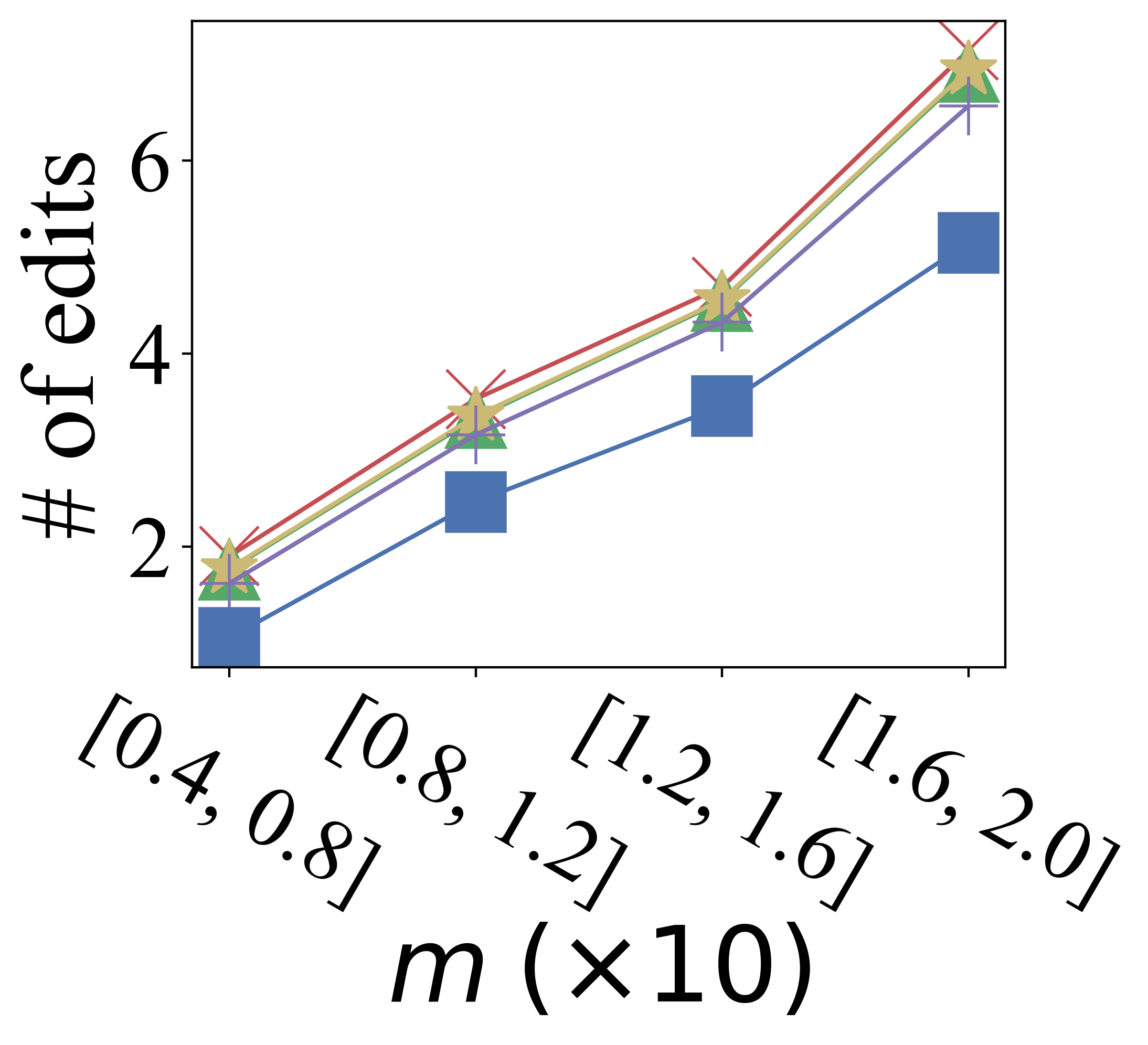}}
		\label{fig:porto_edr_score}}
	\subfigure[][{\scriptsize EDR (Xi'an)}]{
		\scalebox{0.19}[0.19]{\includegraphics{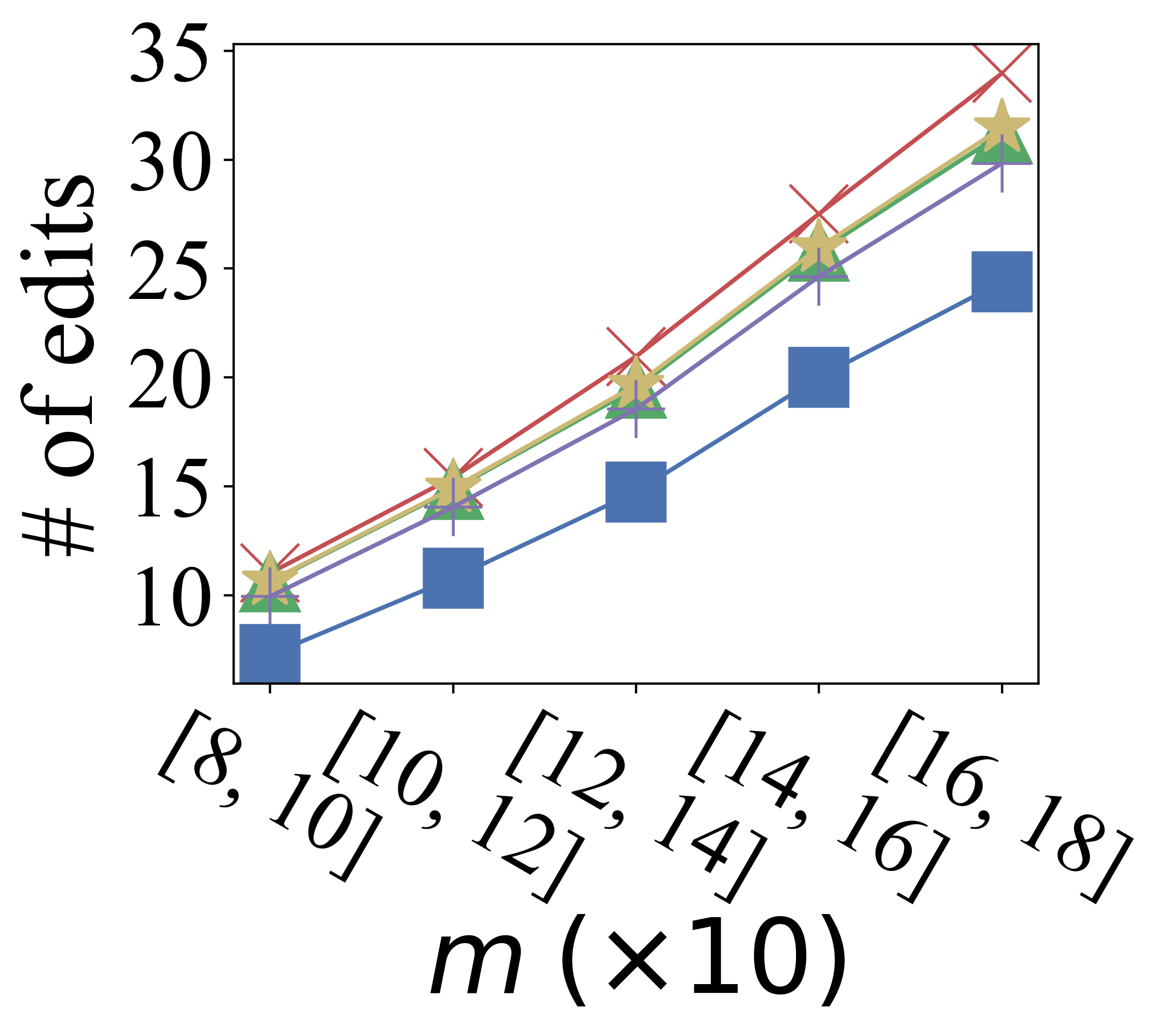}}
		\label{fig:xian_edr_score}}
	\subfigure[][{\scriptsize EDR (Beijing)}]{
		\scalebox{0.19}[0.19]{\includegraphics{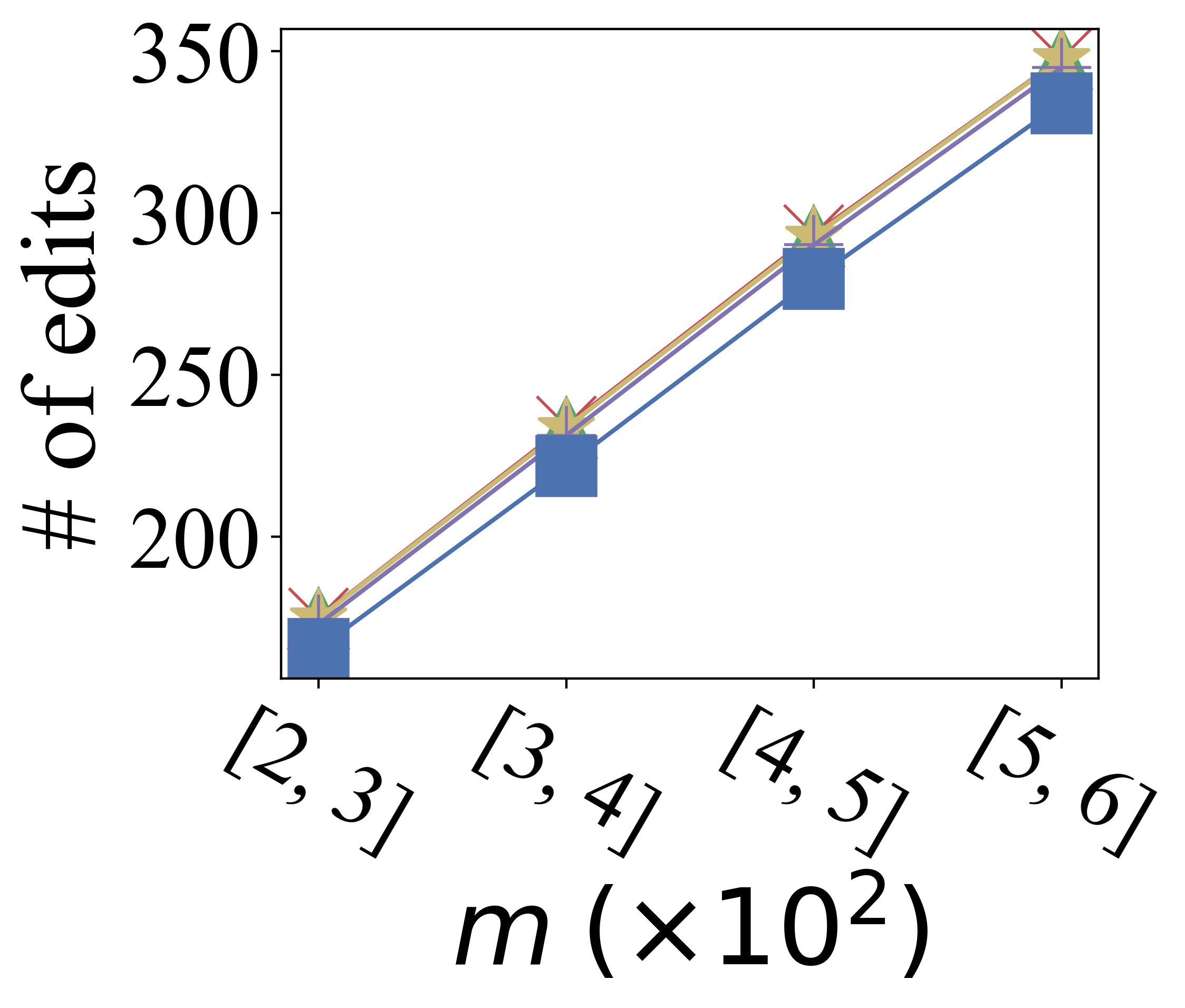}}
		\label{fig:beijing_edr_score}}
	\subfigure[][{\scriptsize EDR (Porto)}]{
		\scalebox{0.19}[0.19]{\includegraphics{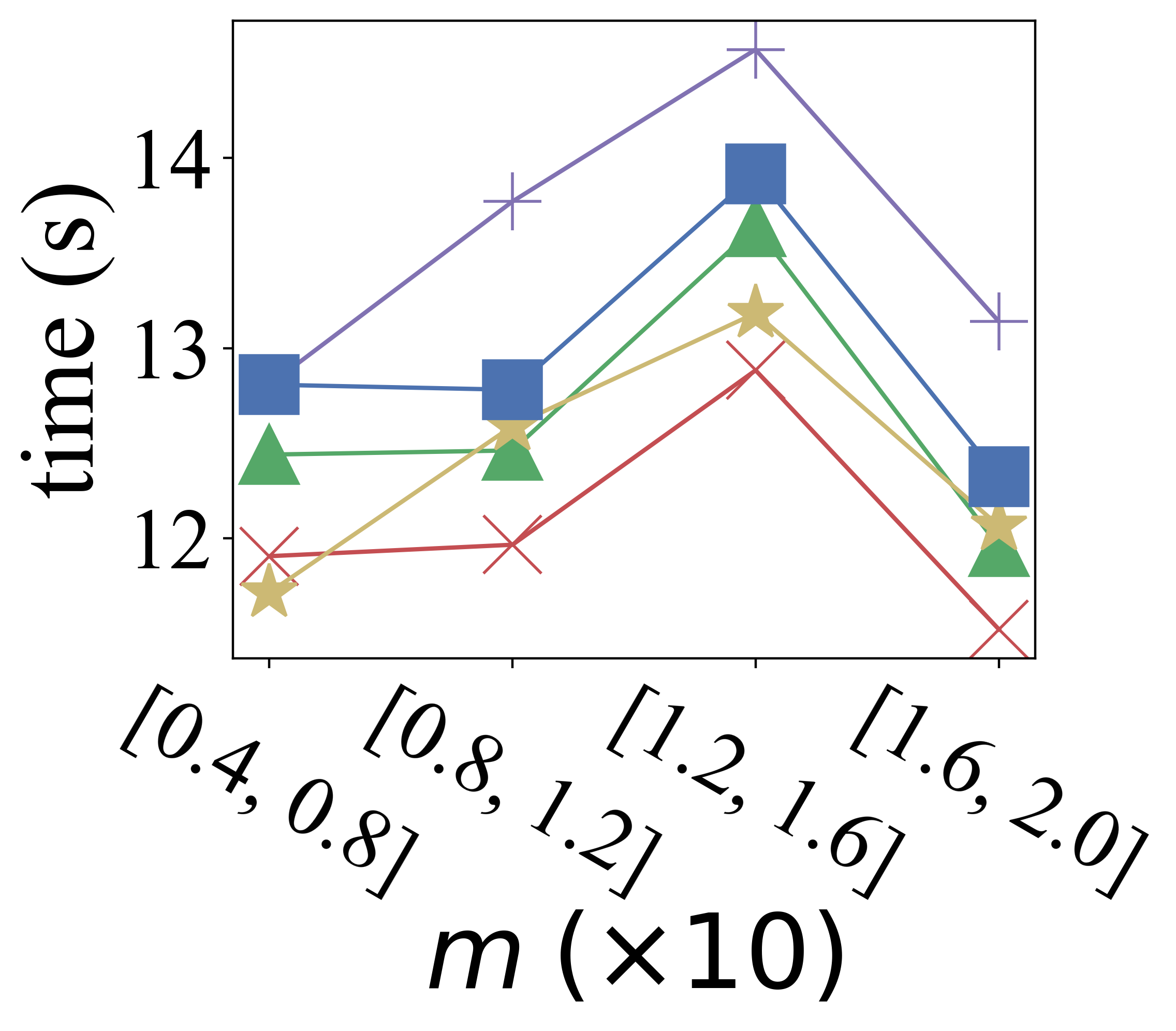}}
		\label{fig:porto_edr}}
	\subfigure[][{\scriptsize EDR (Xi'an)}]{
		\scalebox{0.19}[0.19]{\includegraphics{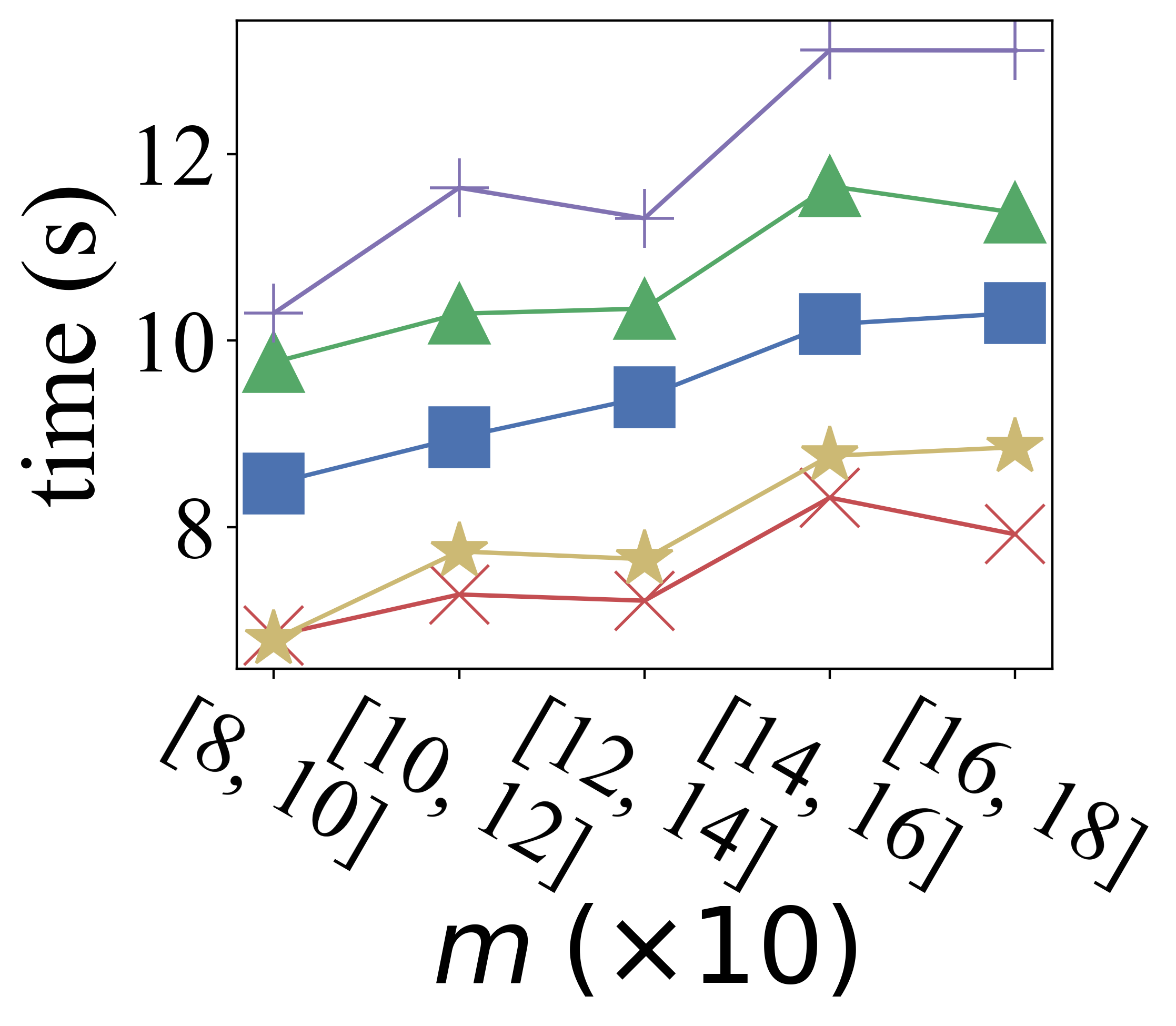}}
		\label{fig:xian_edr}}
	\subfigure[][{\scriptsize EDR (Beijing)}]{
		\scalebox{0.19}[0.19]{\includegraphics{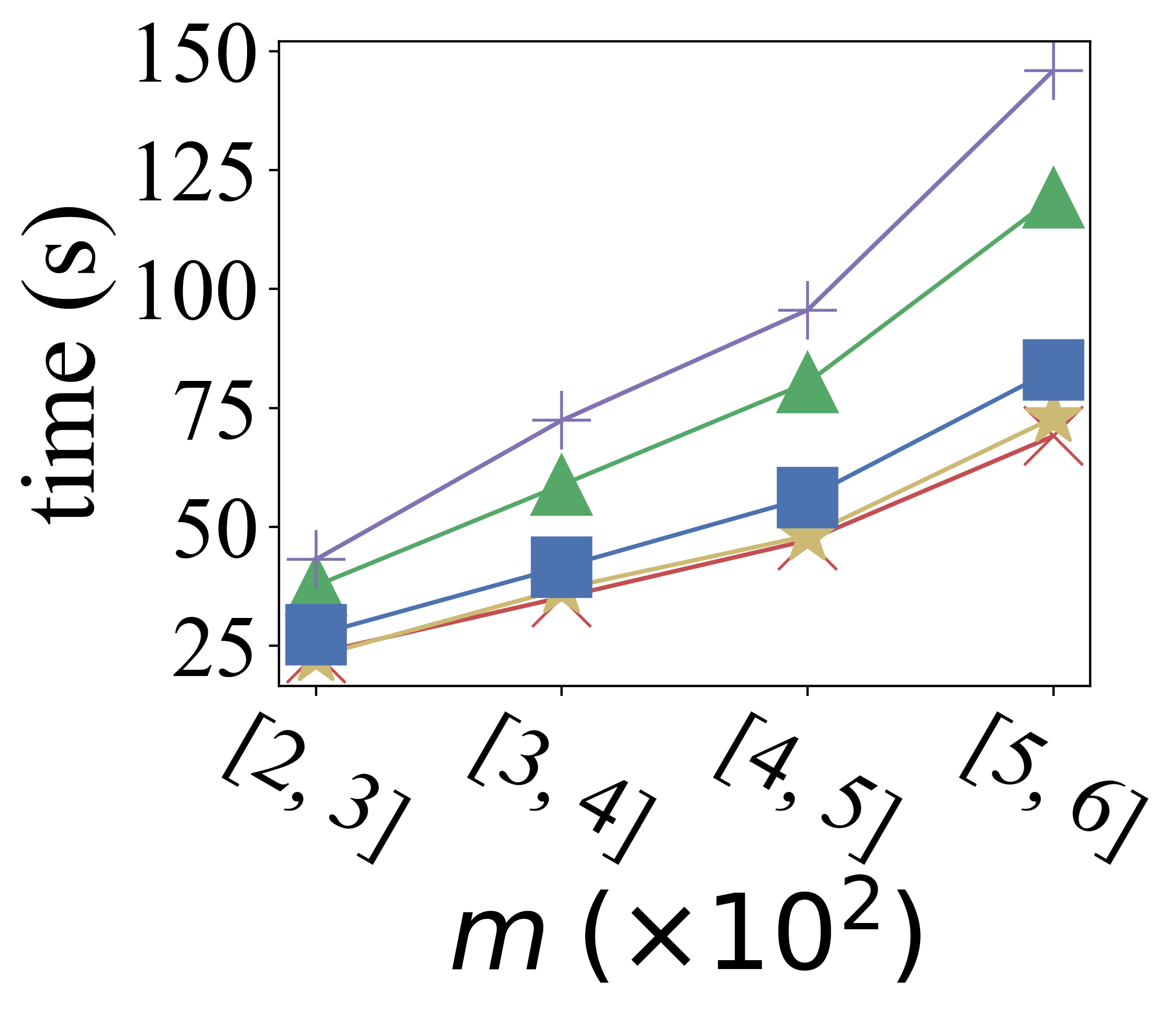}}
		\label{fig:beijing_edr}}
	\\
	\subfigure[][{\scriptsize ERP (Porto)}]{
		\scalebox{0.19}[0.19]{\includegraphics{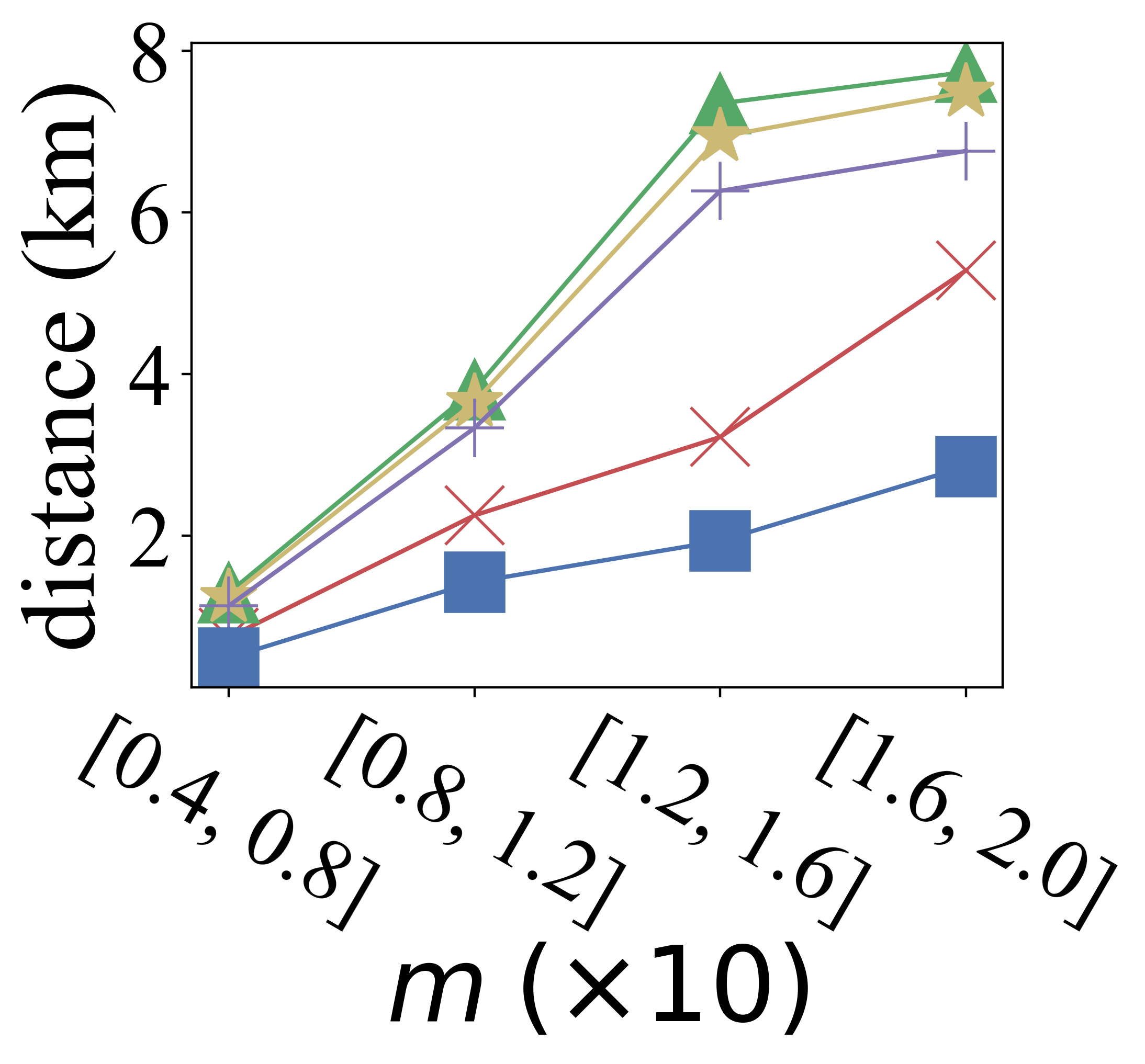}}
		\label{fig:porto_erp_score}}
	\subfigure[][{\scriptsize ERP (Xi'an)}]{
		\scalebox{0.19}[0.19]{\includegraphics{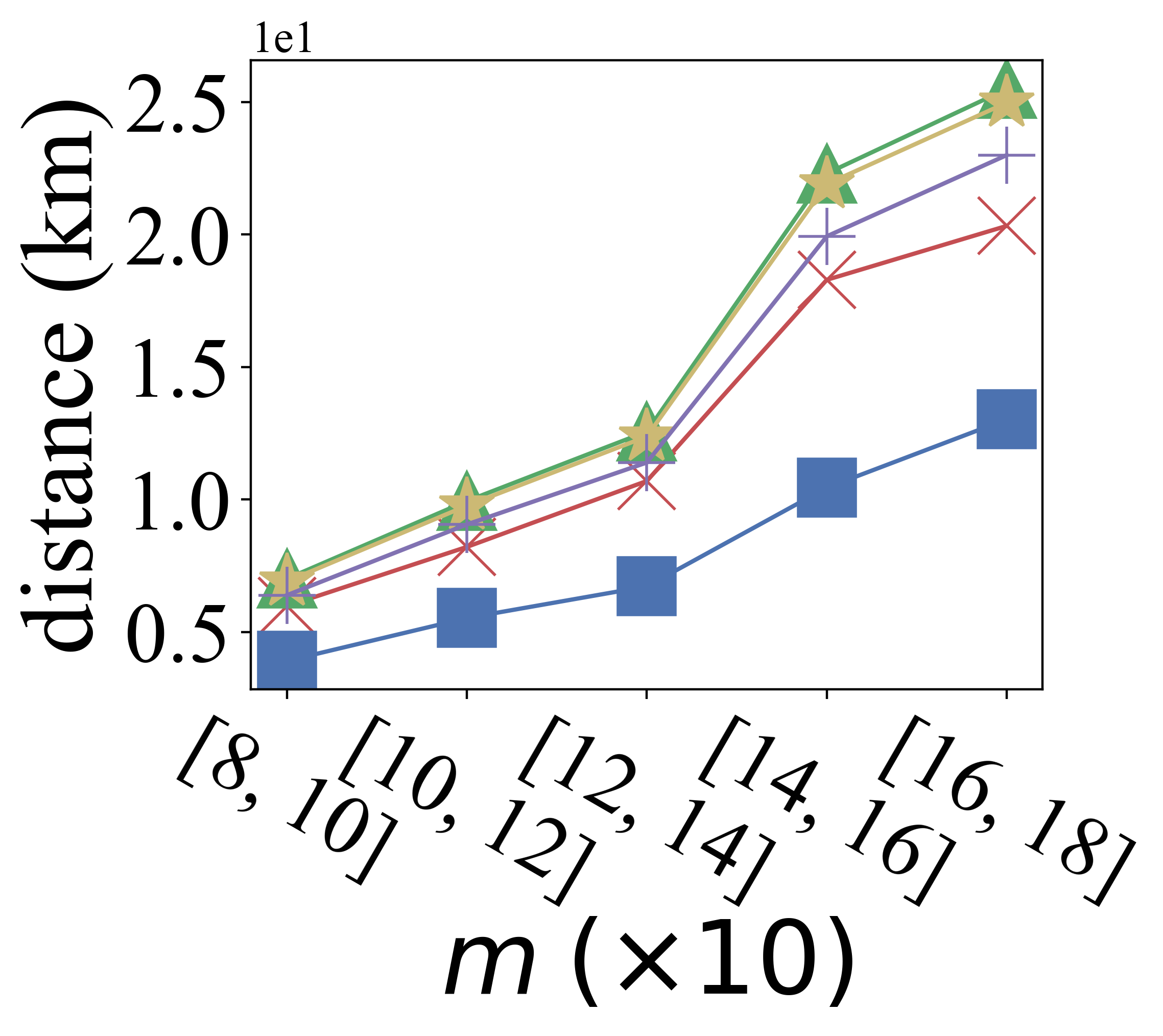}}
		\label{fig:xian_erp_score}}
	\subfigure[][{\scriptsize ERP (Beijing)}]{
		\scalebox{0.19}[0.19]{\includegraphics{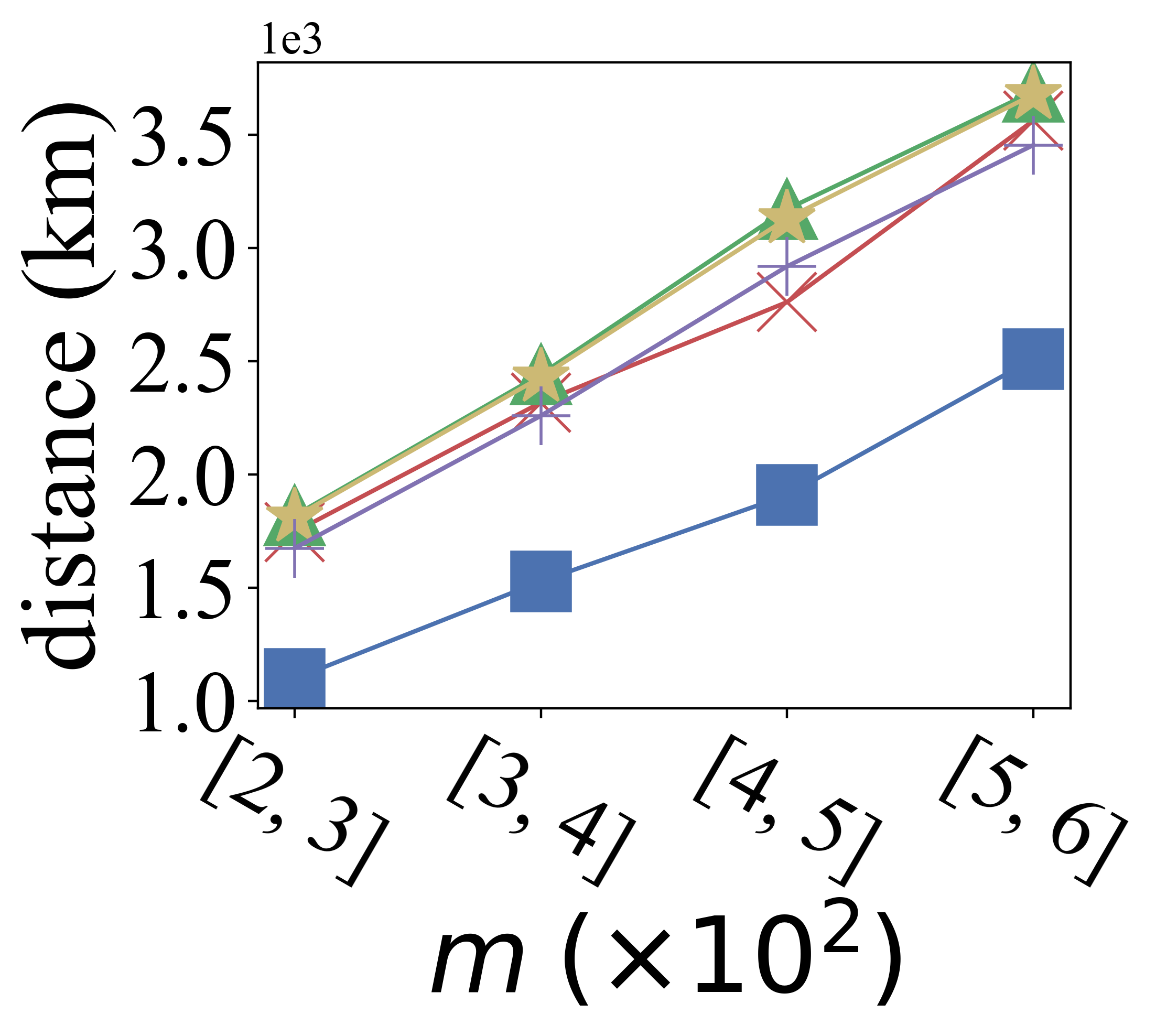}}
		\label{fig:beijing_erp_score}}
	\subfigure[][{\scriptsize ERP (Porto)}]{
		\scalebox{0.19}[0.19]{\includegraphics{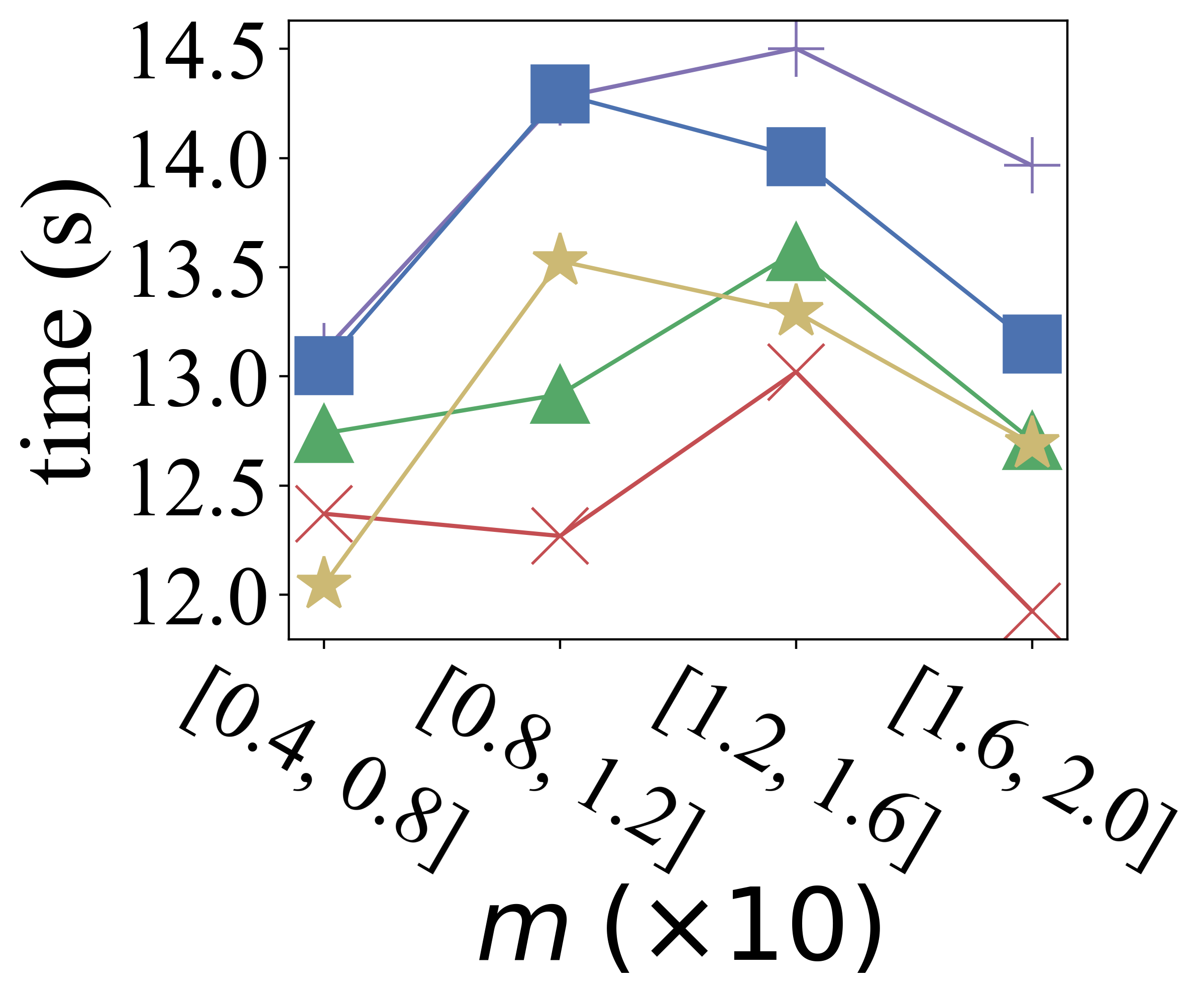}}
		\label{fig:porto_erp}}
	\subfigure[][{\scriptsize ERP (Xi'an)}]{
		\scalebox{0.19}[0.19]{\includegraphics{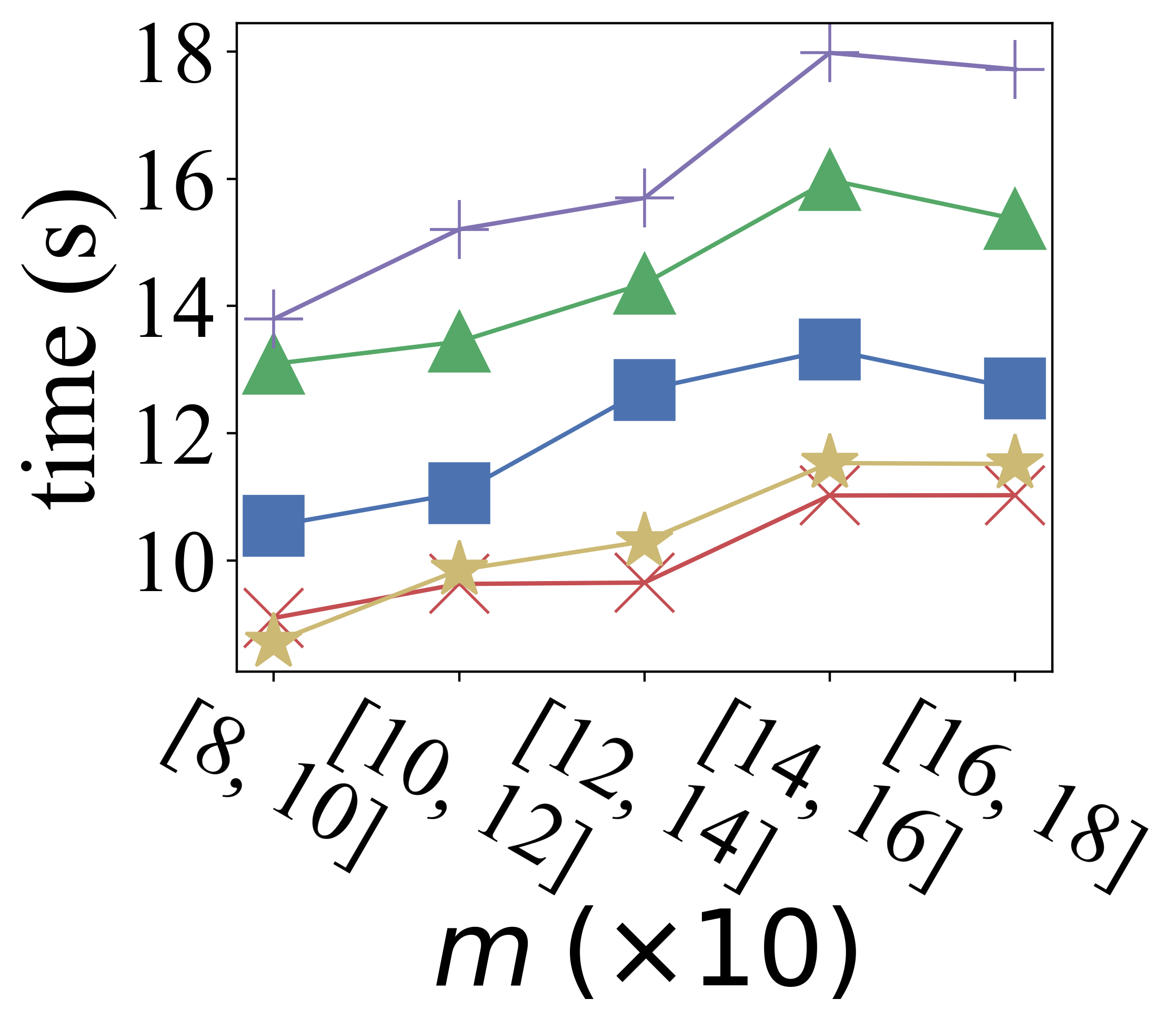}}
		\label{fig:xian_erp}}
	\subfigure[][{\scriptsize ERP (Beijing)}]{
		\scalebox{0.19}[0.19]{\includegraphics{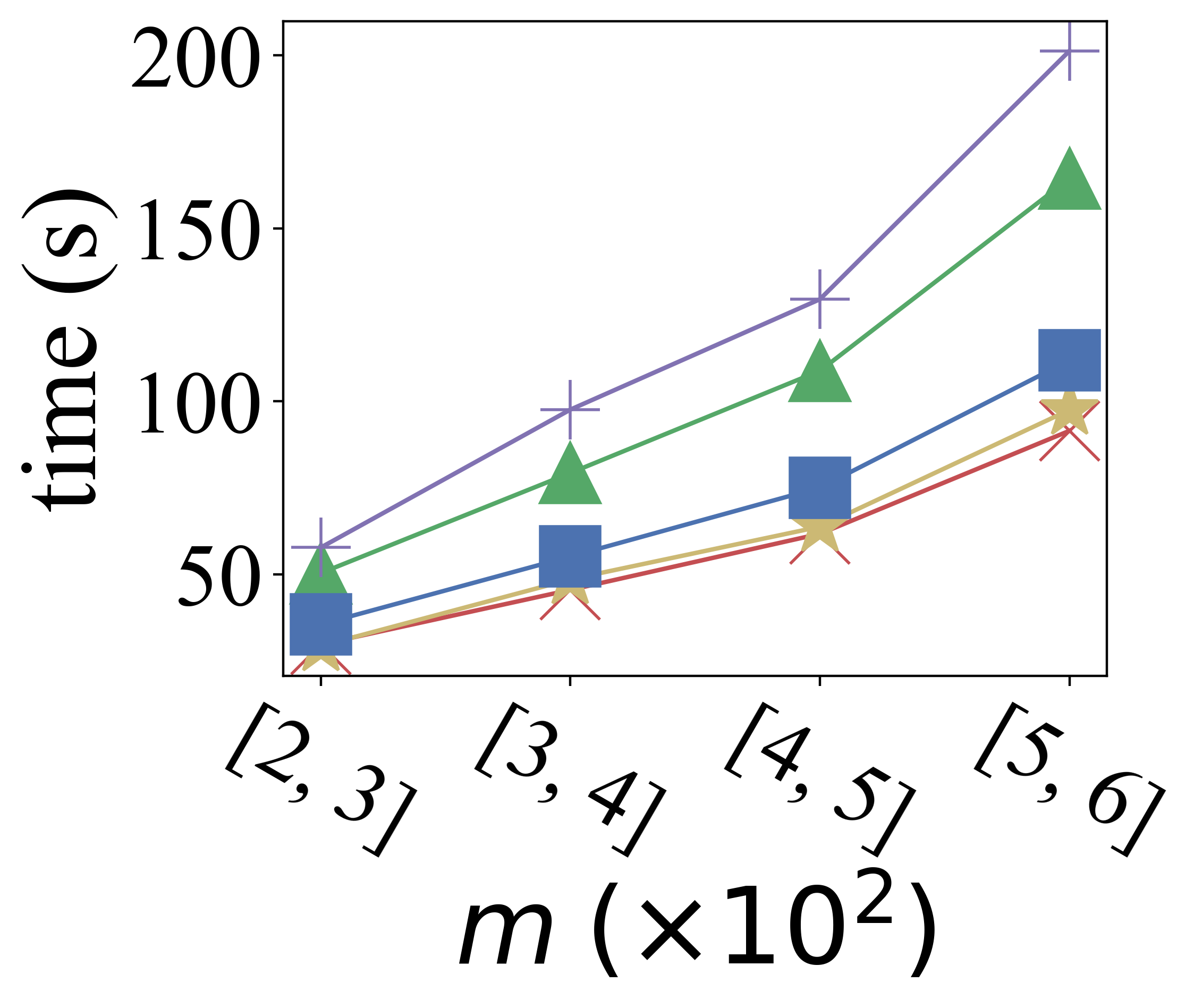}}
		\label{fig:beijing_erp}}
	\\
	\subfigure[][{\scriptsize FD (Porto)}]{
		\scalebox{0.19}[0.19]{\includegraphics{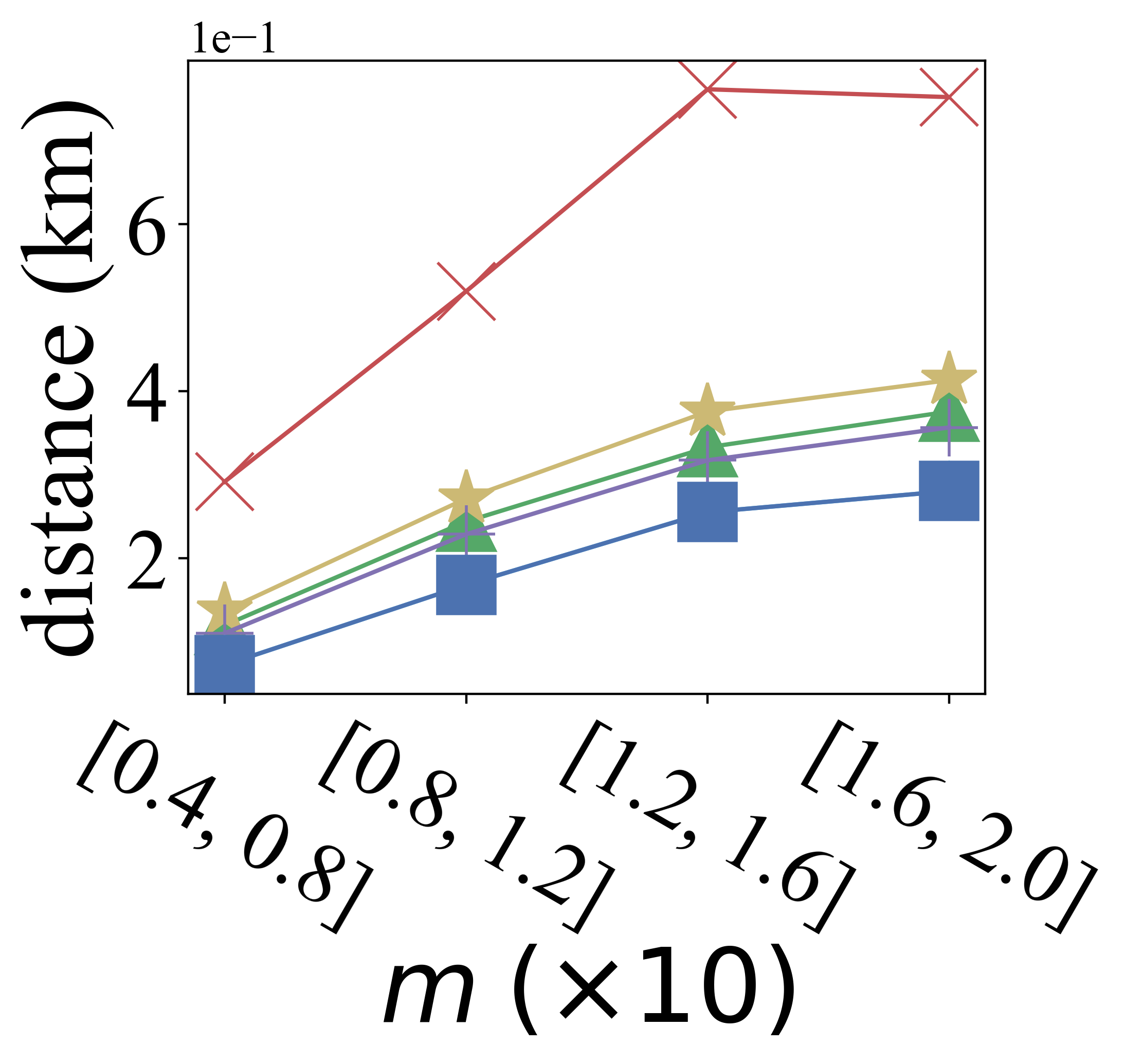}}
		\label{fig:porto_fc_score}}
	\subfigure[][{\scriptsize FD (Xi'an)}]{
		\scalebox{0.19}[0.19]{\includegraphics{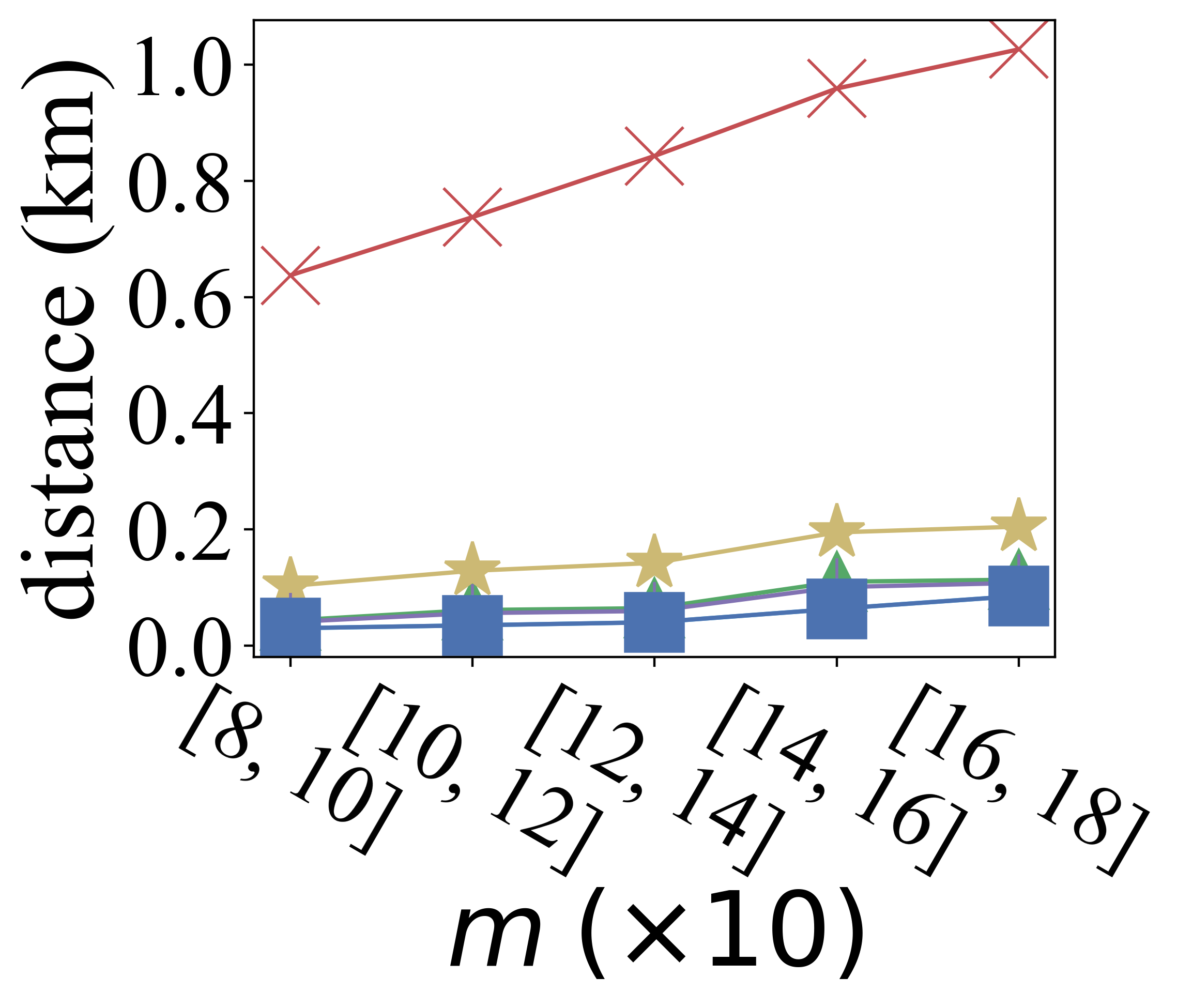}}
		\label{fig:xian_fc_score}}
	\subfigure[][{\scriptsize FD (Beijing)}]{
		\scalebox{0.19}[0.19]{\includegraphics{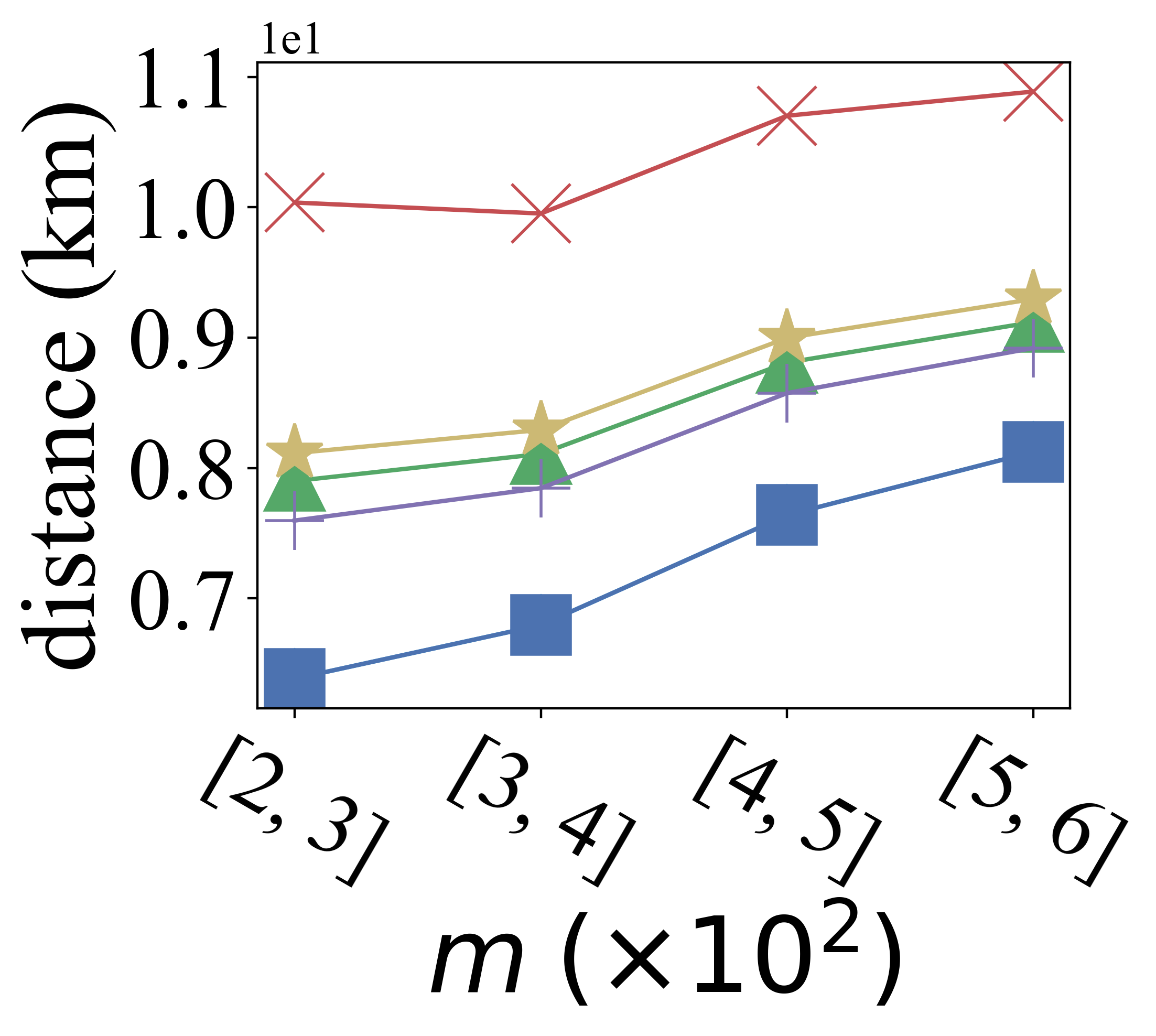}}
		\label{fig:beijing_fc_score}}
	\subfigure[][{\scriptsize FD (Porto)}]{
		\scalebox{0.19}[0.19]{\includegraphics{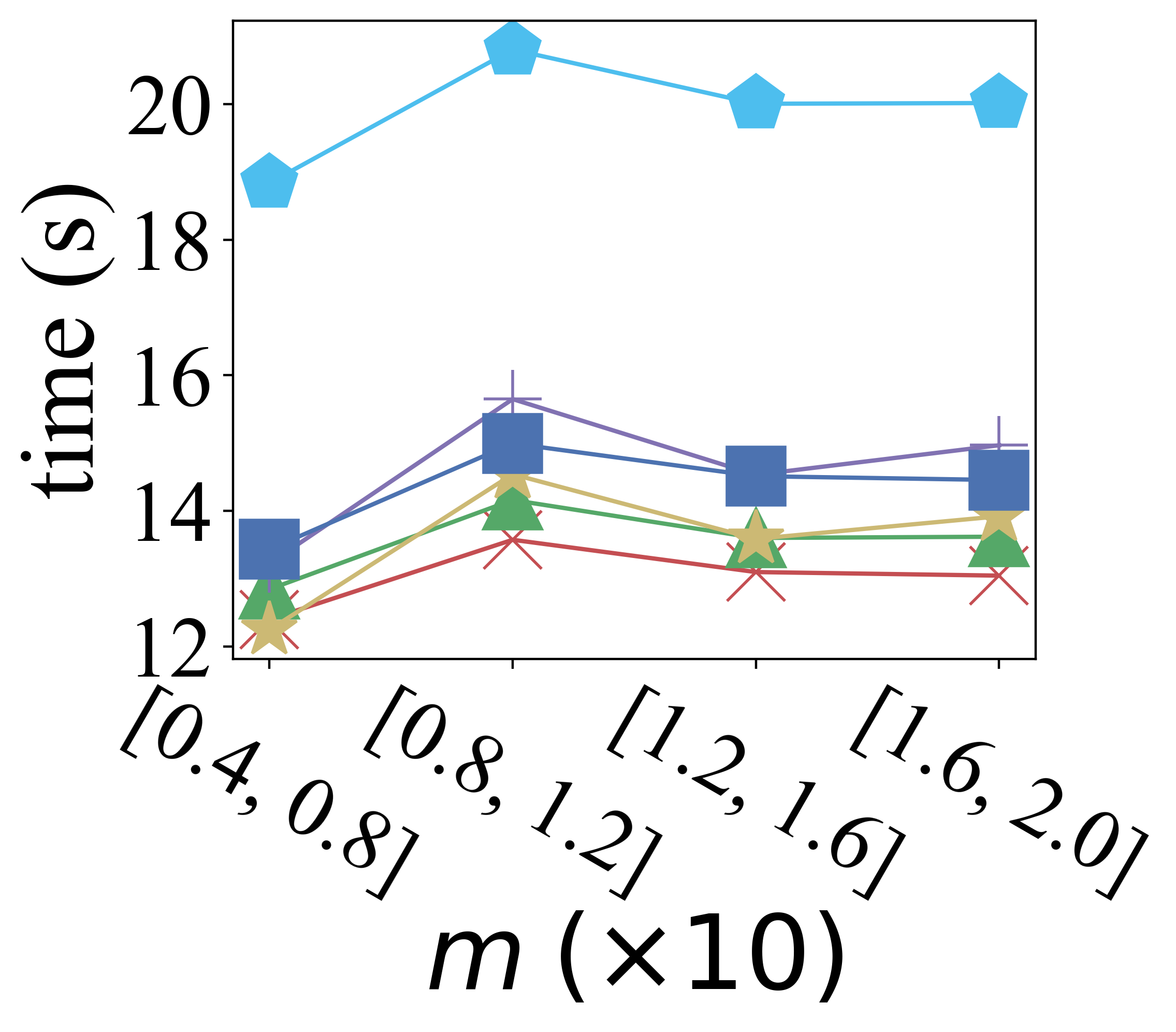}}
		\label{fig:porto_fc}}
	\subfigure[][{\scriptsize FD (Xi'an)}]{
		\scalebox{0.19}[0.19]{\includegraphics{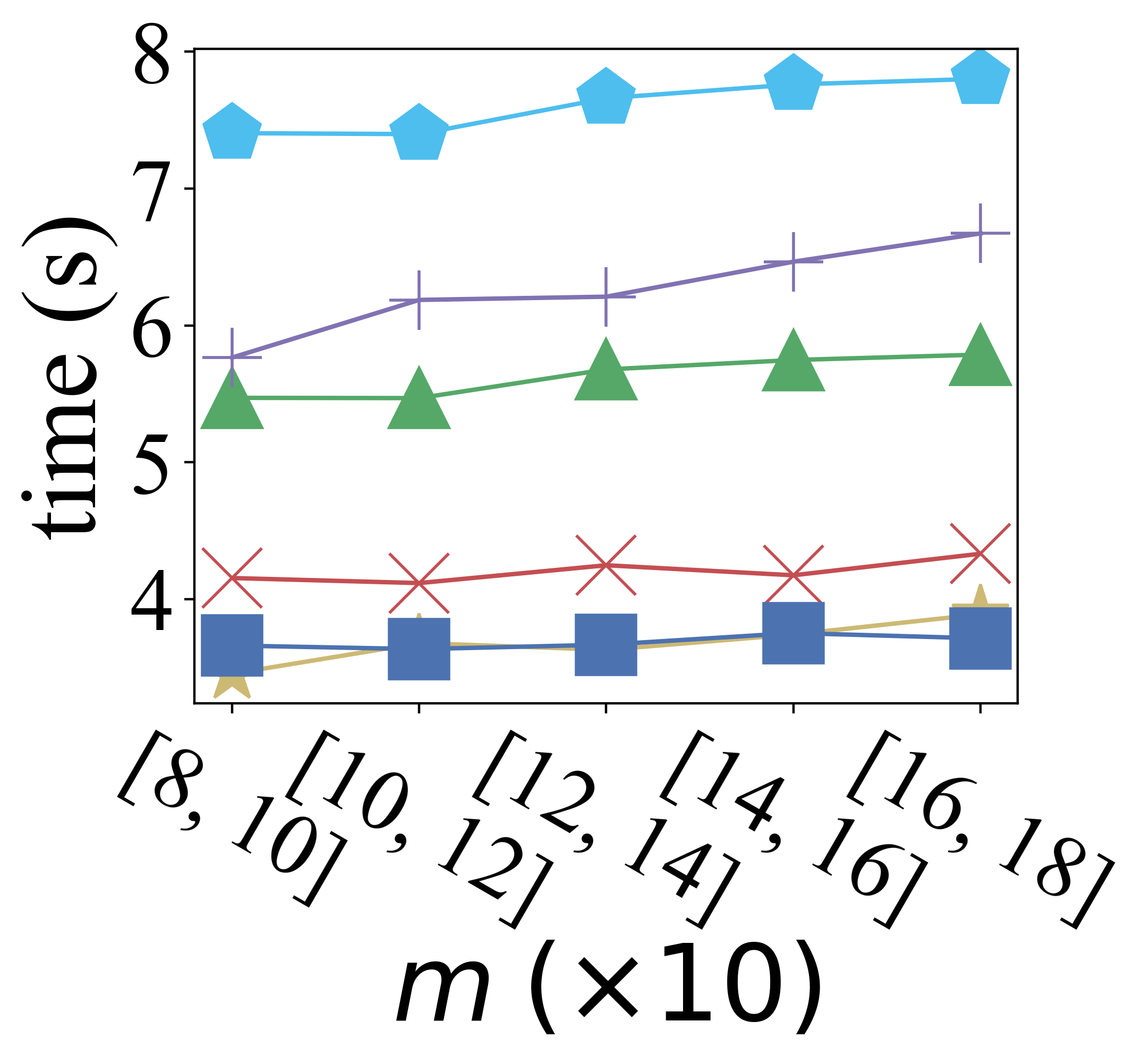}}
		\label{fig:xian_fc}}
	\subfigure[][{\scriptsize FD (Beijing)}]{
		\scalebox{0.19}[0.19]{\includegraphics{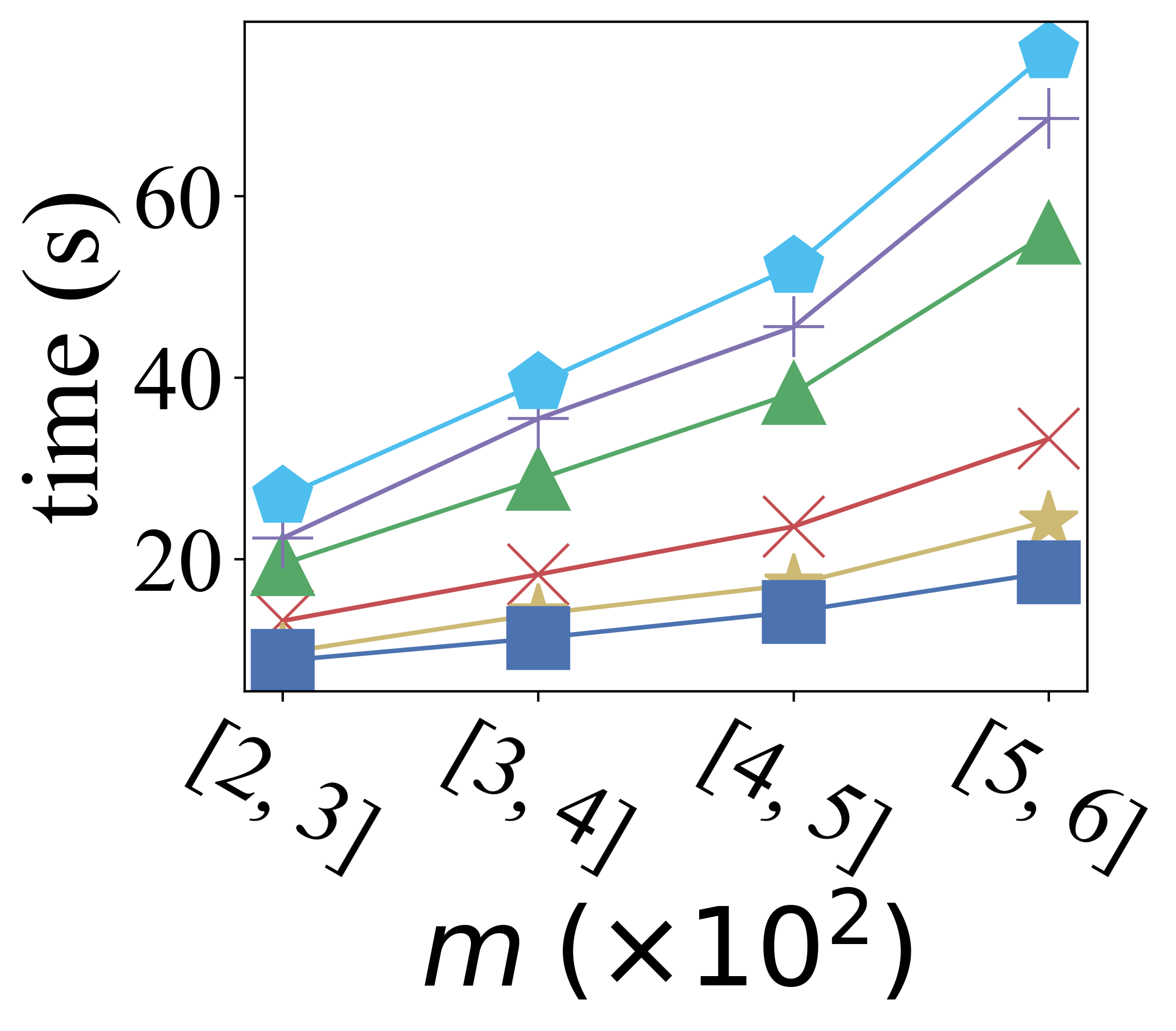}}
		\label{fig:beijing_fc}}
	\caption{\small Effectiveness and efficiency with varying query lengths}
	\label{fig:efficiencyLength}
\end{figure*}

Considering that there are a large number of data trajectories in the database, to improve the efficiency of searching the optimal subtrajectories, we use the pruning methods in the subsequent experiments to filter out the data trajectories that are different from the query trajectories. This paper proposes two modules to filter data trajectories that are not similar to the query trajectory: Filter with Key Points (FKP) and Grid-Based Pruning (GBP). They are compared with the SOTA pruning method, OSF~\cite{KoideXI20}. The details of the pruning methods and their experimental results can be found in the Appendixes B and C.

\noindent\textbf{Metrics.}
We will compare our CMA algorithm with the existing algorithms regarding efficiency and effectiveness. For a given query trajectory, we evaluate the efficiency of an algorithm in terms of the time to find the most similar subtrajectory from all data trajectories. We use four evaluation metrics identical to those used in previous work to evaluate the solutions found by different algorithms in this experiment: \revision{(1) Distance. It refers to the raw distance between a query trajectory and the optimal subtrajectory of the data trajectory found by the search algorithm. (2) Approximate Ratio (AR). Given a distance function, AR represents the ratio of the distance between the query trajectory and the subtrajectory found by an approximate algorithm to the distance between the query trajectory and the optimal solution. (3) Mean Rank (MR). It denotes the rank of the distance between the optimal subtrajectory found by the algorithm and the query trajectory among all subtrajectories of the original data trajectory. In particular, MR$=1$ indicates that the algorithm finds the optimal solution. (4) Relative Rank (RR). It is the percentage of all subtrajectories of the data trajectory that is better than the result returned by the algorithm.}

\noindent\textbf{Evaluation Platform.} 
The methods are implemented in C++14. The experiments are conducted on a Linux server equipped with 48-cores of Intel(R) Xeon(R) 2.20GHz CPUs and with 128.00 GB RAM.
\subsection{Experimental Results}
\label{sec:search}

\noindent\textbf{Effectiveness compared with other algorithms}
We used different algorithms for each distance function in different datasets to find the subtrajectories of the data trajectories with the smallest distance from the query trajectory. The experimental results are shown in Table \ref{tab:performance}. The approximation algorithms have substantial uncertainty in terms of effectiveness. Although the subtrajectory found by these approximation algorithms when using ERP as the distance function is close to the optimal subtrajectory, the subtrajectory found by approximate algorithms when using DTW as the distance function is far from the optimal subtrajectory. POS and PSS tend to select the trajectories with the same length as the query trajectory due to the higher cost of deleting a point when ERP is used as the distance function. In contrast, the length of the optimal subtrajectory tends to vary when DTW is used as the distance function. In addition, the subtrajectories found by the RLS and RLS-Skip algorithms learned based on reinforcement learning are also far from the optimal subtrajectories. CMA can find the exact optimal solution in all cases.

\noindent\textbf{Efficiency compared with other algorithms}
With the pruning algorithm, we can find the optimal subtrajectory from many data trajectories faster. Compared with ExactS, the efficiency of CMA has improved nearly 200 times on Xi'an datasets and nearly 50 times on the Porto dataset according to the Table \ref{tab:efficiencyT}. The longer the length of the trajectory, the more the improvement of CMA over ExactS. CMA can find the optimal subtrajectory relatively quickly regardless of the distance function. POS and RLS-Skip are the fastest, but they are approximate algorithms. \revision{The experimental results in Table \ref{tab:efficiencyT} indicate that CMA exhibits superior efficiency compared to other precise algorithms. Compared to CMA, Spring requires many additional computations. In addition to finding the optimal subtrajectory, Spring can identify all subtrajectories whose distances to the query trajectory are less than a given threshold (without overlaps between these subtrajectories). To achieve this, Spring continuously checks the DP matrix for subtrajectories that satisfy the criteria and outputs them, resulting in some additional computations. CMA, on the other hand, performs only one check after completing the calculation of the DP matrix. Spring is specifically designed for large-scale streaming data. The search space of GB is $O(mn)$. However, during the algorithm's execution, backtracking is required repeatedly, which can result in some nodes being searched multiple times. In contrast, each cell in the DP matrix of CMA is computed only once, making its efficiency slightly higher than that of GB.}

\begin{figure*}[t!]
	\centering\vspace{-1ex}
	\subfigure{
		\scalebox{0.40}[0.40]{\includegraphics{figures/legend.png}}}\hfill\\
	\addtocounter{subfigure}{-1}\vspace{-2ex}
	\subfigure[][{\scriptsize DTW (Beijing)}]{
		\scalebox{0.2}[0.2]{\includegraphics{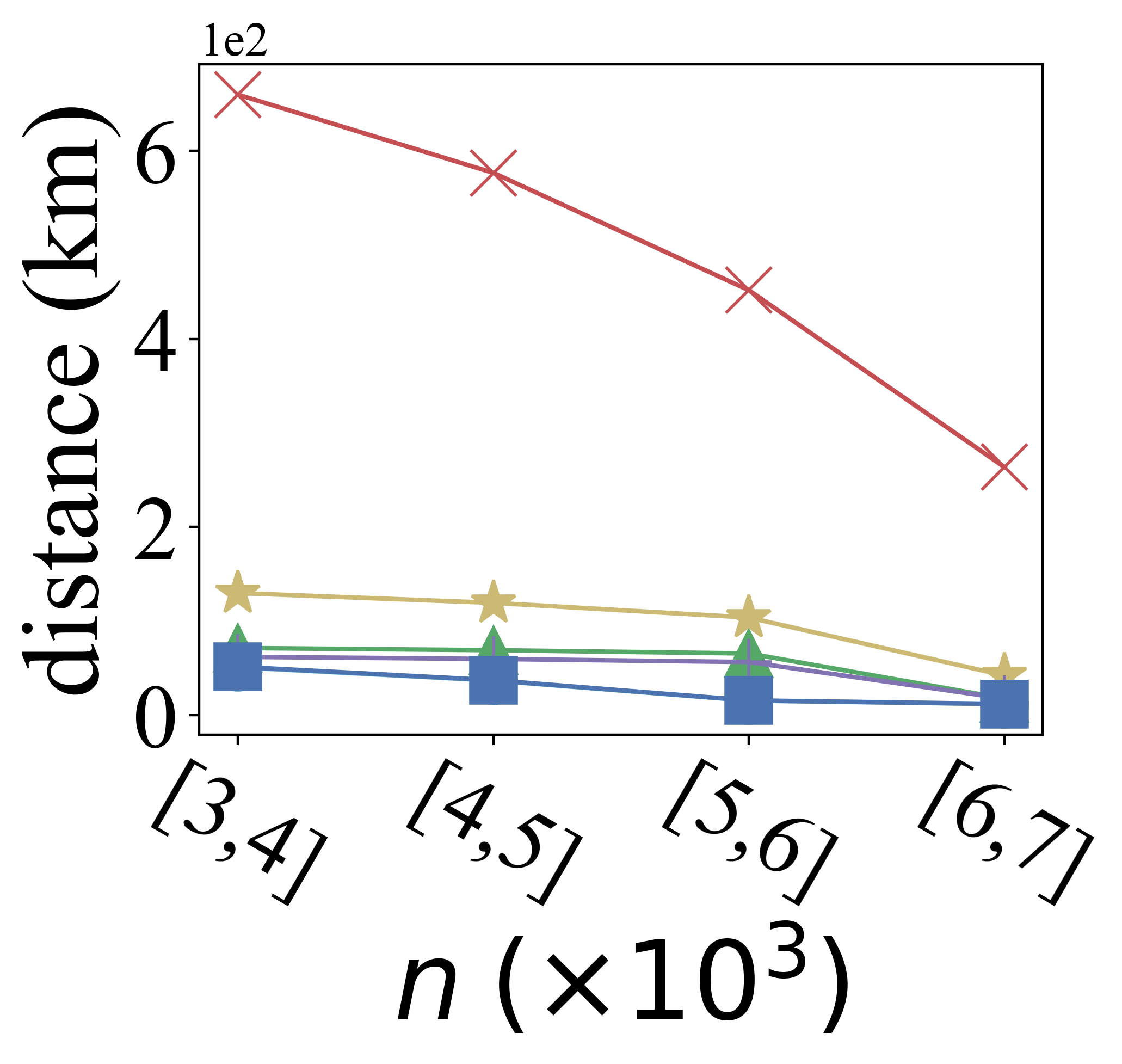}}
		\label{fig:dataLength_beijing_dtw_score}}
	\subfigure[][{\scriptsize EDR (Beijing)}]{
		\scalebox{0.2}[0.2]{\includegraphics{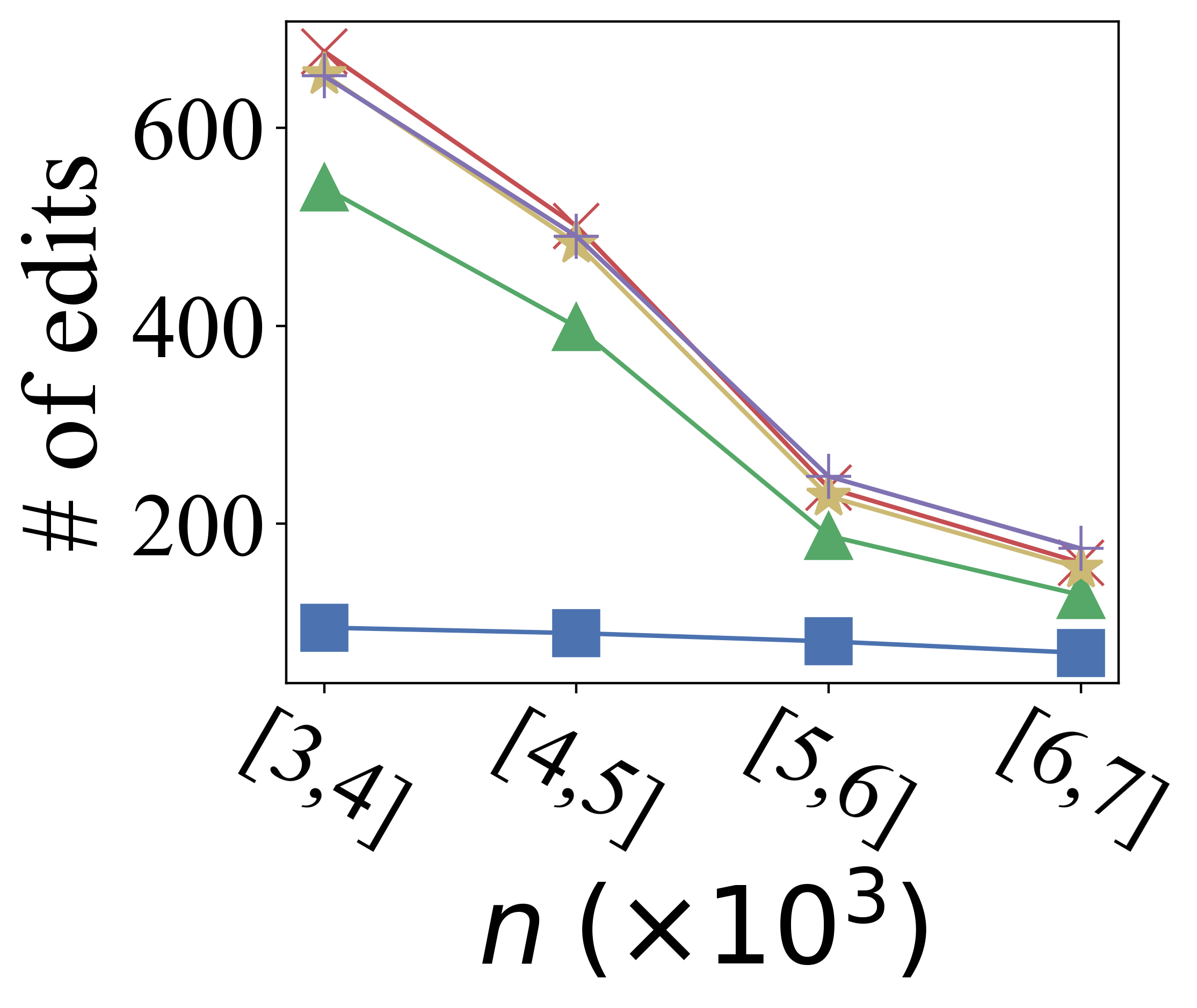}}
		\label{fig:dataLength_beijing_edr_score}}
	\subfigure[][{\scriptsize ERP (Beijing)}]{
		\scalebox{0.2}[0.2]{\includegraphics{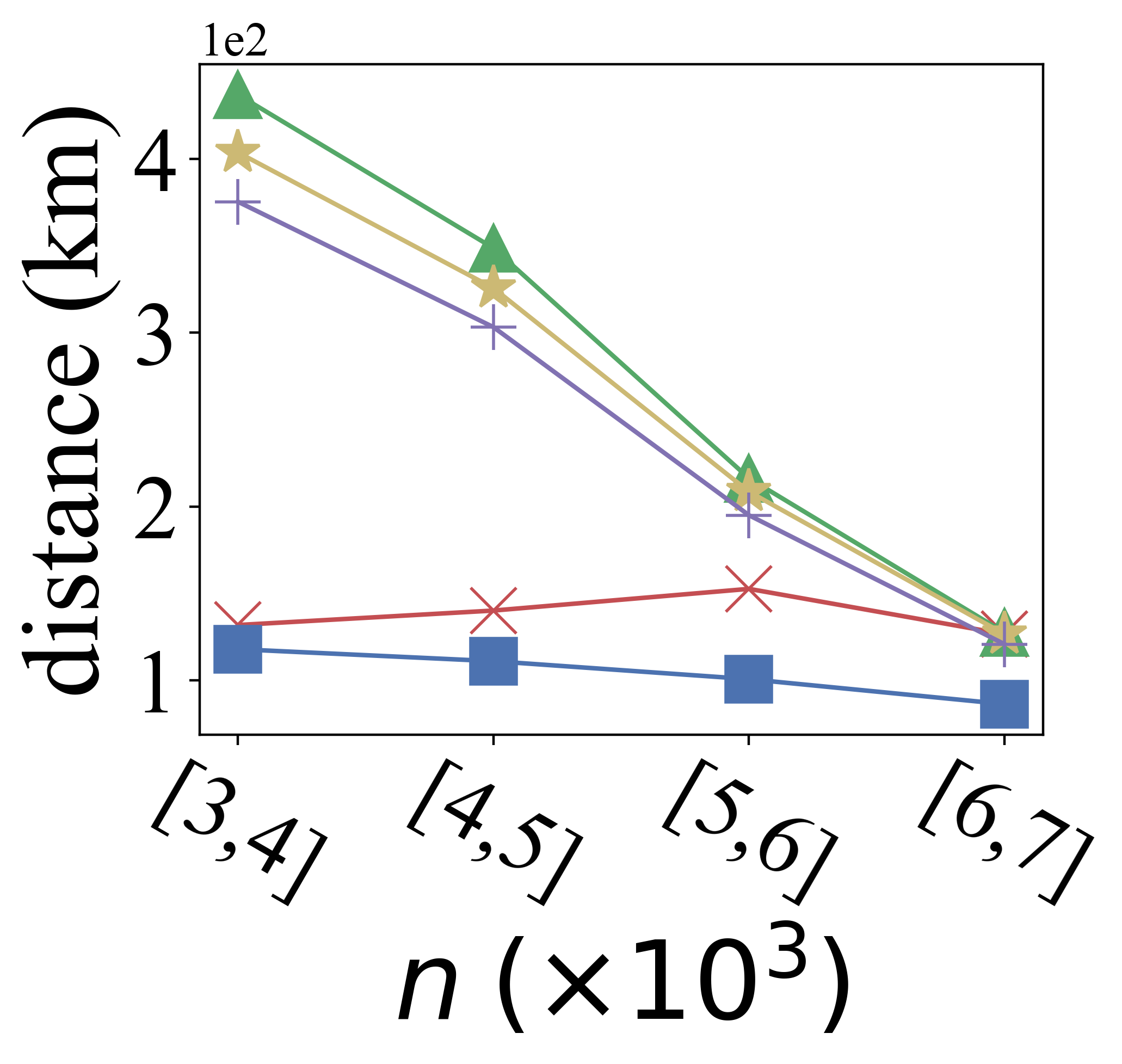}}
		\label{fig:dataLength_beijing_erp_score}}
	\subfigure[][{\scriptsize FD (Beijing)}]{
		\scalebox{0.2}[0.2]{\includegraphics{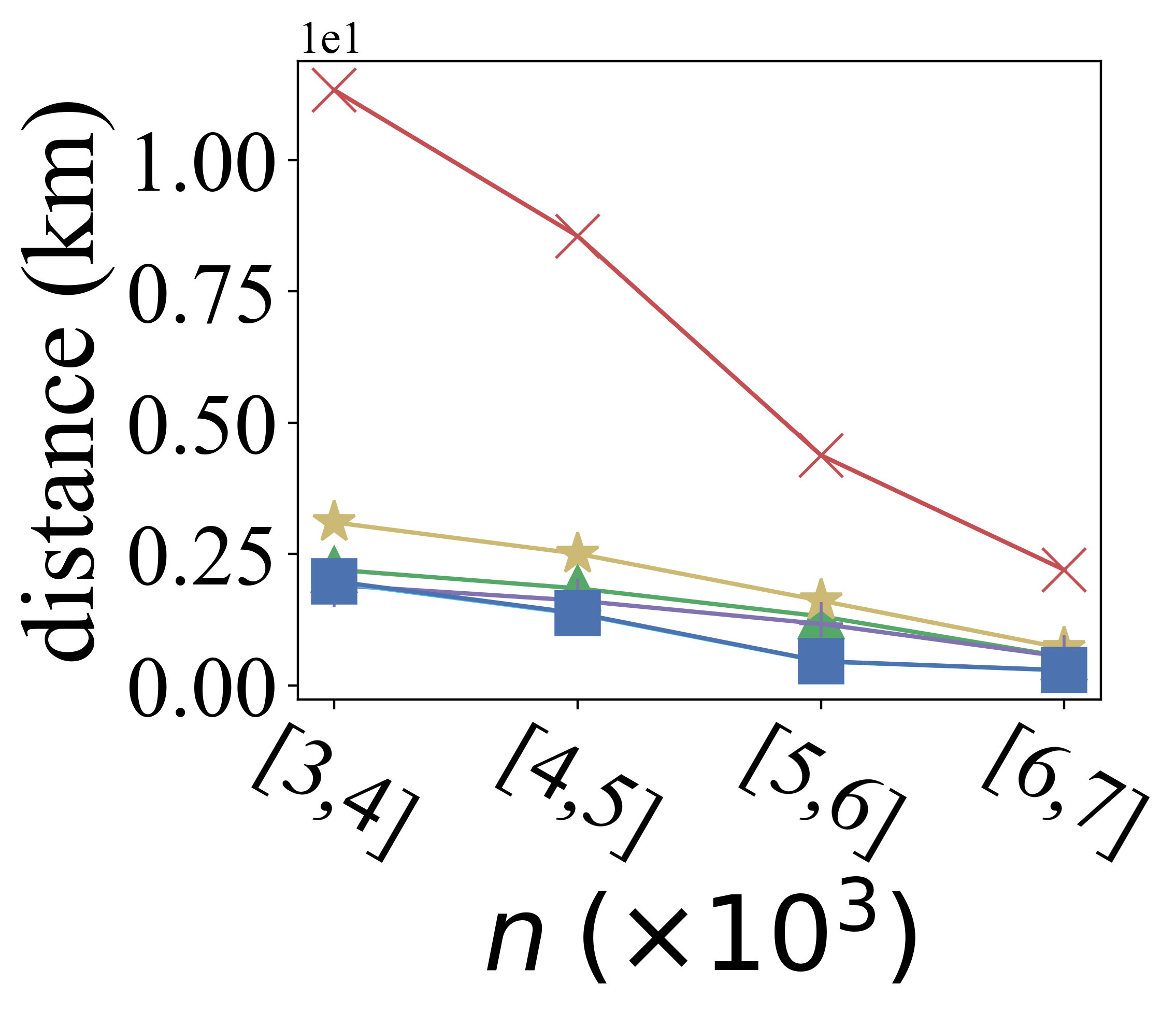}}
		\label{fig:dataLength_beijing_FC_score}}
	\\
	\subfigure[][{\scriptsize DTW (Beijing)}]{
		\scalebox{0.2}[0.2]{\includegraphics{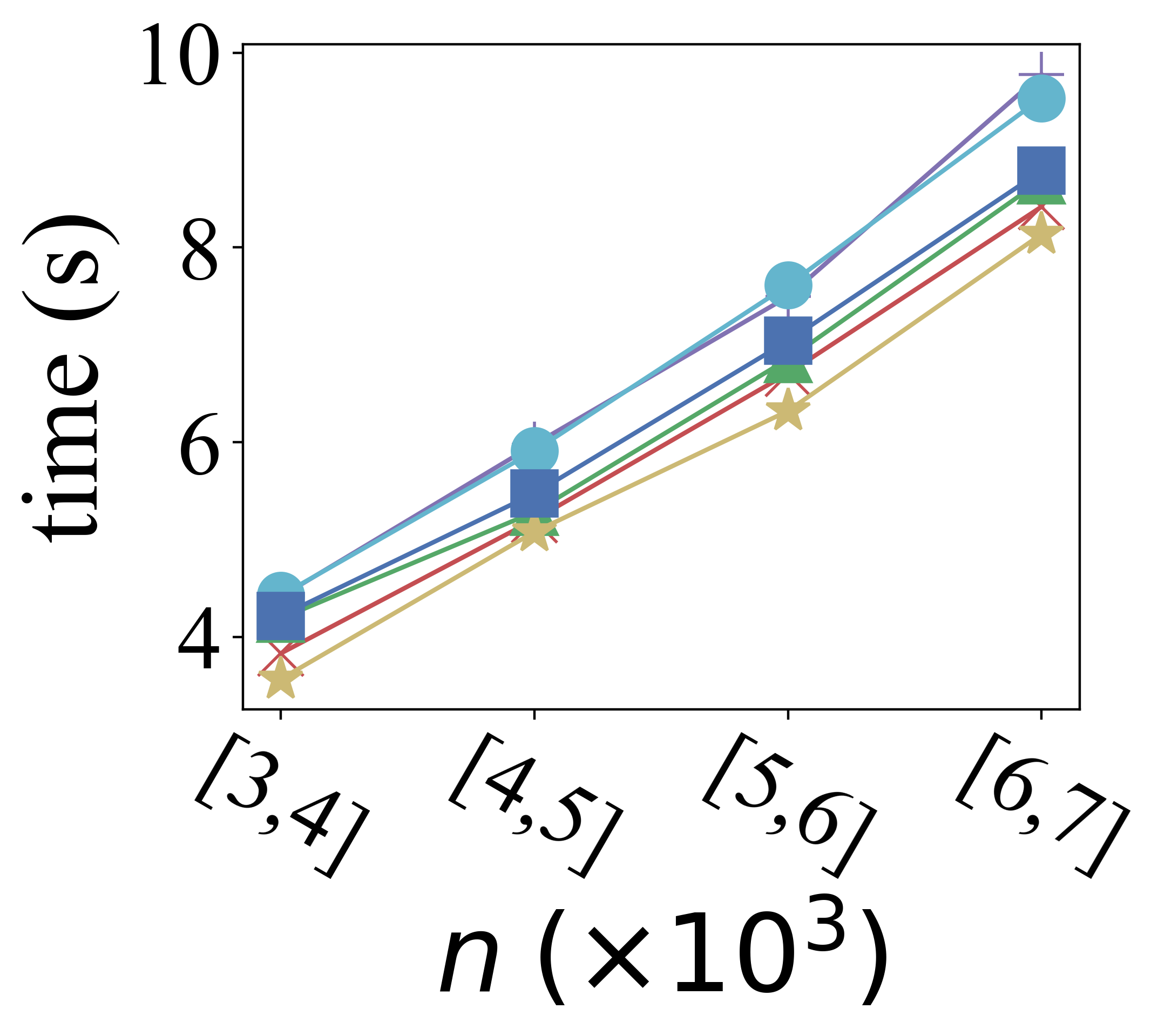}}
		\label{fig:dataLength_beijing_dtw}}
	\subfigure[][{\scriptsize EDR (Beijing)}]{
		\scalebox{0.2}[0.2]{\includegraphics{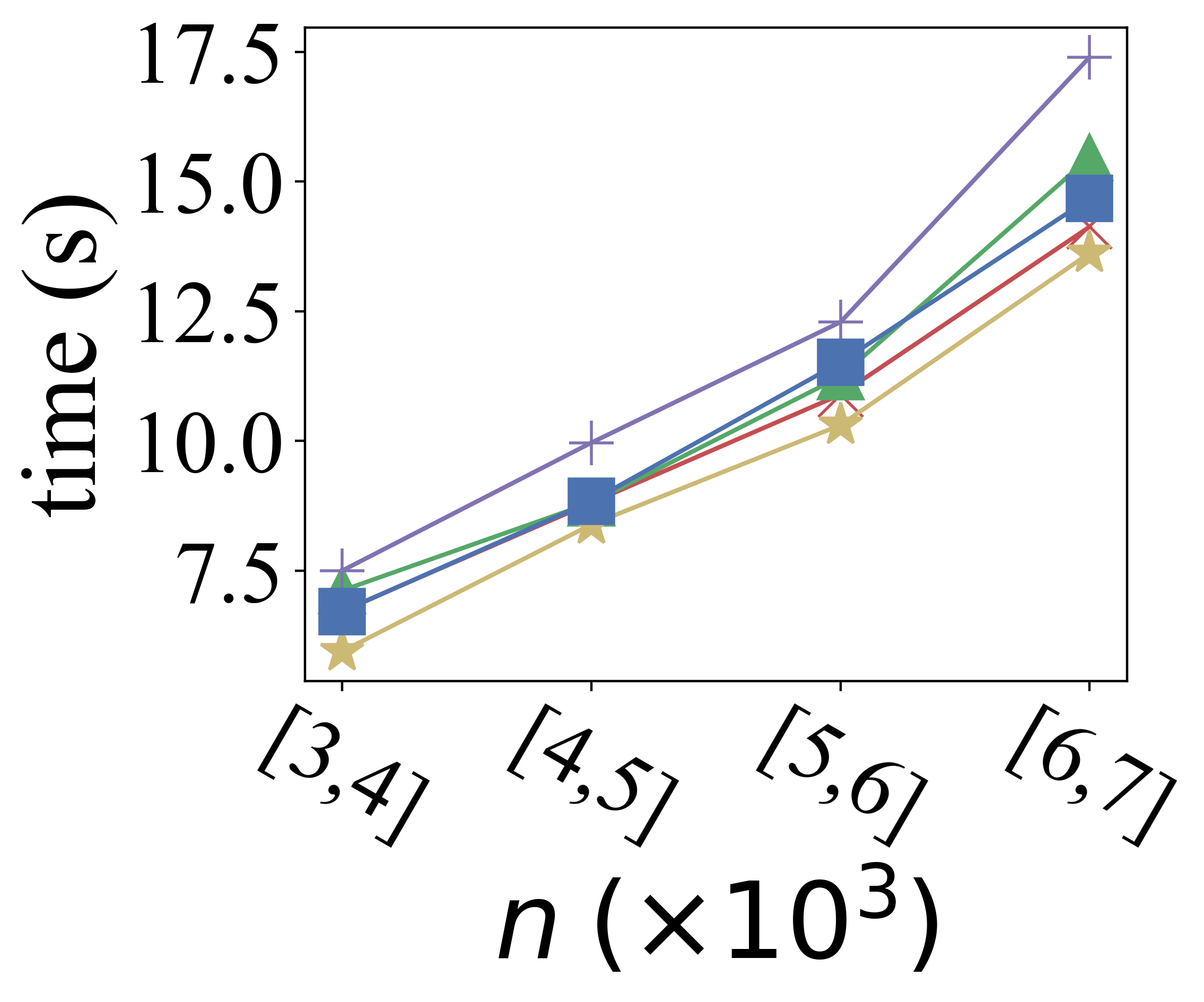}}
		\label{fig:dataLength_beijing_edr}}
	\subfigure[][{\scriptsize ERP (Beijing)}]{
		\scalebox{0.2}[0.2]{\includegraphics{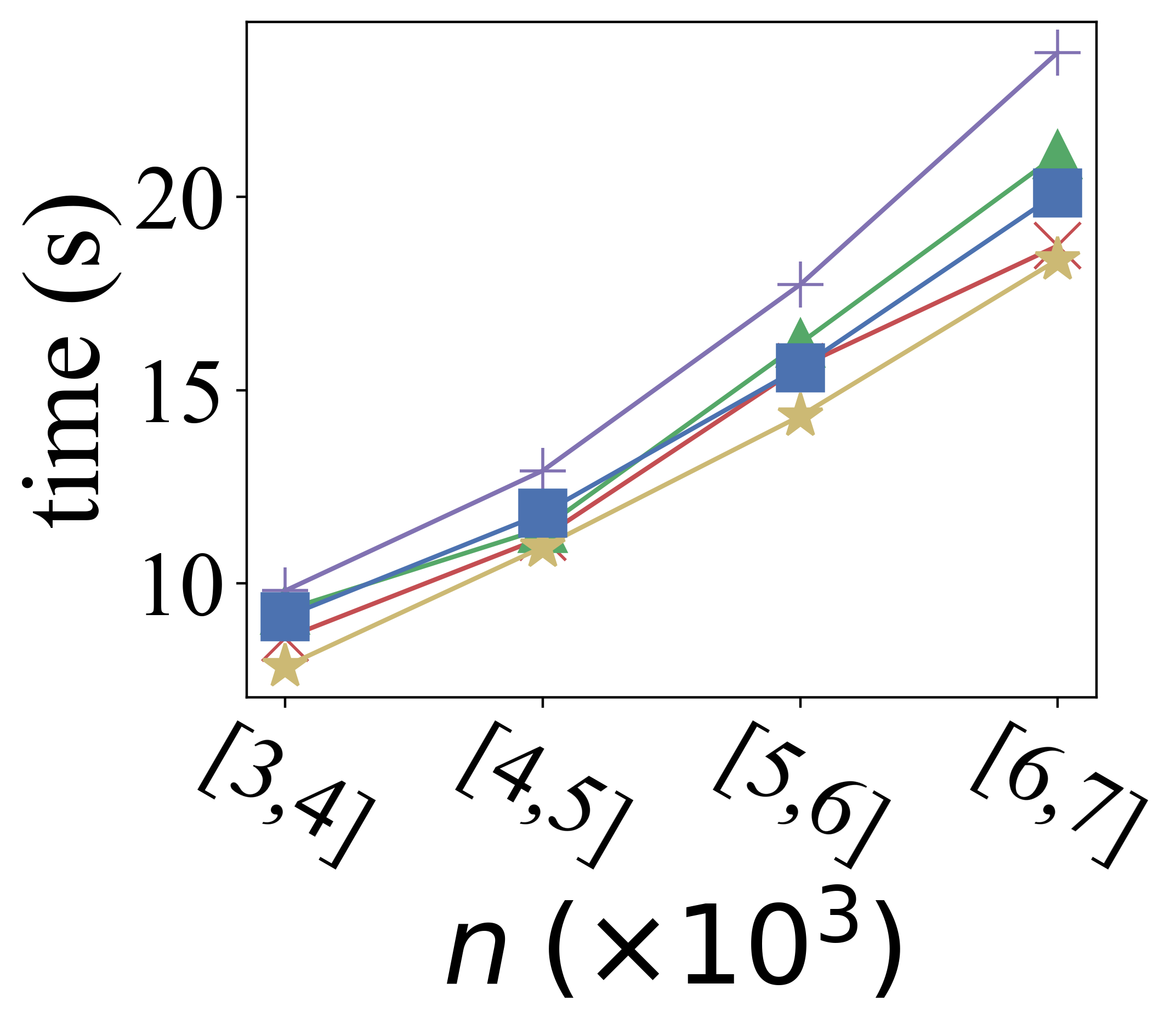}}
		\label{fig:dataLength_beijing_erp}}
	\subfigure[][{\scriptsize FD (Beijing)}]{
		\scalebox{0.2}[0.2]{\includegraphics{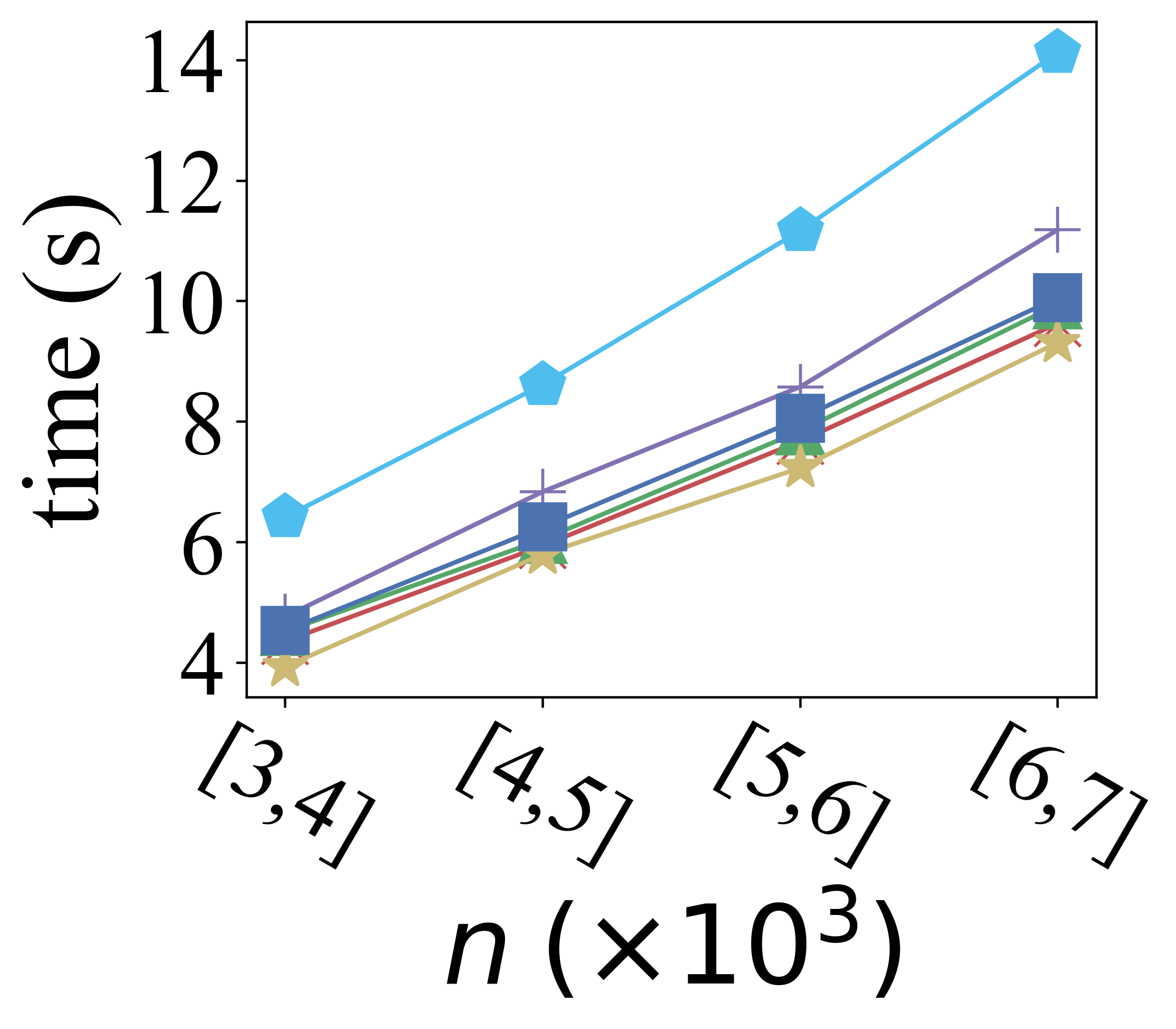}}
		\label{fig:dataLength_beijing_FC}}
	\caption{\small Effectiveness and efficiency with varying data lengths}
	\label{fig:dataLength_effectivenessLength}
\end{figure*}

\begin{table}[t!]
	\begin{center}
		{\small 
			\caption{\small Efficiency of Algorithms.} \label{tab:efficiencyT}
			\vspace{1ex}
			\begin{tabular}{c|c|c|c|c|c}
				\hline
				\multirow{2}{*}{Dataset}&\multirow{2}{*}{Algorithm}&\multicolumn{4}{c}{Time Cost (s)} \\\cline{3-6}
				&&DTW&EDR&ERP&FD\\ \hline\hline
				\multirow{8}{*}{Porto}&POS&16.32&17.75&16.91&18.42 \\ 
				&PSS&18.06&16.90&17.14&18.05 \\ 
				&RLS&17.84&19.87&19.39&19.62 \\ 
				&RLS-Skip&16.62&15.28&17.92&18.90 \\ 
				&CMA&18.78&14.64&19.26&18.78 \\ 
				&ExactS&7794.59&6731.42&7225.32&8334.16 \\ 
				&Spring&20.04&-&-&-\\ 
				&GB&-&-&-&29.01\\ 
				\hline
				\multirow{8}{*}{Xi'an}&POS&6.69&9.69&13.12&5.48 \\ 
				&PSS&8.03&12.47&16.21&7.12 \\ 
				&RLS&7.93&14.66&18.33&7.74 \\ 
				&RLS-Skip&5.79&9.30&13.45&4.91 \\ 
				&CMA&5.65&9.79&14.08&4.31 \\ 
				&ExactS&1625.58&2789.93&3429.52&1312.26 \\ 
				&Spring&7.37&-&-&-\\ 
				&GB&-&-&-&10.76\\ 
				\hline
				\multirow{8}{*}{Beijing}&POS&17.53&35.15&45.54&18.30 \\ 
				&PSS&26.95&58.79&79.31&28.81 \\ 
				&RLS&33.17&72.37&97.63&35.46 \\ 
				&RLS-Skip&13.35&36.94&48.57&13.94 \\ 
				&CMA&10.81&41.53&55.18&11.29 \\ 
				&ExactS&overtime&overtime&overtime&overtime \\ 
				&Spring&16.46&-&-&-\\ 
				&GB&-&-&-&75.86\\ 
				\hline
			\end{tabular}
		}
	\end{center}\vspace{-1ex}
\end{table}

\noindent\textbf{Effectiveness of the length of query trajectory.}
\revision{We use different length ranges of query trajectories for each dataset in our experiments. For the Beijing dataset, we select the length ranges $m$ as $[200, 300]$, $[300, 400]$, $[400, 500]$, and $[500, 600]$. For the Xi'an dataset, we choose the length ranges $[80, 100]$, $[100, 120]$, $[120, 140]$, $[140, 160]$, and $[160, 180]$. For the Porto dataset, we use the length ranges $[4, 8]$, $[8, 12]$, $[12, 16]$, and $[16, 20]$ for the query trajectories.}
Figure \ref{fig:efficiencyLength} shows that the execution time of the algorithm increases with the length of the trajectory regardless of the dataset and distance function because the search algorithm takes less time to find the optimal subtrajectory for each query trajectory when the trajectory length is small. However, in the dataset of Porto, the execution time increases and then decreases with the length of the query trajectory, which may be attributed to the fact that there are fewer trajectories similar to the query trajectory in the dataset when the size of the query trajectory becomes longer. Thus, most trajectories are screened out in the filtering phase, resulting in a decrease in the final search time. RLS takes more time than other algorithms in almost all cases. 
All algorithms except CMA have poor performance when DTW is used as the distance function, regardless of the length of the query trajectory. It is because that DTW allows different points in query trajectory matches the same point in data trajectory.
Furthermore, the effectiveness of the approximation algorithm is improved as the length of the query trajectory increases when EDR is used as the distance function. As the query trajectory length increases, the number of eligible data trajectories decreases, thus the approximation algorithm has a higher probability of finding the optimal solution. Both algorithms, RLS and RLS-Skip, also find much worse subtrajectories than CMA, where RLS has a worse execution time than CMA in almost all cases.

\begin{table*}[t!]
	\begin{center}\vspace{-2ex}
		{\small 
			\caption{\small Summary of subtrajectory similarity search algorithms.} \label{tab:complexities}
			\vspace{1ex}
			\begin{tabular}{c|c|c|c|c|c|c|c|c|c|c}
				\hline
				\multirow{2}{*}{ Algorithms} & \multirow{2}{*}{ Accuracy} & \multicolumn{7}{c|}{ order-insensitive}& \multicolumn{2}{c}{order-sensitive} \\ \cline{3-11}
				& &{\bf  DTW}&{\bf  ERP}&{\bf EDR} & \makecell[c]{\bf FD} & {\bf NetERP} & {\bf NetEDR} & {\bf SURS} & {\bf LCSS} & {\bf LCRS}\\\hline\hline
				\textbf{CMA (Ours)} &exact & $O(mn)$ & $O(mn)$ & $O(mn)$& $O(mn)$& $O(mn)$ & $O(mn)$ & $O(mn)$& -& -\\ \hline
				ExactS~\cite{WangLCL20}  &exact & $O(mn^2)$ & $O(mn^2)$ & $O(mn^2)$& $O(mn^2)$& $O(mn^2)$ & $O(mn^2)$ & $O(mn^2)$& $O(mn^2)$& $O(mn^2)$\\ \hline
				Spring~\cite{SakuraiFY07} &exact & $O(mn)$ & - & -& - & - & -& -& -& -\\ \hline
				\makecell[c]{Greedy
					Backtracking (GB)~\cite{abs-2203-10364}} &exact & - &-& -& $O(mn)$& - & -& -& -& - \\ \hline
				POS~\cite{WangLCL20} &approx. & $O(mn)$ & $O(mn)$ & $O(mn)$& $O(mn)$& $O(mn)$ & $O(mn)$ & $O(mn)$& $O(mn)$& $O(mn)$\\ \hline
				PSS~\cite{WangLCL20} &approx. & $O(mn)$ & $O(mn)$ & $O(mn)$& $O(mn)$ & $O(mn)$ & $O(mn)$ & $O(mn)$& $O(mn)$& $O(mn)$\\ \hline
				RLS~\cite{WangLCL20} &approx. & $O(mn)$ & $O(mn)$ & $O(mn)$& $O(mn)$ & $O(mn)$ & $O(mn)$ & $O(mn)$& $O(mn)$& $O(mn)$\\ \hline
				RLS-Skip~\cite{WangLCL20} &approx. & $O(mn)$ & $O(mn)$ & $O(mn)$& $O(mn)$& $O(mn)$ & $O(mn)$ & $O(mn)$& $O(mn)$& $O(mn)$\\ \hline
			\end{tabular}
		}
	\end{center}
\end{table*}

\revision{
\noindent\textbf{Effectiveness of the length of data trajectories}
We tested the efficiency and effectiveness of different distance functions on the Beijing city dataset by varying the length of data trajectories. In the experiment, we selected 1000 trajectories, each with lengths in the intervals [3000,4000], [4000,5000], [5000,6000], and [6000,7000], from all trajectories in Beijing city. The experimental results are presented in Figure \ref{fig:dataLength_effectivenessLength}. The figure shows that the time to find the optimal solution increases linearly with the length of the data trajectory for all algorithms. Additionally, the distance of the subtrajectories found by the CMA, Spring, and GB algorithms decreases as the length of the data trajectory increases, indicating that longer data trajectories are more likely to contain subtrajectories that are more similar to the query trajectory. We also observed that longer trajectory lengths make it easier for approximation algorithms to find better solutions. This means that there are more subtrajectories similar to the query trajectory with the increase of the length of data trajectories, which enables the approximation algorithms to find better solutions.
}

\noindent\textbf{Performance of Spring and GB}
In this paper, we also explore the performance of Spring and GB; the experimental results of Spring are shown in Figure \ref{fig:xian_dtw_score} $\sim$ \ref{fig:porto_dtw}, while the results of GB are shown in Figure \ref{fig:xian_fc_score} $\sim$ \ref{fig:porto_fc}. The experimental results show that the AR of Spring and GB is 1 in all cases, which means that both algorithms can find the optimal solution. However,   the execution time of Spring is similar to that of CMA, while GB is less efficient.

\noindent\textbf{Summary of Results}.
We verify that CMA can accurately find the nearest subtrajectory from the data trajectory to the query trajectory. Meanwhile, the execution time of CMA is about the same as the two approximation methods (i.e., PSS and POS), and is much smaller than ExactS. Therefore, the proposed algorithm can quickly and accurately find the subtrajectories of the closest data trajectory for each query trajectory.

%% file: relatedWork.tex
\section{Related Work}
\label{sec:related}
\noindent\textbf{Trajectory Distance Function.}
Many works have proposed metrics to measure the distance between two trajectories~\cite{YiJF98, ChenN04, ChenOO05, KoideXI20, VlachosGK02, 0007BCXLQ18, Yuan019, Xie14, AltG95}. We can divide these distance functions into two categories: order-insensitive and order-sensitive. The order-insensitive distance functions are independent of the position of the point in the trajectory; the order-sensitive distance functions are just the opposite. For example, the order-insensitive functions, DTW~\cite{YiJF98} and Fréchet distance (FD)~\cite{AltG95}, define the distance between trajectories as the cost of turning one trajectory into another through substitution operations. DTW allows different points in one trajectory to be mapped to the same point in another trajectory, enabling DTW to deal well with the case where two trajectories are sampled at different frequencies. Compared with DTW, edit distance with real penalty (ERP)~\cite{ChenN04} introduces  the insert and delete operations. The cost of inserting a point and deleting a point equals replacing it with a pre-defined default point. However, when the position of the default point is not set reasonably, the cost of deleting and inserting a point can be much greater than replacing it. Therefore, edit distance on real sequences (EDR)~\cite{ChenOO05} fixes this issue by introducing an upper bound. Specifically, when the distance between a point in the trajectory and its replacement is greater than this upper bound, the replacement cost equals the deletion cost. WED is a generic distance function that allows users to customize the cost of deletion, insertion, and replacement. The order-sensitive distance functions  (e.g., longest common subsequence (LCSS)~\cite{VlachosGK02}, longest overlapping road segments (LORS)~\cite{0007BCXLQ18}, and longest common road segments (LCRS)~\cite{Yuan019}) calculate the distance of a point in a trajectory from another trajectory considering the point positions in the trajectories.

\noindent\textbf{Subtrajectory Search.}
The previous work~\cite{KoideXI20} divides the subtrajectory search into two stages: filtering and verification. In the filtering phase, most of the trajectories whose distance from the query trajectory exceeds a given threshold are filtered out to reduce the number of validations \cite{faloutsos1994fast, WangLCL20}; in the validation phase, the execution time of the validation phase is simplified with the help of indexes. Unfortunately, this work invokes the trajectory distance function calculation method for all candidate subtrajectories within a trajectory during the validation phase, which makes the validation phase take much time. Another work~\cite{WangLCL20} focuses on how to find the subtrajectory with the minimum distance from the query trajectory in the data trajectory given a query trajectory of length $m$ and a data trajectory of length $n$. ExactS~\cite{WangLCL20} is proposed to find the optimal subtrajectory in time complexity of {\small$O(mn^2)$}. Meanwhile, this work also proposes approximate algorithms (e.g., POS and PSS~\cite{WangLCL20}) with $O(mn)$ time complexity. In addition to these traditional methods, this work proposes two reinforcement learning-based approximate methods (RLS, RLS-Skip~\cite{WangLCL20}) to find the optimal subtrajectory. Furthermore, RLS and RLS-Skip can adaptively select appropriate split points to improve the efficiency of the search. With DTW as the distance function, Spring~\cite{SakuraiFY07} can find the optimal subtrajectory exactly in $O(mn)$ time complexity. Besides, GB~\cite{abs-2203-10364} can find the exact optimal similar subtrajectory with $O(mn)$ time complexity on FD. However, Spring and GB do not apply to other distance functions.  In contrast, our CMA can be applied to most order-insensitive distance functions.  Table \ref{tab:complexities} summarizes the existing subtrajectory search methods. \revision{Due to space limitation, experiments on NetERP, NetEDR and SURS can be found in Appendix D.}

\noindent\textbf{Applications of Subtrajectory search.}
Some previous studies~\cite{WangZX14, WauryJKIX19} implement the travel time estimation of a segment of the trajectory by a similar subtrajectory search. One specific process is to search the most similar subtrajectory from the database and then use its time as an estimate of the current trajectory's communication time. The advantage of subtrajectory search is that it can solve the sparsity of trajectories in the database and thus find more similar trajectories. Another common application is to analyze the movement and behavioral performance of players on the sports ground through subtrajectory search~\cite{WangLCJ19}. In addition, subtrajectory search can be used to count the frequency of a given road section in the database for better road planning~\cite{ChenSZ11, KoideTYXI18, LuoT0N13}.

%% file: conclusion.tex
\section{Conclusion}
\label{sec:conclusion}
\revision{This paper focuses on a similar subtrajectory search problem, i.e., finding the subtrajectory of the data trajectory with the minimum distance for the query trajectory. We convert the problem of finding the optimal subtrajectory to finding the optimal matching sequence. For a given query trajectory of length $m$ and a data trajectory of length $n$, we propose the CMA algorithm to find the subtrajectory with the minimum distance to the query trajectory from the data trajectory in the time complexity of $O(mn)$. Finally, we conduct sufficient experiments on the datasets of Xi'an, Beijing and Porto, and the experimental results show that our CMA algorithm can find efficiently the exact optimal subtrajectory for each query trajectory.}

%% file: ack.tex
\section{acknowledgment}
\label{sec:ack}

Peng Cheng's work is supported by the National Natural Science Foundation of China under Grant No. 62102149 and Open Foundation of Key Laboratory of Transport Industry of Big Data Application Technologies for Comprehensive Transport. Lei Chen’s work is supported by National Key Research and Development Program of China (2022YFE0200500), National Science Foundation of China (NSFC) under Grant No. U22B2060, the Hong Kong RGC GRF Project 16209519, CRF Project C6030-18G, C2004-21GF, AOE Project AoE/E-603/18, RIF Project R6020-19, Theme-based project TRS T41-603/20R, China NSFC No. 61729201, Guangdong Basic and Applied Basic Research Foundation 2019B151530001, Hong Kong ITC ITF grants MHX/078/21 and PRP/004/22FX, Microsoft Research Asia Collaborative Research Grant, HKUST-Webank joint research lab grant and HKUST Global Strategic Partnership Fund (2021 SJTU-HKUST). Xuemin Lin's work is supported by NSFC U2241211 and U20B2046. Wenjie Zhang's work is supported by the Australian Research Council FT210100303 and DP230101445. Corresponding author: Peng Cheng.

%% file: appendix.tex
\appendix

\centerline{\large APPENDIX}

\section{Proof of Equivalence of Distance Functions}
In this section, we will demonstrate that $wed$ and $dtw$ are special cases of the general distance function $\Theta$.
\begin{theorem}[WED]
	$wed$ is a specific case of $\Theta$, that is, $wed(\tau_q, \tau_d) = \Theta(\tau_q,\tau_d)$, when
	{\scriptsize\begin{eqnarray}
			\left\{
			\begin{array}{ll}
				Cost_{del}(\tau_q[i],$ $ \tau_d[j])=del(\tau_q[i]) \\
				Cost_{sub}(\tau_q[i], \tau_d[j])=sub(\tau_q[i],\tau_d[j])\\
				Cost_{ins(k)}(\tau_q[i], \tau_d[j])=ins(\tau_d[k+1:j-1])+sub(\tau_q[i],\tau_d[j])\\
			\end{array}
			\right. 
			\label{eq:cost_wed}
	\end{eqnarray}}
\end{theorem}
\begin{proof}
	We will prove {\scriptsize$wed(\tau_q[1:m], \tau_d[1:n]) \leq \Theta(\tau_q,\tau_d)$} and {\scriptsize$wed(\tau_q[1:m], \tau_d[1:n]) \geq \Theta(\tau_q,\tau_d)$} to establish {\scriptsize$wed(\tau_q, \tau_d) = \Theta(\tau_q,\tau_d)$}.
	
	For any given $\mathcal{A}{\tau_q:\tau_d}'$, we can recursively define the function $wed'(\tau_q[1:m],\tau_d[1:n])$ as follows:
	
	{\scriptsize\begin{eqnarray}
			&&wed'(\tau_q[1:i], \tau_d[1:j]) \\ \notag
			&=&\left\{
			\begin{array}{ll}
				wed'(\tau_q[1:i-1], \tau_d[1:j]) + del(\tau_q[i]), &a_i=j, a_{i-1}=j\\
				wed'(\tau_q[1:i-1], \tau_d[1:j-1]) + sub(\tau_q[i],\tau_d[j]),\\
				&a_i=j, a_{i-1}<j\\
				wed'(\tau_q[1:i], \tau_d[1:j-1]) + ins(\tau_d[j]), &otherwise\\
			\end{array}
			\right. 
			\label{eq:wed_fake}
	\end{eqnarray}}
	
	Meanwhile, we define {\scriptsize$wed(\tau_q[i:j], \tau_{\emptyset})=del(\tau_q[i:j])=\sum_{k=i}^{j}{del(\tau_q[k])}$} and {\scriptsize$wed(\tau_{\emptyset}, \tau_d[i:j])=ins(\tau_d[i:j])=\sum_{k=i}^{j}{ins(\tau_d[k])}$}. We can obtain the inequality {\scriptsize$wed'(\tau_q[1:m], \tau_d[1:n])\geq wed(\tau_q[1:m], \tau_d[1:n])$} by mathematical induction as follows.
	{\scriptsize
		\begin{eqnarray}
			&&wed'(\tau_q[1:i], \tau_d[1:j]) \\\notag
			&\geq &\min \{wed'(\tau_q[1:i-1], \tau_d[1:j-1]) + sub(\tau_q[i],\tau_d[j]), \\\notag
			&& wed'(\tau_q[1:i-1], \tau_d[1:j]) + del(\tau_q[i]),\\\notag
			&& wed'(\tau_q[1:i], \tau_d[1:j-1]) + ins(\tau_d[j])\}\\\notag
			&\geq &\min \{wed(\tau_q[1:i-1], \tau_d[1:j-1]) + sub(\tau_q[i],\tau_d[j]), \\\notag
			&& wed(\tau_q[1:i-1], \tau_d[1:j]) + del(\tau_q[i]),\\\notag
			&& wed(\tau_q[1:i], \tau_d[1:j-1]) + ins(\tau_d[j])\}\\\notag
			&=&wed(\tau_q[1:i], \tau_d[1:j])
		\end{eqnarray}
	}
	Next, we will prove that {\scriptsize$wed'(\tau_q[1:i], \tau_d[1:a_i]) = wed'(\tau_q[1:i-1], \tau_d[1:a_{i-1}]) + \text{Cost}(\tau_q[1:i], \tau_d[1:a_i])$ $(2\leq i)$}. We will discuss different cases:
	\begin{enumerate}
		\item When $a_i=j$ and $a_{i-1}=j$, we have {\scriptsize$wed'(\tau_q[1:i], \tau_d[1:a_i])=wed'(\tau_q[1:i-1], \tau_d[1:a_{i-1}])+Cost_{del}(\tau_q[1:i], \tau_d[1:a_i])$}.
		\item When $a_i=j$ and $a_{i-1}=j-1$, we have {\scriptsize$wed'(\tau_q[1:i], \tau_d[1:a_i])=wed'(\tau_q[1:i-1], \tau_d[1:a_{i-1}])+Cost_{sub}(\tau_q[1:i], \tau_d[1:a_i])$}.
		\item When $a_i=j$ and $a_{i-1}<j-1$, we have
		{\scriptsize
			\begin{eqnarray}
				&&wed'(\tau_q[1:i], \tau_d[1:a_i]) \\ \notag
				&=&wed'(\tau_q[1:i-1], \tau_d[1:a_i-1])+sub(\tau_q[i], \tau_d[a_i]) \\ \notag
				&=&wed'(\tau_q[1:i-1],  \tau_d[1:a_{i-1}])+\sum_{k=a_{i-1}+1}^{a_{i}-1}ins(\tau_d[k])+sub(\tau_q[i], \tau_d[a_i]) \\ \notag
				&=&wed'(\tau_q[1:i-1], \tau_d[1:a_{i-1}])+Cost_{ins(a_{i-1})}(\tau_q[i], \tau_d[a_i])
			\end{eqnarray}
		}
	\end{enumerate}
	
	Therefore, we can establish the following equation holds.
	{\scriptsize
		\begin{eqnarray}
			&&wed'(\tau_q[1:m], \tau_d[1:n]) \\\notag
			&=&wed'(\tau_q[1:m], \tau_d[1:a_m])+Insert(\tau_d[a_m+1:n])\\\notag
			&=&wed'(\tau_q[1:1], \tau_d[1:a_1])+\sum_{i=2}^{m}Cost(\tau_q[i], \tau_d[a_i])+Insert(\tau_d[a_m+1:n])\\ \label{eq:pr3}
			&=&wed'(\tau_q[1:0], \tau_d[1:a_1-1])+Cost(\tau_q[1], \tau_d[a_1])\\\notag
			&&+\sum_{i=2}^{m}Cost(\tau_q[i], \tau_d[a_i])+Insert(\tau_d[a_m+1:n])\\\label{eq:pr4}
			&=&Insert(\tau_d[1:a_1-1])+\sum_{i=1}^{m}Cost(\tau_q[i], \tau_d[a_i])+Insert(\tau_d[a_m+1:n])\\\notag
		\end{eqnarray}
	}
	
	We can derive Equation \ref{eq:pr3} from Equation \ref{eq:pr4} mainly because $\tau_q[1]$ will only be deleted or substituted by $\tau_d[a_1]$. Therefore, for any $\mathcal{A}{\tau_q:\tau_d}'$, we have {\scriptsize$wed(\tau_q[1:m], \tau_d[1:n]) \leq Insert(\tau_d[1:a_1]) + \sum_{i=1}^{m} Cost(\tau_q[i], \tau_d[a_i]) + Insert(\tau_d[a_m+1:n])$}. Thus, we can conclude that {\scriptsize$wed(\tau_q[1:m], \tau_d[1:n]) \leq \Theta(\tau_q,\tau_d)$}. Next, we will prove that {\scriptsize$wed(\tau_q[1:m], \tau_d[1:n]) \geq \Theta(\tau_q,\tau_d)$}.
	
	The equation \ref{eq:raw_wed} defines the recursive computation process of WED. Given the dynamic programming process, we can construct a matching sequence $\mathcal{A}{\tau_q:\tau_d}$. We start the recursion with the computation of $wed(\tau_q[1:m],\tau_d[1:n])$. Whenever {\scriptsize$wed(\tau_q[1:i], \tau_d[1:j]) = wed(\tau_q[1:i-1], \tau_d[1:j-1]) + sub(\tau_q[i],\tau_d[j])$}, we assign a value to $a_i$ and denote $a_i = j$. Thus, we obtain a sequence with assigned values $\mathbb{k} = \{a_{k_1}, a_{k_2}, \dots, a_{k_t}\}$. Given two adjacent assigned values $a_{k_u}$ and $a_{k_{u+1}}$, we need to delete $\tau_q[k_u+1:k_{u+1}-1]$ and insert $\tau_d[a_{k_u}+1:a_{k_{u+1}}-1]$. Therefore, we let $a_v = a_{k_u}$ for $(k_u+1 \leq v \leq k_{u+1}-1)$. We let $a_v = a_{k_1}$ for $1 \leq v \leq k_1$ and let $a_v = a_t$ for $a_t+1 \leq v \leq m$.
	
	After obtaining the constructed matching sequence $\mathcal{A}{\tau_q:\tau_d}$, we need to prove that the matching cost of this sequence is equivalent to $wed(\tau_q[1:m], \tau_d[1:n])$. Assuming $\tau_q[a_{k_u}]$ and $\tau_q[a_{k_{u+1}}]$ are adjacent subsituted points in the dynamic programming process, we have {\scriptsize$wed(\tau_q[1:{k_{u+1}}], \tau_d[1:a_{k_{u+1}}]) = wed(\tau_q[1:{k_u}], \tau_d[1:a_{k_u}]) + del(\tau_q[{k_u}+1:k_{u+1}-1]) + ins(\tau_d[a_{k_u}+1:a_{k_{u+1}}-1])$}.
	
	Therefore, we have 
	{\scriptsize
		\begin{eqnarray}
			&&wed(\tau_q[1:m],\tau_d[1:n]) \\\notag
			&=& del(\tau_q[1:{k_1}-1]) + del(\tau_q[{k_t}+1:m]) +  ins(\tau_d[1:a_{k_1}-1])\\ \notag 
			&& + ins(\tau_d[a_{k_t}+1:n]) + sub(\tau_q[k_1],\tau_d[k])+ \\\notag
			&&\sum_{a_{k_u}, a_{k_{u+1}}\in \mathbb{k}}del(\tau_q[{k_u}+1:k_{u+1}-1]) + \\ \notag
			&& ins(\tau_d[a_{k_u}+1:a_{k_{u+1}}-1])+sub(\tau_q[k_{u+1}], \tau_d[a_{k_{u+1}}])  \\ \notag 
			&=& \sum_{i=1}^{k_1-1}Cost_{del}(\tau_q[i])+\sum_{i={k_t}+1}^{m}Cost_{del}(\tau_q[i]) \\ \notag
			&& + Insert(\tau_d[1:a_1-1])+Insert(\tau_d[a_m+1:n]) + \\\notag
			&& \sum_{a_{k_u}, a_{k_{u+1}}\in \mathbb{k}}{\sum_{i=k_u+1}^{k_{u+1}-1}Cost_{del}(\tau_q[i]) + Cost_{ins(a_{k_u})}(\tau_q[k_{u+1}], \tau_d[[a_{k_{u+1}}])} \\ \notag
			&=& \sum_{a_i \in \mathcal{A}_{\tau_q:\tau_d}}{Cost(\tau_q[i], \tau_d[a_i])} + Insert(\tau_d[1:a_1-1])+Insert(\tau_d[a_m+1:n]) \\ \notag
			&\geq&  \Theta(\tau_q:\tau_d)
		\end{eqnarray}
	}
	
	Therefore, we can conclude that $wed(\tau_q[1:m], \tau_d[1:n]) \geq \Theta(\tau_q,\tau_d)$. Furthermore, combining this result with our previous conclusion, we can state that $wed(\tau_q[1:m], \tau_d[1:n]) = \Theta(\tau_q,\tau_d)$.
	
\end{proof}

Next, we will demonstrate that $dtw$ is also a special case of $\Theta$. Equation \ref{eq:raw_dtw} provides a detailed calculation process for $dtw$. During the recursive process, each invocation of Equation \ref{eq:raw_dtw} considers establishing an edge between $\tau_q[i]$ and $\tau_d[j]$. We denote the set of all edges as $E$. We define the weight of an edge $e$ as $e.w = sub(\tau_q[i], \tau_d[j])$ and thus, we have $dtw(\tau_q, \tau_d) = \sum_{e \in E} e.w$. According to the definition of $dtw$, each point has at least one edge. We define the points with multiple edges as \textit{multi-points}. We will now prove that there are no edges between multi-points.

\begin{lemma}
	\label{lem:only}
	Given the set of edges $E$ obtained from the recursive process of $dtw$, there are no edges between multi-points.
\end{lemma}
\begin{proof}
	Suppose there exist two multi-points $\tau_q[i]$ and $\tau_d[j]$ with an adjacent edge $e \in E$ connecting them. Since $\tau_q[i]$ is a multi-point, we can assume that $\tau_q[i]$ is connected to $\tau_d[j-1]$. Furthermore, as $\tau_d[j]$ is also a multi-point, there exists an edge between $\tau_d[j]$ and $\tau_q[i+1]$. We can construct a set $E' = E \setminus \{e\}$ such that $dtw'(\tau_q, \tau_d) < dtw(\tau_q, \tau_d)$, which contradicts the definition of $dtw$.
\end{proof}

Lemma \ref{lem:only} demonstrates a property of $dtw$. Although there are no edges connecting multi-points to each other, it is possible to have points $\tau_q[i]$ and $\tau_d[j]$ connected by an edge, where neither $\tau_q[i]$ nor $\tau_d[j]$ are multi-points. Therefore, we can classify the edges in set $E$ into three categories:
\begin{enumerate}
	\item Connecting two non-multi-points $\tau_q[i]$ and $\tau_d[j]$. \label{cond:1}
	\item Connecting a multi-point $\tau_q[k]$ in the query trajectory with multiple consecutive points $\tau_d[i:j]$ in the data trajectory. \label{cond:2}
	\item Connecting a multi-point $\tau_d[k]$ in the data trajectory with multiple consecutive points $\tau_q[i:j]$ in the query trajectory. \label{cond:3}
\end{enumerate}

We will place $\tau_q[i]$ from \ref{cond:1}, $\tau_q[k]$ from \ref{cond:2}, and $\tau_d[k]$ from \ref{cond:3} into a set $G$. Every $\tau_q[k] \in G$ in the query trajectory is connected to one or more points $\tau_d[i:j]$ ($j \geq i$). Similarly, every point $\tau_d[k] \in G$ in the data trajectory is connected to multiple points $\tau_q[i:j]$ ($j > i$). Thus, we have {\scriptsize$dtw(\tau_q, \tau_d) = \sum_{\tau_q[k] \in G} \sum_{t=i}^{j} sub(\tau_q[k], \tau_d[t]) + \sum_{\tau_d[k] \in G} \sum_{t=i}^{j} sub(\tau_q[t], \tau_d[k])$}. Based on these definitions and Lemma \ref{lem:only}, we will now prove that Theorem \ref{eq:q_dtw} holds true.

\begin{theorem}[DTW]
	\label{eq:q_dtw}
	$dtw$ is a specific case of $\Theta$, that is, $dtw(\tau_q, \tau_d)$ $= \Theta(\tau_q,\tau_d)$, when
	{\scriptsize\begin{eqnarray}
			\left\{
			\begin{array}{ll}
				Cost_{del}(\tau_q[i],$ $ \tau_d[j])=&sub(\tau_q[i],\tau_d[j]) \\
				Cost_{sub}(\tau_q[i], \tau_d[j])=&sub(\tau_q[i],\tau_d[j])\\
				Cost_{ins(k)}(\tau_q[i], \tau_d[j])=&\mathop{\min}\limits_{k\leq t\leq j-1}$ 
				$\sum_{p=k+1}^{t}{sub(\tau_q[i-1],\tau_d[p])} \\
				&+ \sum_{p=t+1}^{j}{sub(\tau_q[i],\tau_d[p])}\\
			\end{array}
			\right. 
	\end{eqnarray}}
\end{theorem}
\begin{proof}
	We can show that $dtw(\tau_q, \tau_d) = \Theta(\tau_q,\tau_d)$ by demonstrating that both $dtw(\tau_q[1:m], \tau_d[1:n]) \leq \Theta(\tau_q,\tau_d)$ and $dtw(\tau_q[1:m], \tau_d[1:n]) \geq \Theta(\tau_q,\tau_d)$ hold true.
	
	For any given $\mathcal{A}{\tau_q:\tau_d}'$,	{\scriptsize $Cost(\tau_q[i], \tau_d[j])=Cost_{ins(k)}(\tau_q[i], \tau_d[j])$} when $a_i=j$ and $a_{i-1}<j-1$. We define {\scriptsize$t_{k}^{*}=\mathop{\arg\min}\limits_{k\leq t\leq j-1}$ $\sum_{p=k+1}^{t}{sub(\tau_q[i-1],\tau_d[p])}$ $+ \sum_{p=t+1}^{j}{sub(\tau_q[i],\tau_d[p])}$} (Note that for not all k, $t_{k}^{*}$ is defined).  Then, we have
	{\scriptsize$$
	Cost_{ins(k)}(\tau_q[i], \tau_d[j])=\sum_{p=k+1}^{t_{k}^{*}}{sub(\tau_q[i-1],\tau_d[p])}+ \sum_{p=t_{k}^{*}+1}^{j}{sub(\tau_q[i],\tau_d[p])}
	$$}
	
	Next, we can recursively define the function $dtw'(\tau_q[1:m],\tau_d[1:n])$ as follows where $i>2$:
	
	{\scriptsize\begin{eqnarray}
			&&dtw'(\tau_q[1:i], \tau_d[1:j]) \\ \notag
			&=&\left\{
			\begin{array}{ll}
				dtw'(\tau_q[1:i-1], \tau_d[1:j]) + sub(\tau_q[i],\tau_d[j]), &a_i=j, a_{i-1}=j\\
				dtw'(\tau_q[1:i-1], \tau_d[1:j-1]) + sub(\tau_q[i],\tau_d[j]),\\
				&a_i=j, a_{i-1}=j-1\\
				dtw'(\tau_q[1:i-1], \tau_d[1:j-1]) + sub(\tau_q[i],\tau_d[j]),\\
				&t_{a_{i-1}}* \text{exists and} j=t*_{a_{i-1}} + 1\\
				dtw'(\tau_q[1:i], \tau_d[1:j-1]) + sub(\tau_q[i],\tau_d[j]), &otherwise\\
			\end{array}
			\right. 
			\label{eq:dtw_fake}
	\end{eqnarray}}
	
	Meanwhile, we define {\scriptsize$dtw'(\tau_q[1:1], \tau_d[1:j])=\sum_{k=1}^{j}{sub(\tau_q[1],\tau_d[k])}$}. We can obtain the inequality $dtw'(\tau_q[1:m], \tau_d[1:n])\geq dtw(\tau_q[1:m], \tau_d[1:n])$ by similar process as $wed$.
	Next, we will prove that {\scriptsize$dtw'(\tau_q[1:i], \tau_d[1:a_i]) = dtw'(\tau_q[1:i-1], \tau_d[1:a_{i-1}]) + \text{Cost}(\tau_q[1:i], \tau_d[1:a_i])$ $(2\leq i)$}. We will discuss different cases:
	\begin{enumerate}
		\item When $a_i=j$ and $a_{i-1}=j$, we have {\scriptsize$dtw'(\tau_q[1:i], \tau_d[1:a_i])=dtw'(\tau_q[1:i-1], \tau_d[1:a_{i-1}])+Cost_{del}(\tau_q[1:i], \tau_d[1:a_i])$}.
		\item When $a_i=j$ and $a_{i-1}=j-1$, we have {\scriptsize$dtw'(\tau_q[1:i], \tau_d[1:a_i])=dtw'(\tau_q[1:i-1], \tau_d[1:a_{i-1}])+Cost_{sub}(\tau_q[1:i], \tau_d[1:a_i])$}.
		\item When $a_i=j$ and $a_{i-1}<j-1$, we have
		{\scriptsize
			\begin{eqnarray}
				&&dtw'(\tau_q[1:i], \tau_d[1:a_i]) \\ \notag
				&=&dtw'(\tau_q[1:i], \tau_d[1:t*_{a_{i-1}}+1])+\sum_{k=t*_{a_{i-1}}+2}^{a_i}sub(\tau_q[i], \tau_d[k]) \\ \notag
				&=&dtw'(\tau_q[1:i-1],  \tau_d[1:t*_{a_{i-1}}])+\sum_{k=t*_{a_{i-1}}+1}^{a_i}sub(\tau_q[i], \tau_d[k]) \\ \notag
				&=&dtw'(\tau_q[1:i-1],  \tau_d[1:a_{i-1}])+\sum_{k=a_{i-1}+1}^{t*_{a_{i-1}}}sub(\tau_q[i-1], \tau_d[k]) \\ \notag
				&&+\sum_{k=t*_{a_{i-1}}+1}^{a_i}sub(\tau_q[i], \tau_d[k]) \\ \notag
				&=&dtw'(\tau_q[1:i-1], \tau_d[1:a_{i-1}])+Cost_{ins(a_{i-1})}(\tau_q[i], \tau_d[a_i])
			\end{eqnarray}
		}
	\end{enumerate}
	
	Therefore, we can establish the following equation holds.
	{\scriptsize
		\begin{eqnarray}
			&&dtw'(\tau_q[1:m], \tau_d[1:n]) \\\notag
			&=&dtw'(\tau_q[1:m], \tau_d[1:a_m])+Insert(\tau_d[a_m+1:n])\\\notag
			&=&dtw'(\tau_q[1:1], \tau_d[1:a_1])+\sum_{i=2}^{m}Cost(\tau_q[i], \tau_d[a_i])+Insert(\tau_d[a_m+1:n])\\ \label{eq:pr5}
			&=&Insert(\tau_d[1:a_1-1])+\sum_{i=1}^{m}Cost(\tau_q[i], \tau_d[a_i])+Insert(\tau_d[a_m+1:n])\\\notag
		\end{eqnarray}
	}
	
	Therefore, for any $\mathcal{A}{\tau_q:\tau_d}'$, we have {\scriptsize$dtw(\tau_q[1:m], \tau_d[1:n]) \leq Insert(\tau_d[1:a_1]) + \sum_{i=1}^{m} Cost(\tau_q[i], \tau_d[a_i]) + Insert(\tau_d[a_m+1:n])$}. Thus, we can conclude that {\scriptsize$dtw(\tau_q[1:m], \tau_d[1:n]) \leq \Theta(\tau_q,\tau_d)$}. Next, we will prove that {\scriptsize$dtw(\tau_q[1:m], \tau_d[1:n]) \geq \Theta(\tau_q,\tau_d)$}.
	
	The equation \ref{eq:raw_dtw} defines the recursive computation process of DTW. Given the dynamic programming process, we can construct a matching sequence $\mathcal{A}{\tau_q:\tau_d}$ by the multi-points set $G$. For each multi-point $\tau_d[k]$ and all points $\tau_q[i:j]$ that it connects, we set $a_t=k$ ($i\leq t\leq j$). For each multi-point $\tau_q[k]$ and all points $\tau_d[i:j]$ that it connects, we set $a_k=i$.
	
	For each multi-point $\tau_q[k]$ ($k<m$), we have
	{\scriptsize
	\begin{eqnarray}
		\sum_{t=i}^{j}sub(\tau_q[k],\tau_d[t])&=&Cost_{sub}(\tau_q[k], \tau_d[i]) + \sum_{t=i}^{j}sub(\tau_q[k],\tau_d[t]) \\ \notag
		&\geq& Cost_{sub}(\tau_q[k], \tau_d[i]) + Cost_{ins(i)}(\tau_q[k+1], \tau_d[j+1]) \\\notag
		&&- sub(\tau_q[k+1], \tau_d[j+1]) \notag
	\end{eqnarray}
}
    When $k=m$, {\scriptsize$\sum_{t=i}^{j}sub(\tau_q[k],\tau_d[t])=Cost_{sub}(\tau_q[k],\tau_d[i])+Insert(\tau_q[i+1:j])$}.
    
    For each multi-point $\tau_d[k]$, we have
    {\scriptsize
    \begin{eqnarray}
    	\sum_{t=i}^{j}sub(\tau_q[t],\tau_d[k])&=&Cost_{sub}(\tau_q[i], \tau_d[k]) + \sum_{t=i+1}^{j}Cost_del(\tau_q[i],\tau_d[t]) 
    \end{eqnarray}
}
    Finally, we can get
    {\scriptsize
    \begin{eqnarray}
    	dtw(\tau_q,\tau_d)&=&\sum_{\tau_q[k] \in G}{\sum_{t=i}^{j}sub(\tau_q[k],\tau_d[t])}+\sum_{\tau_d[k] \in G}{\sum_{t=i}^{j}sub(\tau_q[t],\tau_d[k])} \\ \notag
    	&\geq& \sum_{i=1}^{m}Cost(\tau_q[i], \tau_d[a_i])+Insert(\tau_d[a_m+1:n]) \\ \notag
    	&\geq& \Theta(\tau_q,\tau_d) \\ \notag
    \end{eqnarray}
}
	Therefore, we can conclude that $dtw(\tau_q[1:m], \tau_d[1:n]) \geq \Theta(\tau_q,\tau_d)$. Furthermore, combining this result with our previous conclusion, we can state that $dtw(\tau_q[1:m], \tau_d[1:n]) = \Theta(\tau_q,\tau_d)$.
	
\end{proof}

\section{Pruning Algorithm}

In the previous section, we proposed an algorithm with the time complexity of $O(mn)$ to find the subtrajectory with the minimum distance for the query trajectory in a data trajectory. The computation complexity for solving the SSS problem is thus optimized to $O(Nmn)$. Although the Algorithm \ref{algo:efficiencyA} is efficient enough to complete a similar subtrajectory search within $10ms$, considering many data trajectories in the database (i.e., usually millions of data), the time to find the optimal subtrajectory from the data trajectory database is about 15$\sim$30 minutes for a given query trajectory. This is still not affordable for most of the applications. In practice, most of the data trajectories in the database are far distant from the query trajectory. Thus, we apply two heuristic pruning methods on the raw data trajectories before passing them to our CMA algorithm.

\noindent\textbf{Grid Based Pruning (GBP).}
\revision{GBP divides the space into multiple grids and stores the IDs of the data trajectories passing through each grid using inverted indexing. The intuition of GBP is: if a data trajectory does not contain enough points in the grids that the query trajectory passes by or is adjacent to, the data trajectory is unlikely to have subtrajectories similar to the query trajectory.}

We divide the map into a square grid with $\varepsilon $ as side lengths. For a point $\tau_d^{(k)}[j]$ in the data trajectory, we denote the grid it locates in as $g(\tau_d^{(k)}[j])$. Each grid is surrounded by 8 neighboring grids; we denote the set consisting of the grid $g(\tau_d^{(k)}[j])$ and its neighboring grids as $B(\tau_d^{(k)}[j])$. We consider $\tau_d^{(k)}[j]$ to be \textit{close} to the points in the grid it is located in and its neighboring grids, i.e., a point $\tau_q^{(t)}[i]$ is close to $\tau_d^{(k)}[j]$ if and only if $g(\tau_q^{(t)}[i]) \in B(\tau_d^{(k)}[j])$. If there exists a point in a query trajectory that is close to $\tau_d^{(k)}[j]$, then we consider this query trajectory to be close to $\tau_d^{(k)}[j]$. Meanwhile, we denote the set formed by the points of all query trajectories close to $\tau_d^{(k)}[j]$ as $H(\tau_d^{(k)}[j])$.

{\small\begin{eqnarray}
		H(\tau_d^{(k)}[j])=\{\tau_q^{(t)}[i]|g(\tau_q^{(t)}[i]) \in B(\tau_d^{(k)}[j])\} \label{equ:eq4}
\end{eqnarray}}

Given a data trajectory $\tau_d^{(k)}$, we can count how many points in the query trajectory $\tau_q$ are close to that query trajectory and denote it as $close(\tau_q,\tau_d^{(k)})$, that is, 
{\small\begin{eqnarray}
		close(\tau_q,\tau_d^{(k)})=\left|\left\{\tau_q[i]\left|\tau_q[i] \in \mathop{\cup}\limits_{1\leq j \leq n}H(\tau_d^{(k)}[j])\right.\right\}\right| \label{equ:eq3}
\end{eqnarray}}
The larger $close(\tau_q,\tau_d^{(k)})$ means that more points of the query trajectory are distributed around the data trajectory $\tau_d^{(k)}$, then there is likely a segment of sub-trajectories in $\tau_d^{(k)}$ that are similar to $\tau_q$. We define a constant $\mu$ ($0<\mu<1$) and we call the algorithm \ref{algo:efficiencyA} to search for the optimal sub-trajectory in $\tau_d^{(k)}$ whenever $close(\tau_q,\tau_d^{(k)})\geq \mu \cdot m$.

\noindent\textbf{Key Points Filter (KPF).} \revision{The intuition of KPF is that: if an estimated lower bound of distance  between the query trajectory $\tau_q$ and a data trajectory $\tau'_d$ is larger than the distance between the query trajectory and the current found optimal similar subtrajectory, we can skip running CMA on $\tau'_d$ (i.e., pruning $\tau'_d$). In KPF, we propose a method to estimate the lower bound of the distance between a $\tau_q$ and a $\tau'_d$, then prune the data trajectories by the lower bound. Specifically, to speed up the calculation of the estimation of the lower bound, we apply a sampling method in KPF, which samples $r\cdot|\tau_q|$ (i.e., $r$ is the sampling rate) key points from $\tau_q$ and estimates the lower bound of the total distance between the $r\cdot|\tau_q|$ key points and the data trajectory $\tau'_d$, then enlarge the result with $\frac{1}{r}$ to represent the lower bound of the distance between $\tau_q$ and $\tau'_d$. In our overall process, all trajectories that are not pruned by the GBP module are passed to KPF module. KPF is synchronous with the CMA searching process, and uses the current found optimal similar subtrajectory to further prune the new data trajectories. If a data trajectory is not pruned by GBP nor KPF, we will invoke a CMA algorithm on it. Next, we will provide a detailed explanation of the KPF module.}

For each point $\tau_q[i]$ in the query trajectory $\tau_q$, we can convert it to a point in the data trajectory by substitution or deletion. Without considering other points, each point has a minimum cost of conversion. We use $minCost(\tau_q[i],\tau_d)$ to denote the lower bound of the cost of converting $\tau_q[i]$ to a point in $\tau_d$. We give a formal definition of $minCost(\tau_q[i],\tau_d)$ as follows.
{\small\begin{eqnarray}
		minCost(\tau_q[i],\tau_d)=\min\{del(\tau_q[i]), \mathop{\min}\limits_{1\leq j\leq n}sub(\tau_q[i],\tau_d[j])\} \notag
\end{eqnarray}}
Next, we introduce the lower bound on the conversion cost by the Theorem \ref{the:lbocc}.
\begin{theorem}[Lower Bound of Cost]
	\label{the:lbocc}
	We denote {\small$minCost(\tau_q,\tau_d)$} as {\small$\sum_{i=1}^{m}minCost(\tau_q[i],\tau_d )$} and {\small$minCost(\tau_q,\tau_d)$} is the lower bound on the cost of converting $\tau_q$ to $\tau_d$, i.e., {\scriptsize$minCost(\tau_q,\tau_d)\leq \mathop{\min}\limits_{1\leq j\leq n}C_{m,j}$}.
\end{theorem}
\begin{proof}
	We use WED as an example to prove this theorem using mathematical induction.
	
	(i) When $i=1$, we have {\scriptsize$minCost(\tau_q[1:i],\tau_d)\leq \mathop{\min}\limits_{1\leq j\leq n}sub(\tau_q[1],\tau_d[j])=\mathop{\min}\limits_{1\leq j\leq n}C_{m,j}$}.
	
	(ii) Suppose when $i=k-1$, we have {\small$minCost(\tau_q[1:k-1],\tau_d)\leq \mathop{\min}\limits_{1\leq j\leq n}C_{k-1,j}$}. We will discuss the case when $i=k$ in two scenarios. If $j=1$, then {\small$minCost(\tau_q[1:k],\tau_d)=minCost(\tau_q[1:k-1],\tau_d)+minCost(\tau_q[k],\tau_d)$}. Considering that {\small$minCost(\tau_q[1:k-1],\tau_d)<C_{k-1,1}$} and {\small$minCost(\tau_q[1:k-1],\tau_d)<del(\tau_q[1:i-1])$}, we have {\small$minCost(\tau_q[1:k],\tau_d)\leq \min\{C_{i-1,j}+del(\tau_q[i]), sub(\tau_q[i],\tau_d[1])+del(\tau_q[1:i-1])\}=C_{k,1}$}. For each {\small$\forall j\neq 1$}, we have 
	{\small \begin{eqnarray}
			&&minCost(\tau_q[1:k],\tau_d) \notag \\
			&=&minCost(\tau_q[1:k-1],\tau_d) +minCost(\tau_q[k],\tau_d)\notag \\
			&\leq& \mathop{\min}\limits_{1\leq t\leq j}C_{k-1,t}+sub(\tau_q[i],\tau_d[j]) \notag \\
			&\leq&C_{k,j} \notag
	\end{eqnarray}}
	Thus, we can obtain $minCost(\tau_q[1:k],\tau_d)\leq \mathop{\min}\limits_{1\leq j\leq n}C_{k,j}$.
	
	Finally, by combining (i) and (ii), the theorem \ref{the:lbocc} is proved.
\end{proof}

\begin{figure}
	\begin{small}
		\begin{center}
			\includegraphics[width=0.30\textwidth]{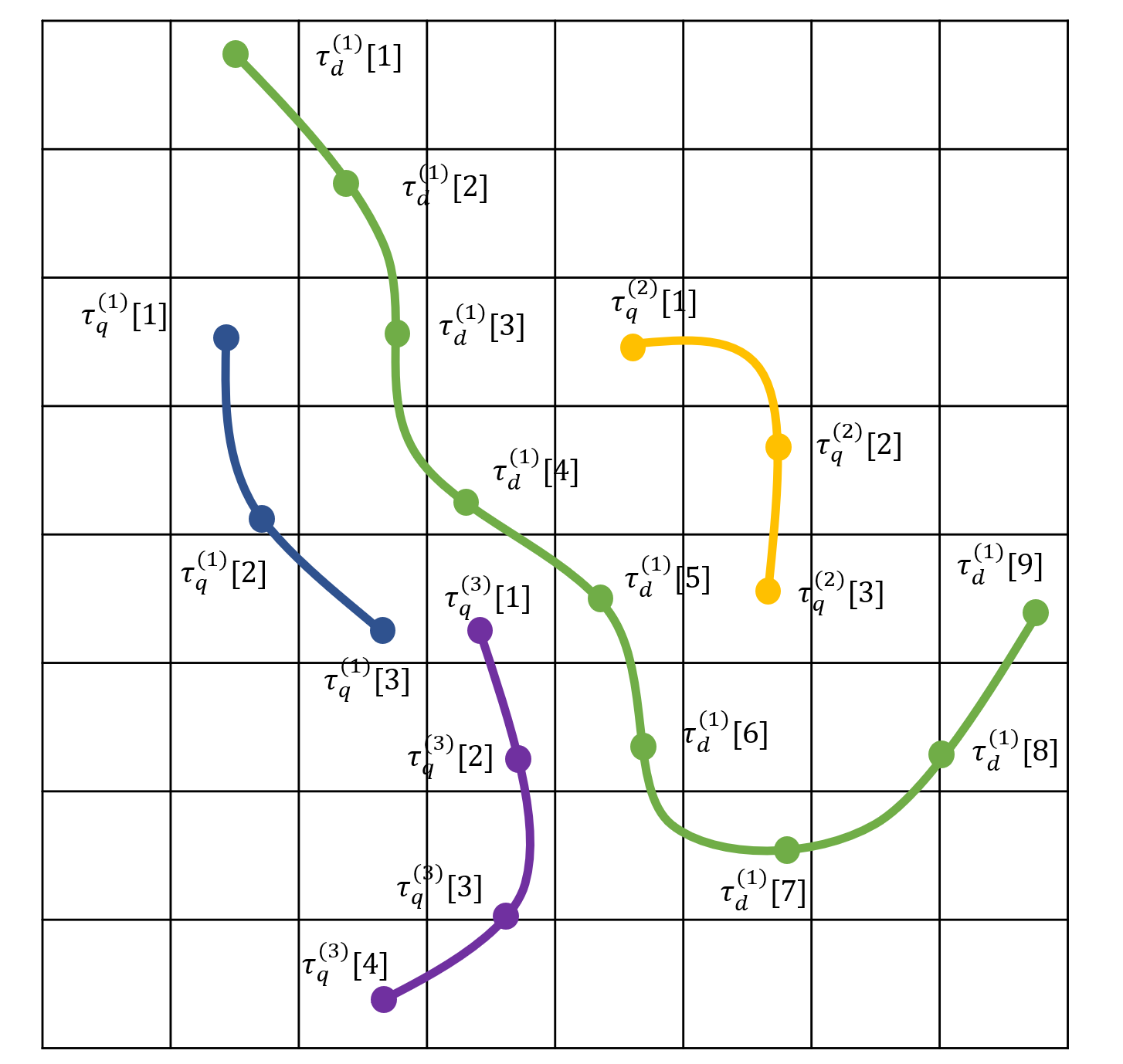}
		\end{center}
		\caption{Pruning Example}
		\label{fig:pruningExample}
	\end{small}
\end{figure}

Nevertheless, we cannot directly compute $minCost(\tau_q,\tau_d)$ in practical applications. It is because the time complexity of computing $minCost(\tau_q,\tau_d)$ is $O(m\cdot n)$, which is the same as the time complexity of computing the optimal subtrajectory directly. Fortunately, we can approximate $minCost(\tau_q,\tau_d)$ due to the continuity of trajectory. In reality, the position of an object cannot change dramatically in a short time, thus the location of a point in a trajectory is close to the location of its neighboring points. Based on the continuity of the trajectory, we can select key points $Key(\tau_q)(=\{\tau_q[e_1],\tau_q[e_2],\dots,\tau_q[e_K]\})$ from the original query trajectory. Based on these key points, we can compute the estimation of $minCost(\tau_q,\tau_d)$ and denote it by $minCost_e(\tau_q,\tau_d)$ as follows.
{\small\begin{equation}
		minCost_e(\tau_q,\tau_d)=\frac{1}{r}\sum_{\tau_q[e_i]\in Key(\tau_q)}minCost(\tau_q[e_i],\tau_d) \label{eq:e2}
\end{equation}}
Here, $r=\frac{|Key(\tau_d)|}{n}$ denoting the ratio of key points selected from the query trajectory. The selection of key points $Key(\tau_q)$ affects the accuracy of estimation; we uniformly select points from $\tau_q$ in this paper. The more key points in $Key(\tau_q)$, the more accurate the estimation of $minCost(\tau_q,\tau_d)$ will be. However, the time complexity of the computation also increases as $Key(\tau_q)$ increases. Therefore, we need to choose as few key points as possible while ensuring the estimation accuracy as much as possible.

With GBP and KPF, the whole similar subtrajectory searching process is shown in Algorithm \ref{algo:efficiencyC}.

\begin{algorithm}[t]
	{\small
		\DontPrintSemicolon
		\KwIn{\small a query trajectory $\tau_q$, data trajectories $\{\tau_d^{(1)},\tau_d^{(2)},\dots,\tau_d^{(N)}\}$}
		\KwOut{\small the optimal subtrajectory $\tau_d^{(k^{*})}[i^{*}:j^{*}]$ for the query trajectory $\tau_q$}
		$\tau_d^{(*)}[i^{*}:j^{*}]\gets \{\}$\;
		\ForAll{$1\leq k \leq N$}{
			For each point in the data trajectory $\tau_d^{(k)}$, we compute $H(\tau_d^{(k)}[j])$ according to the Equation \ref{equ:eq4}\;\label{line:3}
			
			According to the equation \ref{equ:eq3}, we can compute $close(\tau_q,\tau_d^{(k)})$ for each query trajectory $\tau_q$\;
			
			\ForAll{$1\leq t \leq M$}{
				\eIf{$result = (-1,-1,-1)$ }{
					$\tau_d^{(k)}[i^{k}:j^{k}]\gets CMA(\tau_q,\tau_d^{(k)})$\;
					$\tau_d^{(*)}[i^{*}:j^{*}]\gets \tau_d^{(k)}[i^{k}:j^{k}]$\;
				}{
					\If{$close(\tau_q,\tau_d^{(k)})>\mu \cdot m$}{
						calculate the estimation of the lower bound $minCost_e(\tau_q,\tau_d^{(k)})$ by Equation \ref{eq:e2}\;
						\If{$minCost_e(\tau_q,\tau_d^{(k)})<\Theta(\tau_q,\tau_d^{(*)}[i^{*}:j^{*}])$}{
							\label{line:14}
							$\tau_d^{(k)}[i^{k}:j^{k}]\gets CMA(\tau_q,\tau_d^{(k)})$\;
							\If{$\Theta(\tau_q,\tau_d^{(k)}[i^{k}:j^{k}])<\Theta(\tau_q,\tau_d^{(*)}[i^{*}:j^{*}])$}{
								$\tau_d^{(*)}[i^{*}:j^{*}]\gets \tau_d^{(k)}[i^{k}:j^{k}]$\;
							}
						}
					}
				}
			}
		}
		\Return{$\tau_d^{(*)}[i^{*}:j^{*}]$}
		\caption{\small Pruning Algorithm}
		\label{algo:efficiencyC}}
\end{algorithm}



\noindent\textbf{Complexity.} By using the hash table, we can find the points in its neighboring grid for each $\tau_d^{(k)}[j]$ in a constant time. Thus the time complexity of computing $H(\tau_d^{(k)}[j])$ is $O(1)$, making the time complexity of computing for each point in the data trajectory $O(n)$. On the other hand, the computational complexity of Equation \ref{equ:eq3} depends mainly on the number of points in the query trajectory around the data trajectory, which is much smaller than $n$. Therefore, the pruning complexity is $O(Nn)$ in total if the average length of the data trajectories is $n$. Suppose the average length of the query trajectory is $m$. Since we select key points from the query trajectory at a certain ratio $r$, the number of key points is $mr$. Thus, the time complexity of computing the distance lower bound is $mrQn$, where $Q$ denote the count of trajectories satisfying $close(\tau_q,\tau_d^{(k)})>\mu \cdot m$. Finally, the time complexity of the algorithm \ref{algo:efficiencyC} is $O(Nn+mrQn+Q'mn)$, where $Q'$ denotes the count of trajectories satisfying line \ref{line:14}, and we have $N\gg Q\gg Q'$.

\section{Experimental Results of Pruning Algorithm}
\noindent\textbf{Analysis of pruning time and searching time.}
We explored the time of the pruning process and the search process with different pruning algorithms and search algorithms. We can find that the two modules, GBP and KPF, have different effects on the overall process of finding the optimal subtrajectories. For example, Figure \ref{fig:Efficiency_PruningSearching} shows that using the GBP module will significantly speed up the pruning process because we only need a time complexity of $O(n)$ to determine whether this data trajectory $\tau_d$ is similar to the query trajectory. However, the drawback of GBP is that the fixed-parameter $\mu$ is not sufficient to sieve out all the dissimilar data trajectories well, which increases the number of invocations of the search algorithm. In particular, the time to complete the whole search process increases significantly when the complexity of the search algorithm is high.

On the other hand, Figure \ref{fig:Efficiency_PruningSearching} shows that using KPF makes the pruning process significantly more time-consuming since no data structure is used to speed up the operation resulting in that we have to calculate the distance between the key point in the trajectory and the nearest point in the data trajectory for each query trajectory to estimate the lower bound of the distance between the optimal subtrajectory and the query trajectory. KPF will filter out all the data trajectories whose lower bound of distance from the query trajectory is greater than the distance between the current optimal subtrajectory and the query trajectory. The distance between the optimal subtrajectory and the query trajectory decreases as the search process continues, which means that more and more data trajectories will be sieved out, leaving only a few data trajectories that need to be searched. The results shown in Figure \ref{fig:Efficiency_PruningSearching} indicate that the filtering process takes more time than the case of using only GBP when we use KPF, yet the search time is less than the case of using only GBP. Therefore, both the search and pruning processes take less time when we include both the GBP and the KPF module in the pruning process. In addition, the difference in pruning process overhead between using both KPF and GBP modules and using only GBP module is not significant, which indicates that using GBP module can significantly reduce the execution time of KPF module. Compared with OSF, GBP and KPF are able to sieve out more data trajectories that are not similar to the query trajectory, which also makes their search process much less time-consuming than OSF. In particular, the experimental results demonstrate that the search time using the CMA is significantly less than that using POS as the search algorithm when we use DTW as the distance function, which verifies the efficiency of the algorithm proposed in this paper.

\noindent\textbf{Effect of the number of data trajectories on cost time.}
Similarly, we also investigate the relationship between the algorithm execution time and the number of data trajectories. As shown in Figure \ref{fig:varyN}, the time of the whole query process shows a linear relationship with $N$. The pruning time does not increase significantly with the number of data trajectories when we use both the GBP and KPF modules, which shows the excellent scalability of our algorithm. The reason is that using the GBP module to filter the $N$ data trajectories is fast. The increase in time is mainly since the number of data trajectories similar to the query trajectory increases as $N$ becomes more extensive. In addition, the time of the whole query process increases obviously when using only the GBP module or the KPF module. Besides, the experimental results still show that the GBP·KPF proposed in this paper is superior to the effect of OSF. The time overhead of the pruning process while only using GBP or KPF is greater than the time overhead of the OSF algorithm.
\begin{figure}[t!]
	\centering
	\subfigure{
		\scalebox{0.08}[0.08]{\includegraphics{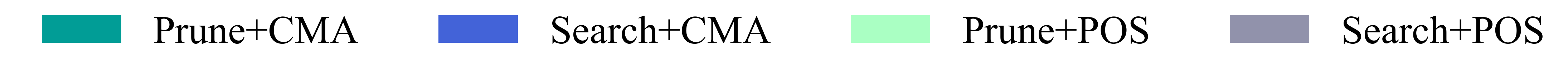}}}\hfill\\
	\addtocounter{subfigure}{-1}\vspace{-2ex}
	\subfigure[][{\scriptsize DTW}]{
		\scalebox{0.19}[0.19]{\includegraphics{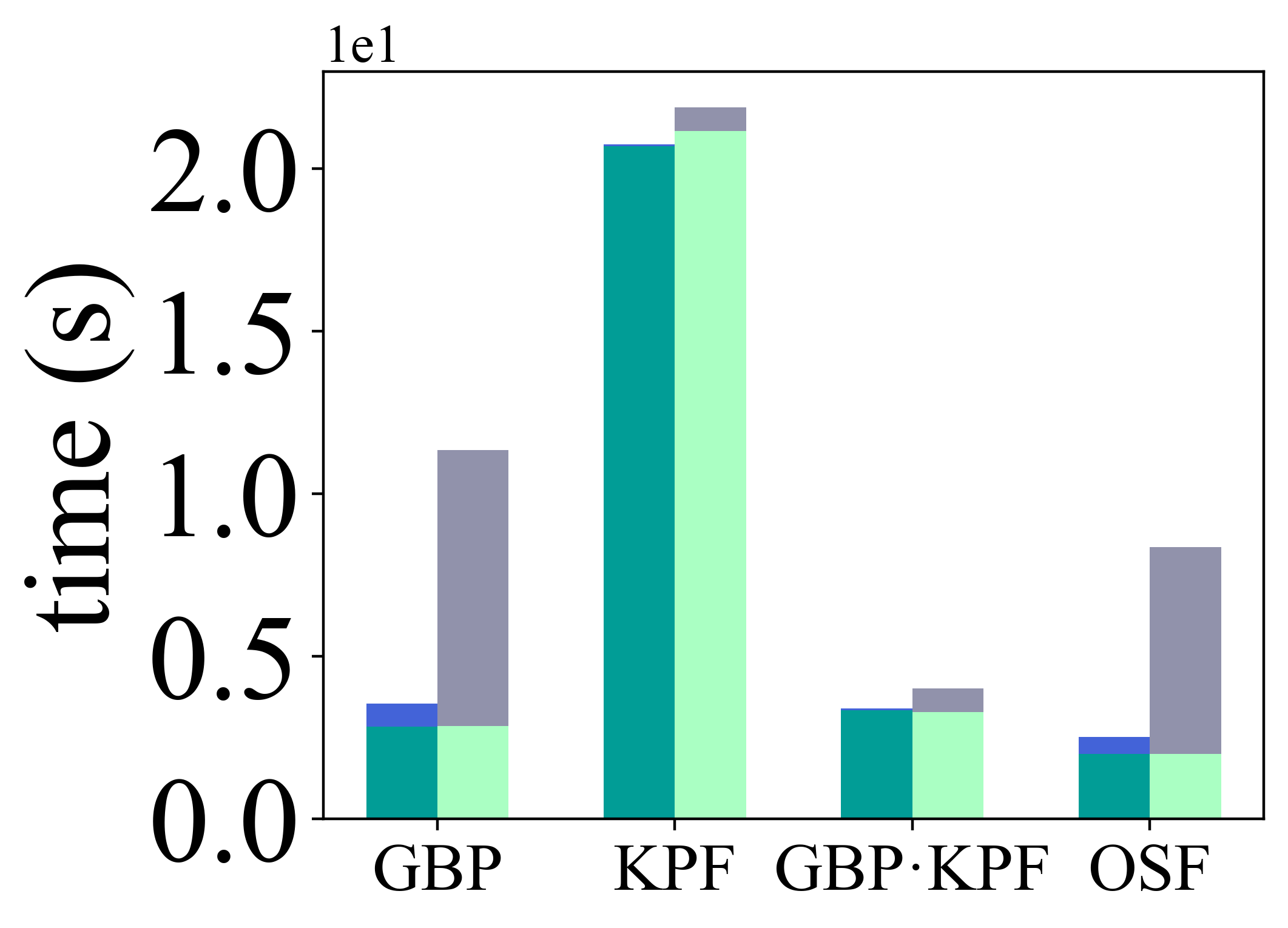}}
		\label{fig:parts_dtw_score}}
	\subfigure[][{\scriptsize EDR}]{
		\scalebox{0.19}[0.19]{\includegraphics{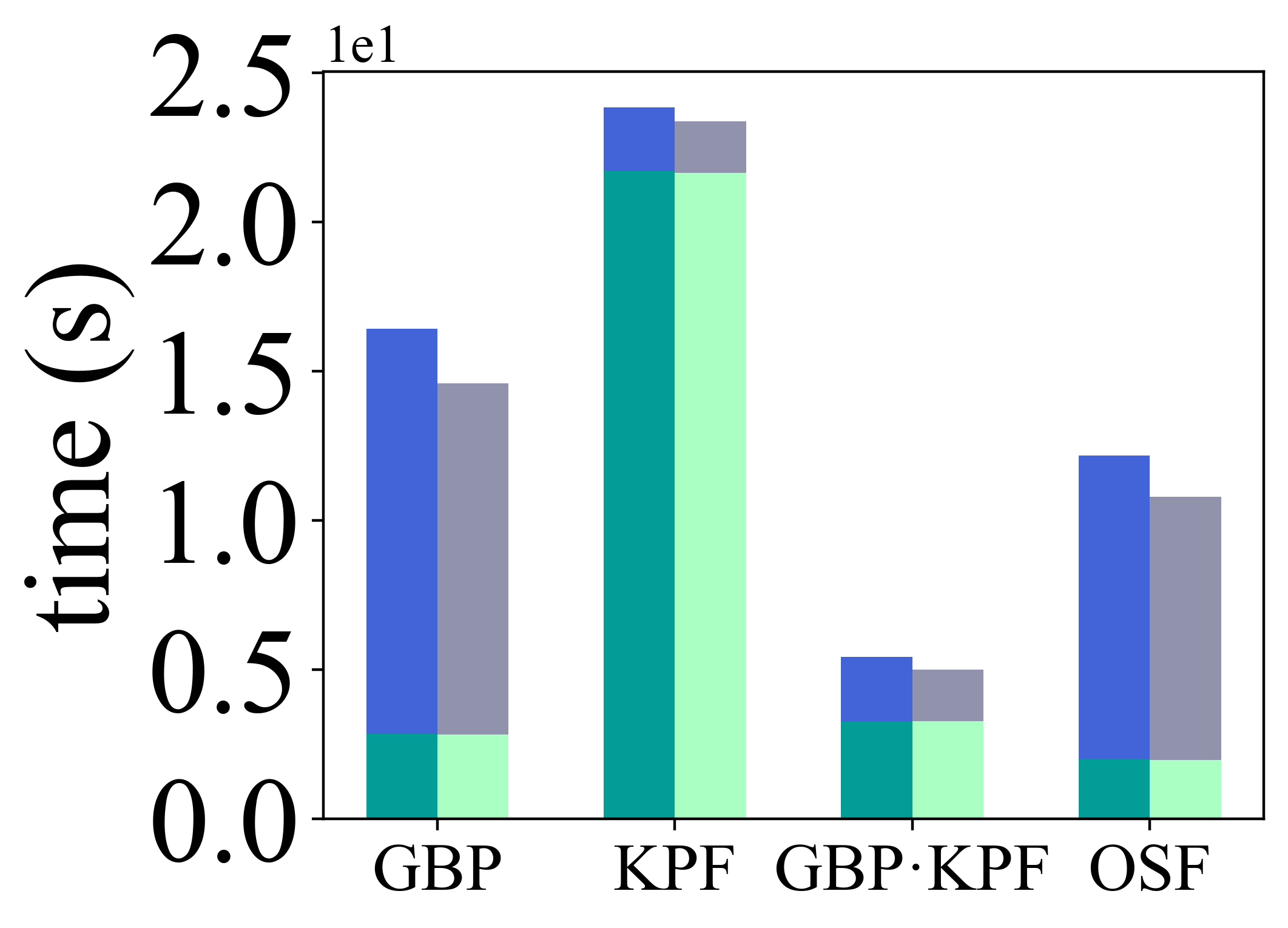}}
		\label{fig:parts_edr_score}}
	\subfigure[][{\scriptsize ERP}]{
		\scalebox{0.19}[0.19]{\includegraphics{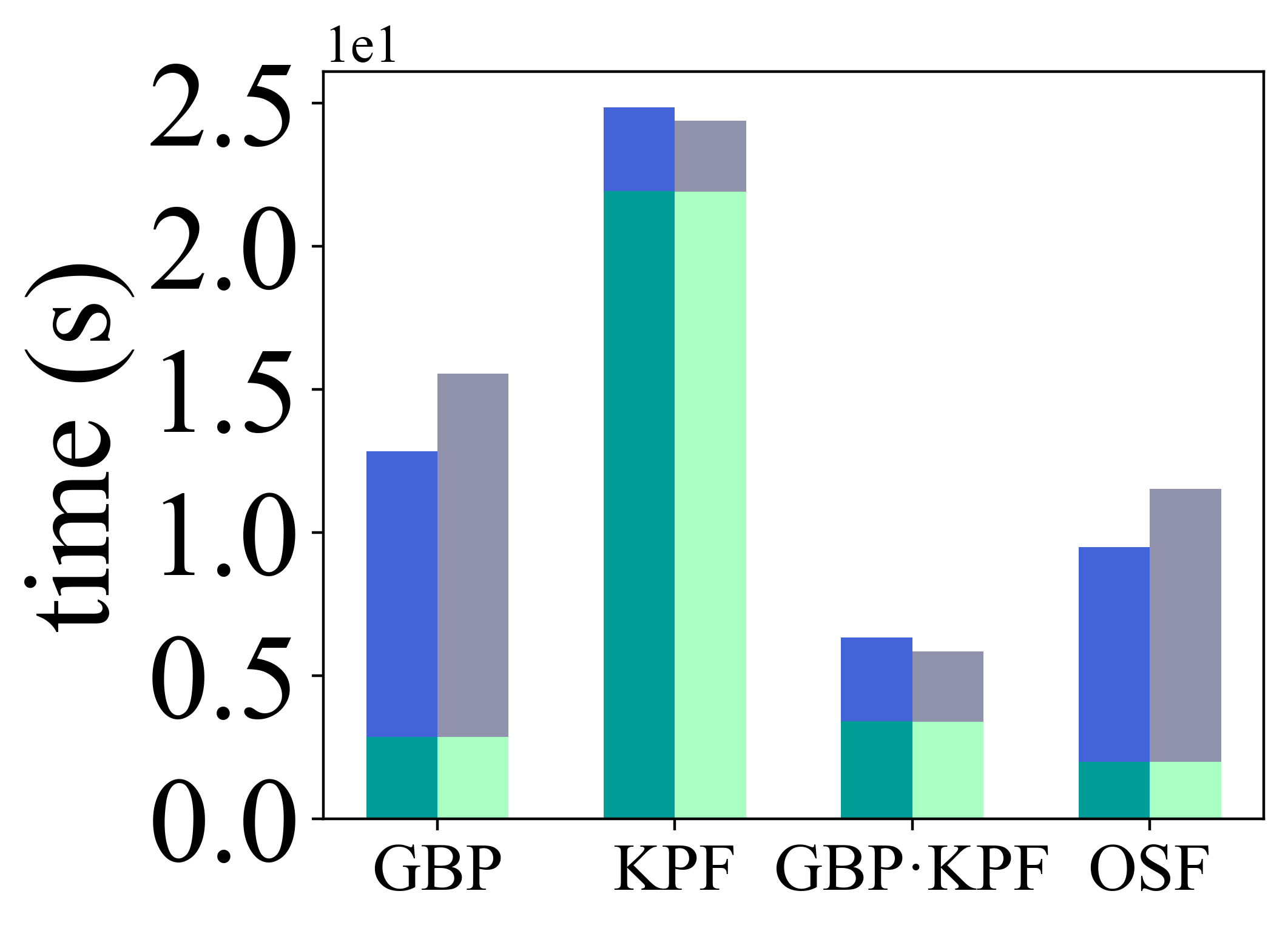}}
		\label{fig:parts_erp_score}}
	\caption{\small Efficiency of Pruning and Searching}
	\label{fig:Efficiency_PruningSearching}
\end{figure}

\begin{figure}[t!]
	\centering	
	\subfigure{
		\scalebox{0.08}[0.08]{\includegraphics{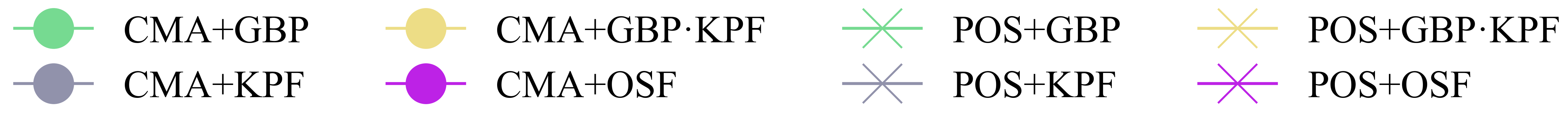}}}\hfill\\
	\addtocounter{subfigure}{-1}\vspace{-2ex}
	\subfigure[][{\scriptsize DTW}]{
		\scalebox{0.19}[0.19]{\includegraphics{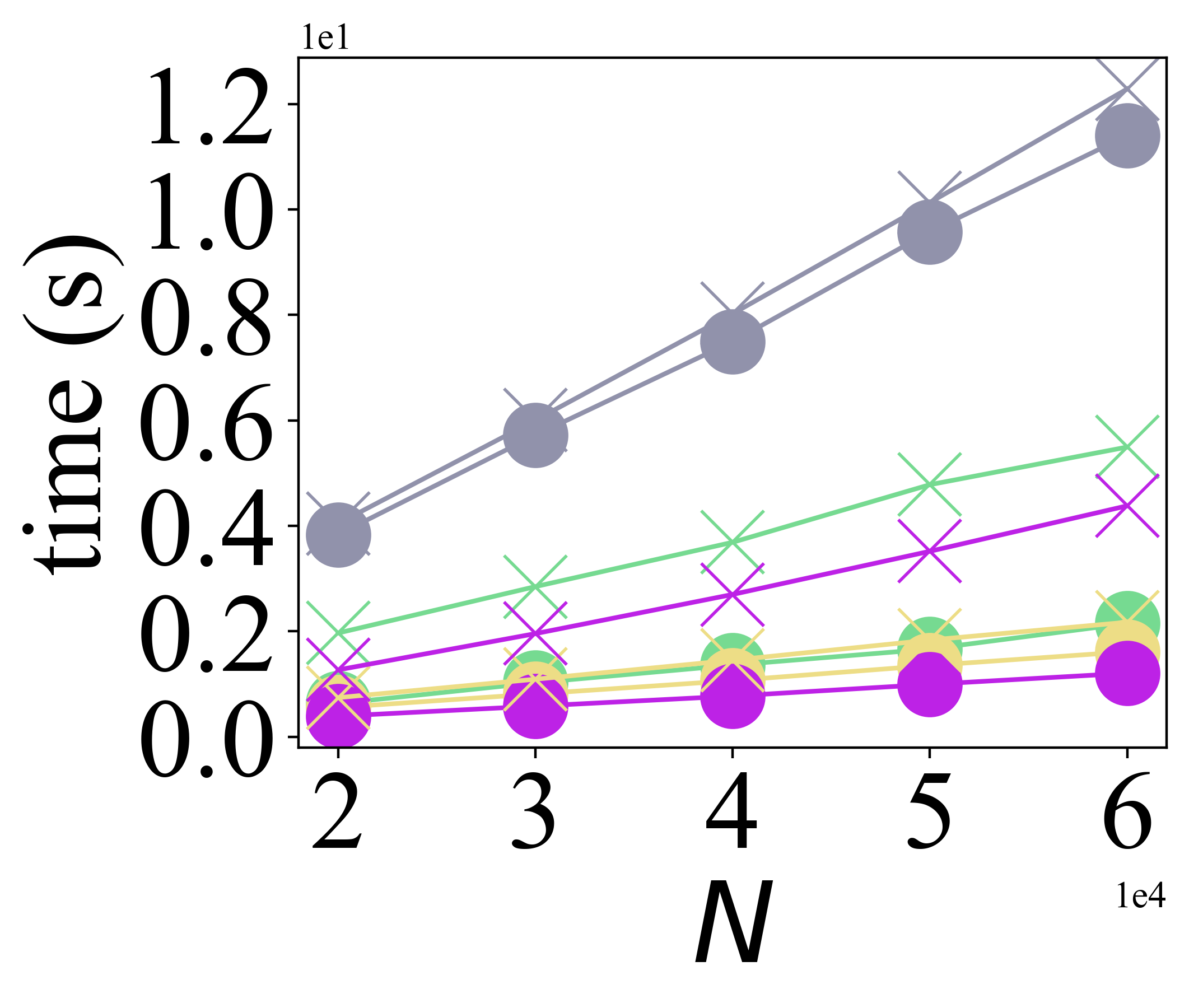}}
		\label{fig:varyN_dtw_score}}
	\subfigure[][{\scriptsize EDR}]{
		\scalebox{0.19}[0.19]{\includegraphics{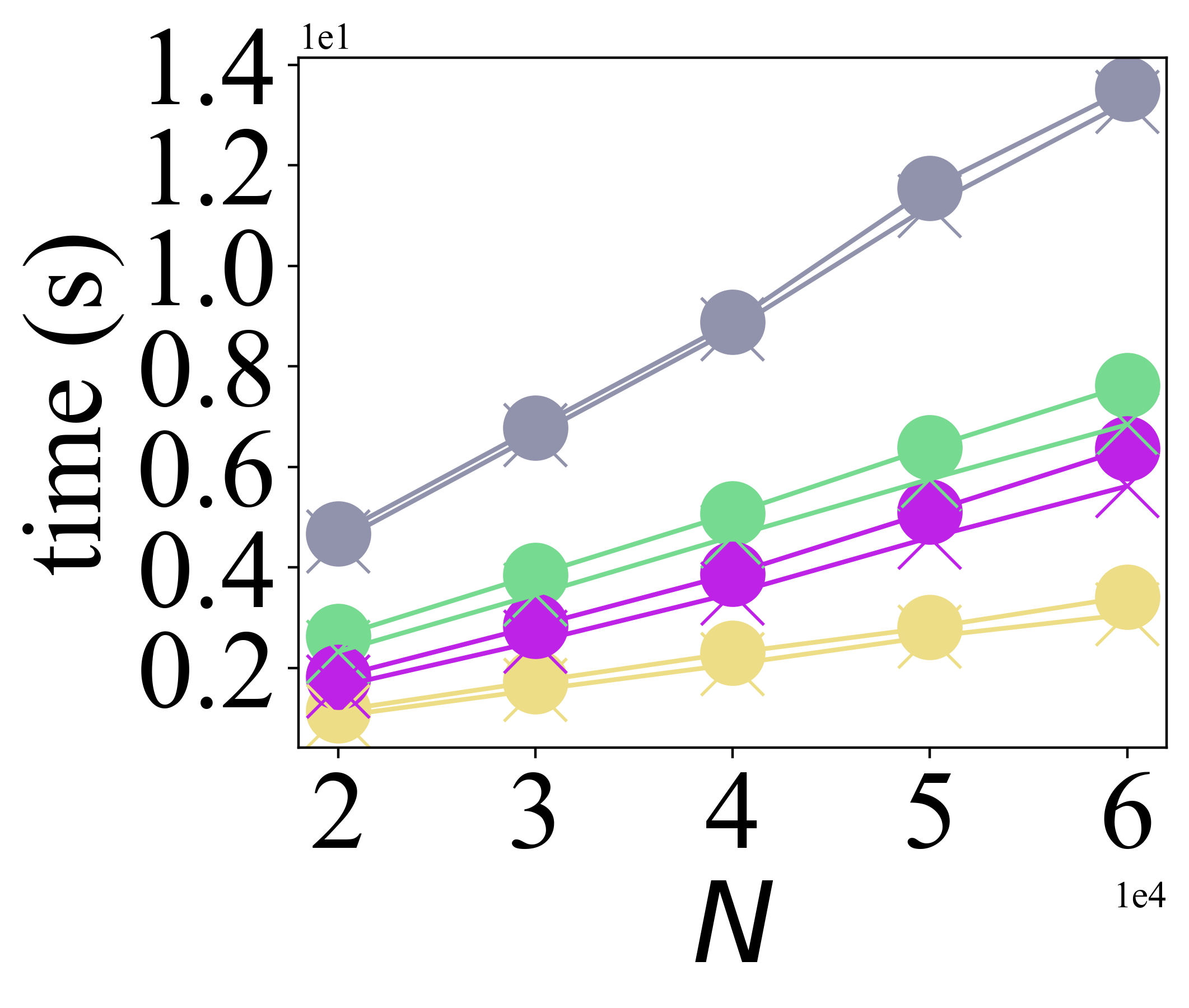}}
		\label{fig:varyN_edr_score}}
	\subfigure[][{\scriptsize ERP}]{
		\scalebox{0.19}[0.19]{\includegraphics{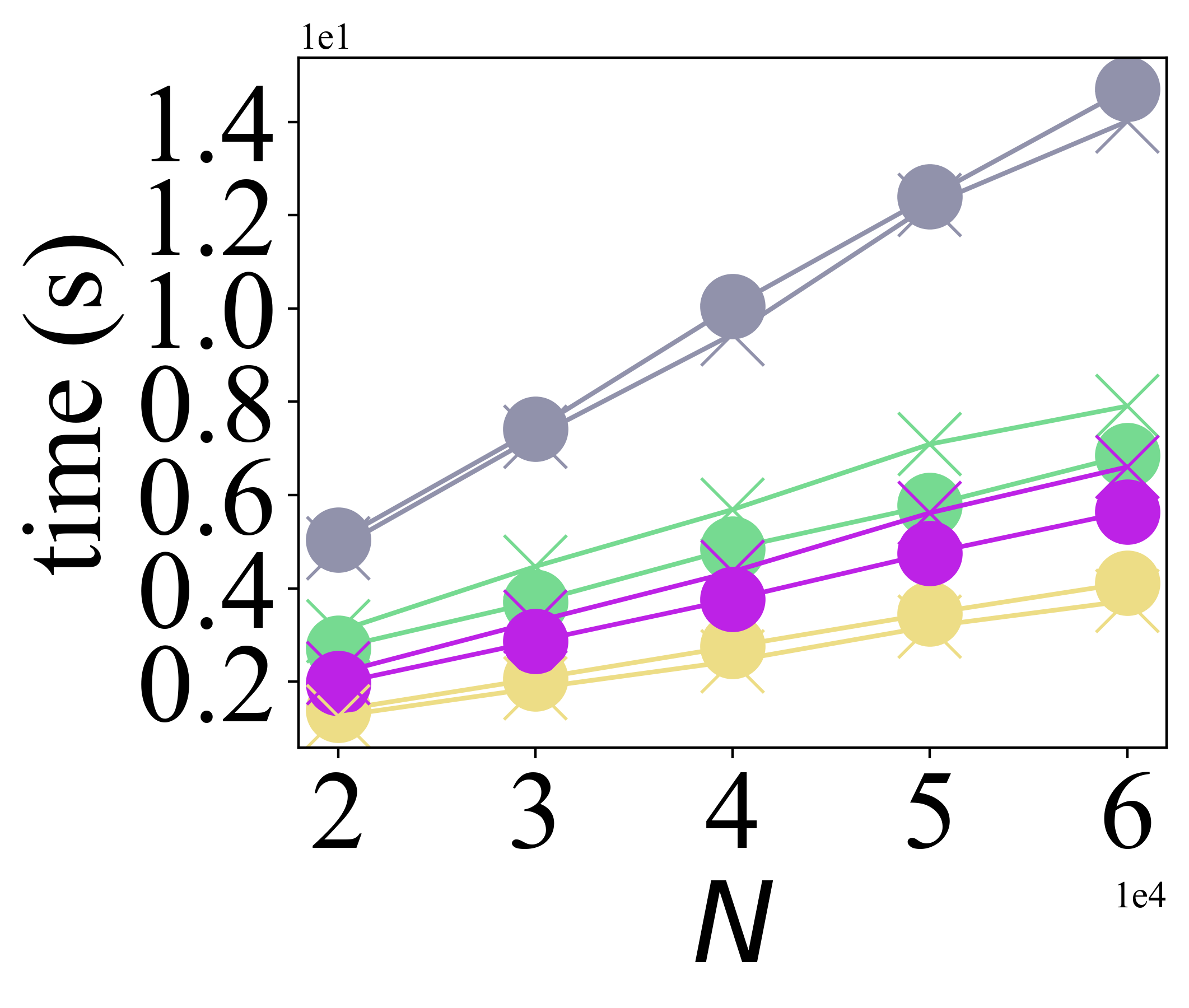}}
		\label{fig:varyN_erp_score}}
	\caption{\small Efficiency with varying data trajectories size $N$}
	\label{fig:varyN}
\end{figure}

\begin{figure*}[t!]
	\centering
	\subfigure{
		\scalebox{0.08}[0.08]{\includegraphics{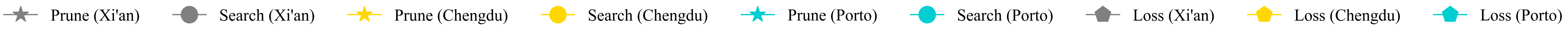}}}\hfill\\
	\addtocounter{subfigure}{-1}\vspace{-2ex}
	\subfigure[][{\scriptsize Time}]{
		\scalebox{0.19}[0.19]{\includegraphics{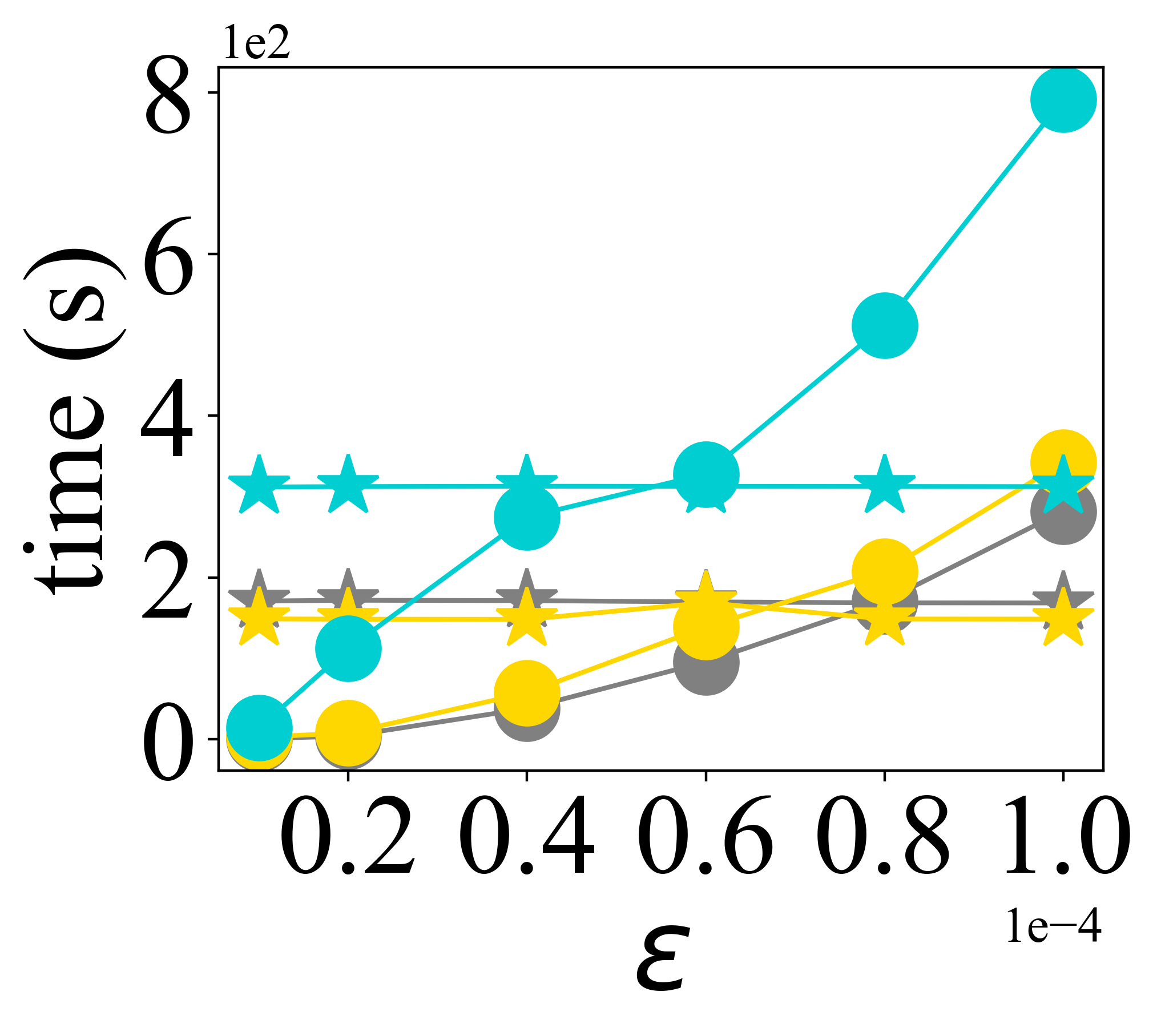}}
		\label{fig:parameter_grid_time_grid}}
	\subfigure[][{\scriptsize Time}]{
		\scalebox{0.19}[0.19]{\includegraphics{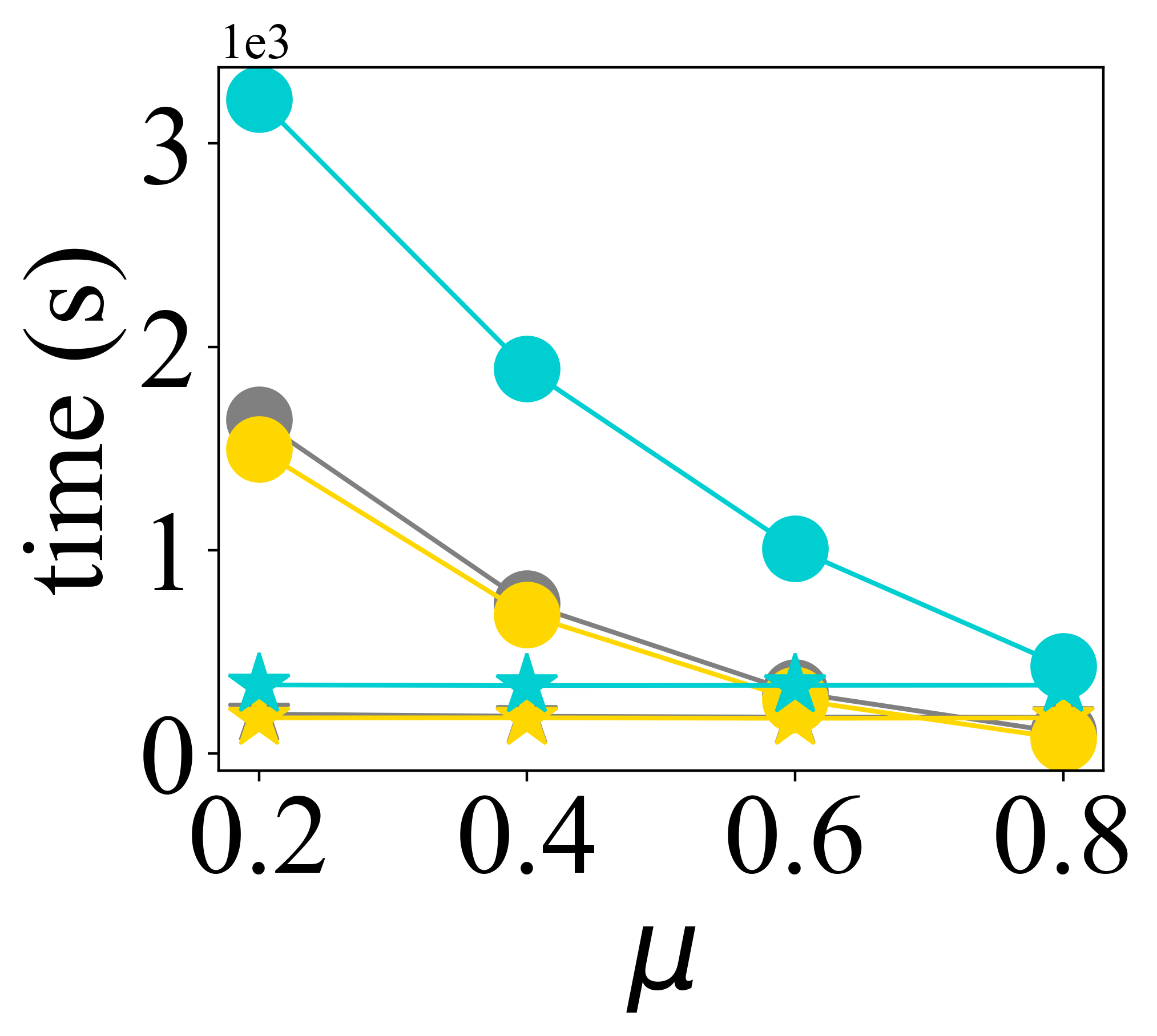}}
		\label{fig:parameter_mu_time_grid}}
	\subfigure[][{\scriptsize Time}]{
		\scalebox{0.19}[0.19]{\includegraphics{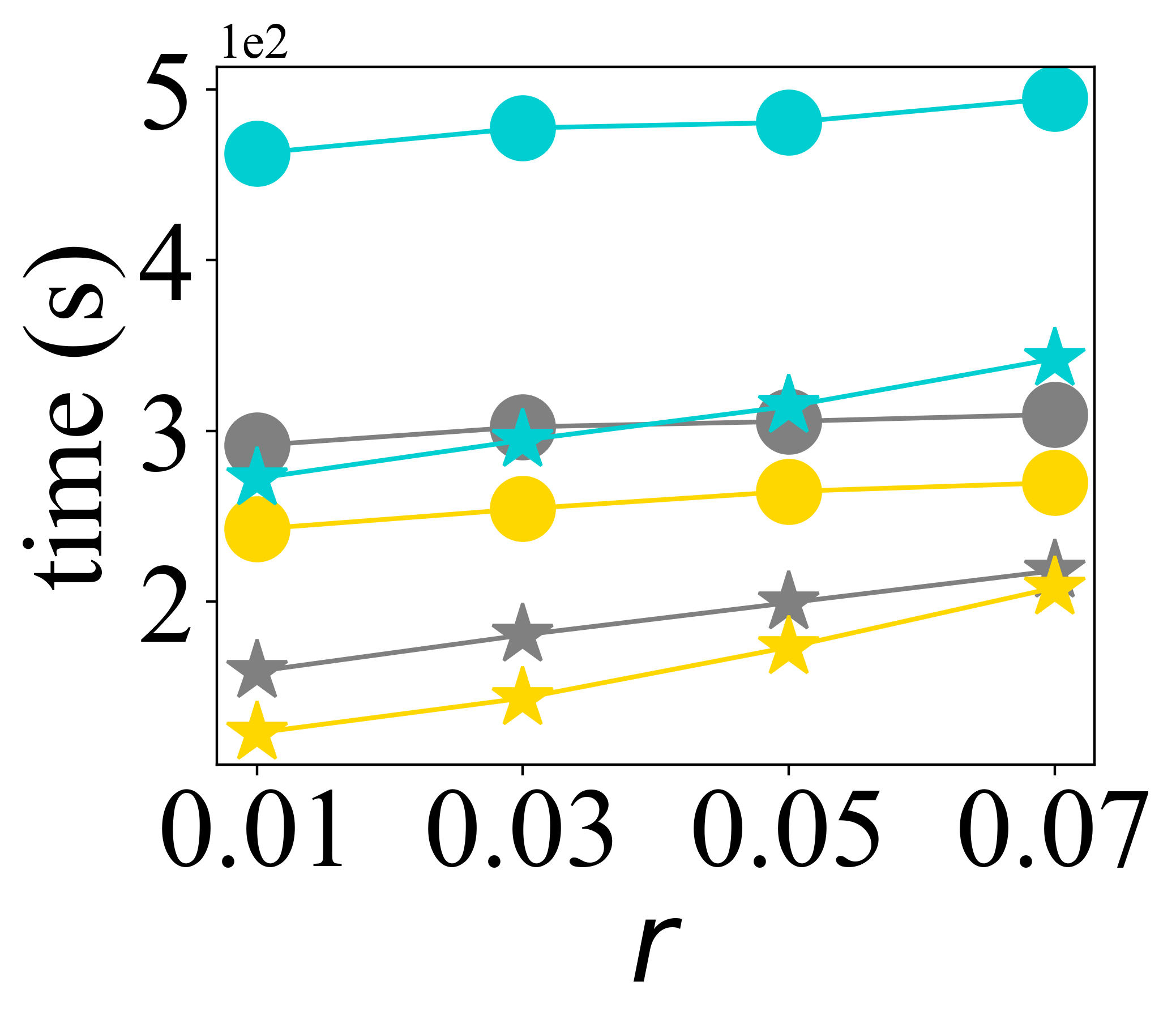}}
		\label{fig:parameter_rate_time_grid}}
	\subfigure[][{\scriptsize Loss}]{
		\scalebox{0.19}[0.19]{\includegraphics{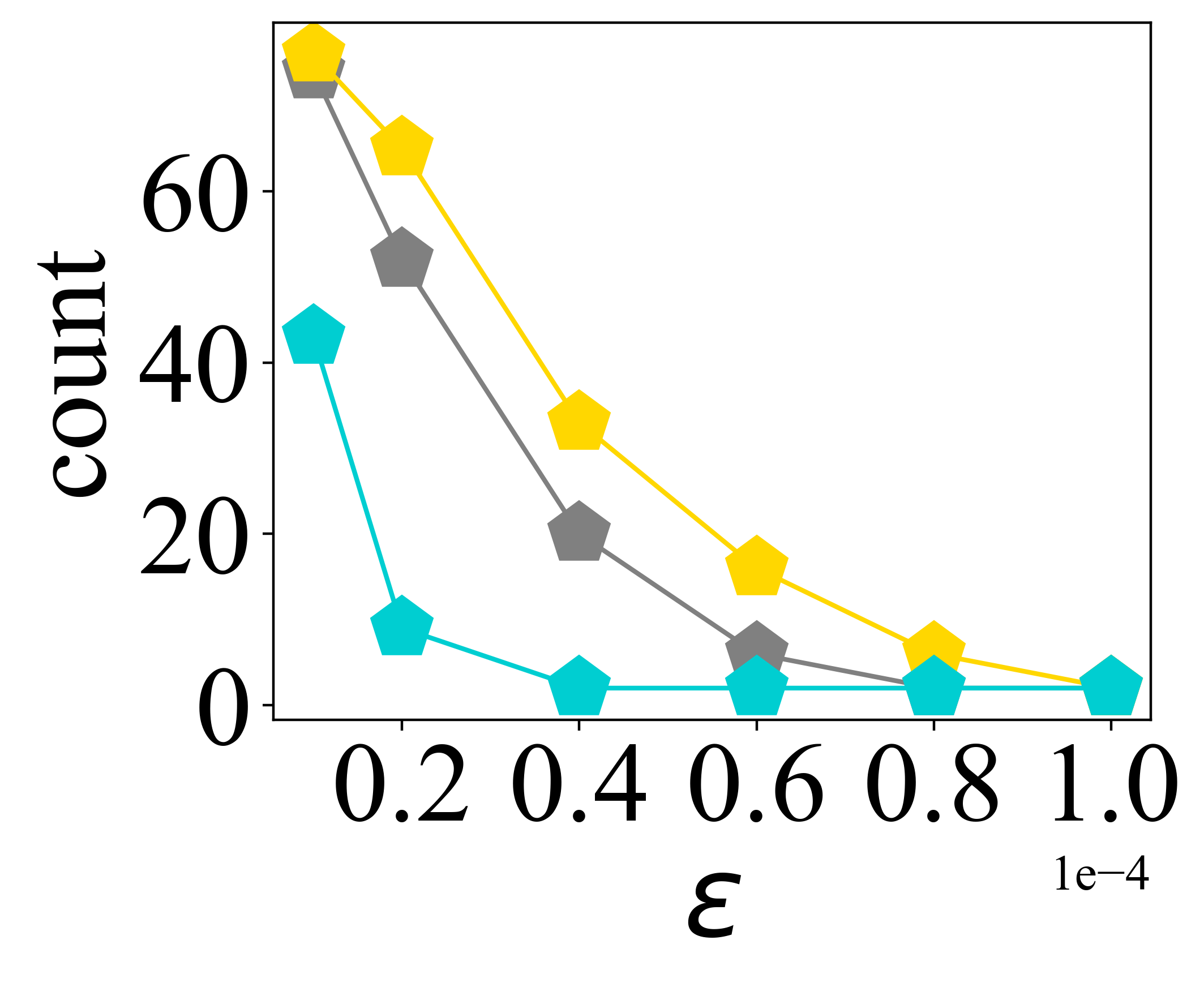}}
		\label{fig:parameter_grid_loss_grid}}
	\subfigure[][{\scriptsize Loss}]{
		\scalebox{0.19}[0.19]{\includegraphics{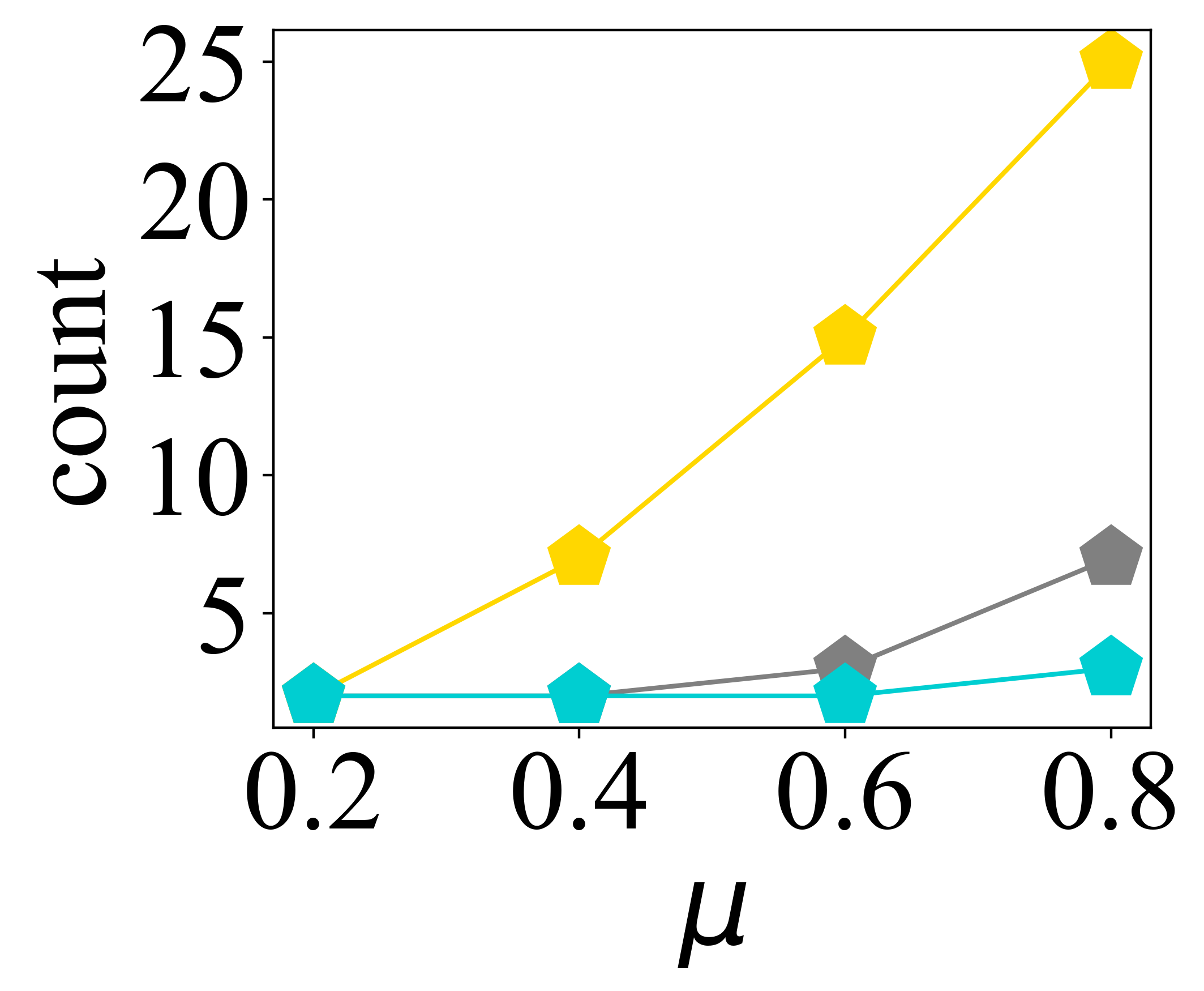}}
		\label{fig:parameter_mu_loss_grid}}
	\subfigure[][{\scriptsize Loss}]{
		\scalebox{0.19}[0.19]{\includegraphics{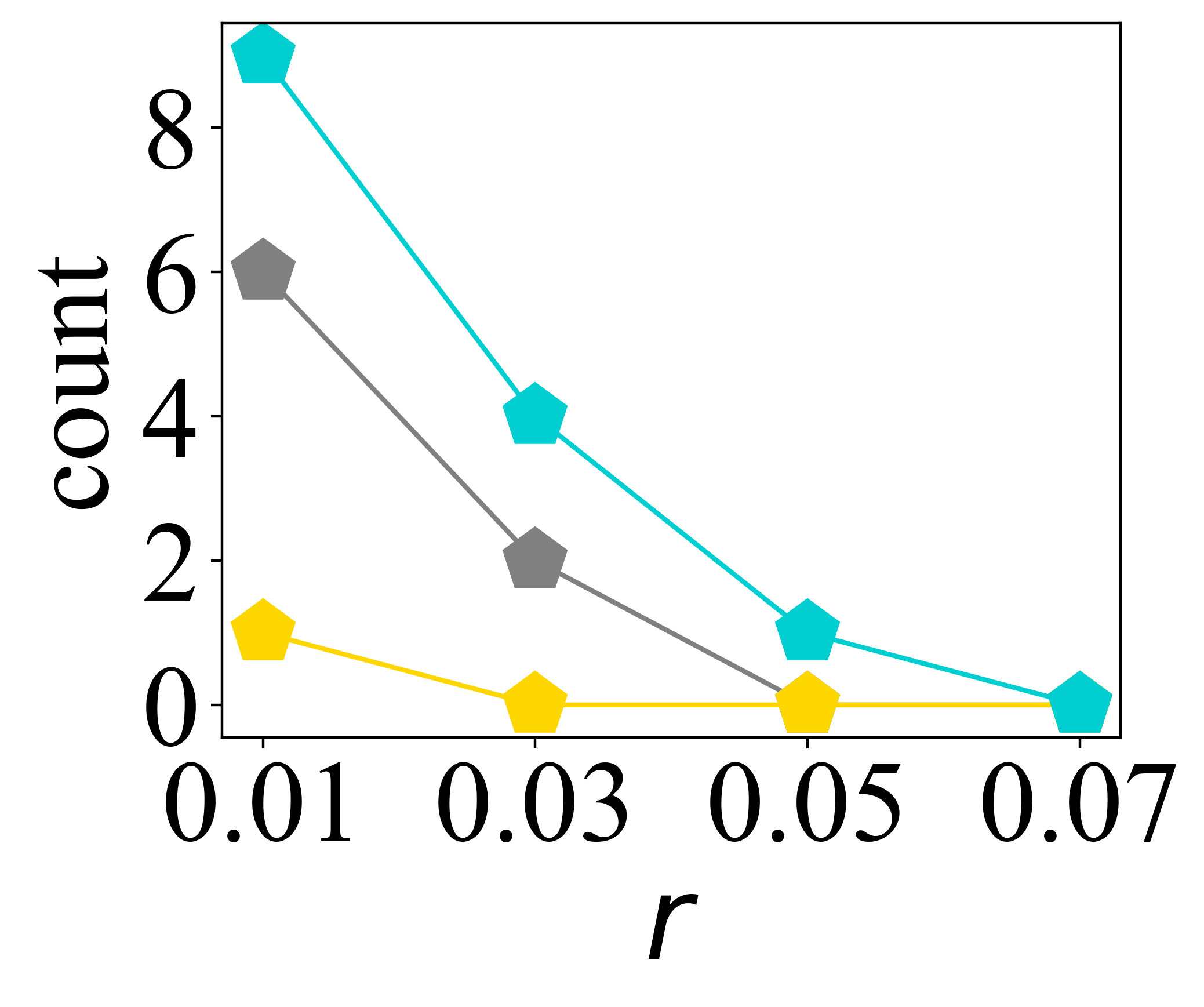}}
		\label{fig:parameter_rate_loss_grid}}
	\caption{\small Effect of different parameters}
	\label{fig:parameters}
\end{figure*}

\noindent\textbf{Effect of $r$, $\varepsilon$ and $\mu$ on pruning algorithm.}
The pruning condition setting can significantly impact the final pruning effect. Most of the trajectories can pass the pruning condition when the pruning condition is set loosely, which will lead to a sharp increase in the time to calculate the optimal subtrajectory. In contrast, when the strict pruning condition will lead to filtering out data trajectories similar to the query trajectory; thus the algorithm cannot find the optimal subtrajectory for some query trajectories. Therefore, two metrics (i.e., time and loss) are used to determine the impact of these three parameters $r$, $\varepsilon$, and $\mu$ on the pruning algorithm. A shorter run time of the algorithm means that the algorithm can filter out more trajectories that are not similar, while loss indicates the number of query trajectories for which the optimal subtrajectory cannot be found with the current parameter settings. We conducted experiments on the Xi'an dataset with different parameter settings based on ERP as a distance function using the CMA algorithm, and the experimental results are shown in Figure \ref{fig:parameters}. \revision{From the experimental results presented in Figure \ref{fig:parameter_grid_time_grid} and \ref{fig:parameter_grid_loss_grid}, it is clear that as $\varepsilon$ increases, the number of trajectories remained after pruning also increases. This leads to an increase in the amount of time needed for subsequent searches. However, as $\varepsilon$ increases, the number of query trajectories for which the optimal subtrajectory is not found decreases. At $\varepsilon=0.8e^{-4}$, the optimal subtrajectories are found for all query trajectories. To minimize the execution time of the framework and guarantee the optimal subtrajectories can be searched out, we set $\varepsilon$ to $0.8e^{-4}$ as the default value in the experiments of our draft.
Similarly, the experimental results show that as $\mu$ increases, more trajectories will be pruned, which leads to a higher probability of missing the optimal subtrajectories. In contrast, if $\mu$ decreases, the time cost required for the subsequent searching increases. Therefore, to minimize the execution time of the framework and guarantee the optimal subtrajectories can be safely searched out, we set $\mu$ to 0.4 as the default value, based on the experimental results.
Figure \ref{fig:parameter_rate_time_grid} and \ref{fig:parameter_rate_loss_grid} illustrates that as $r$ becomes larger, the time overhead of the pruning process also increases. This is because subtrajectory distances are approximated by a more accurate lower bound, making it less likely to miss the optimal subtrajectory. On the other hand, when $r$ is small, the pruning process becomes less time-consuming, but it becomes easier to overlook the optimal subtrajectories. Therefore, we set the sampling rate $r$ to $0.05$ in the experiments presented in this paper. 
}
\begin{figure*}[h!]
	\centering
	\subfigure{
		\scalebox{0.40}[0.40]{\includegraphics{figures/legend.png}}}\hfill\\
	\addtocounter{subfigure}{-1}\vspace{-2ex}
	\subfigure[][{\scriptsize SURS (Porto)}]{
		\scalebox{0.19}[0.19]{\includegraphics{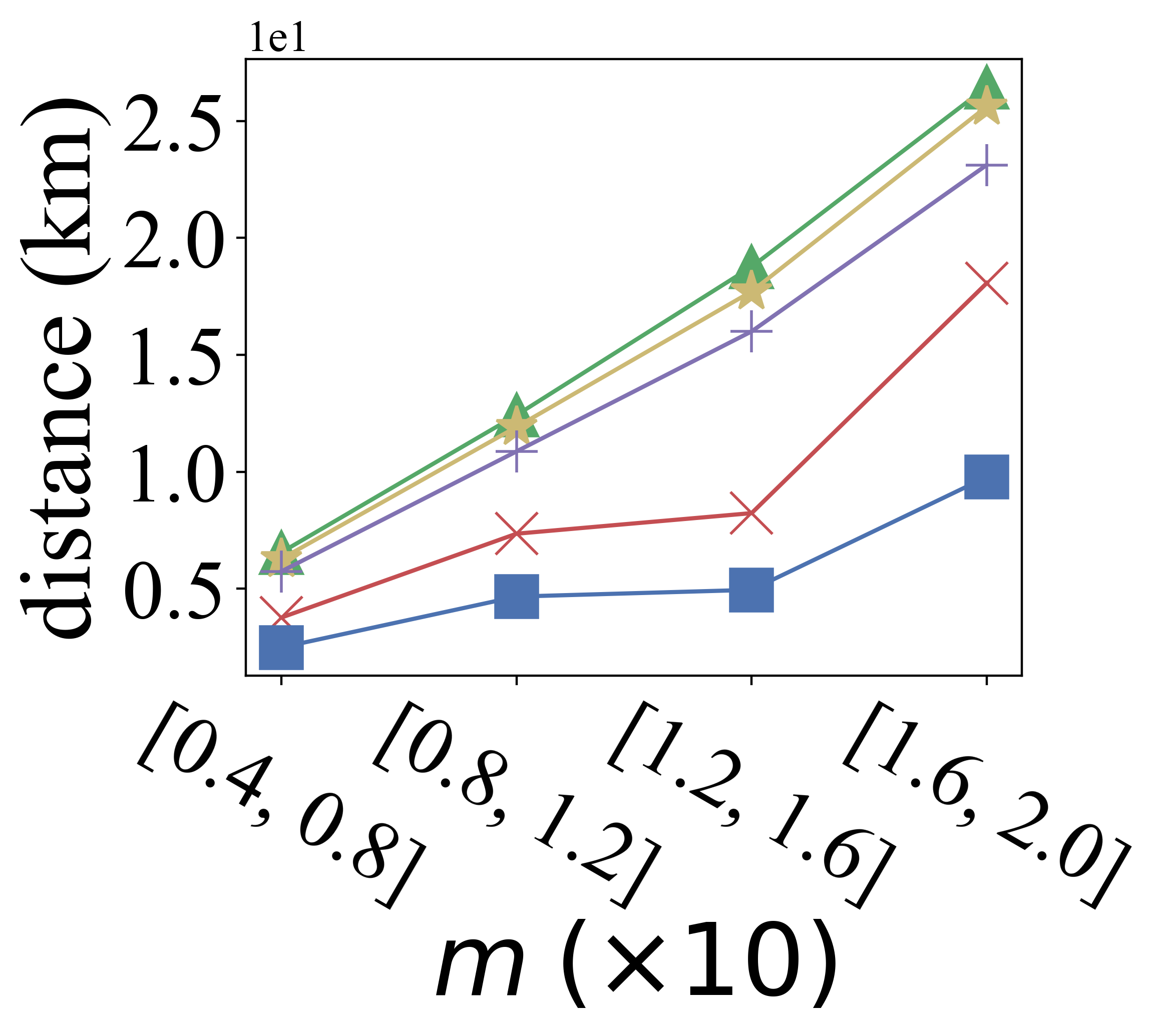}}
		\label{fig:road_porto_SURS_score}}
	\subfigure[][{\scriptsize SURS (Xi'an)}]{
		\scalebox{0.19}[0.19]{\includegraphics{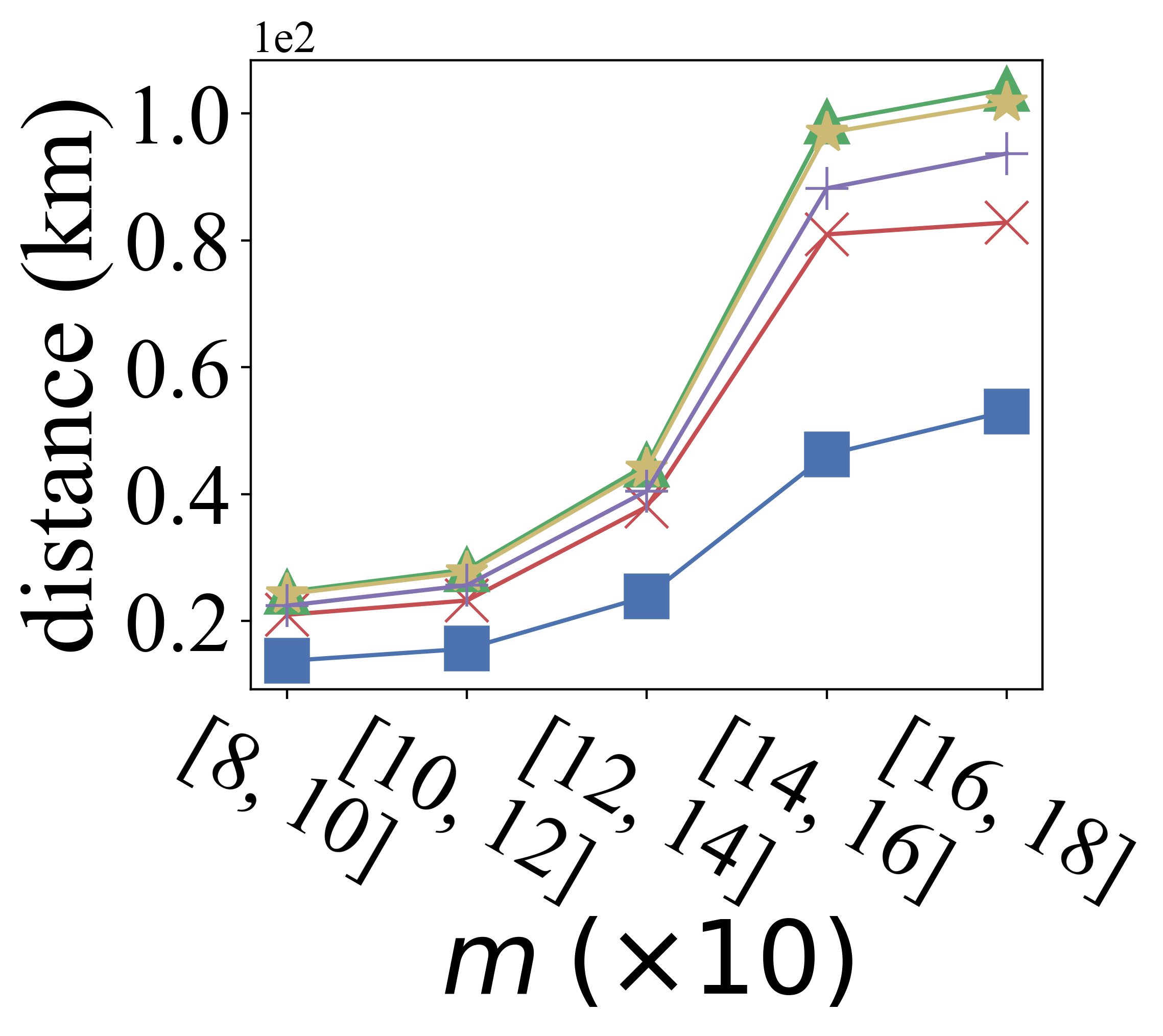}}
		\label{fig:road_xian_SURS_score}}
	\subfigure[][{\scriptsize SURS (Beijing)}]{
		\scalebox{0.19}[0.19]{\includegraphics{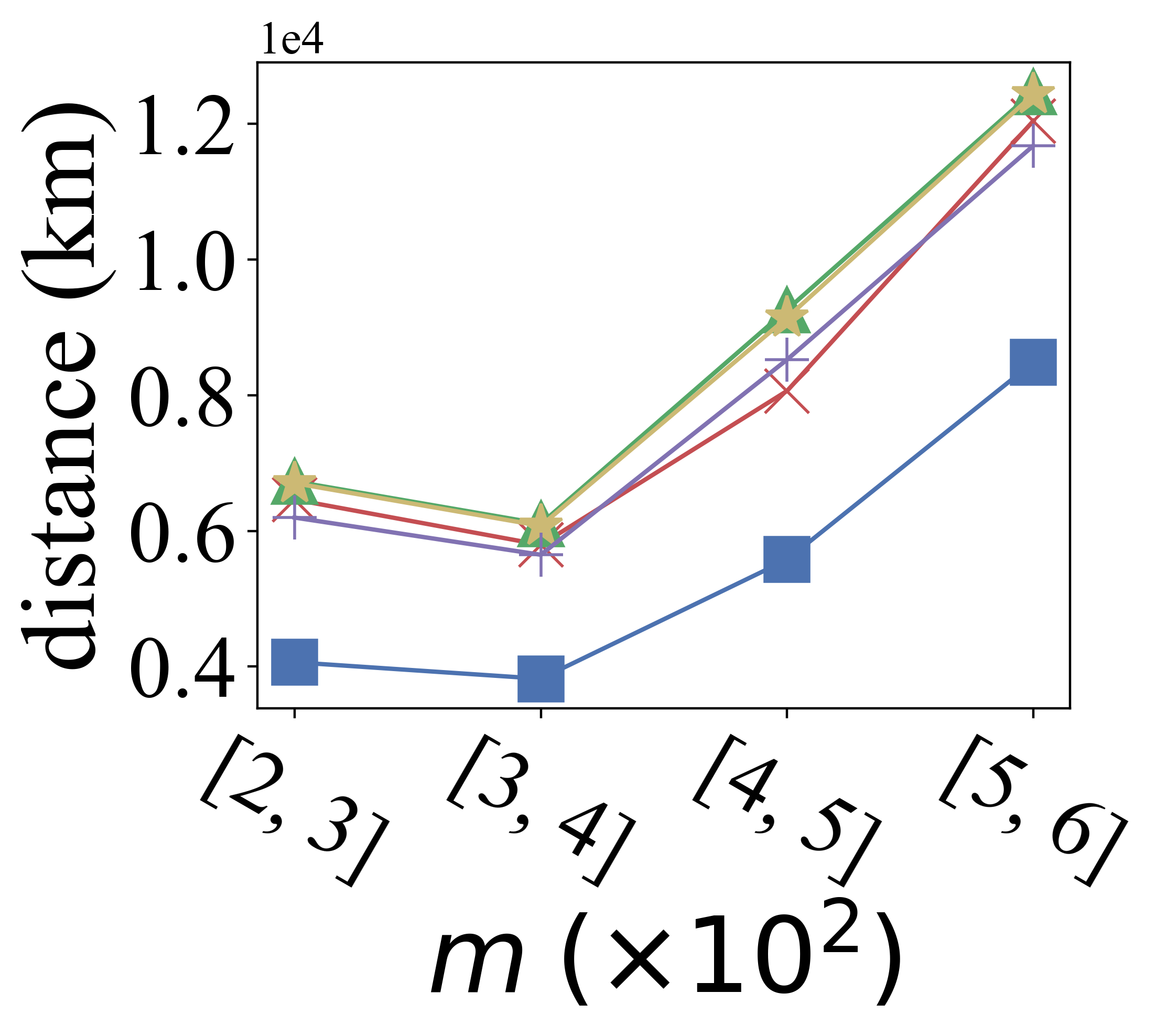}}
		\label{fig:road_beijing_SURS_score}}
	\subfigure[][{\scriptsize SURS (Porto)}]{
		\scalebox{0.19}[0.19]{\includegraphics{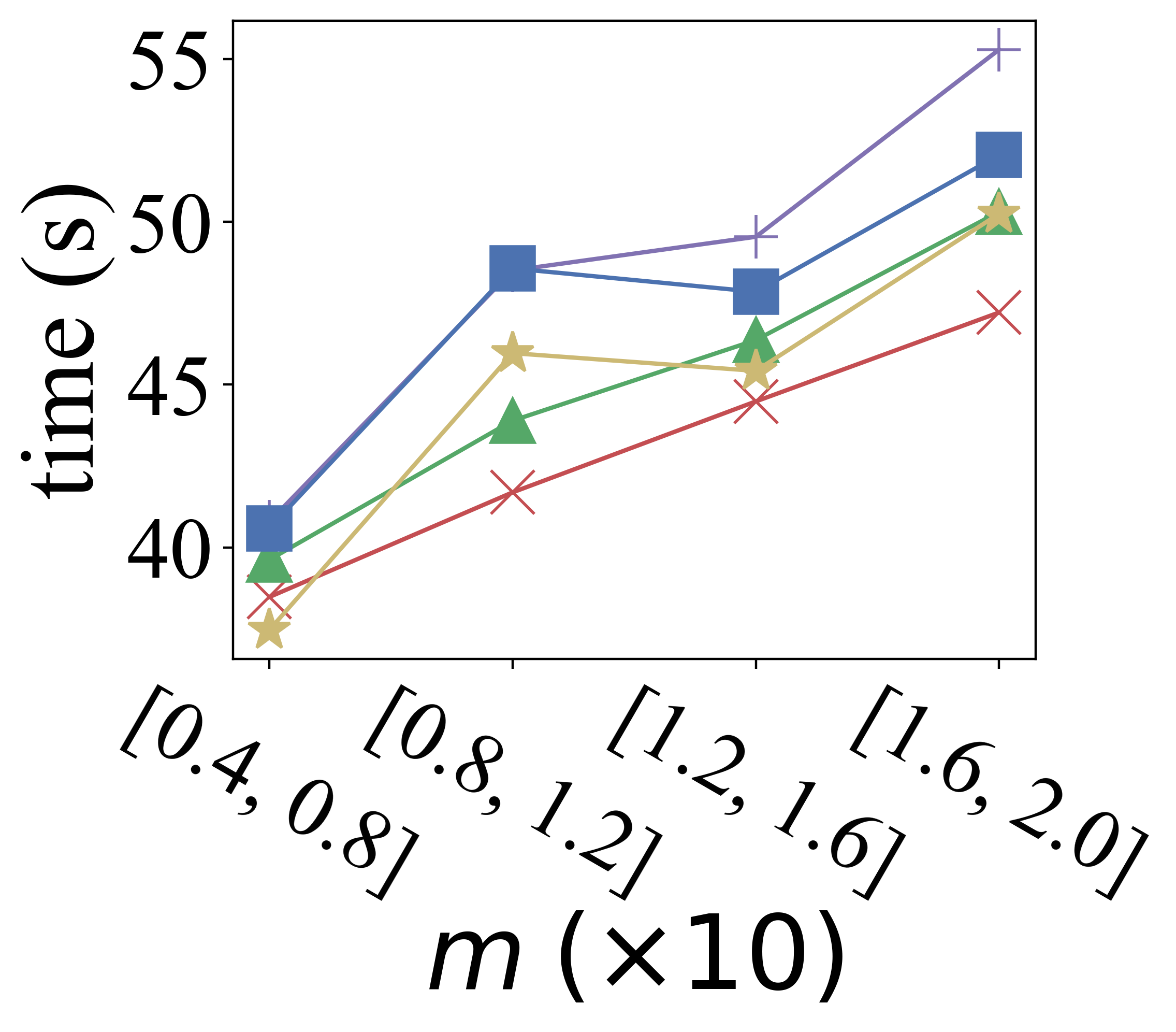}}
		\label{fig:road_porto_SURS}}
	\subfigure[][{\scriptsize SURS (Xi'an)}]{
		\scalebox{0.19}[0.19]{\includegraphics{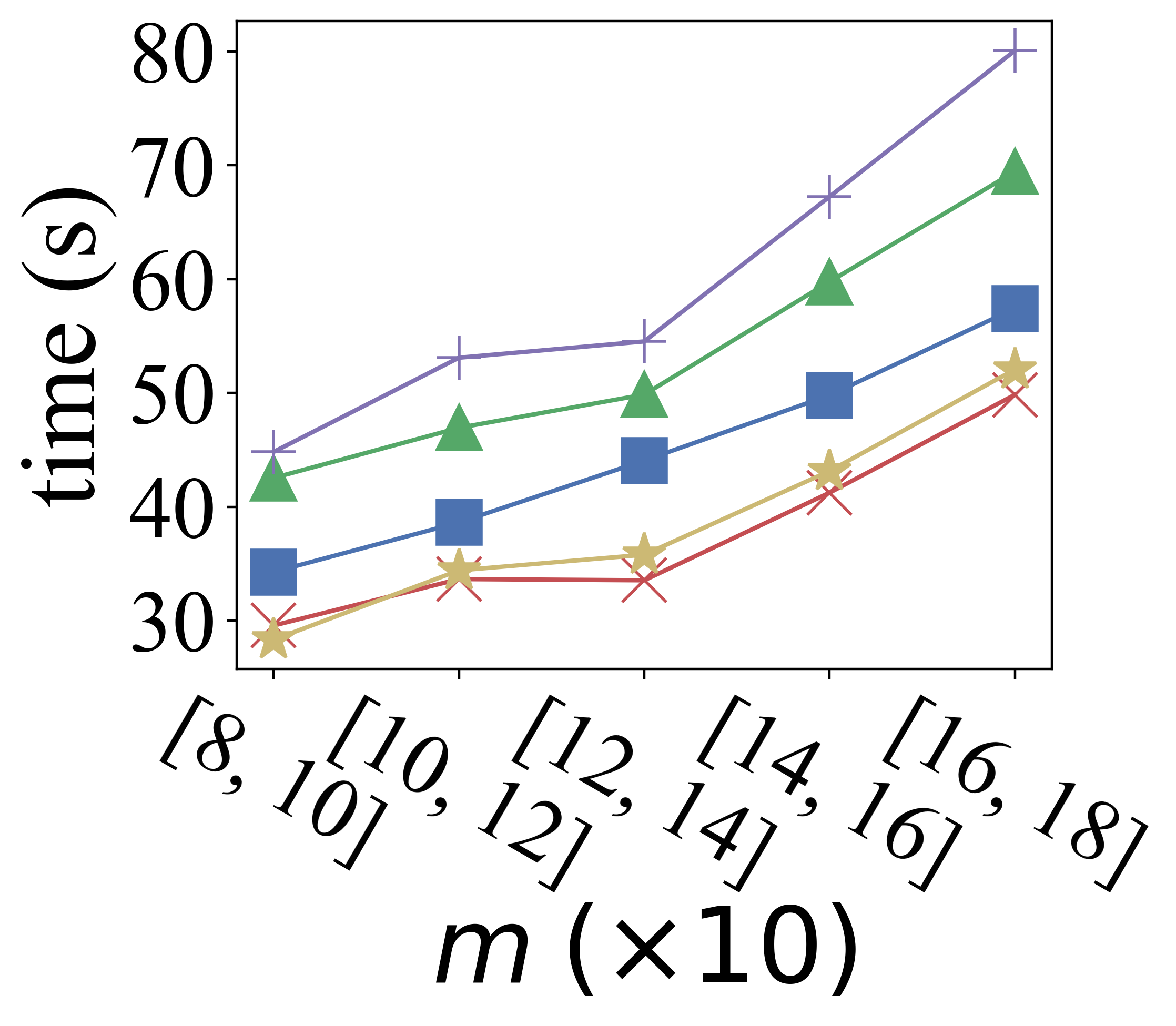}}
		\label{fig:road_xian_SURS}}
	\subfigure[][{\scriptsize SURS (Beijing)}]{
		\scalebox{0.19}[0.19]{\includegraphics{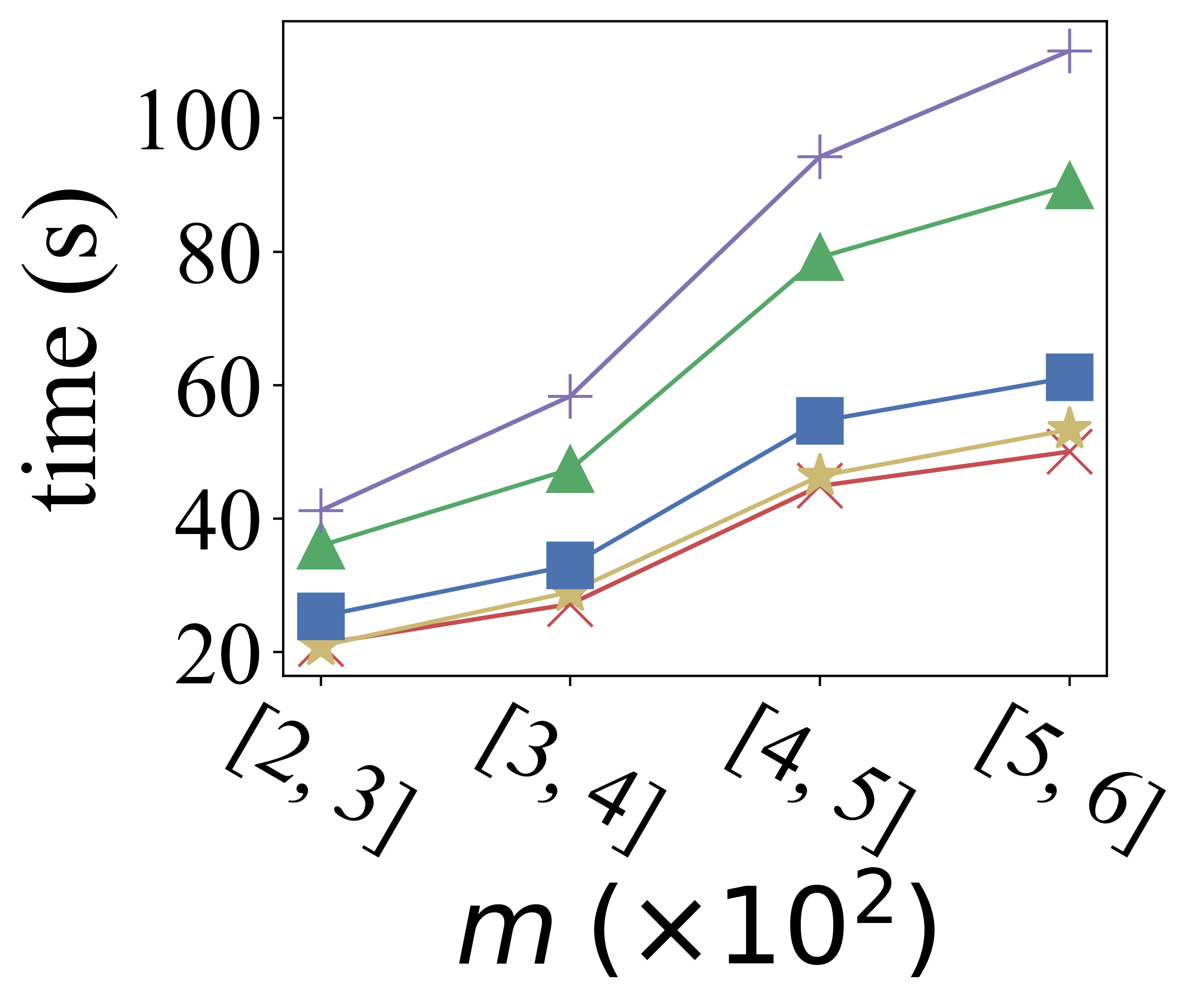}}
		\label{fig:road_beijing_SURS}}
	\\
	\subfigure[][{\scriptsize NetEDR (Porto)}]{
		\scalebox{0.19}[0.19]{\includegraphics{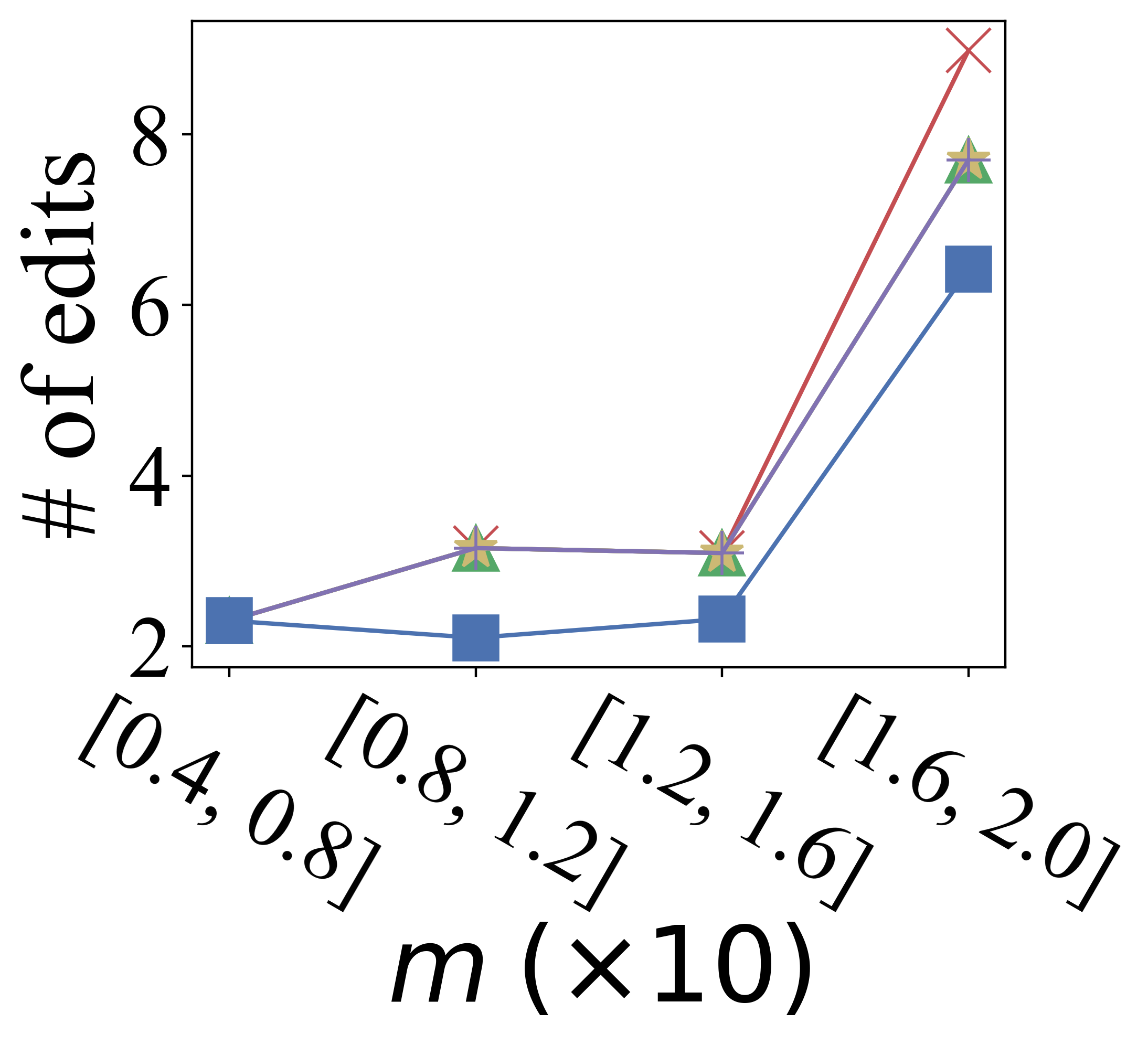}}
		\label{fig:road_porto_NetEDR_score}}
	\subfigure[][{\scriptsize NetEDR (Xi'an)}]{
		\scalebox{0.19}[0.19]{\includegraphics{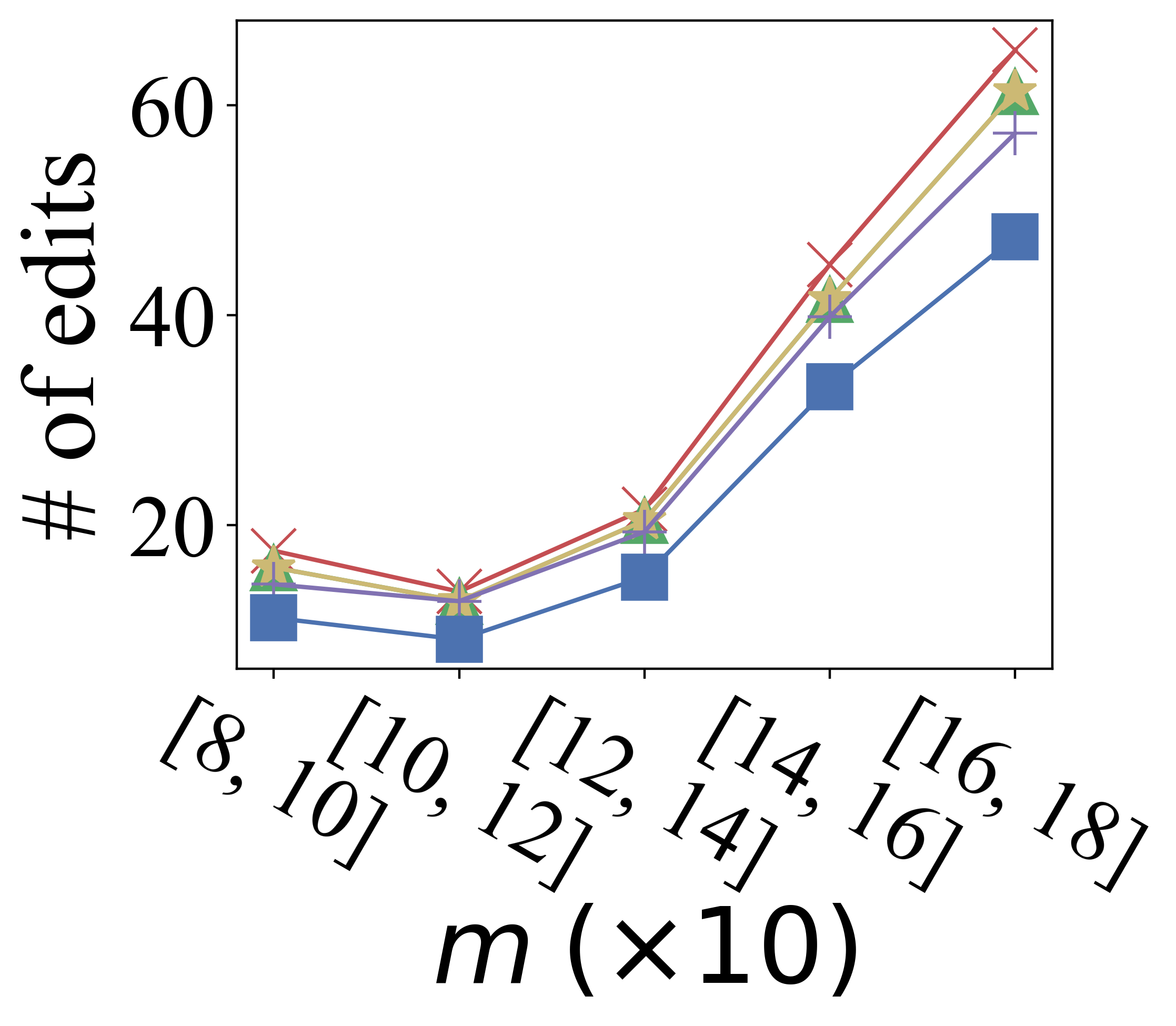}}
		\label{fig:road_xian_NetEDR_score}}
	\subfigure[][{\scriptsize NetEDR (Beijing)}]{
		\scalebox{0.19}[0.19]{\includegraphics{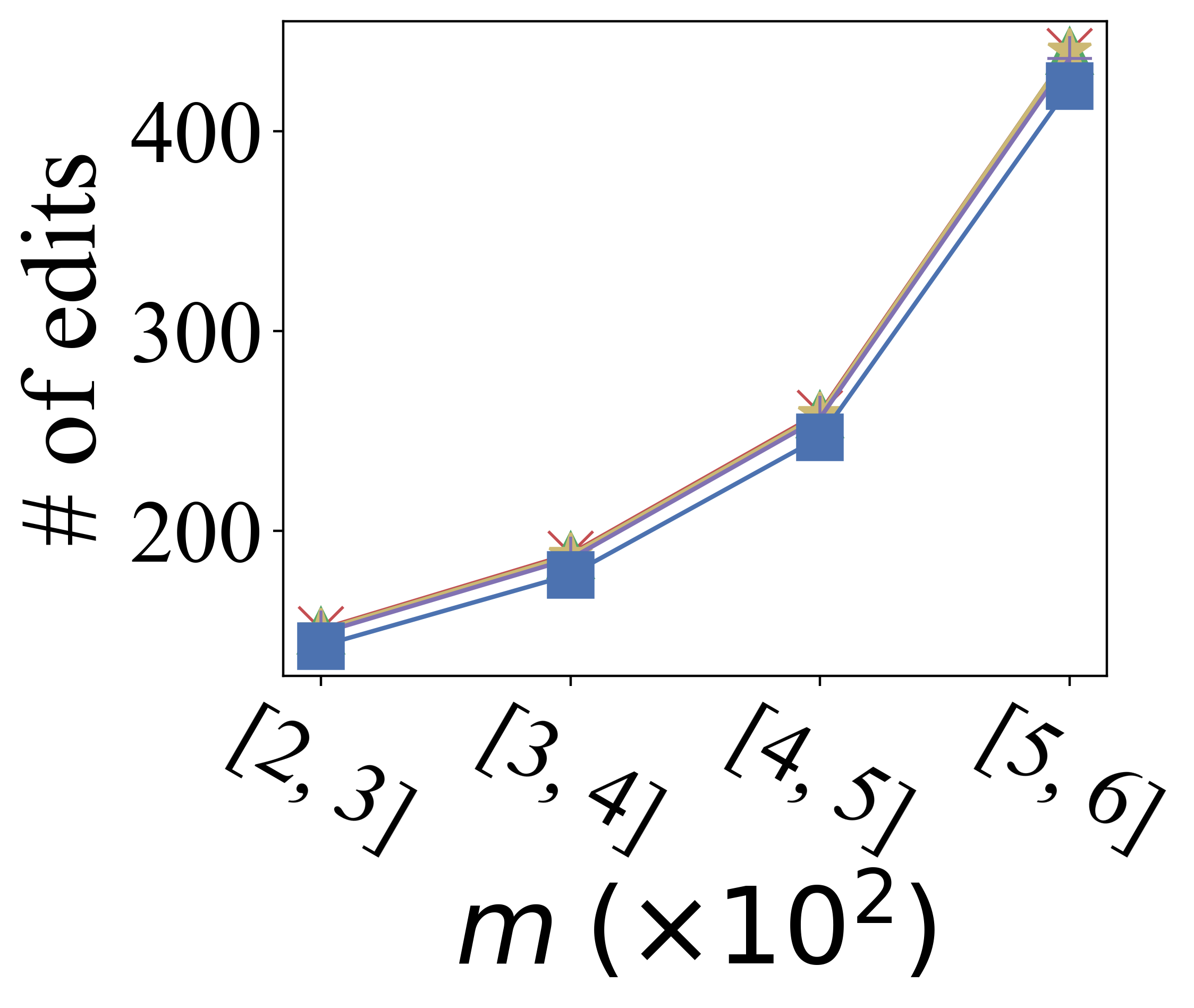}}
		\label{fig:road_beijing_NetEDR_score}}
	\subfigure[][{\scriptsize NetEDR (Porto)}]{
		\scalebox{0.19}[0.19]{\includegraphics{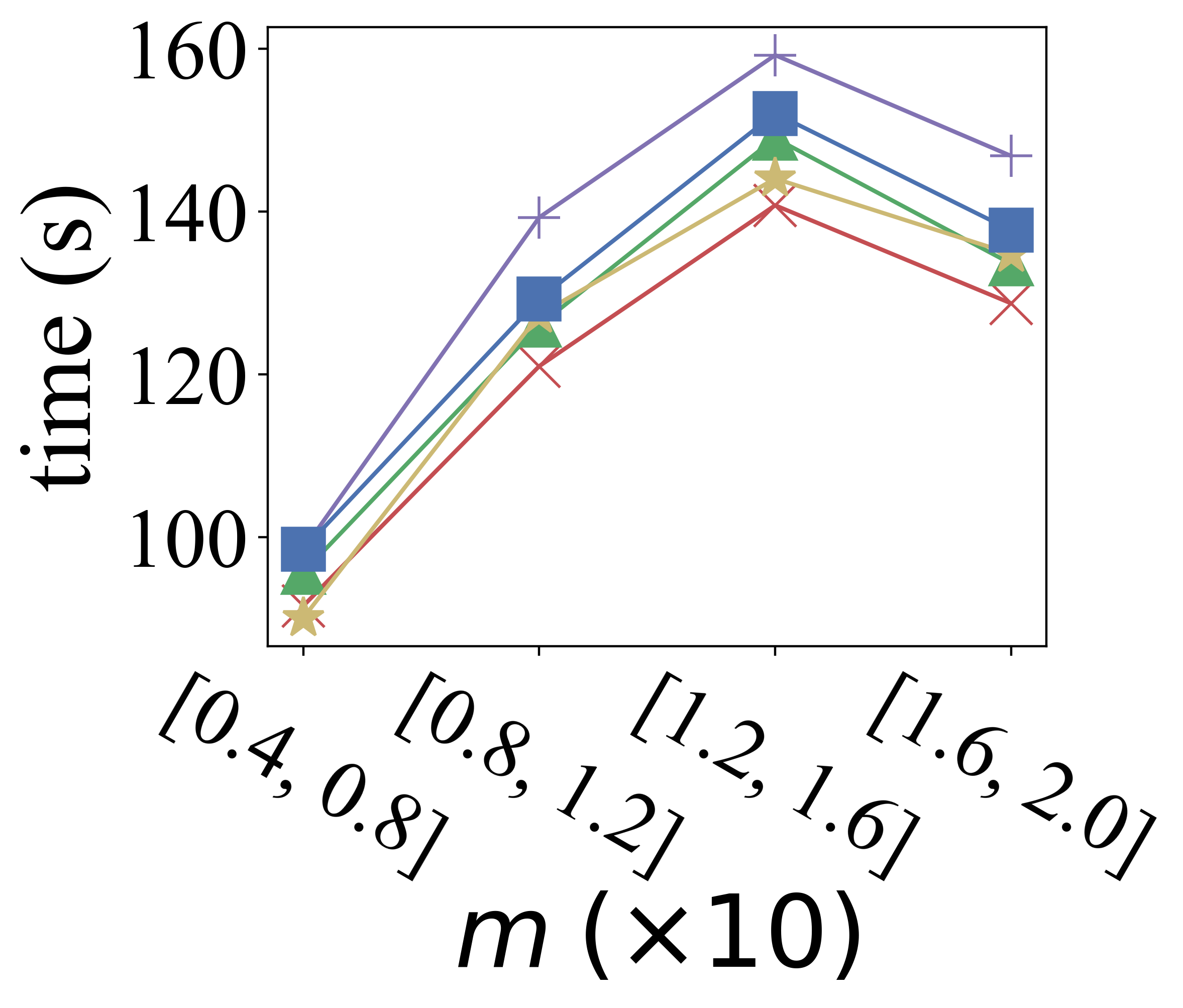}}
		\label{fig:road_porto_NetEDR}}
	\subfigure[][{\scriptsize NetEDR (Xi'an)}]{
		\scalebox{0.19}[0.19]{\includegraphics{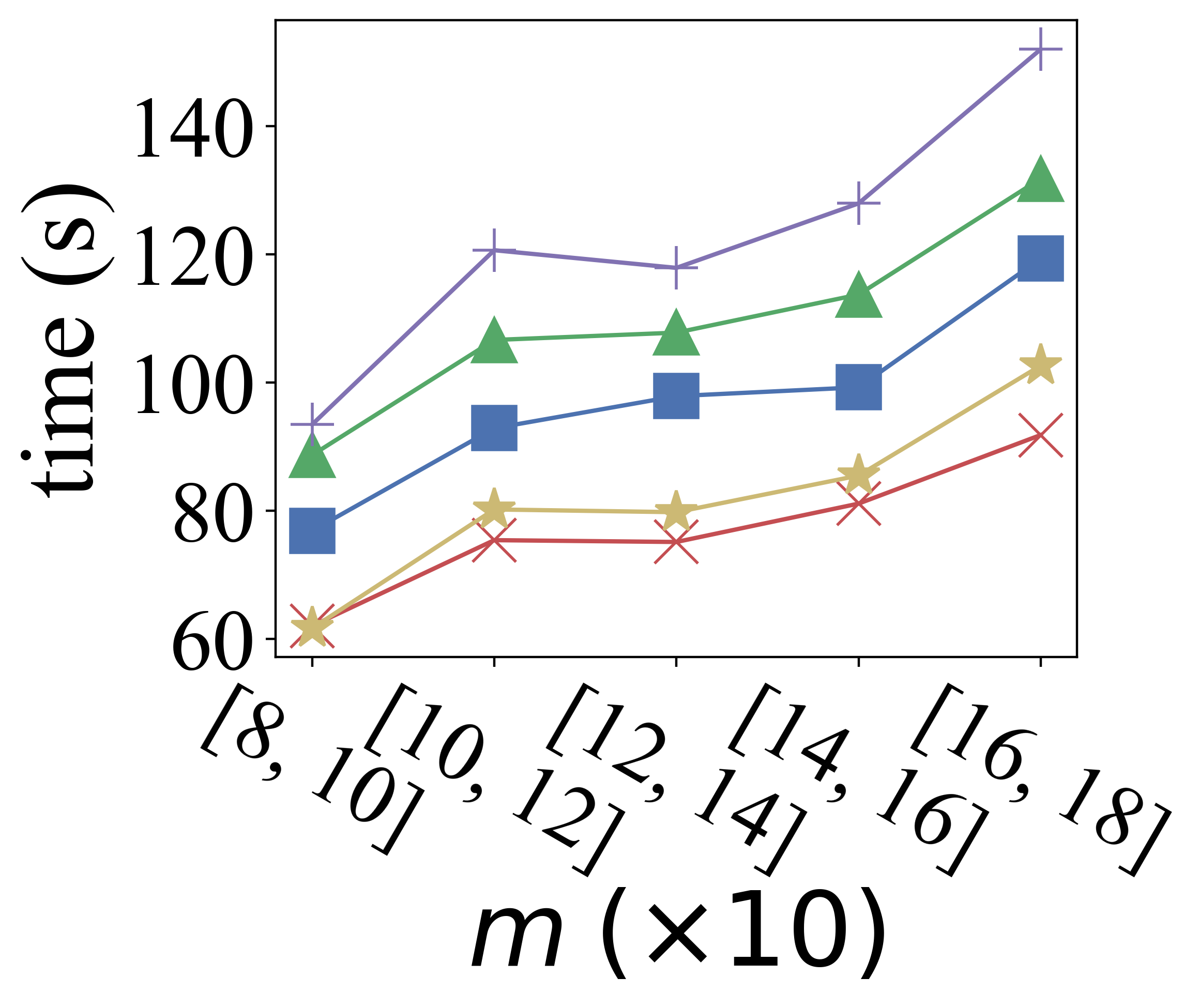}}
		\label{fig:road_xian_NetEDR}}
	\subfigure[][{\scriptsize NetEDR (Beijing)}]{
		\scalebox{0.19}[0.19]{\includegraphics{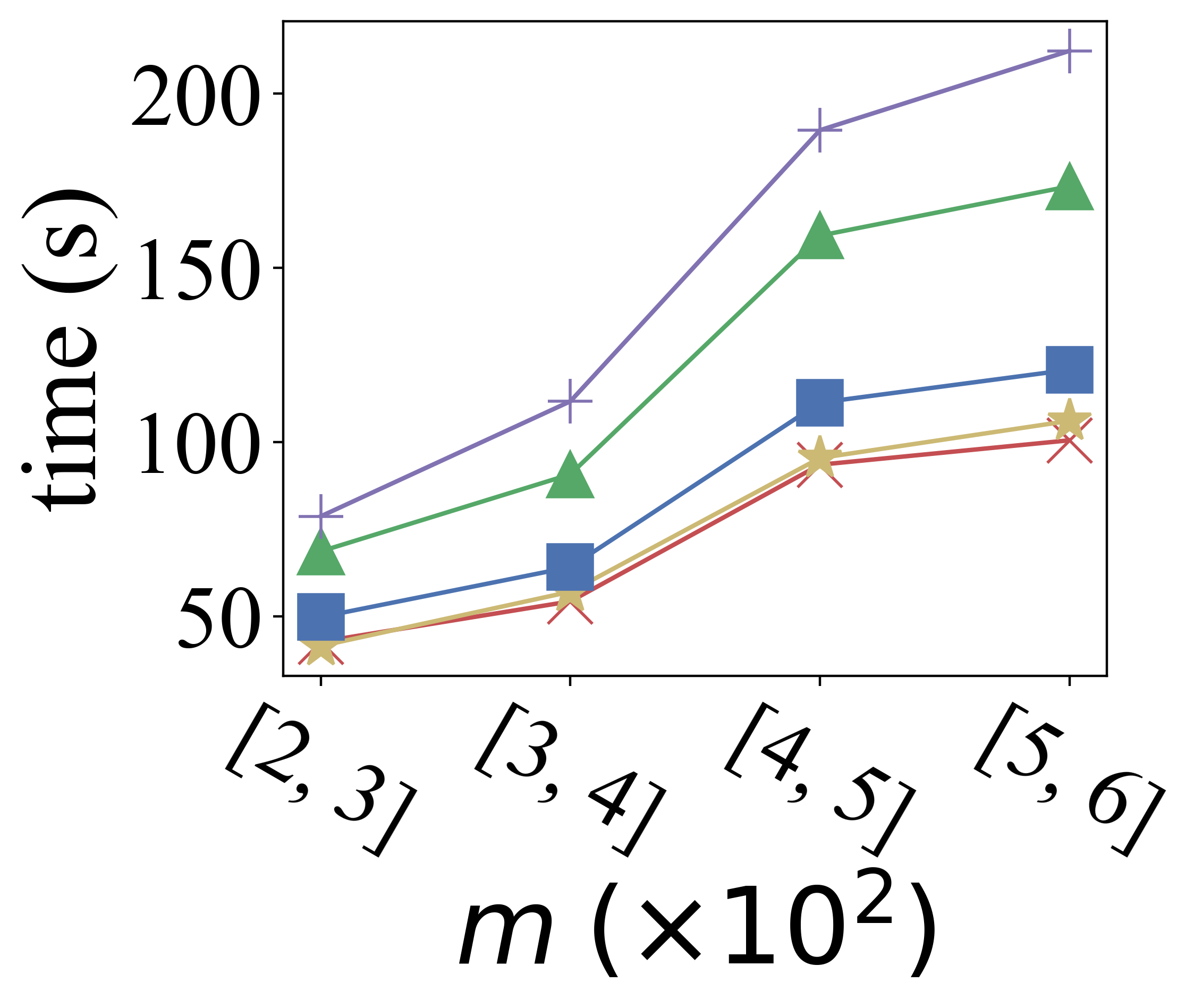}}
		\label{fig:road_beijing_NetEDR}}
	\\
	\subfigure[][{\scriptsize NetERP (Porto)}]{
		\scalebox{0.19}[0.19]{\includegraphics{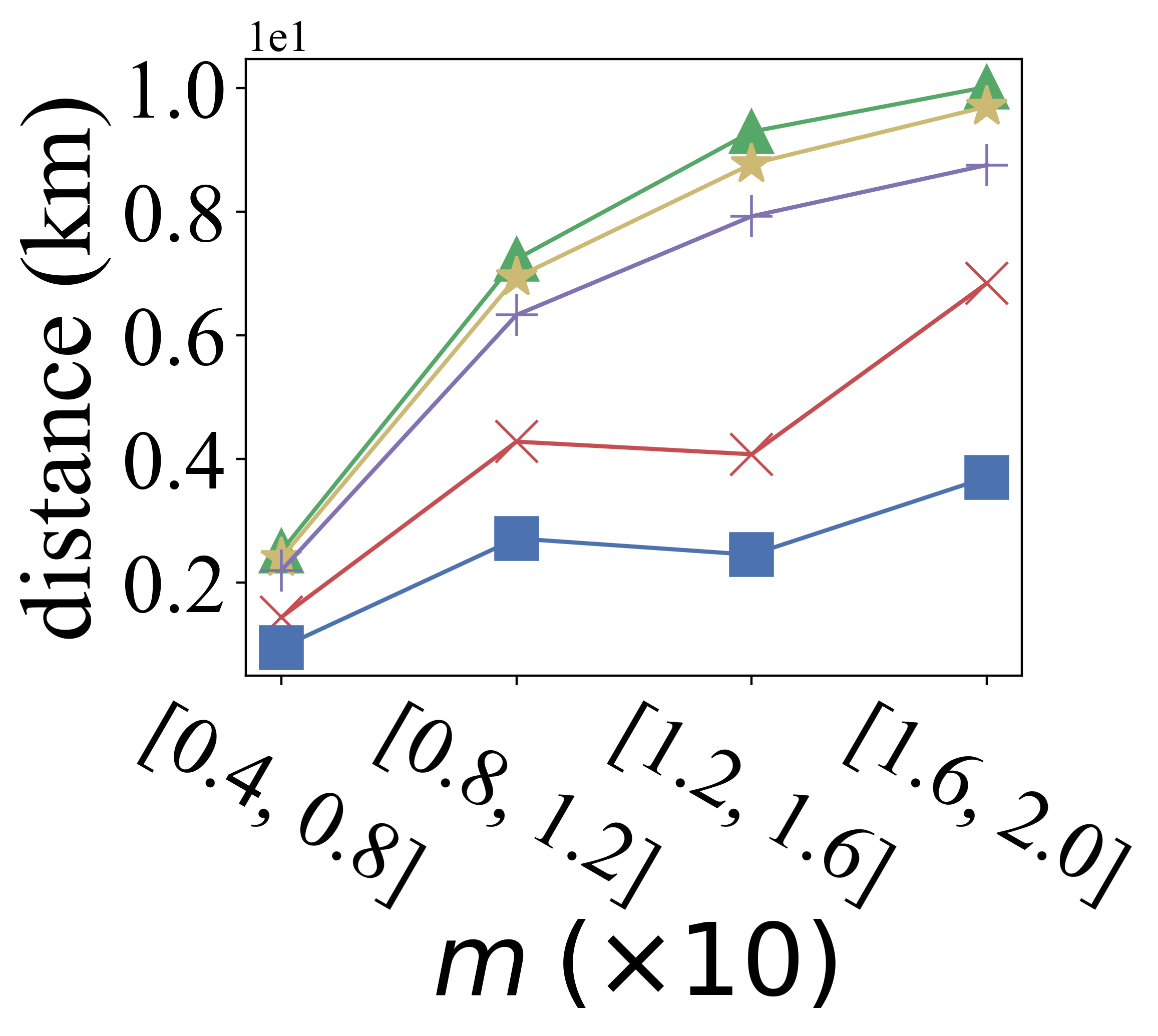}}
		\label{fig:road_porto_NetERP_score}}
	\subfigure[][{\scriptsize NetERP (Xi'an)}]{
		\scalebox{0.19}[0.19]{\includegraphics{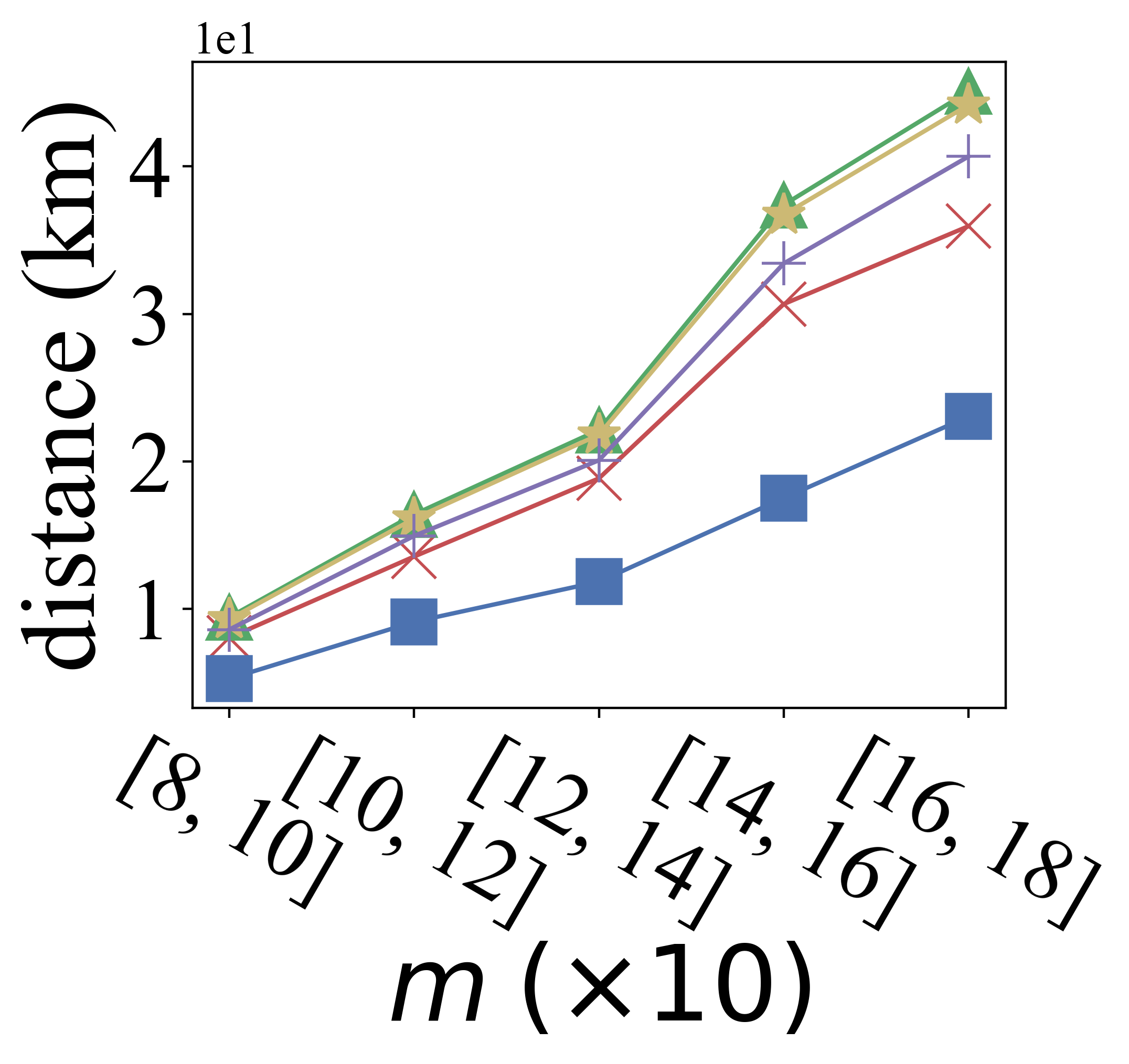}}
		\label{fig:road_xian_NetERP_score}}
	\subfigure[][{\scriptsize NetERP (Beijing)}]{
		\scalebox{0.19}[0.19]{\includegraphics{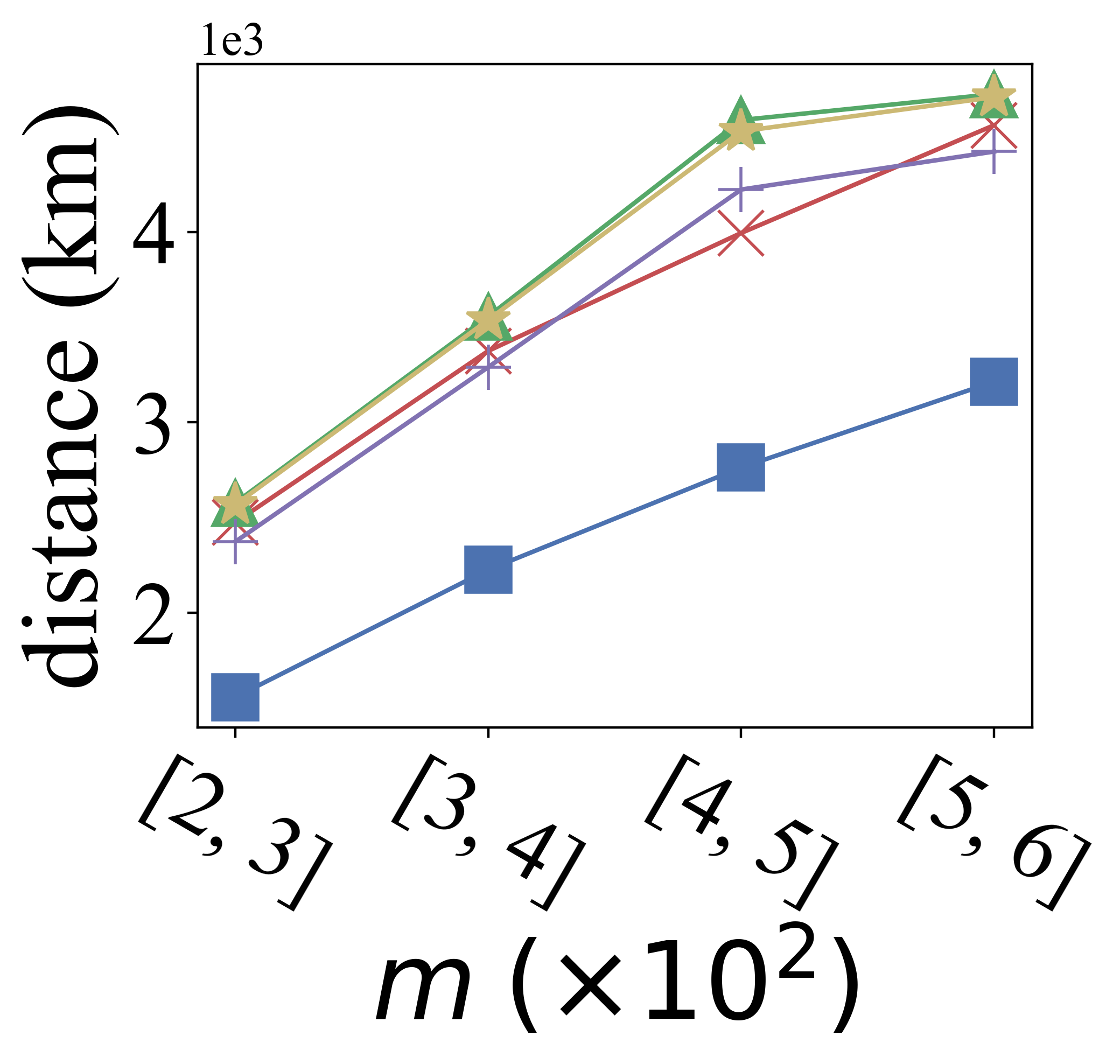}}
		\label{fig:road_beijing_NetERP_score}}
	\subfigure[][{\scriptsize NetERP (Porto)}]{
		\scalebox{0.19}[0.19]{\includegraphics{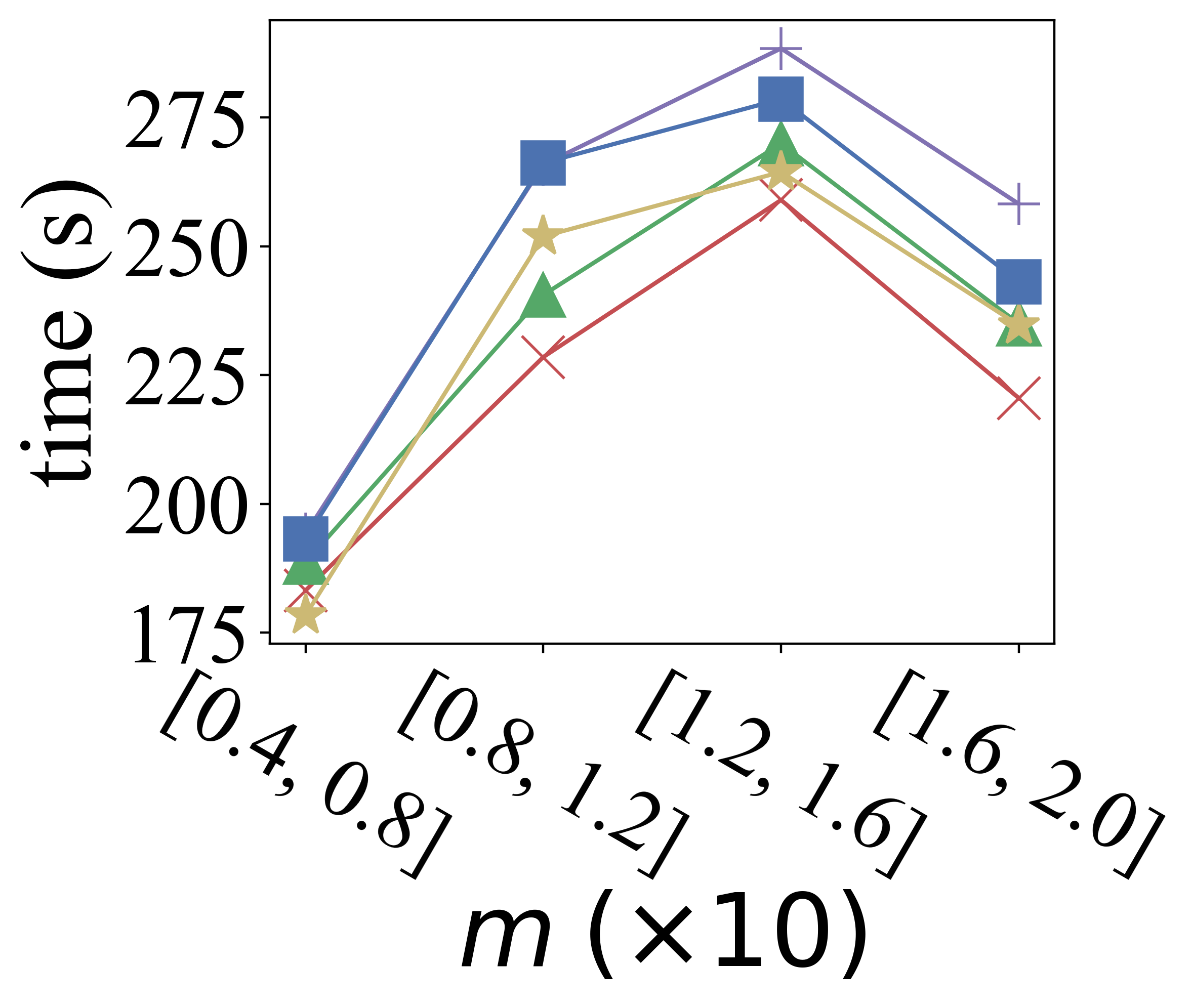}}
		\label{fig:road_porto_NetERP}}
	\subfigure[][{\scriptsize NetERP (Xi'an)}]{
		\scalebox{0.19}[0.19]{\includegraphics{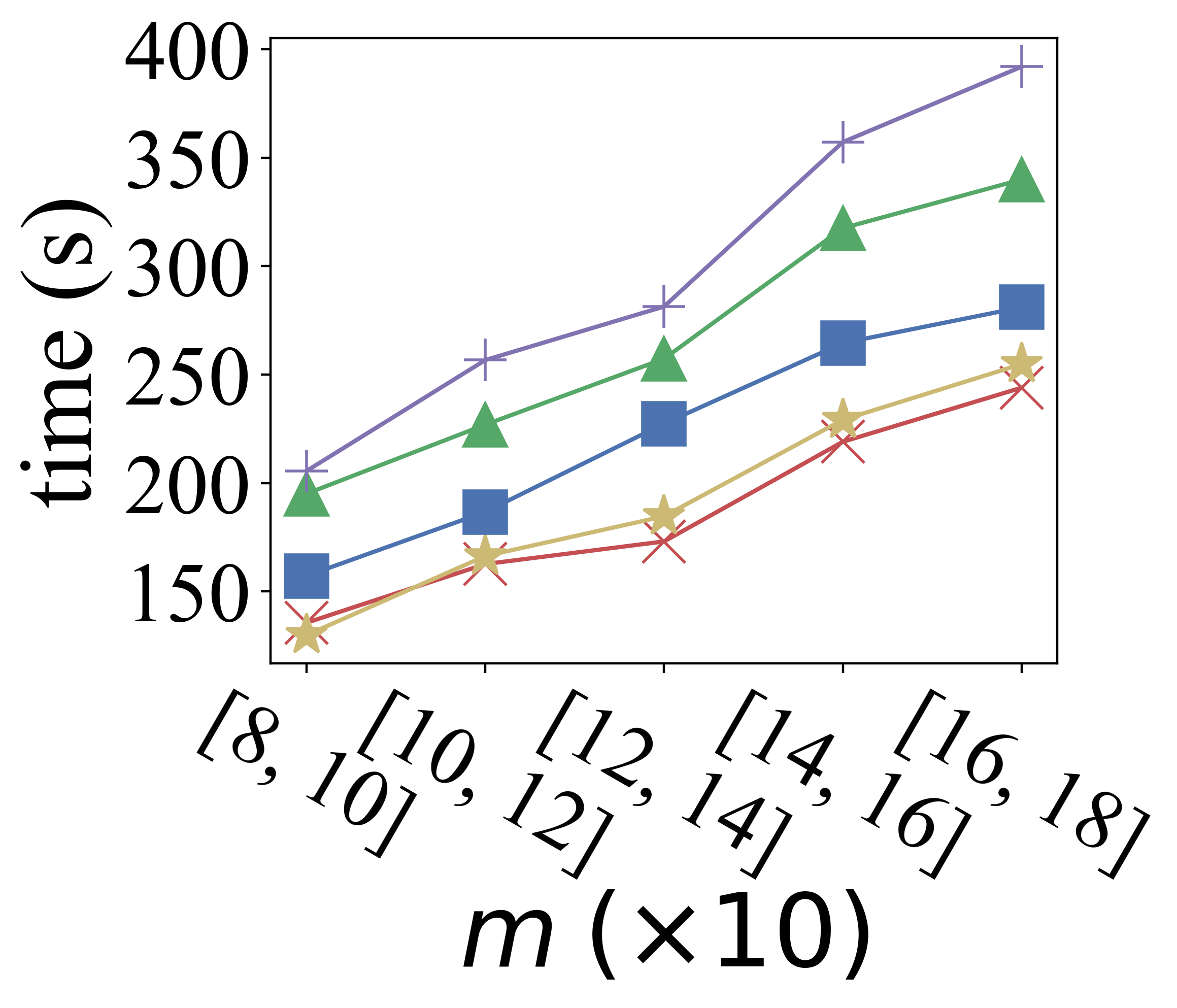}}
		\label{fig:road_xian_NetERP}}
	\subfigure[][{\scriptsize NetERP (Beijing)}]{
		\scalebox{0.19}[0.19]{\includegraphics{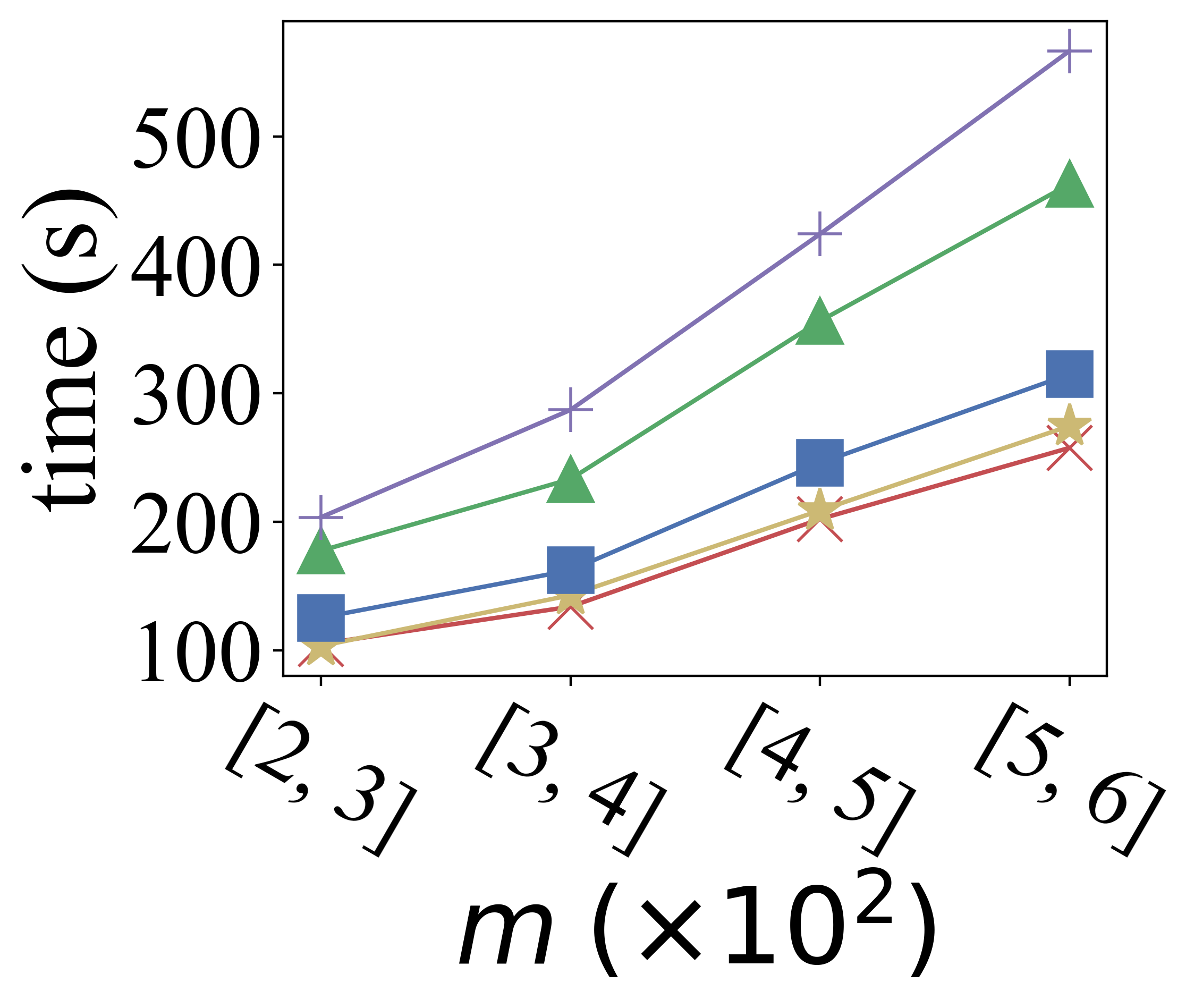}}
		\label{fig:road_beijing_NetERP}}
	\caption{\small Effectiveness and efficiency with varying query lengths}
	\label{fig:efficiencyLengthNet}
\end{figure*}
\section{Experiment on Other Distance Function}
\revision{
Our algorithm proposed in this paper applies not only to trajectories represented by GPS points but also to trajectories represented by other forms such as each point or edge on the road network. Experiments are conducted to test the performance of our method under different representations of trajectory points. Three different distance functions are employed:

\begin{enumerate}
	\item NetERP: In NetERP, each trajectory point is a point on the road network. Unlike ERP, the distance between trajectory points in NetERP is the distance on the road network.
	\item NetEDR: In NetEDR, each trajectory point is a point on the road network. The cost of inserting, deleting, or replacing a trajectory point is 1.
	\item SURS: In SURS, each trajectory consists of a series of edges on the road network. The cost of inserting or deleting a trajectory edge is equal to the weight of the corresponding edge. The cost of replacing one trajectory edge with another is equal to the sum of the weights of the two corresponding edges on the road network.
\end{enumerate}

NetERP, NetEDR, and SURS are all special cases of WED and can be applied to the algorithm proposed in this paper. We used RoutingKit~\footnote{https://github.com/RoutingKit/RoutingKit} to convert our original GPS dataset into a road network. We compared the performance of these three algorithms on different datasets and show the experimental results in Figure \ref{fig:efficiencyLengthNet}:

The experimental results show that the algorithm proposed in this paper can achieve optimal results, regardless of the distance function used. However, using NetEDR and NetERP to search for subtrajectories has a relatively high time complexity since they require the use of the shortest path algorithm to calculate the distance between two trajectory points. As the value of $m$ increases, the time required to find the optimal subtrajectory also increases, and the distance between the optimal solution and the query trajectory continues to increase.

Different distance functions depend on different representations of trajectories. As long as the distance function is location-insensitive, the CMA algorithm can be used to find the optimal subtrajectory for any trajectory representation.
}
\section{Effect of $K$ on Our Framework}
\revision{We conduct additional experiments to investigate the impact of $K$ on the algorithm's execution time. The results shown in Figure \ref{fig:varyK} indicate that the search framework proposed in this paper can still achieve good performance even when $K$ is large. Figure \ref{fig:varyK} shows the summation of distances between the $K$ found sub-trajectories and the query trajectory and the running times. During the search process, we maintain a heap of size $K$. Similar to the process of top-k similar subtrajectory search in the existing VLDB2020 work \cite{WangLCL20}, each time when we invoke a similar subtrajectory search algorithm (such as CMA and POS) on a new data trajectory, we insert the searched optimal subtrajectory of the data trajectory into the heap. The time required to maintain the heap of size $K$ is negligible compared to the similar subtrajectory search algorithms, thus the overall time complexity of the search process mainly depends on the number of invocations of the similar subtrajectory search algorithm.}
\begin{figure*}[h!]
	\centering
	\subfigure{
		\scalebox{0.40}[0.40]{\includegraphics{figures/legend.png}}}\hfill\\
	\addtocounter{subfigure}{-1}\vspace{-2ex}
	\subfigure[][{\scriptsize EDR (Porto)}]{
		\scalebox{0.2}[0.2]{\includegraphics{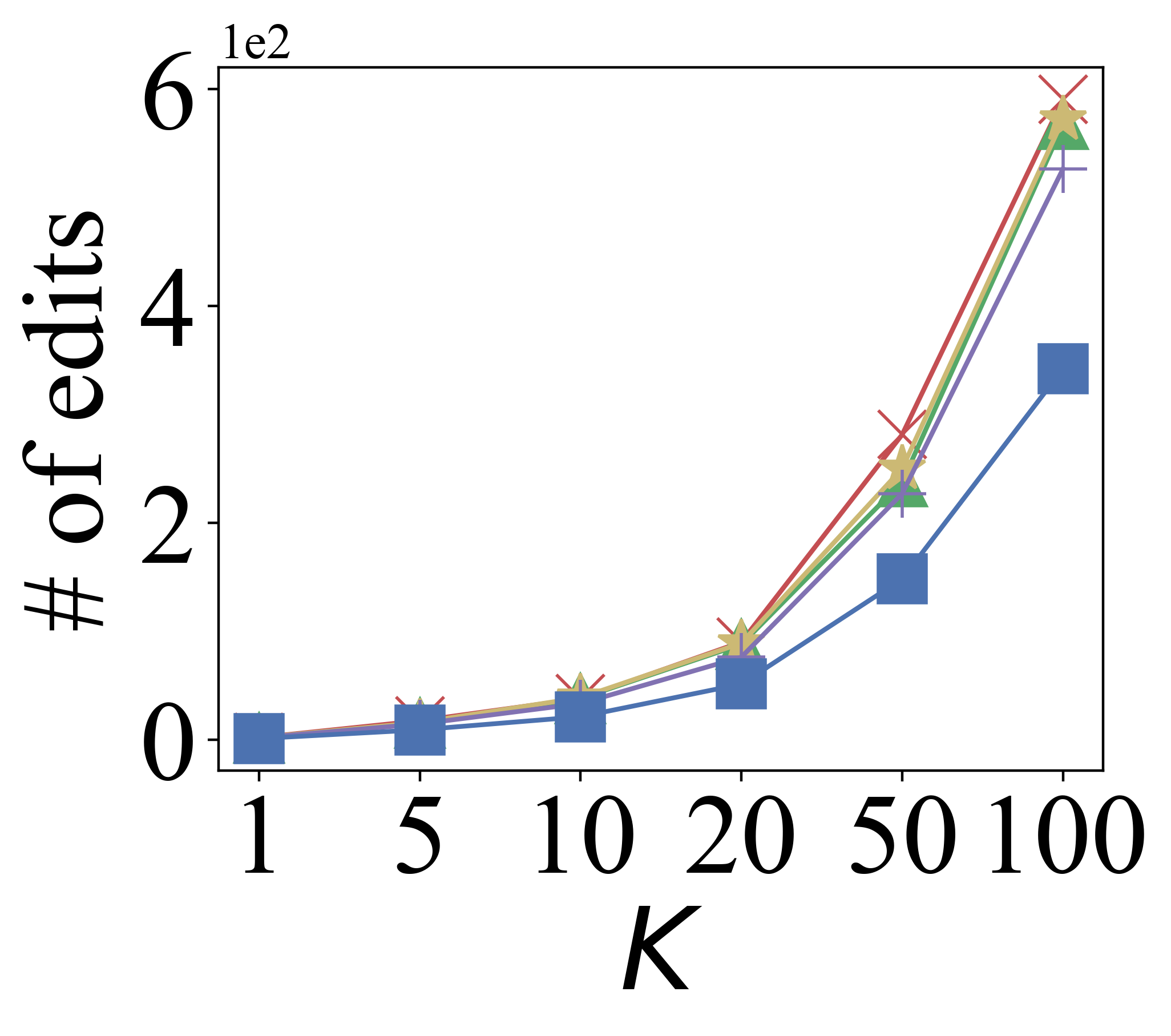}}
		\label{fig:varyK_porto_edr_score}}
	\subfigure[][{\scriptsize EDR (Xi'an)}]{
		\scalebox{0.2}[0.2]{\includegraphics{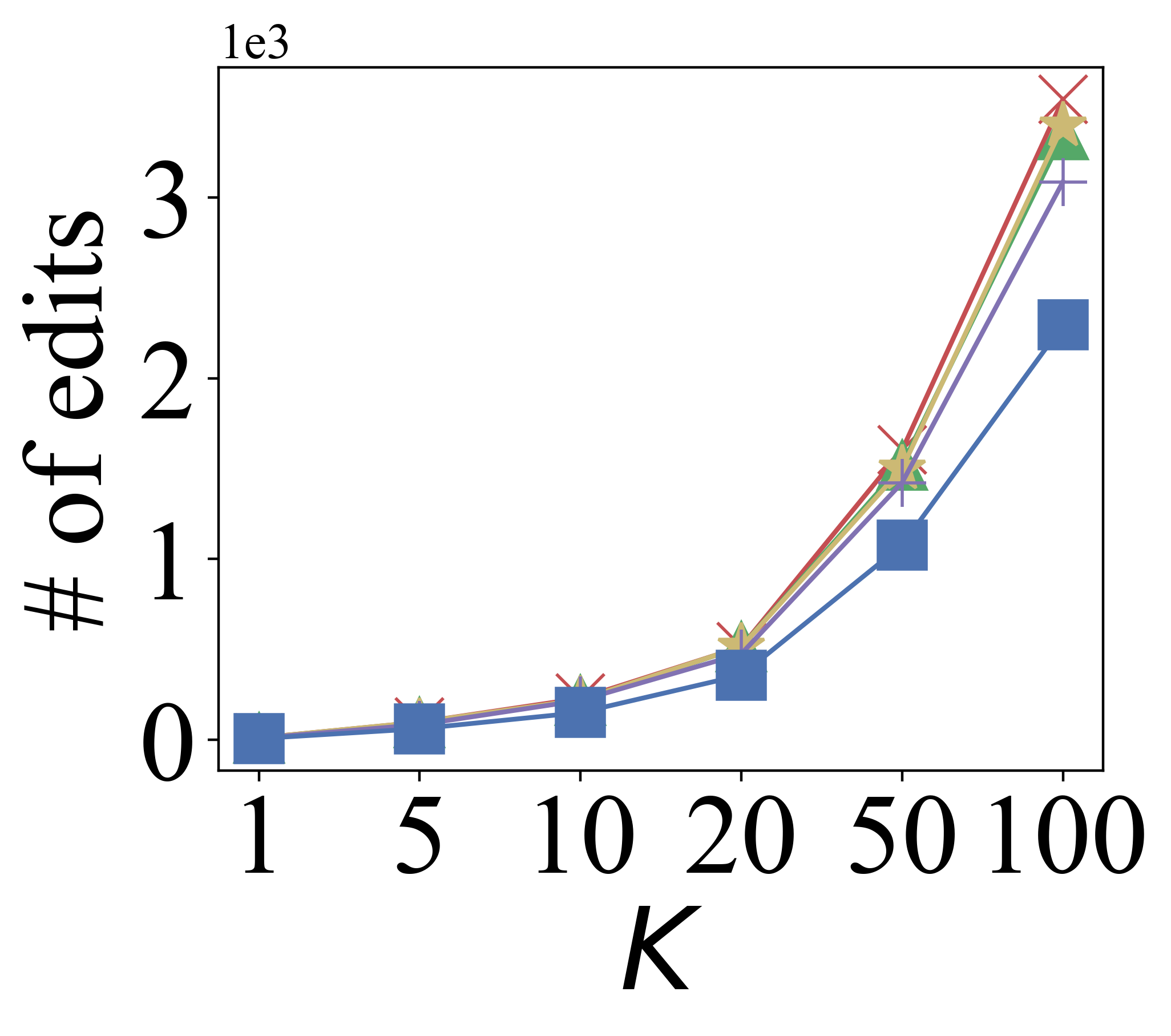}}
		\label{fig:varyK_xian_edr_score}}
	\subfigure[][{\scriptsize EDR (Beijing)}]{
		\scalebox{0.2}[0.2]{\includegraphics{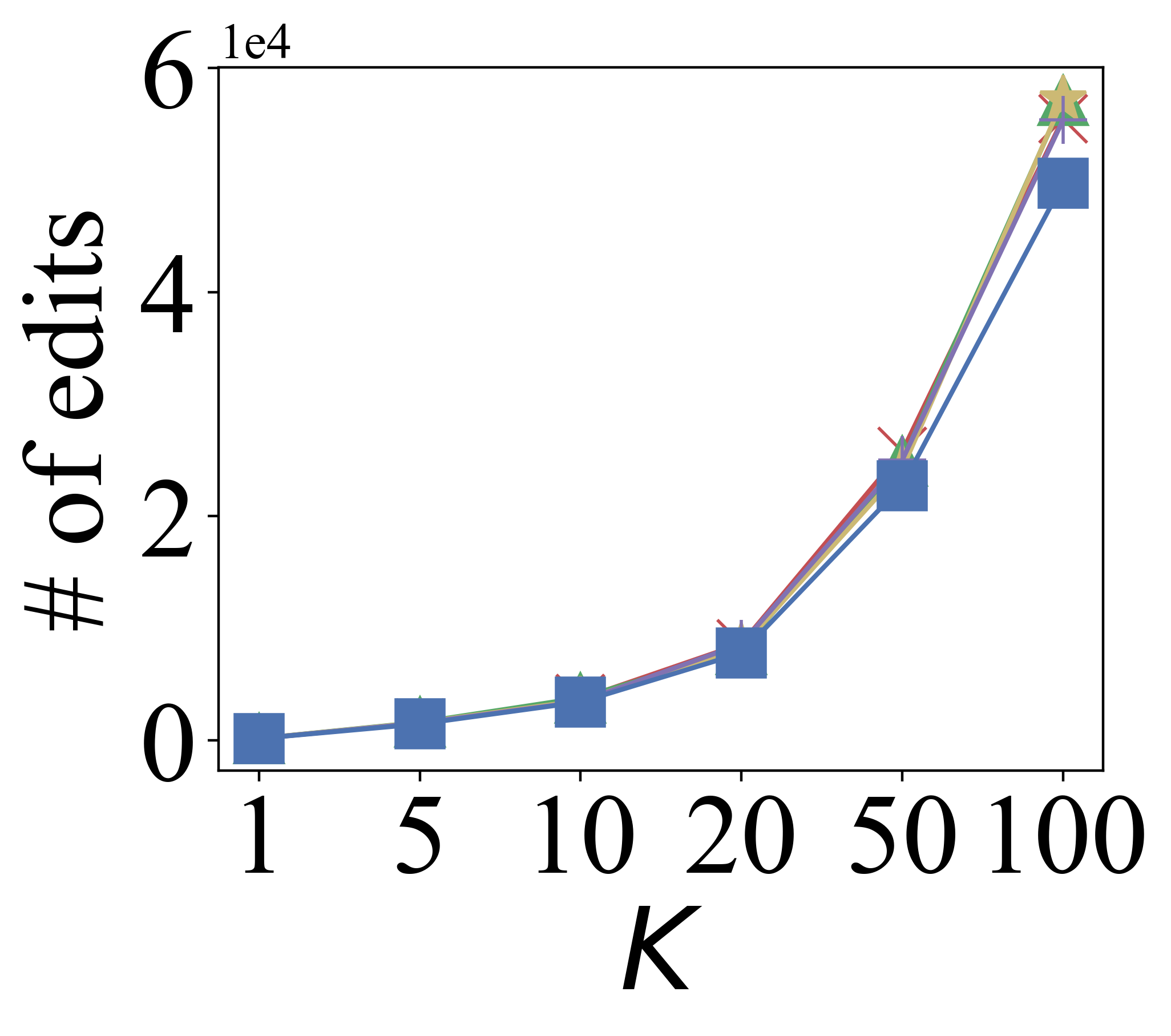}}
		\label{fig:varyK_beijing_edr_score}}
	\subfigure[][{\scriptsize EDR (Porto)}]{
		\scalebox{0.2}[0.2]{\includegraphics{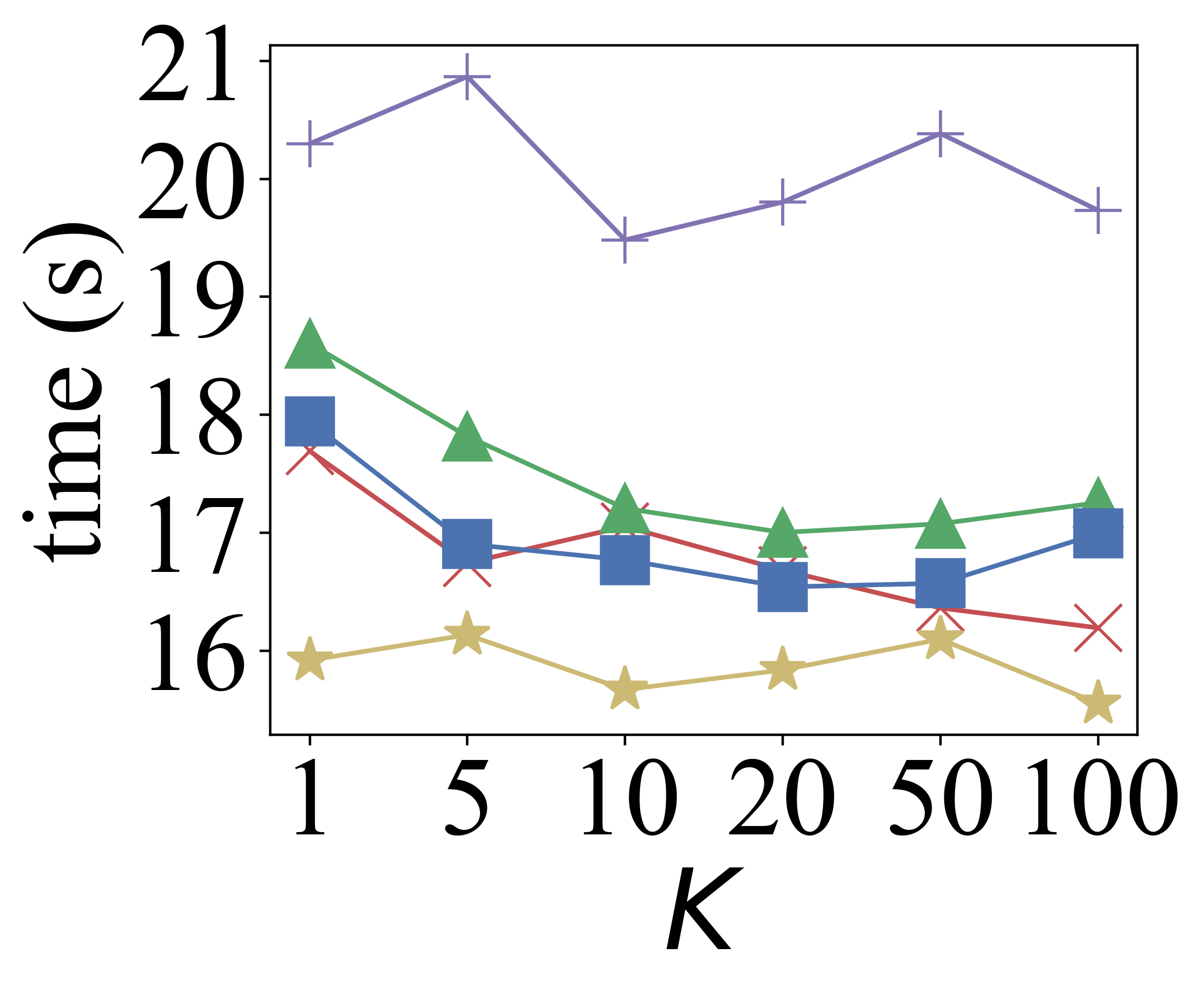}}
		\label{fig:varyK_porto_edr}}
	\subfigure[][{\scriptsize EDR (Xi'an)}]{
		\scalebox{0.2}[0.2]{\includegraphics{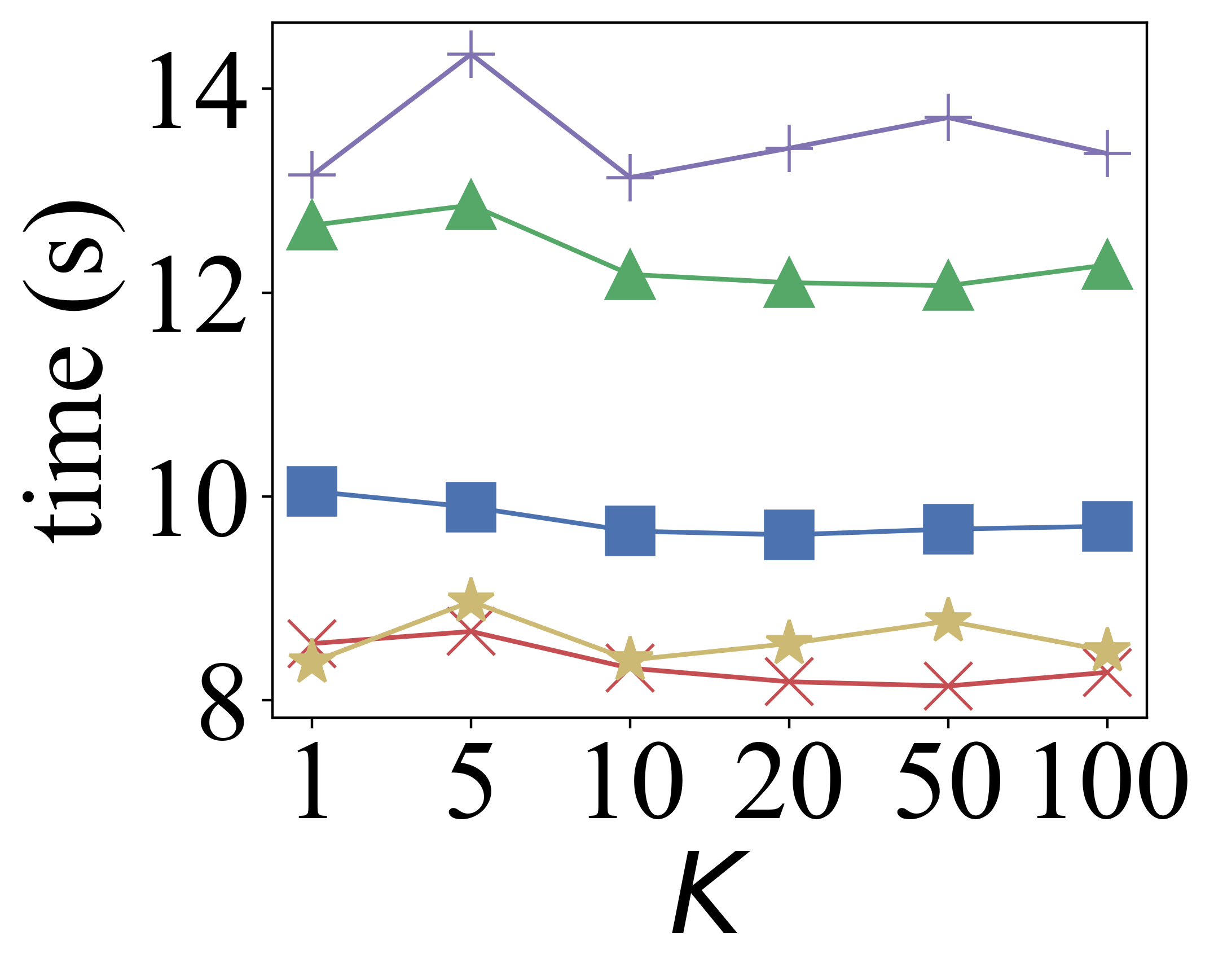}}
		\label{fig:varyK_xian_edr}}
	\subfigure[][{\scriptsize EDR (Beijing)}]{
		\scalebox{0.2}[0.2]{\includegraphics{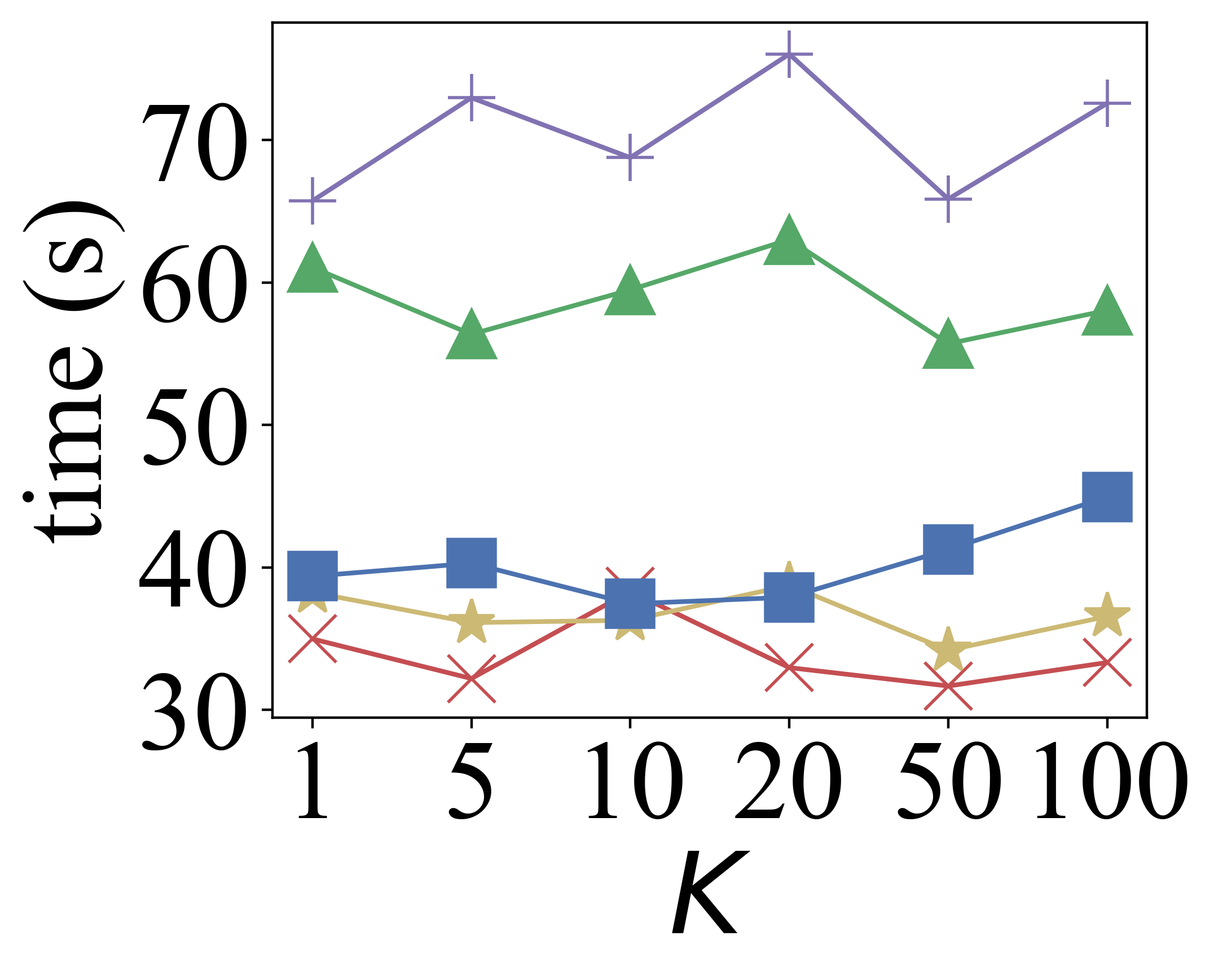}}
		\label{fig:varyK_beijing_edr}}
	\\
	\subfigure[][{\scriptsize DTW (Porto)}]{
		\scalebox{0.2}[0.2]{\includegraphics{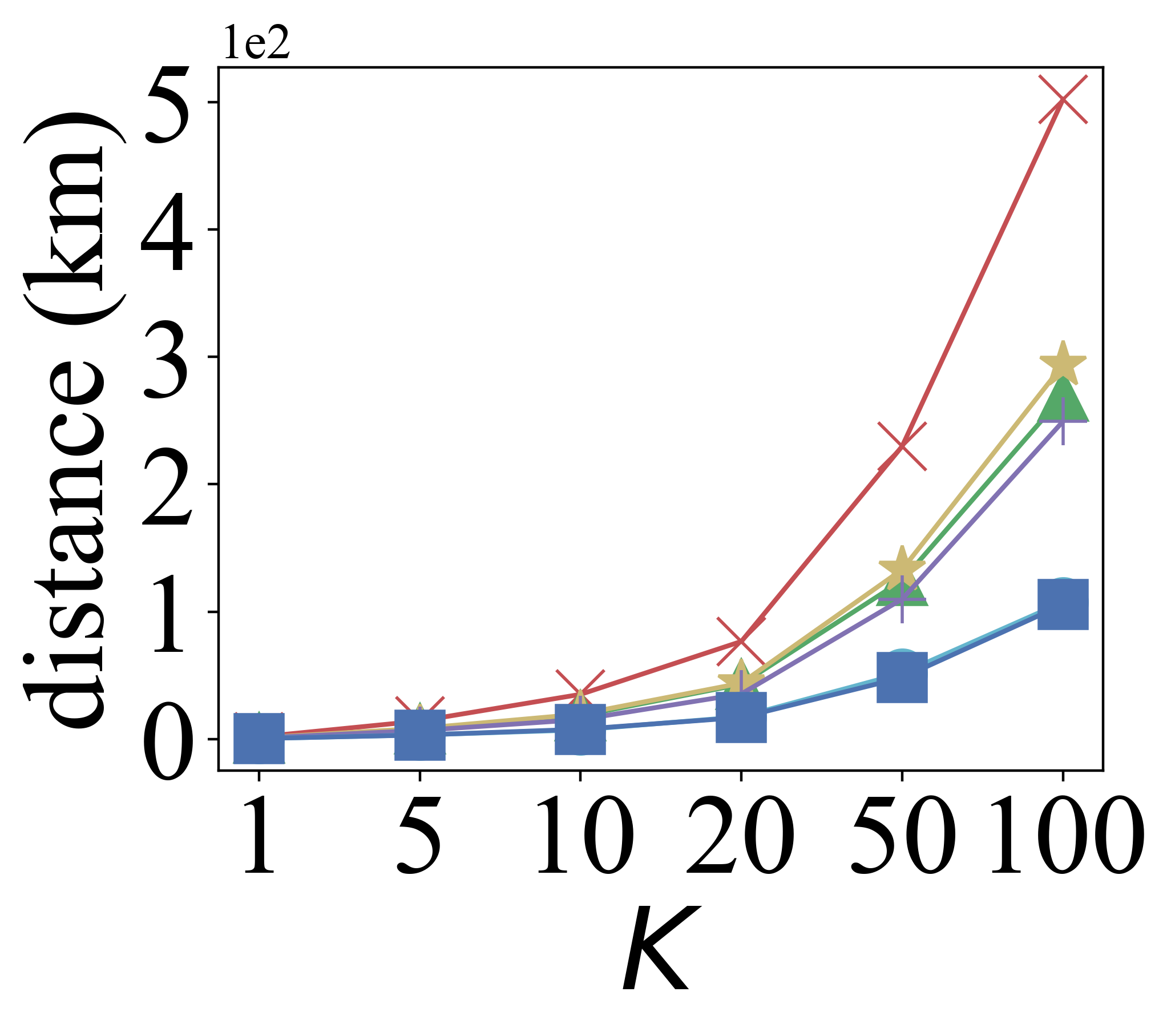}}
		\label{fig:varyK_porto_dtw_score}}
	\subfigure[][{\scriptsize DTW (Xi'an)}]{
		\scalebox{0.2}[0.2]{\includegraphics{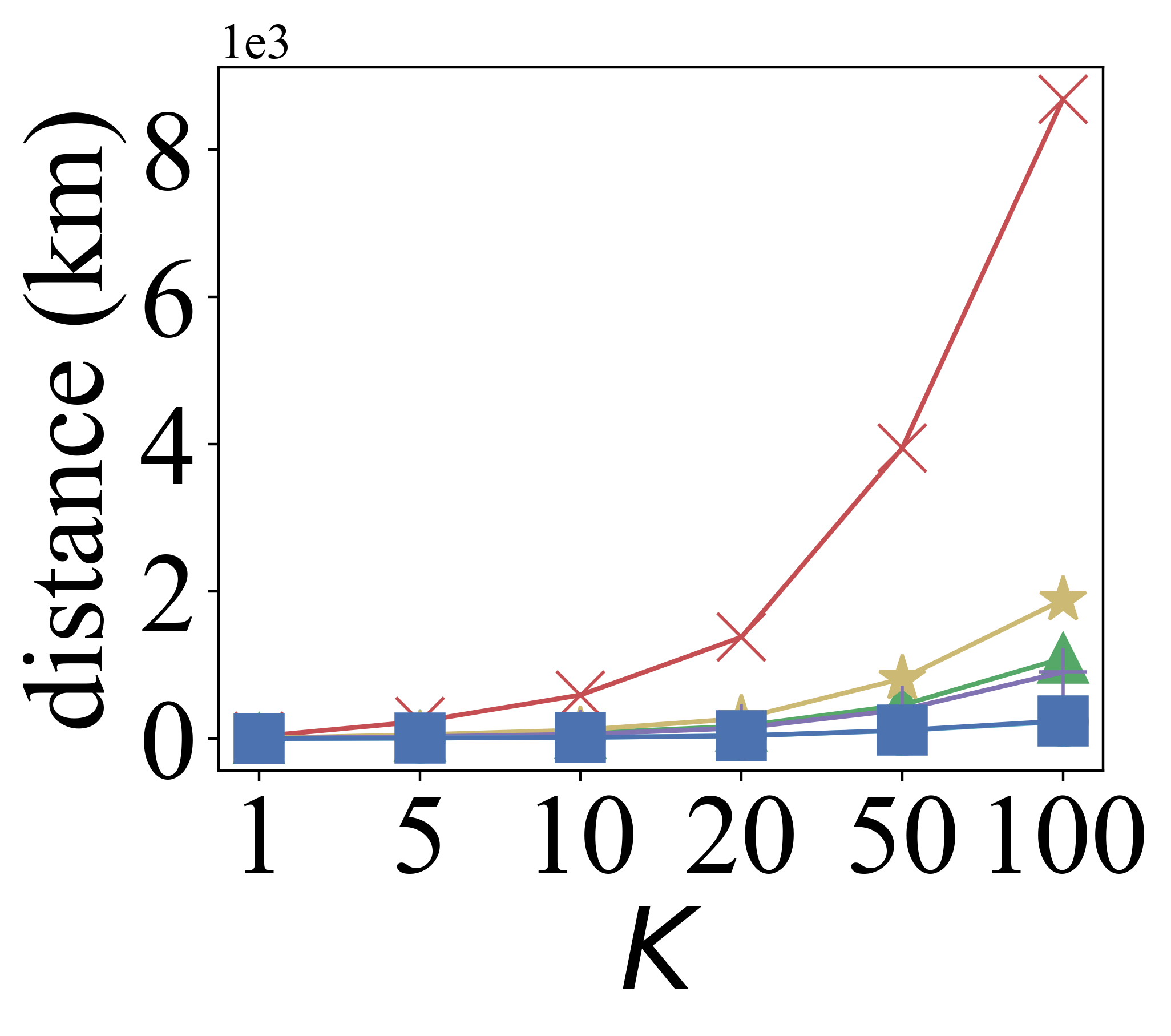}}
		\label{fig:varyK_xian_dtw_score}}
	\subfigure[][{\scriptsize DTW (Beijing)}]{
		\scalebox{0.2}[0.2]{\includegraphics{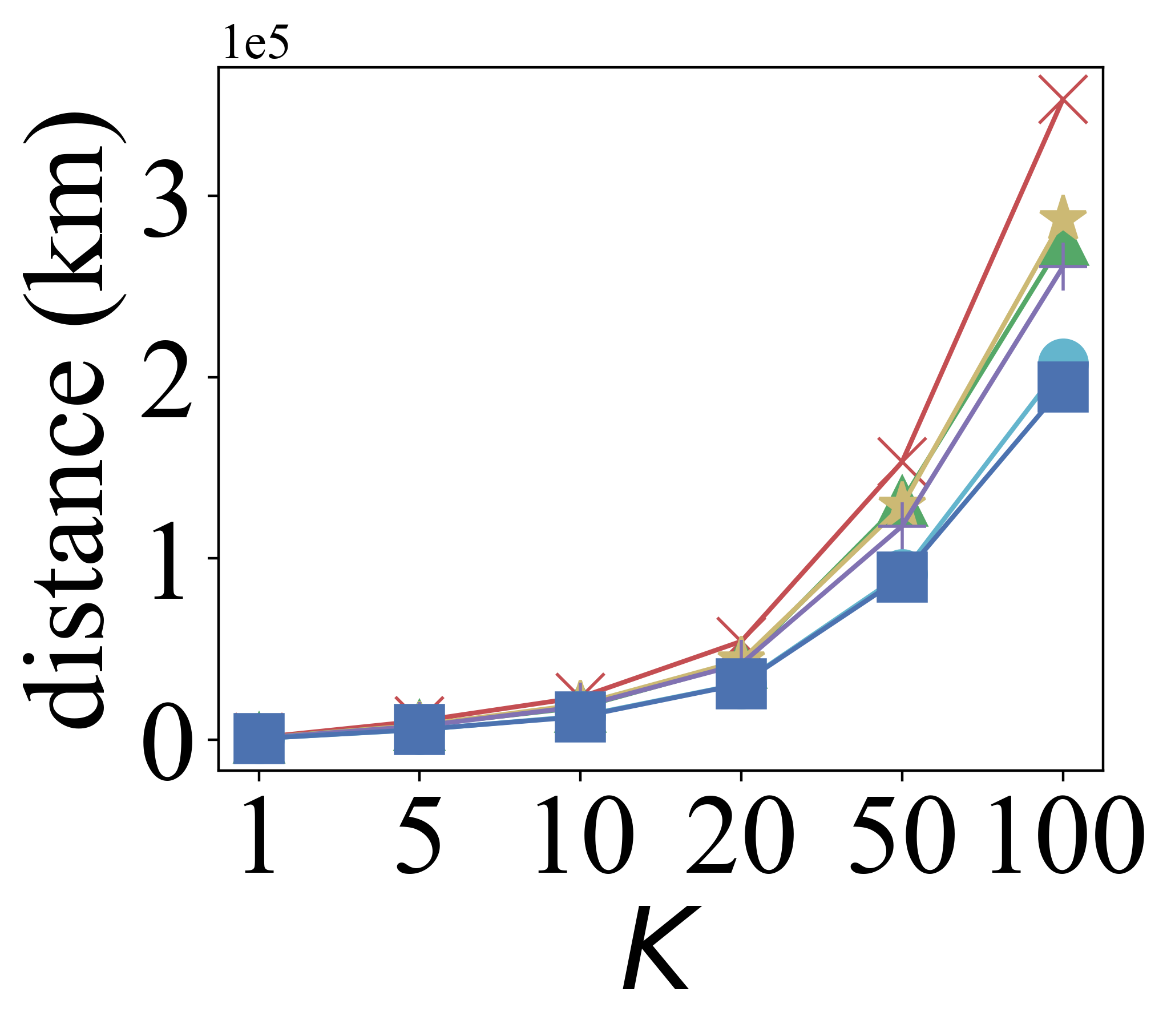}}
		\label{fig:varyK_beijing_dtw_score}}
	\subfigure[][{\scriptsize DTW (Porto)}]{
		\scalebox{0.2}[0.2]{\includegraphics{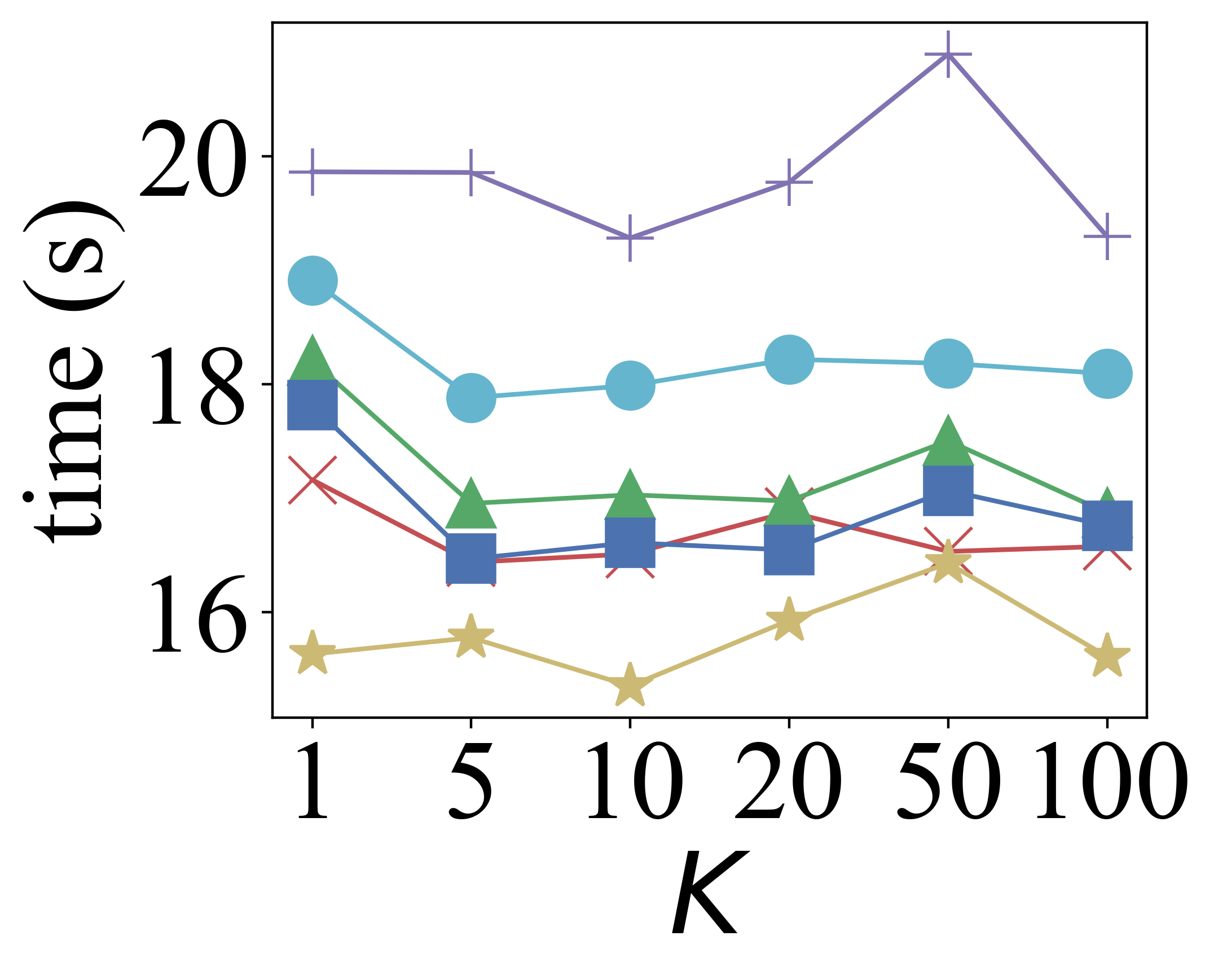}}
		\label{fig:varyK_porto_dtw}}
	\subfigure[][{\scriptsize DTW (Xi'an)}]{
		\scalebox{0.2}[0.2]{\includegraphics{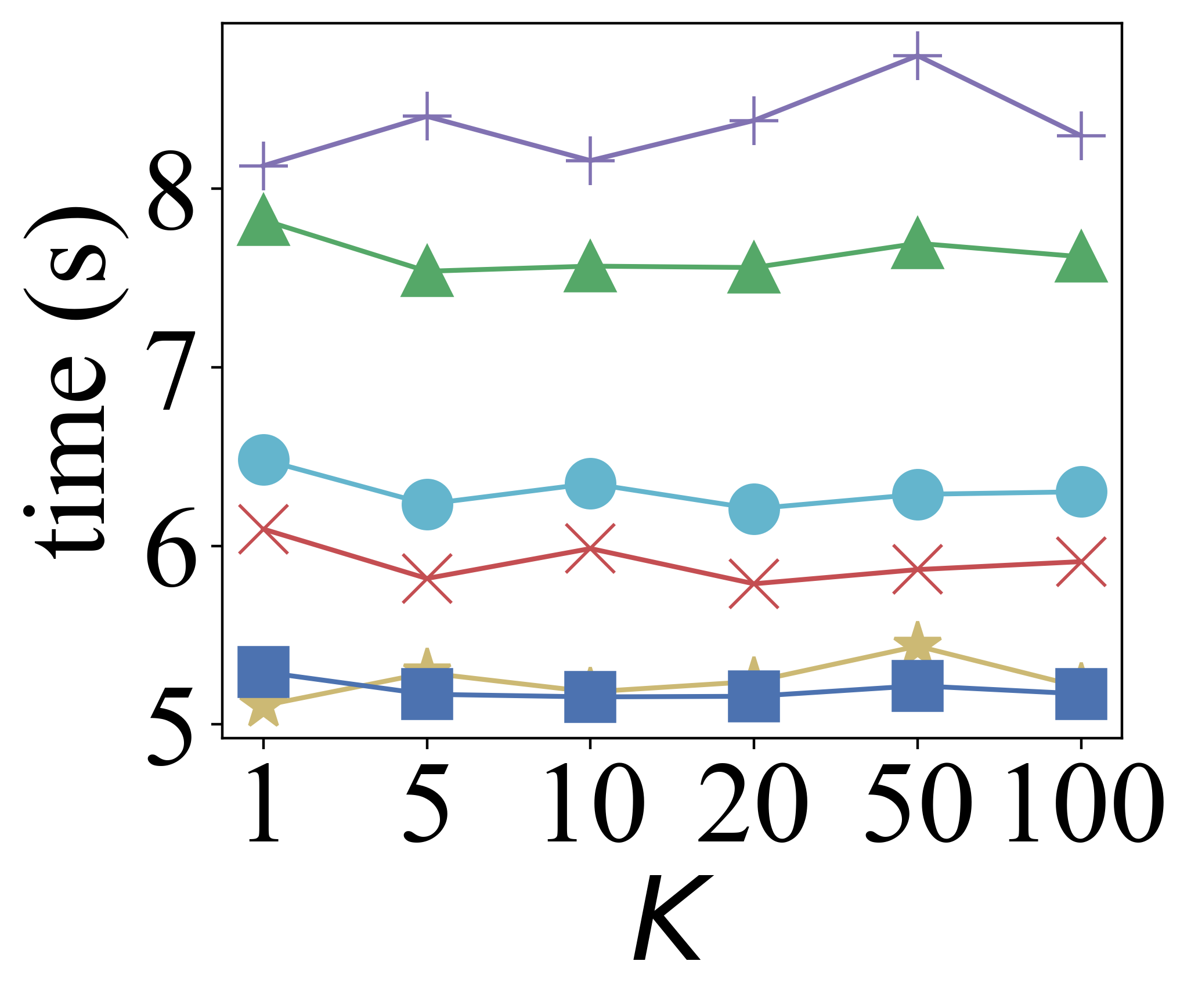}}
		\label{fig:varyK_xian_dtw}}
	\subfigure[][{\scriptsize DTW (Beijing)}]{
		\scalebox{0.2}[0.2]{\includegraphics{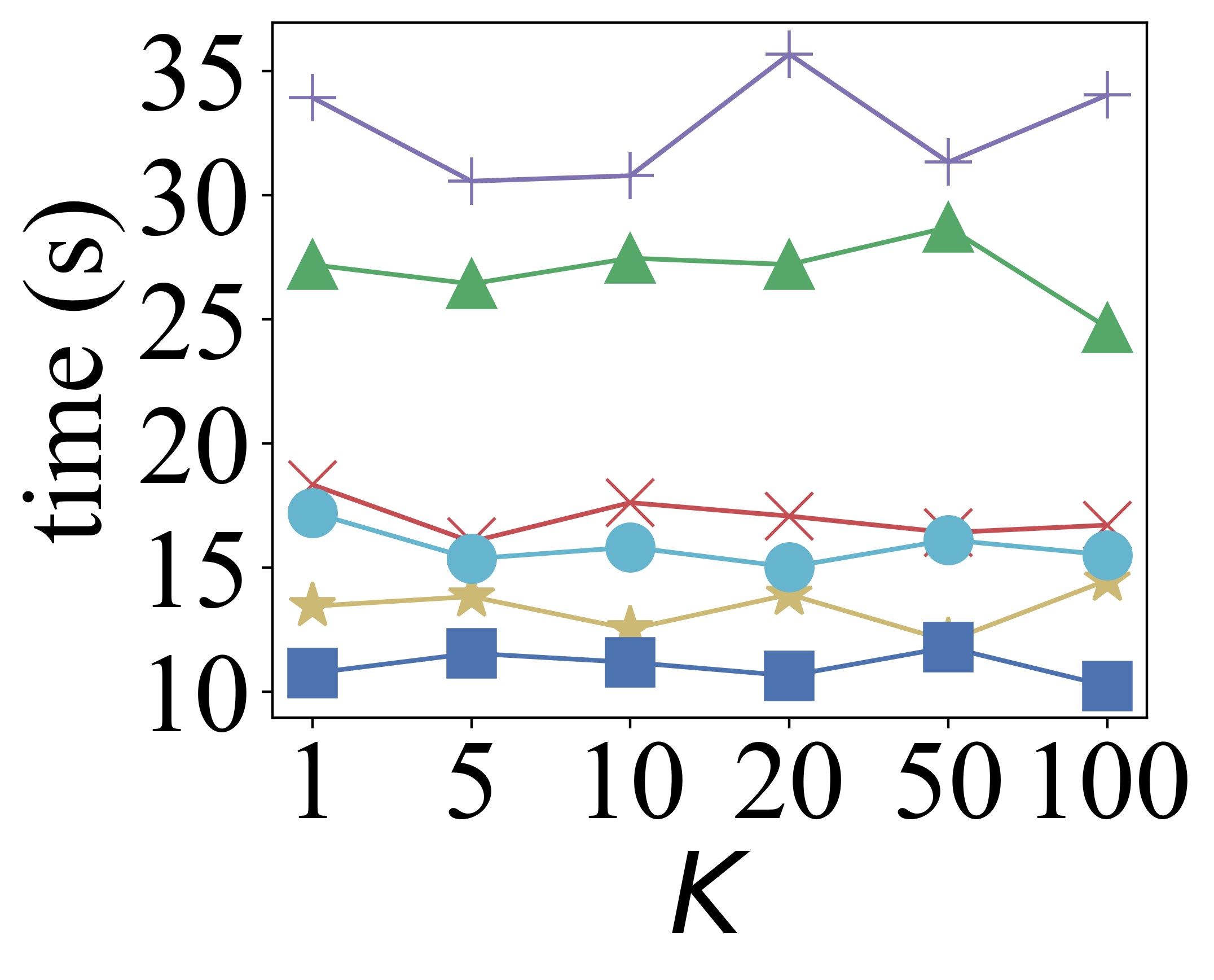}}
		\label{fig:varyK_beijing_dtw}}
	\\	
	\subfigure[][{\scriptsize ERP (Porto)}]{
		\scalebox{0.2}[0.2]{\includegraphics{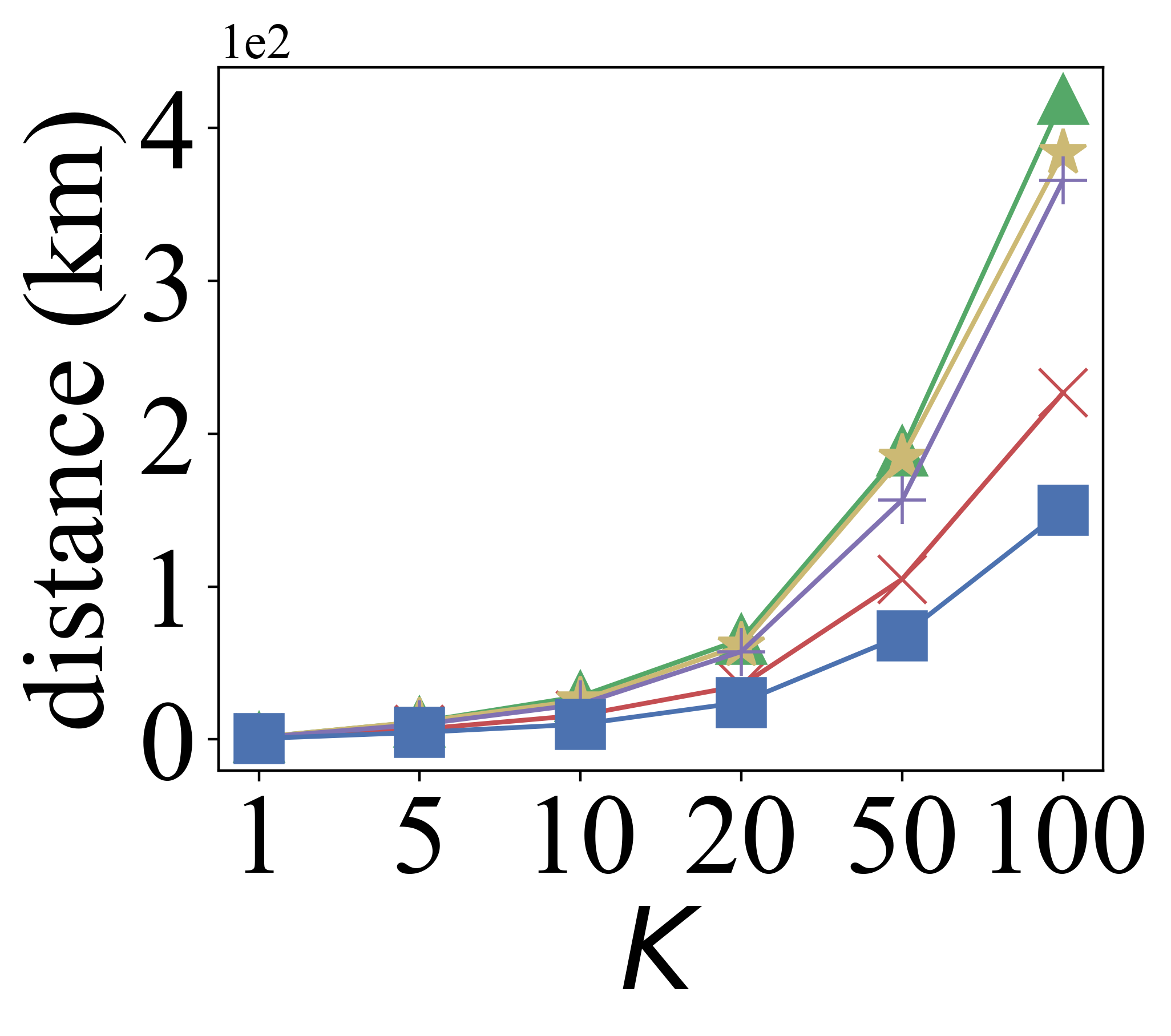}}
		\label{fig:varyK_porto_erp_score}}
	\subfigure[][{\scriptsize ERP (Xi'an)}]{
		\scalebox{0.2}[0.2]{\includegraphics{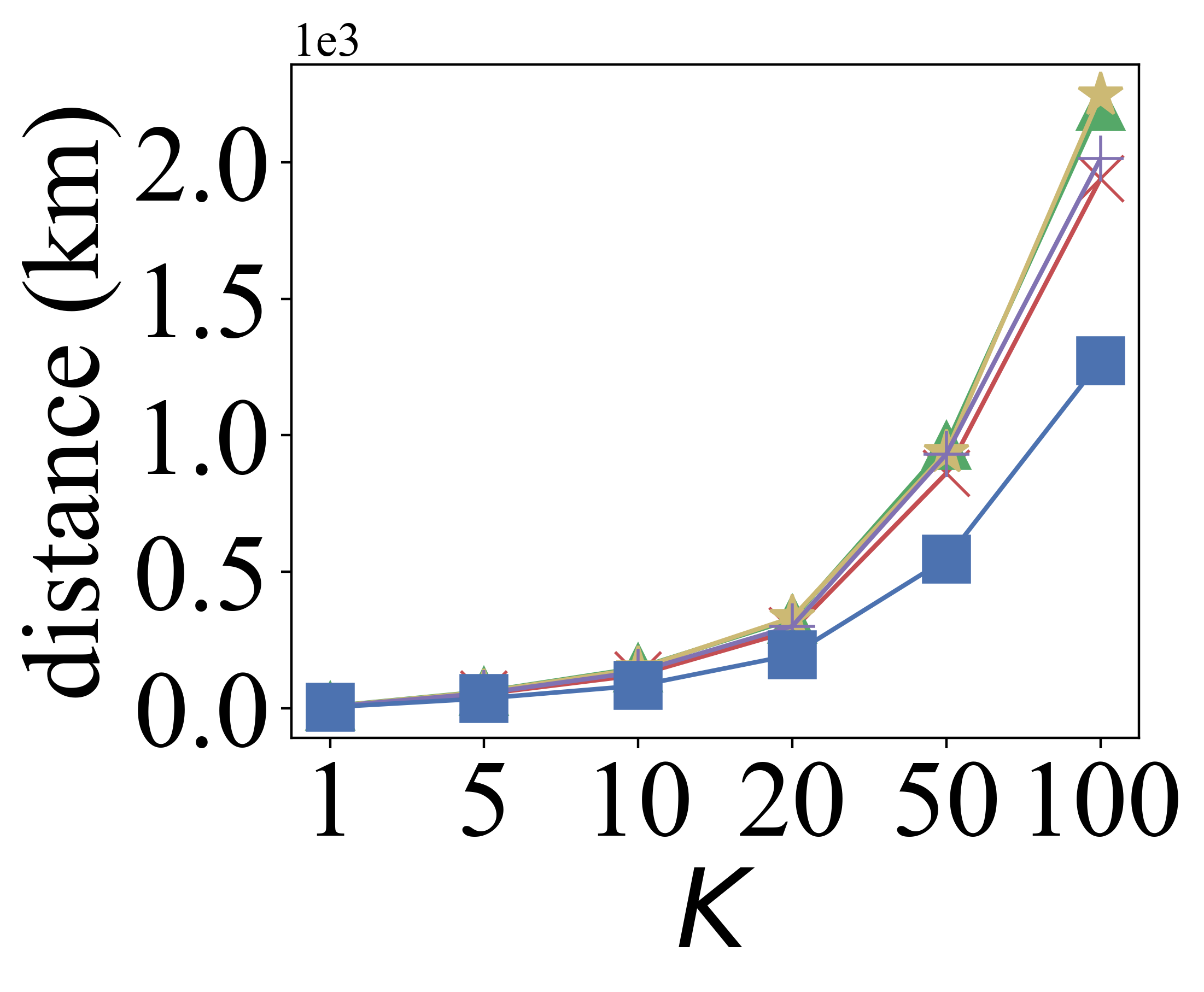}}
		\label{fig:varyK_xian_erp_score}}
	\subfigure[][{\scriptsize ERP (Beijing)}]{
		\scalebox{0.2}[0.2]{\includegraphics{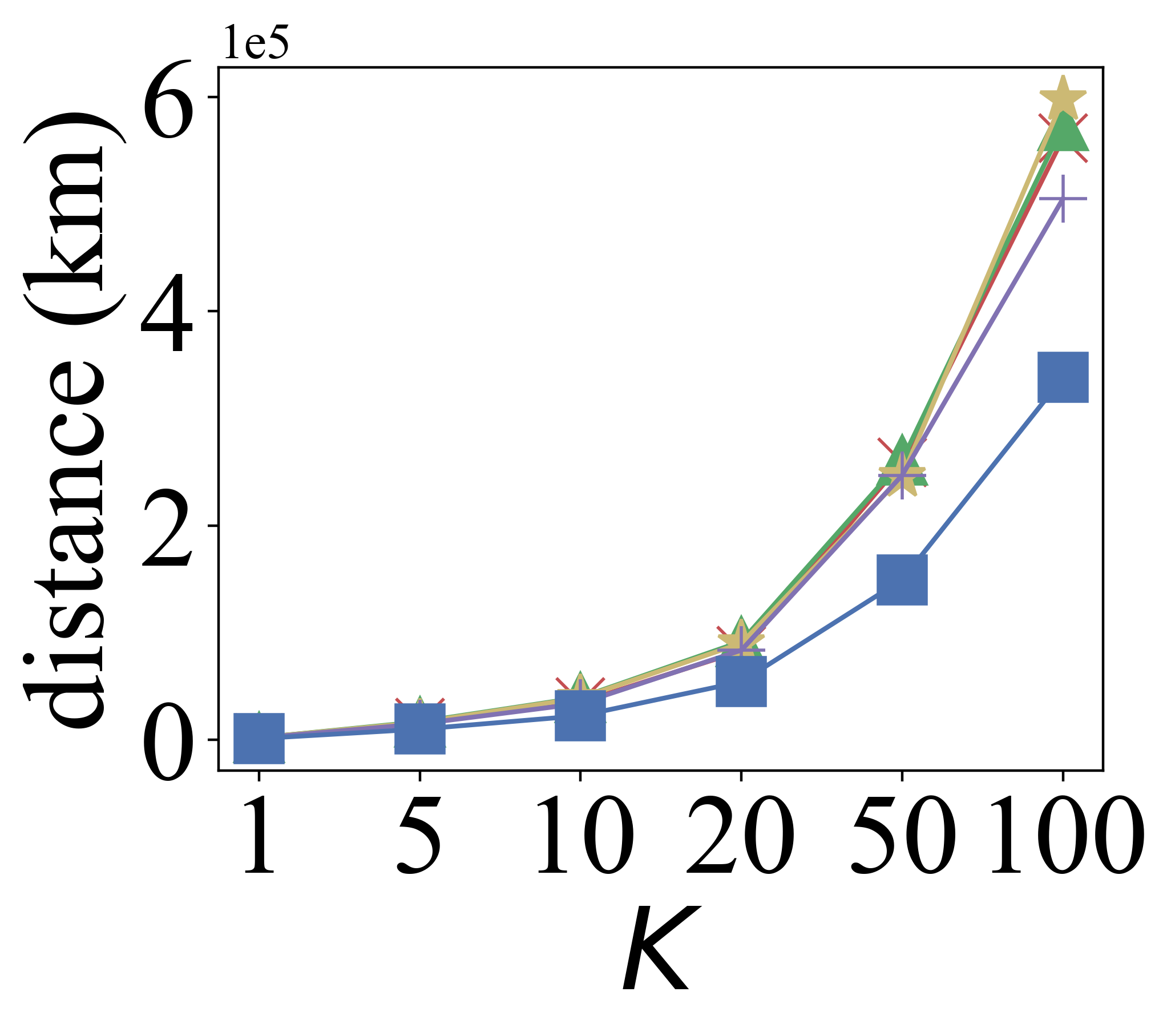}}
		\label{fig:varyK_beijing_erp_score}}
	\subfigure[][{\scriptsize ERP (Porto)}]{
		\scalebox{0.2}[0.2]{\includegraphics{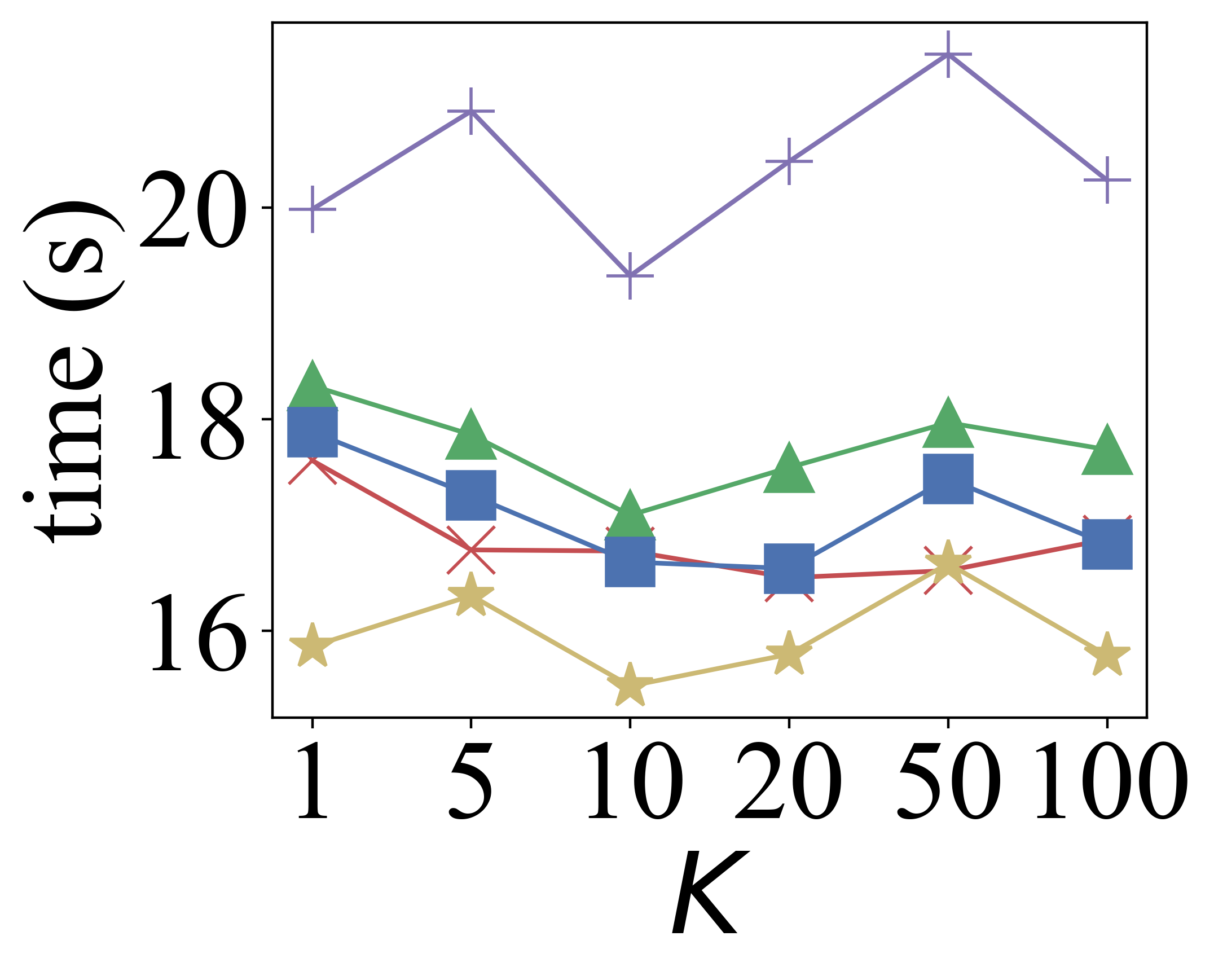}}
		\label{fig:varyK_porto_erp}}
	\subfigure[][{\scriptsize ERP (Xi'an)}]{
		\scalebox{0.2}[0.2]{\includegraphics{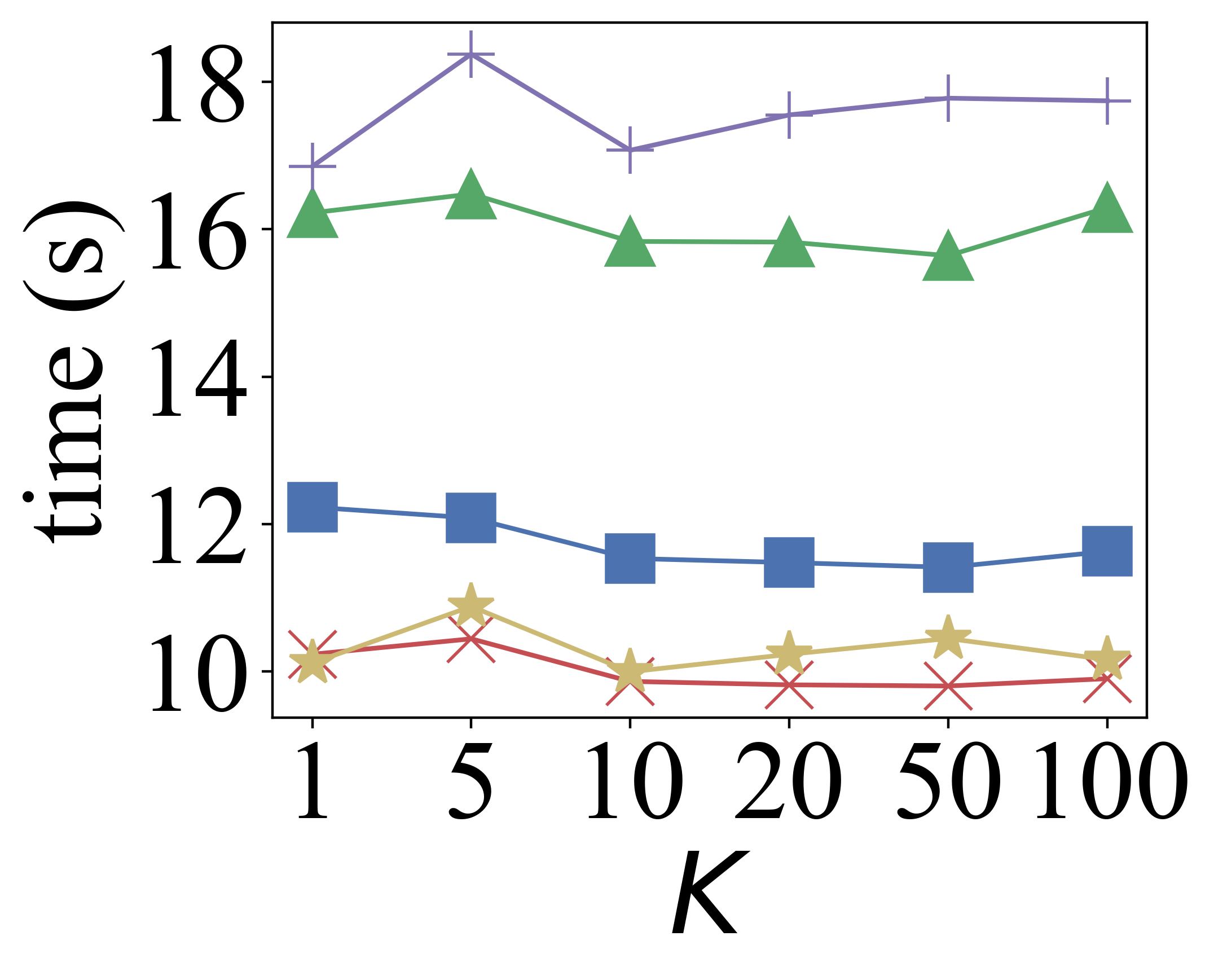}}
		\label{fig:varyK_xian_erp}}
	\subfigure[][{\scriptsize ERP (Beijing)}]{
		\scalebox{0.2}[0.2]{\includegraphics{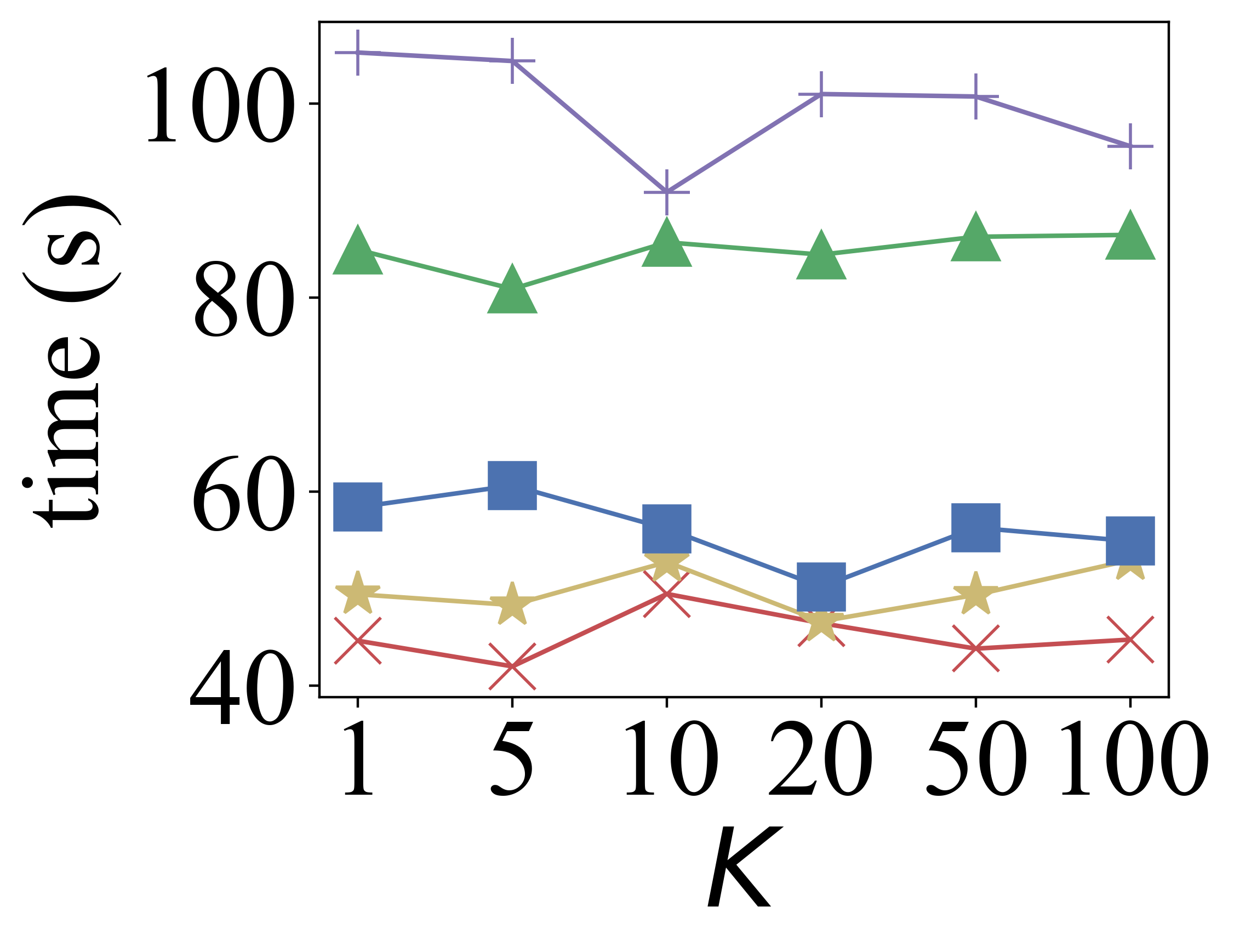}}
		\label{fig:varyK_beijing_erp}}	
	\\
	\subfigure[][{\scriptsize FD (Porto)}]{
		\scalebox{0.2}[0.2]{\includegraphics{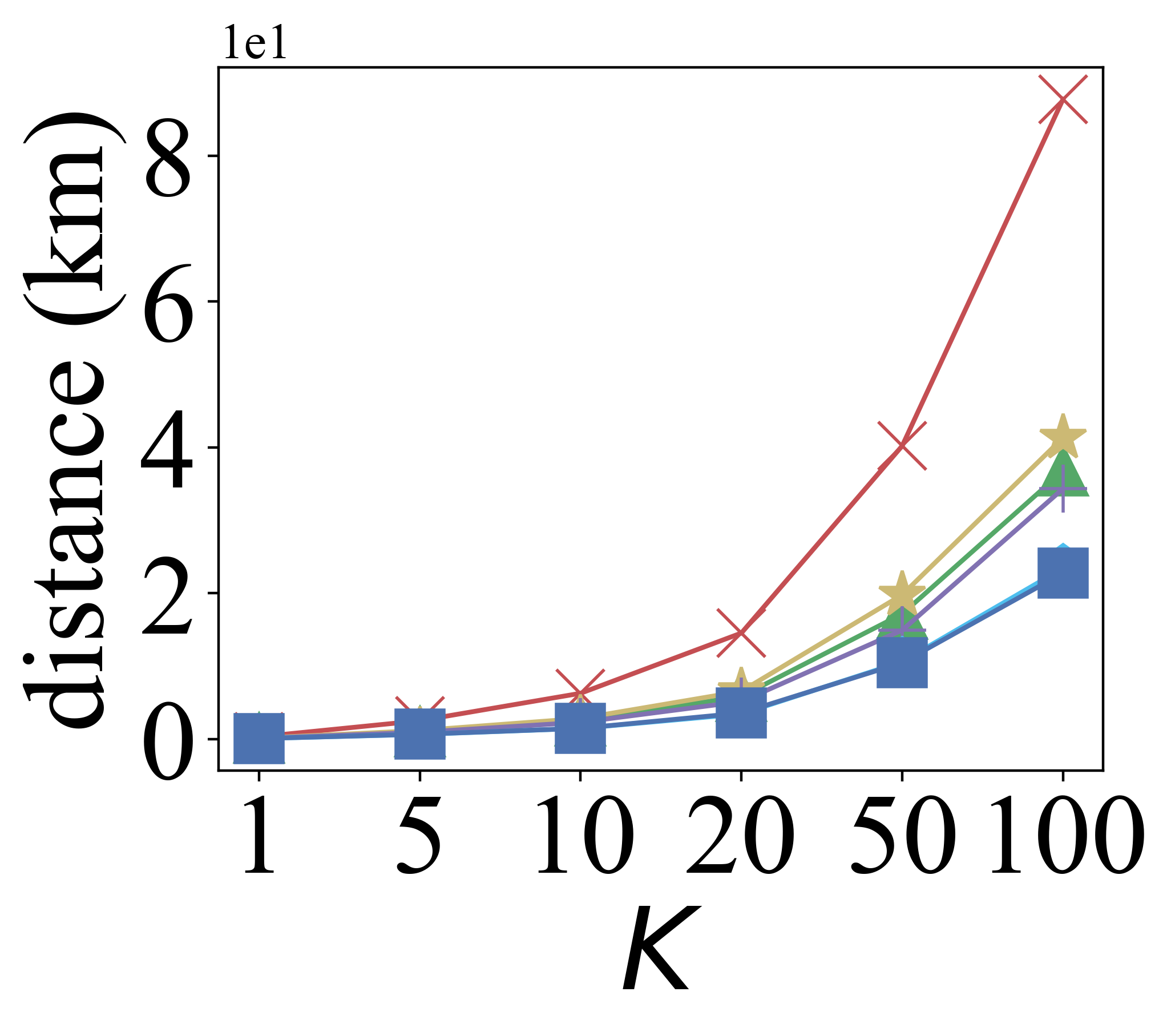}}
		\label{fig:varyK_porto_FC_score}}
	\subfigure[][{\scriptsize FD (Xi'an)}]{
		\scalebox{0.2}[0.2]{\includegraphics{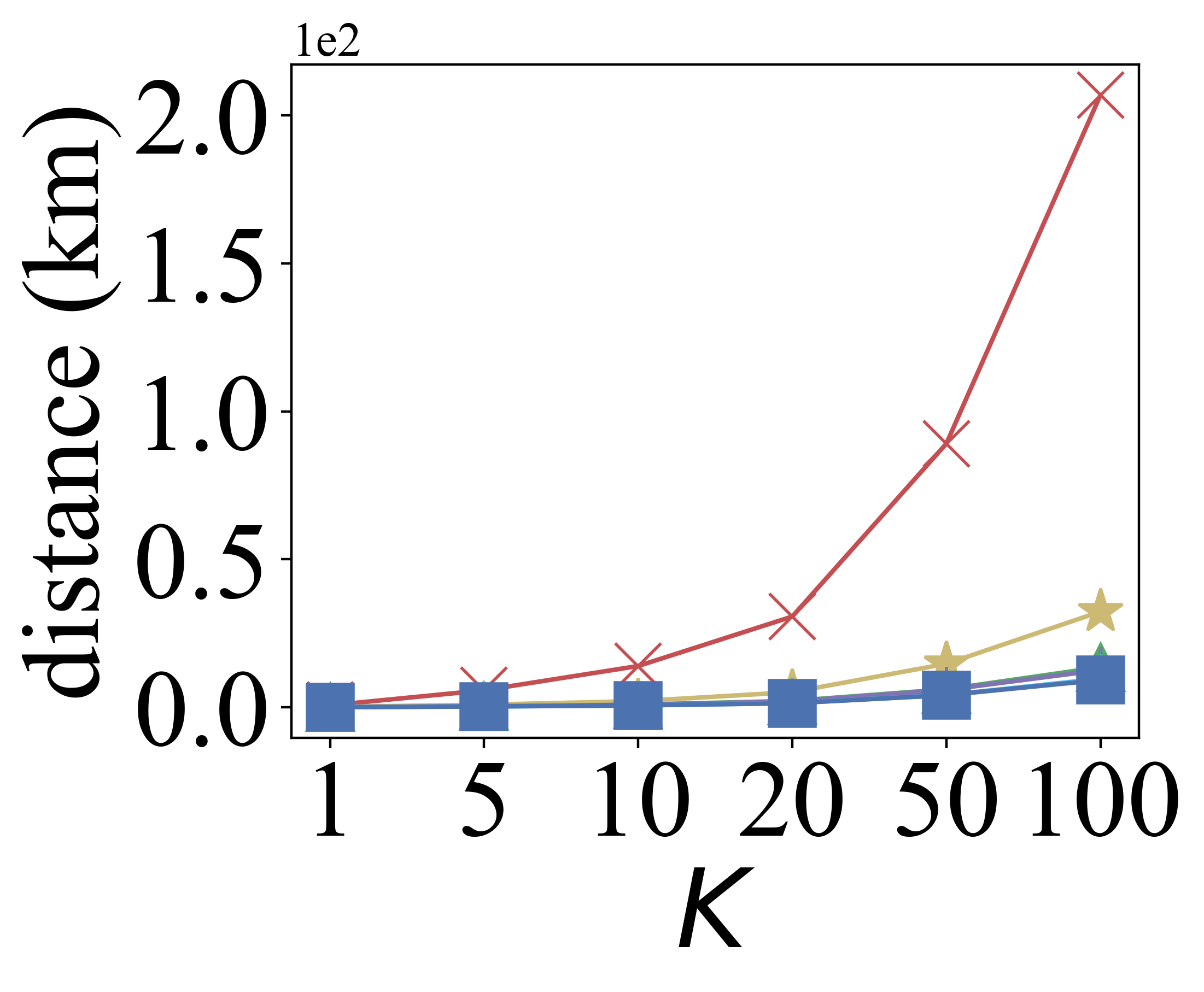}}
		\label{fig:varyK_xian_FC_score}}
	\subfigure[][{\scriptsize FD (Beijing)}]{
		\scalebox{0.2}[0.2]{\includegraphics{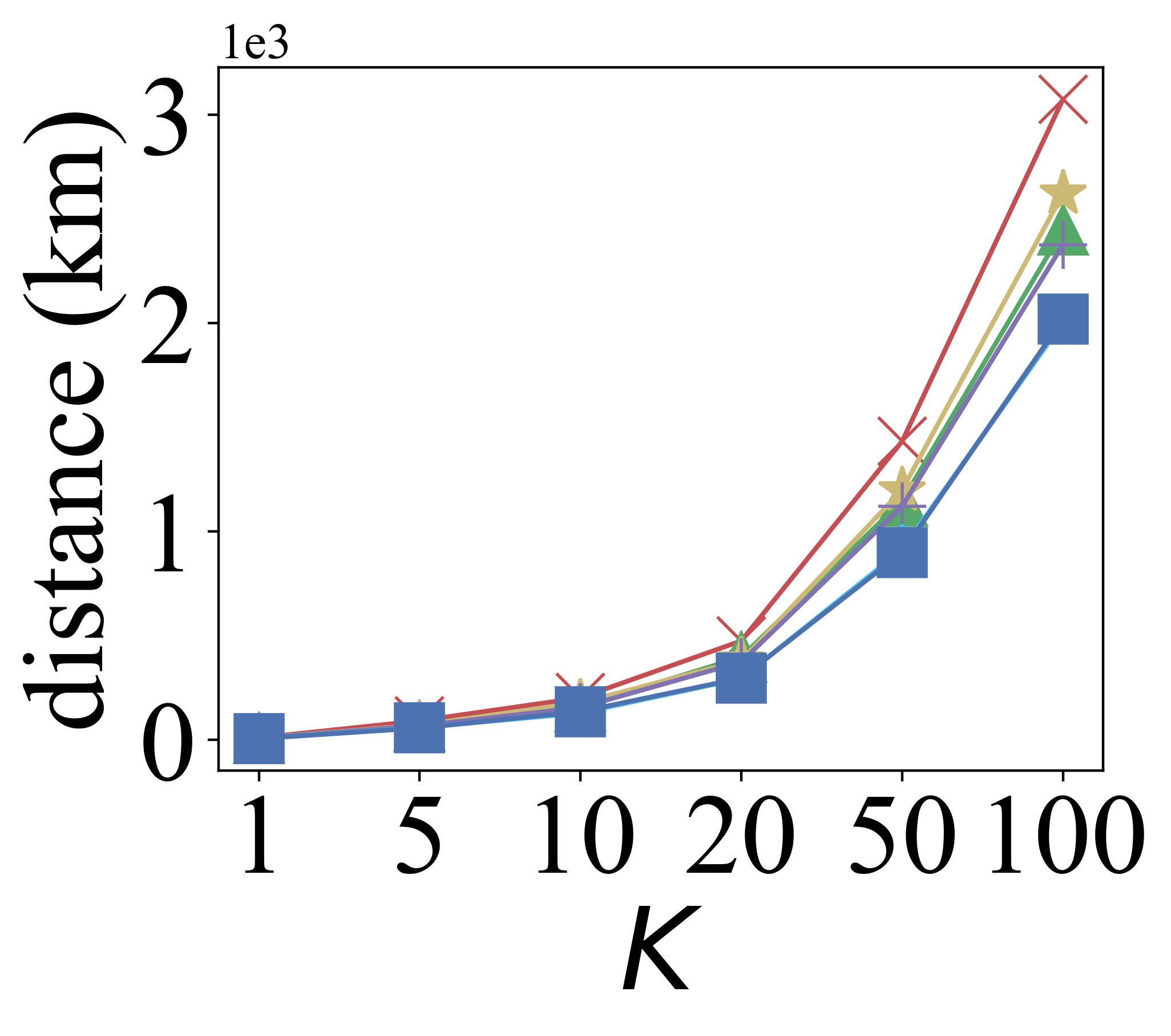}}
		\label{fig:varyK_beijing_FC_score}}
	\subfigure[][{\scriptsize FD (Porto)}]{
		\scalebox{0.2}[0.2]{\includegraphics{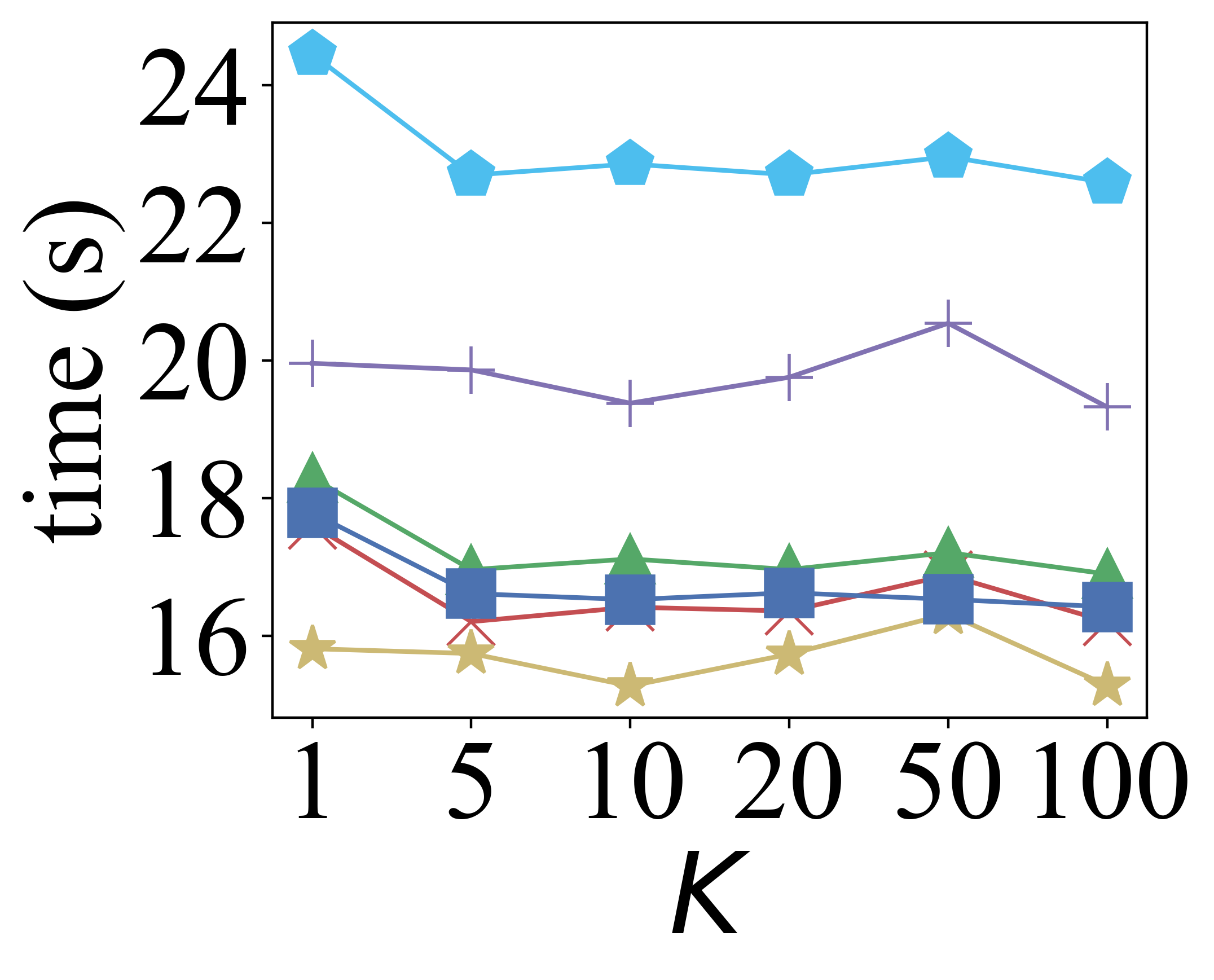}}
		\label{fig:varyK_porto_FC}}
	\subfigure[][{\scriptsize FD (Xi'an)}]{
		\scalebox{0.2}[0.2]{\includegraphics{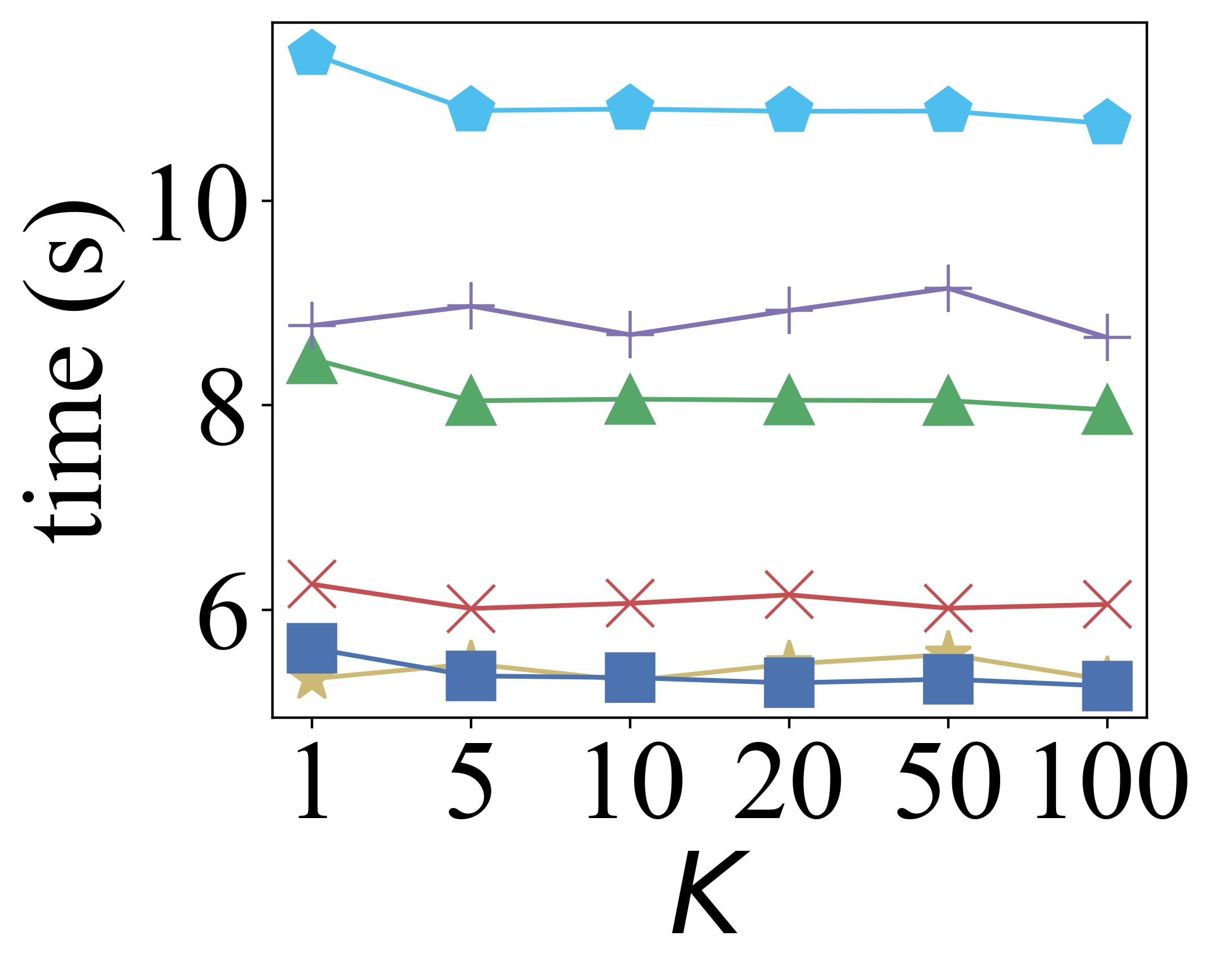}}
		\label{fig:varyK_xian_FC}}
	\subfigure[][{\scriptsize FD (Beijing)}]{
		\scalebox{0.2}[0.2]{\includegraphics{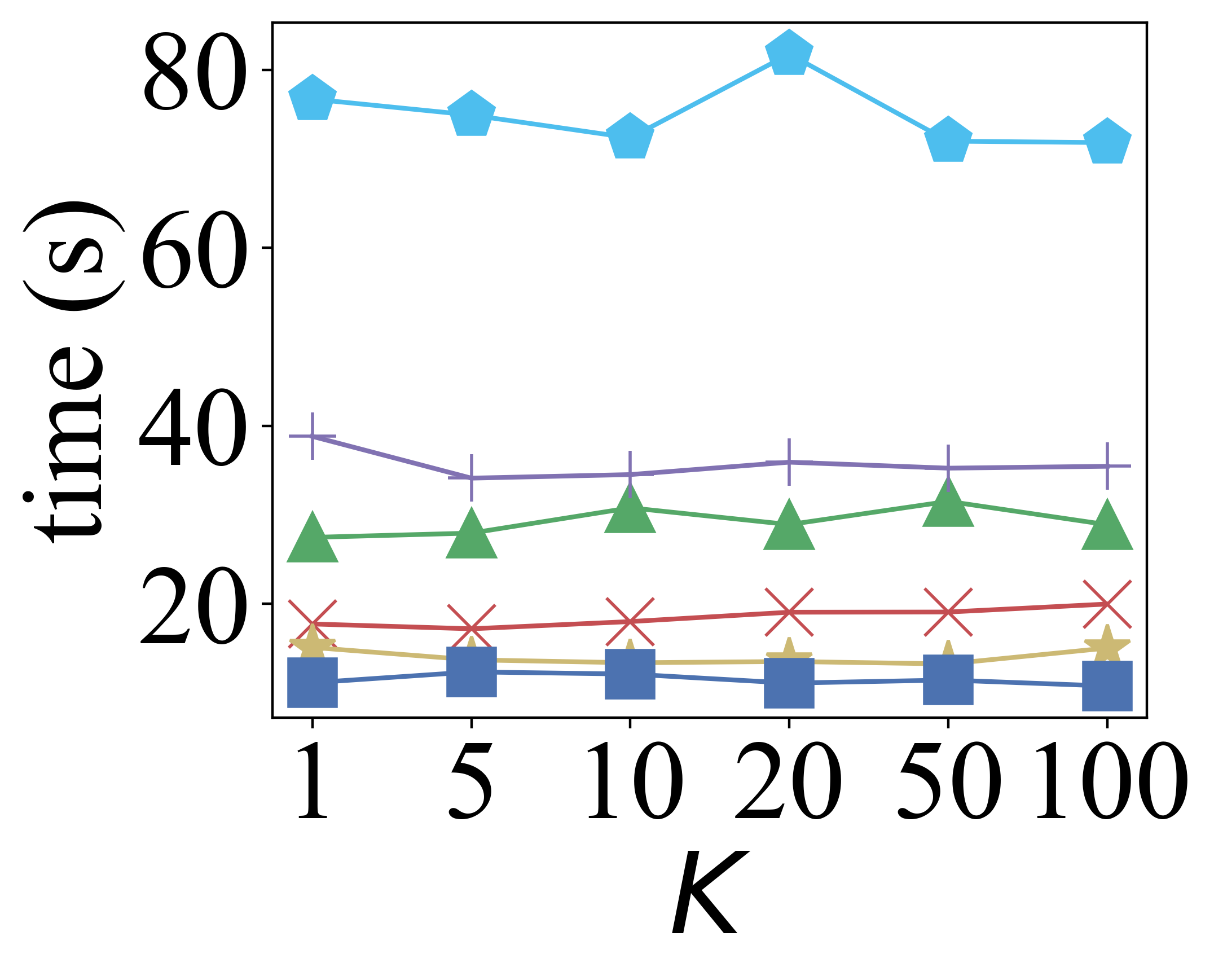}}
		\label{fig:varyK_beijing_FC}}
	\caption{\small Efficiency with varying $K$}\vspace{-2ex}
	\label{fig:varyK}
\end{figure*}

%% file: SubtrajectorySimilarity.bbl

\begin{thebibliography}{32}


\ifx \showCODEN    \undefined \def \showCODEN     #1{\unskip}     \fi
\ifx \showDOI      \undefined \def \showDOI       #1{#1}\fi
\ifx \showISBNx    \undefined \def \showISBNx     #1{\unskip}     \fi
\ifx \showISBNxiii \undefined \def \showISBNxiii  #1{\unskip}     \fi
\ifx \showISSN     \undefined \def \showISSN      #1{\unskip}     \fi
\ifx \showLCCN     \undefined \def \showLCCN      #1{\unskip}     \fi
\ifx \shownote     \undefined \def \shownote      #1{#1}          \fi
\ifx \showarticletitle \undefined \def \showarticletitle #1{#1}   \fi
\ifx \showURL      \undefined \def \showURL       {\relax}        \fi
\providecommand\bibfield[2]{#2}
\providecommand\bibinfo[2]{#2}
\providecommand\natexlab[1]{#1}
\providecommand\showeprint[2][]{arXiv:#2}

\bibitem[\protect\citeauthoryear{Agarwal, Fox, Munagala, Nath, Pan, and
  Taylor}{Agarwal et~al\mbox{.}}{2018}]%
        {agarwal2018subtrajectory}
\bibfield{author}{\bibinfo{person}{Pankaj~K Agarwal}, \bibinfo{person}{Kyle
  Fox}, \bibinfo{person}{Kamesh Munagala}, \bibinfo{person}{Abhinandan Nath},
  \bibinfo{person}{Jiangwei Pan}, {and} \bibinfo{person}{Erin Taylor}.}
  \bibinfo{year}{2018}\natexlab{}.
\newblock \showarticletitle{Subtrajectory clustering: Models and algorithms}.
  In \bibinfo{booktitle}{\emph{Proceedings of the 37th ACM SIGMOD-SIGACT-SIGAI
  Symposium on Principles of Database Systems}}. \bibinfo{pages}{75--87}.
\newblock


\bibitem[\protect\citeauthoryear{Alt and Godau}{Alt and Godau}{1995}]%
        {AltG95}
\bibfield{author}{\bibinfo{person}{Helmut Alt} {and} \bibinfo{person}{Michael
  Godau}.} \bibinfo{year}{1995}\natexlab{}.
\newblock \showarticletitle{Computing the Fr{\'{e}}chet distance between two
  polygonal curves}.
\newblock \bibinfo{journal}{\emph{Int. J. Comput. Geom. Appl.}}
  \bibinfo{volume}{5} (\bibinfo{year}{1995}), \bibinfo{pages}{75--91}.
\newblock


\bibitem[\protect\citeauthoryear{Buchin, Buchin, Gudmundsson, L{\"o}ffler, and
  Luo}{Buchin et~al\mbox{.}}{2011}]%
        {buchin2011detecting}
\bibfield{author}{\bibinfo{person}{Kevin Buchin}, \bibinfo{person}{Maike
  Buchin}, \bibinfo{person}{Joachim Gudmundsson}, \bibinfo{person}{Maarten
  L{\"o}ffler}, {and} \bibinfo{person}{Jun Luo}.}
  \bibinfo{year}{2011}\natexlab{}.
\newblock \showarticletitle{Detecting commuting patterns by clustering
  subtrajectories}.
\newblock \bibinfo{journal}{\emph{International Journal of Computational
  Geometry \& Applications}} \bibinfo{volume}{21}, \bibinfo{number}{03}
  (\bibinfo{year}{2011}), \bibinfo{pages}{253--282}.
\newblock


\bibitem[\protect\citeauthoryear{Chen and Ng}{Chen and Ng}{2004}]%
        {ChenN04}
\bibfield{author}{\bibinfo{person}{Lei Chen} {and} \bibinfo{person}{Raymond~T.
  Ng}.} \bibinfo{year}{2004}\natexlab{}.
\newblock \showarticletitle{On The Marriage of Lp-norms and Edit Distance}. In
  \bibinfo{booktitle}{\emph{{VLDB}}}. \bibinfo{publisher}{Morgan Kaufmann},
  \bibinfo{pages}{792--803}.
\newblock


\bibitem[\protect\citeauthoryear{Chen, {\"{O}}zsu, and Oria}{Chen
  et~al\mbox{.}}{2005}]%
        {ChenOO05}
\bibfield{author}{\bibinfo{person}{Lei Chen}, \bibinfo{person}{M.~Tamer
  {\"{O}}zsu}, {and} \bibinfo{person}{Vincent Oria}.}
  \bibinfo{year}{2005}\natexlab{}.
\newblock \showarticletitle{Robust and Fast Similarity Search for Moving Object
  Trajectories}. In \bibinfo{booktitle}{\emph{{SIGMOD} Conference}}.
  \bibinfo{publisher}{{ACM}}, \bibinfo{pages}{491--502}.
\newblock


\bibitem[\protect\citeauthoryear{Chen, Shen, and Zhou}{Chen
  et~al\mbox{.}}{2011}]%
        {ChenSZ11}
\bibfield{author}{\bibinfo{person}{Zaiben Chen}, \bibinfo{person}{Heng~Tao
  Shen}, {and} \bibinfo{person}{Xiaofang Zhou}.}
  \bibinfo{year}{2011}\natexlab{}.
\newblock \showarticletitle{Discovering popular routes from trajectories}. In
  \bibinfo{booktitle}{\emph{{ICDE}}}. \bibinfo{publisher}{{IEEE} Computer
  Society}, \bibinfo{pages}{900--911}.
\newblock


\bibitem[\protect\citeauthoryear{Faloutsos, Ranganathan, and
  Manolopoulos}{Faloutsos et~al\mbox{.}}{1994}]%
        {faloutsos1994fast}
\bibfield{author}{\bibinfo{person}{Christos Faloutsos},
  \bibinfo{person}{Mudumbai Ranganathan}, {and} \bibinfo{person}{Yannis
  Manolopoulos}.} \bibinfo{year}{1994}\natexlab{}.
\newblock \showarticletitle{Fast subsequence matching in time-series
  databases}.
\newblock \bibinfo{journal}{\emph{Acm Sigmod Record}} \bibinfo{volume}{23},
  \bibinfo{number}{2} (\bibinfo{year}{1994}), \bibinfo{pages}{419--429}.
\newblock


\bibitem[\protect\citeauthoryear{Gudmundsson, Seybold, and Pfeifer}{Gudmundsson
  et~al\mbox{.}}{2021}]%
        {abs-2203-10364}
\bibfield{author}{\bibinfo{person}{Joachim Gudmundsson},
  \bibinfo{person}{Martin~P. Seybold}, {and} \bibinfo{person}{John Pfeifer}.}
  \bibinfo{year}{2021}\natexlab{}.
\newblock \showarticletitle{On Practical Nearest Sub-Trajectory Queries under
  the Fr{\'{e}}chet Distance}. In \bibinfo{booktitle}{\emph{{SIGSPATIAL/GIS}}}.
  \bibinfo{publisher}{{ACM}}, \bibinfo{pages}{596--605}.
\newblock


\bibitem[\protect\citeauthoryear{Hui, Yan, Chen, and Ku}{Hui
  et~al\mbox{.}}{2021a}]%
        {Hui0CK21}
\bibfield{author}{\bibinfo{person}{Bo Hui}, \bibinfo{person}{Da Yan},
  \bibinfo{person}{Haiquan Chen}, {and} \bibinfo{person}{Wei{-}Shinn Ku}.}
  \bibinfo{year}{2021}\natexlab{a}.
\newblock \showarticletitle{TrajNet: {A} Trajectory-Based Deep Learning Model
  for Traffic Prediction}. In \bibinfo{booktitle}{\emph{{KDD}}}.
  \bibinfo{publisher}{{ACM}}, \bibinfo{pages}{716--724}.
\newblock


\bibitem[\protect\citeauthoryear{Hui, Yan, Chen, and Ku}{Hui
  et~al\mbox{.}}{2021b}]%
        {Hui2021}
\bibfield{author}{\bibinfo{person}{Bo Hui}, \bibinfo{person}{Da Yan},
  \bibinfo{person}{Haiquan Chen}, {and} \bibinfo{person}{Wei-Shinn Ku}.}
  \bibinfo{year}{2021}\natexlab{b}.
\newblock \showarticletitle{Trajectory {WaveNet}: A Trajectory-Based Model for
  Traffic Forecasting}. In \bibinfo{booktitle}{\emph{2021 {IEEE} International
  Conference on Data Mining ({ICDM})}}. \bibinfo{publisher}{{IEEE}}.
\newblock
\urldef\tempurl%
\url{https://doi.org/10.1109/icdm51629.2021.00131}
\showDOI{\tempurl}


\bibitem[\protect\citeauthoryear{Koide, Tadokoro, Yoshimura, Xiao, and
  Ishikawa}{Koide et~al\mbox{.}}{2018}]%
        {KoideTYXI18}
\bibfield{author}{\bibinfo{person}{Satoshi Koide}, \bibinfo{person}{Yukihiro
  Tadokoro}, \bibinfo{person}{Takayoshi Yoshimura}, \bibinfo{person}{Chuan
  Xiao}, {and} \bibinfo{person}{Yoshiharu Ishikawa}.}
  \bibinfo{year}{2018}\natexlab{}.
\newblock \showarticletitle{Enhanced Indexing and Querying of Trajectories in
  Road Networks via String Algorithms}.
\newblock \bibinfo{journal}{\emph{{ACM} Trans. Spatial Algorithms Syst.}}
  \bibinfo{volume}{4}, \bibinfo{number}{1} (\bibinfo{year}{2018}),
  \bibinfo{pages}{3:1--3:41}.
\newblock


\bibitem[\protect\citeauthoryear{Koide, Xiao, and Ishikawa}{Koide
  et~al\mbox{.}}{2020}]%
        {KoideXI20}
\bibfield{author}{\bibinfo{person}{Satoshi Koide}, \bibinfo{person}{Chuan
  Xiao}, {and} \bibinfo{person}{Yoshiharu Ishikawa}.}
  \bibinfo{year}{2020}\natexlab{}.
\newblock \showarticletitle{Fast Subtrajectory Similarity Search in Road
  Networks under Weighted Edit Distance Constraints}.
\newblock \bibinfo{journal}{\emph{Proc. {VLDB} Endow.}} \bibinfo{volume}{13},
  \bibinfo{number}{11} (\bibinfo{year}{2020}), \bibinfo{pages}{2188--2201}.
\newblock


\bibitem[\protect\citeauthoryear{Lee, Han, and Whang}{Lee
  et~al\mbox{.}}{2007}]%
        {lee2007trajectory}
\bibfield{author}{\bibinfo{person}{Jae-Gil Lee}, \bibinfo{person}{Jiawei Han},
  {and} \bibinfo{person}{Kyu-Young Whang}.} \bibinfo{year}{2007}\natexlab{}.
\newblock \showarticletitle{Trajectory clustering: a partition-and-group
  framework}. In \bibinfo{booktitle}{\emph{Proceedings of the 2007 ACM SIGMOD
  international conference on Management of data}}. \bibinfo{pages}{593--604}.
\newblock


\bibitem[\protect\citeauthoryear{Li, Zhao, Cong, Jensen, and Wei}{Li
  et~al\mbox{.}}{2018}]%
        {li2018deep}
\bibfield{author}{\bibinfo{person}{Xiucheng Li}, \bibinfo{person}{Kaiqi Zhao},
  \bibinfo{person}{Gao Cong}, \bibinfo{person}{Christian~S Jensen}, {and}
  \bibinfo{person}{Wei Wei}.} \bibinfo{year}{2018}\natexlab{}.
\newblock \showarticletitle{Deep representation learning for trajectory
  similarity computation}. In \bibinfo{booktitle}{\emph{2018 IEEE 34th
  international conference on data engineering (ICDE)}}. IEEE,
  \bibinfo{pages}{617--628}.
\newblock


\bibitem[\protect\citeauthoryear{Luo, Tan, Chen, and Ni}{Luo
  et~al\mbox{.}}{2013}]%
        {LuoT0N13}
\bibfield{author}{\bibinfo{person}{Wuman Luo}, \bibinfo{person}{Haoyu Tan},
  \bibinfo{person}{Lei Chen}, {and} \bibinfo{person}{Lionel~M. Ni}.}
  \bibinfo{year}{2013}\natexlab{}.
\newblock \showarticletitle{Finding time period-based most frequent path in big
  trajectory data}. In \bibinfo{booktitle}{\emph{{SIGMOD} Conference}}.
  \bibinfo{publisher}{{ACM}}, \bibinfo{pages}{713--724}.
\newblock


\bibitem[\protect\citeauthoryear{online}{online}{2016a}]%
        {gaia}
\bibfield{author}{\bibinfo{person}{online}.} \bibinfo{year}{2016}\natexlab{a}.
\newblock \bibinfo{title}{GAIA Open Dataset}.
\newblock
  \bibinfo{howpublished}{\url{https://outreach.didichuxing.com/appEn-vue/ChengDuOct2016?id=7}}.
\newblock


\bibitem[\protect\citeauthoryear{online}{online}{2016b}]%
        {porto}
\bibfield{author}{\bibinfo{person}{online}.} \bibinfo{year}{2016}\natexlab{b}.
\newblock \bibinfo{title}{Porto Dataset}.
\newblock
  \bibinfo{howpublished}{\url{https://www.kaggle.com/c/pkdd-15-predict-taxi-service-trajectory-i/data}}.
\newblock


\bibitem[\protect\citeauthoryear{Ranu, Deepak, Telang, Deshpande, and
  Raghavan}{Ranu et~al\mbox{.}}{2015}]%
        {ranu2015indexing}
\bibfield{author}{\bibinfo{person}{Sayan Ranu}, \bibinfo{person}{Padmanabhan
  Deepak}, \bibinfo{person}{Aditya~D Telang}, \bibinfo{person}{Prasad
  Deshpande}, {and} \bibinfo{person}{Sriram Raghavan}.}
  \bibinfo{year}{2015}\natexlab{}.
\newblock \showarticletitle{Indexing and matching trajectories under
  inconsistent sampling rates}. In \bibinfo{booktitle}{\emph{2015 IEEE 31st
  International conference on data engineering}}. IEEE,
  \bibinfo{pages}{999--1010}.
\newblock


\bibitem[\protect\citeauthoryear{Sakurai, Faloutsos, and Yamamuro}{Sakurai
  et~al\mbox{.}}{2007}]%
        {SakuraiFY07}
\bibfield{author}{\bibinfo{person}{Yasushi Sakurai}, \bibinfo{person}{Christos
  Faloutsos}, {and} \bibinfo{person}{Masashi Yamamuro}.}
  \bibinfo{year}{2007}\natexlab{}.
\newblock \showarticletitle{Stream Monitoring under the Time Warping Distance}.
  In \bibinfo{booktitle}{\emph{{ICDE}}}. \bibinfo{publisher}{{IEEE} Computer
  Society}, \bibinfo{pages}{1046--1055}.
\newblock


\bibitem[\protect\citeauthoryear{Tampakis, Doulkeridis, Pelekis, and
  Theodoridis}{Tampakis et~al\mbox{.}}{2020}]%
        {tampakis2020distributed}
\bibfield{author}{\bibinfo{person}{Panagiotis Tampakis},
  \bibinfo{person}{Christos Doulkeridis}, \bibinfo{person}{Nikos Pelekis},
  {and} \bibinfo{person}{Yannis Theodoridis}.} \bibinfo{year}{2020}\natexlab{}.
\newblock \showarticletitle{Distributed subtrajectory join on massive
  datasets}.
\newblock \bibinfo{journal}{\emph{ACM Transactions on Spatial Algorithms and
  Systems (TSAS)}} \bibinfo{volume}{6}, \bibinfo{number}{2}
  (\bibinfo{year}{2020}), \bibinfo{pages}{1--29}.
\newblock


\bibitem[\protect\citeauthoryear{Vlachos, Gunopulos, and Kollios}{Vlachos
  et~al\mbox{.}}{2002}]%
        {VlachosGK02}
\bibfield{author}{\bibinfo{person}{Michail Vlachos}, \bibinfo{person}{Dimitrios
  Gunopulos}, {and} \bibinfo{person}{George Kollios}.}
  \bibinfo{year}{2002}\natexlab{}.
\newblock \showarticletitle{Discovering Similar Multidimensional Trajectories}.
  In \bibinfo{booktitle}{\emph{{ICDE}}}. \bibinfo{publisher}{{IEEE} Computer
  Society}, \bibinfo{pages}{673--684}.
\newblock


\bibitem[\protect\citeauthoryear{Wang, Cheng, Zheng, Feng, Chen, Lin, and
  Wang}{Wang et~al\mbox{.}}{2020a}]%
        {WangCZFCLW20}
\bibfield{author}{\bibinfo{person}{Jiachuan Wang}, \bibinfo{person}{Peng
  Cheng}, \bibinfo{person}{Libin Zheng}, \bibinfo{person}{Chao Feng},
  \bibinfo{person}{Lei Chen}, \bibinfo{person}{Xuemin Lin}, {and}
  \bibinfo{person}{Zheng Wang}.} \bibinfo{year}{2020}\natexlab{a}.
\newblock \showarticletitle{Demand-Aware Route Planning for Shared Mobility
  Services}.
\newblock \bibinfo{journal}{\emph{Proc. {VLDB} Endow.}} \bibinfo{volume}{13},
  \bibinfo{number}{7} (\bibinfo{year}{2020}), \bibinfo{pages}{979--991}.
\newblock


\bibitem[\protect\citeauthoryear{Wang, Bao, Culpepper, Xie, Liu, and Qin}{Wang
  et~al\mbox{.}}{2018}]%
        {0007BCXLQ18}
\bibfield{author}{\bibinfo{person}{Sheng Wang}, \bibinfo{person}{Zhifeng Bao},
  \bibinfo{person}{J.~Shane Culpepper}, \bibinfo{person}{Zizhe Xie},
  \bibinfo{person}{Qizhi Liu}, {and} \bibinfo{person}{Xiaolin Qin}.}
  \bibinfo{year}{2018}\natexlab{}.
\newblock \showarticletitle{Torch: {A} Search Engine for Trajectory Data}. In
  \bibinfo{booktitle}{\emph{{SIGIR}}}. \bibinfo{publisher}{{ACM}},
  \bibinfo{pages}{535--544}.
\newblock


\bibitem[\protect\citeauthoryear{Wang, Zheng, and Xue}{Wang
  et~al\mbox{.}}{2014}]%
        {WangZX14}
\bibfield{author}{\bibinfo{person}{Yilun Wang}, \bibinfo{person}{Yu Zheng},
  {and} \bibinfo{person}{Yexiang Xue}.} \bibinfo{year}{2014}\natexlab{}.
\newblock \showarticletitle{Travel time estimation of a path using sparse
  trajectories}. In \bibinfo{booktitle}{\emph{{KDD}}}.
  \bibinfo{publisher}{{ACM}}, \bibinfo{pages}{25--34}.
\newblock


\bibitem[\protect\citeauthoryear{Wang, Long, Cong, and Ju}{Wang
  et~al\mbox{.}}{2019}]%
        {WangLCJ19}
\bibfield{author}{\bibinfo{person}{Zheng Wang}, \bibinfo{person}{Cheng Long},
  \bibinfo{person}{Gao Cong}, {and} \bibinfo{person}{Ce Ju}.}
  \bibinfo{year}{2019}\natexlab{}.
\newblock \showarticletitle{Effective and Efficient Sports Play Retrieval with
  Deep Representation Learning}. In \bibinfo{booktitle}{\emph{{KDD}}}.
  \bibinfo{publisher}{{ACM}}, \bibinfo{pages}{499--509}.
\newblock


\bibitem[\protect\citeauthoryear{Wang, Long, Cong, and Liu}{Wang
  et~al\mbox{.}}{2020b}]%
        {WangLCL20}
\bibfield{author}{\bibinfo{person}{Zheng Wang}, \bibinfo{person}{Cheng Long},
  \bibinfo{person}{Gao Cong}, {and} \bibinfo{person}{Yiding Liu}.}
  \bibinfo{year}{2020}\natexlab{b}.
\newblock \showarticletitle{Efficient and Effective Similar Subtrajectory
  Search with Deep Reinforcement Learning}.
\newblock \bibinfo{journal}{\emph{Proc. {VLDB} Endow.}} \bibinfo{volume}{13},
  \bibinfo{number}{11} (\bibinfo{year}{2020}), \bibinfo{pages}{2312--2325}.
\newblock


\bibitem[\protect\citeauthoryear{Waury, Jensen, Koide, Ishikawa, and
  Xiao}{Waury et~al\mbox{.}}{2019}]%
        {WauryJKIX19}
\bibfield{author}{\bibinfo{person}{Robert Waury}, \bibinfo{person}{Christian~S.
  Jensen}, \bibinfo{person}{Satoshi Koide}, \bibinfo{person}{Yoshiharu
  Ishikawa}, {and} \bibinfo{person}{Chuan Xiao}.}
  \bibinfo{year}{2019}\natexlab{}.
\newblock \showarticletitle{Indexing Trajectories for Travel-Time Histogram
  Retrieval}. In \bibinfo{booktitle}{\emph{{EDBT}}}.
  \bibinfo{publisher}{OpenProceedings.org}, \bibinfo{pages}{157--168}.
\newblock


\bibitem[\protect\citeauthoryear{Xie}{Xie}{2014}]%
        {Xie14}
\bibfield{author}{\bibinfo{person}{Min Xie}.} \bibinfo{year}{2014}\natexlab{}.
\newblock \showarticletitle{{EDS:} a segment-based distance measure for
  sub-trajectory similarity search}. In \bibinfo{booktitle}{\emph{{SIGMOD}
  Conference}}. \bibinfo{publisher}{{ACM}}, \bibinfo{pages}{1609--1610}.
\newblock


\bibitem[\protect\citeauthoryear{Yi, Jagadish, and Faloutsos}{Yi
  et~al\mbox{.}}{1998}]%
        {YiJF98}
\bibfield{author}{\bibinfo{person}{Byoung{-}Kee Yi}, \bibinfo{person}{H.~V.
  Jagadish}, {and} \bibinfo{person}{Christos Faloutsos}.}
  \bibinfo{year}{1998}\natexlab{}.
\newblock \showarticletitle{Efficient Retrieval of Similar Time Sequences Under
  Time Warping}. In \bibinfo{booktitle}{\emph{{ICDE}}}.
  \bibinfo{publisher}{{IEEE} Computer Society}, \bibinfo{pages}{201--208}.
\newblock


\bibitem[\protect\citeauthoryear{Yuan and Li}{Yuan and Li}{2019}]%
        {Yuan019}
\bibfield{author}{\bibinfo{person}{Haitao Yuan} {and} \bibinfo{person}{Guoliang
  Li}.} \bibinfo{year}{2019}\natexlab{}.
\newblock \showarticletitle{Distributed In-memory Trajectory Similarity Search
  and Join on Road Network}. In \bibinfo{booktitle}{\emph{{ICDE}}}.
  \bibinfo{publisher}{{IEEE}}, \bibinfo{pages}{1262--1273}.
\newblock


\bibitem[\protect\citeauthoryear{Yuan, Zheng, Xie, and Sun}{Yuan
  et~al\mbox{.}}{2011}]%
        {yuan2011driving}
\bibfield{author}{\bibinfo{person}{Jing Yuan}, \bibinfo{person}{Yu Zheng},
  \bibinfo{person}{Xing Xie}, {and} \bibinfo{person}{Guangzhong Sun}.}
  \bibinfo{year}{2011}\natexlab{}.
\newblock \showarticletitle{Driving with knowledge from the physical world}. In
  \bibinfo{booktitle}{\emph{Proceedings of the 17th ACM SIGKDD international
  conference on Knowledge discovery and data mining}}.
  \bibinfo{pages}{316--324}.
\newblock


\bibitem[\protect\citeauthoryear{Yuan, Zheng, Zhang, Xie, Xie, Sun, and
  Huang}{Yuan et~al\mbox{.}}{2010}]%
        {yuan2010t}
\bibfield{author}{\bibinfo{person}{Jing Yuan}, \bibinfo{person}{Yu Zheng},
  \bibinfo{person}{Chengyang Zhang}, \bibinfo{person}{Wenlei Xie},
  \bibinfo{person}{Xing Xie}, \bibinfo{person}{Guangzhong Sun}, {and}
  \bibinfo{person}{Yan Huang}.} \bibinfo{year}{2010}\natexlab{}.
\newblock \showarticletitle{T-drive: driving directions based on taxi
  trajectories}. In \bibinfo{booktitle}{\emph{Proceedings of the 18th
  SIGSPATIAL International conference on advances in geographic information
  systems}}. \bibinfo{pages}{99--108}.
\newblock


\end{thebibliography}
